\documentclass[reqno]{amsart}
\usepackage[lite]{amsrefs}
\usepackage{mathtools,amsmath,amssymb}
\usepackage[pdfencoding=auto, psdextra]{hyperref}
\hypersetup{colorlinks=true,allcolors=[rgb]{0,0,0.6}}
\usepackage[hmargin= 1.3 in,vmargin=1 in]{geometry}
\usepackage{bbm,mathrsfs}
\usepackage{xcolor}
\usepackage[all,cmtip]{xy}
\usepackage{appendix}
\usepackage{mathrsfs}
\DeclareFontFamily{U}{mathx}{}
\DeclareFontShape{U}{mathx}{m}{n}{<-> mathx10}{}
\DeclareSymbolFont{mathx}{U}{mathx}{m}{n}
\DeclareMathAccent{\widehat}{0}{mathx}{"70}
\DeclareMathAccent{\widecheck}{0}{mathx}{"71}

\newcommand{\R}{\mathbb{R}}
\newcommand{\Z}{\mathbb{Z}}
\newcommand{\N}{\mathbb{N}}

\newcommand{\ii}{\mathrm{i}}

\newcommand{\FF}{\mathrm{F}}
\newcommand{\BB}{\mathrm{B}}
\newcommand{\Bolt}{\mathrm{Bolt}}

\usepackage{dsfont}

\numberwithin{equation}{section}

\theoremstyle{plain} 
\newtheorem{theorem}{Theorem}[section]
\newtheorem*{theorem*}{Theorem}
\newtheorem{lemma}[theorem]{Lemma}
\newtheorem*{lemma*}{Lemma}
\newtheorem{corollary}[theorem]{Corollary}
\newtheorem*{corollary*}{Corollary}
\newtheorem{proposition}[theorem]{Proposition}
\newtheorem*{proposition*}{Proposition}

\newtheorem*{conjecture*}{Conjecture}

\theoremstyle{definition} 
\newtheorem{definition}[theorem]{Definition}
\newtheorem*{definition*}{Definition}

\newtheorem*{example*}{Example}
\newtheorem{remark}[theorem]{Remark}
\newtheorem*{remark*}{Remark}
\newtheorem{assumption}[theorem]{Assumption}
\newtheorem*{assumption*}{Assumption}

\newcommand{\ee}{\mathrm{e}}
\newcommand{\dd}{\mathrm{d}}

\usepackage{bm}
\let\vec \bm 

\DeclareMathOperator{\sgn}{sgn}
\DeclareMathOperator{\Tr}{Tr}
\DeclareMathOperator{\id}{id}

\title[large-mass limit of quantum gases in the continuum]{The large-mass limit of interacting quantum gases in the continuum}
\author{Spyros Garouniatis}
\address{Department of Mathematics, Brandeis University, 415 South Street, Waltham, MA, 02453, United States of America.}
\email{garouniatis@brandeis.edu}
\author{Grega Saksida}
\address{University of Warwick, Mathematics Institute, Zeeman Building, Coventry CV4 7AL, United Kingdom.}
\email{Grega.Saksida@warwick.ac.uk}
\author{Vedran Sohinger}
\address{University of Warwick, Mathematics Institute, Zeeman Building, Coventry CV4 7AL, United Kingdom.}
\email{V.Sohinger@warwick.ac.uk.}

\begin{document}

\maketitle

\begin{abstract}
We study the \emph{large-mass} limit of interacting quantum (Bose or Fermi) gases in thermal equilibrium. 
We show that in the suitably-defined large-mass limit, the system gives rise to a gas of classical interacting particles.
The corresponding question for bosons on a lattice was previously addressed by Fr\"{o}hlich, Knowles, Schlein, and the third author in \cite{FKSS_2023}. In this work, we study the continuum regime which requires us to suitably tune the chemical potential. The starting point of our analysis is the \emph{Ginibre loop ensemble} \cite{Ginibre1,Ginibre2,Ginibre3,Ginibre}, which allows one to describe a system of interacting quantum gases in thermal equilibrium in terms of an ensemble of interacting Brownian paths. In a finite volume, our analysis is performed for stable and H\"{o}lder continuous interaction potentials and we are able to obtain explicit rates of convergence. When the interaction potential is nonnegative and satisfies suitable integrability conditions, we study the associated infinite-volume problem by means of cluster expansions.   
\end{abstract}

\section{Introduction}

\subsection{Setup of the problem}
For $d \in \N^*$ and $L \geq 1$, we consider the spatial domain 
\begin{equation*}
\Lambda \equiv \Lambda_{L,d}=[-L/2,L/2]^d=\R^d/L \Z^d\,,
\end{equation*}
which is the  $d$-dimensional torus of sidelength $L$. A system of $n$ spinless bosons or fermions of mass $m>0$ confined to $\Lambda$ is governed by the \emph{$n$-body Hamiltonian}
\begin{equation}
\label{Hamiltonian_H}
\mathbb{H}_n\coloneqq -\sum_{i=1}^{n} \frac{\Delta_i}{2m}+\lambda \sum_{1 \leq i < j \leq n} v(x_i-x_j)\,.
\end{equation}
In \eqref{Hamiltonian_H}, $\Delta_i$ is the Laplacian acting in the variable $x_i$ with periodic boundary conditions, $v \colon \Lambda \rightarrow \R$ is a two-body interaction potential and $\lambda \geq 0$ is a coupling constant.

When working with bosons, the Hamiltonian given in \eqref{Hamiltonian_H} acts on the $n$-particle bosonic Hilbert space, which we denote by 
\begin{equation}
\label{n_particle_space}
\mathcal{H}_n \equiv \mathcal{H}_n^{\BB} \coloneqq P_{n}^{+} L^2(\Lambda^n)\,. 
\end{equation}
In \eqref{n_particle_space}, we write
\begin{equation}
\label{symmetric_projector}
P_{n}^{+}f(x_{1},\ldots,x_{n})\coloneqq \frac{1}{n!}\, \sum_{\pi \in {S_n}} f(x_{\pi(1)}, \ldots, x_{\pi(n)})\,,
\end{equation}
where $S_n$ is the set of permutations of $\{1,\ldots,n\}$. 

When working with fermions, the Hamiltonian in \eqref{Hamiltonian_H} acts on the $n$-particle fermionic Hilbert space, which we denote by
\begin{equation}
\label{n_particle_space_Fermi}
\mathcal{H}_n \equiv \mathcal{H}_n^{\FF} \coloneqq P_{n}^{-} L^2(\Lambda^n)\,,
\end{equation}
where
\begin{equation}
\label{antisymmetric_projector}
P_{n}^{-}f(x_{1},\ldots,x_{n})\coloneqq \frac{1}{n!}\, \sum_{\pi \in {S_n}} \sgn(\pi) f(x_{\pi(1)}, \ldots, x_{\pi(n)})\,.
\end{equation}

We study the interacting quantum gases in the \textbf{\emph{grand-canonical ensemble}} at positive temperature. Mathematically, the latter is a sequence of operators $(\rho_n)_{n \in \mathbb{N}}$ of the form
\begin{equation}
\label{grand_canonical_ensemble_1}
\rho_n\coloneqq \frac{1}{\Xi} \: \ee^{-\beta(\mathbb{H}_{n}-\mu n)}\,,\qquad \Xi\coloneqq 1+\sum_{n=1}^{\infty} \mathrm{Tr}_{\mathcal{H}_n} \ee^{-\beta(\mathbb{H}_{n}-\mu n)}\,.
\end{equation}
In \eqref{grand_canonical_ensemble_1}, $\beta>0$ denotes the inverse temperature. Furthermore, the quantity $\mu<0$ denotes the chemical potential. The quantity $\Xi$ is the \textbf{\emph{grand-canonical partition function}}. 
For the remainder of the paper, we set 
\begin{equation*}
\beta=1\,.
\end{equation*}
The case of general $\beta>0$ can be obtained from this case by a suitable rescaling of the parameters.

\subsubsection{Notation and conventions} 
\label{Notation and conventions}
Throughout the paper, $\N^*=\{1,2,3,\ldots\}$ denotes the set of positive integers and $\N=\{0,1,2, \ldots\}$ denotes the set of nonnegative integers. Given $A,B>0$ and a finite list of parameters $a_1,\ldots,a_k$, we say that 
\begin{equation*}
A \lesssim_{a_1,\ldots,a_k} B
\end{equation*}
if $A \leq C B$ for some quantity $C>0$ depending on $a_1,\ldots,a_k$. Sometimes, we also write $C=C(a_1,\ldots,a_k)$. Alternatively, we also write $A=\mathcal{O}_{a_1,\ldots,a_k}(B)$. We say that
\begin{equation*}
A \sim_{a_1,\ldots,a_k} B
\end{equation*}
if $A \lesssim_{a_1,\ldots,a_k} B$ and $B \lesssim_{a_1,\ldots,a_k} A.$ For $x \in \R^d$, we denote by $\langle x \rangle  \equiv \sqrt{1+|x|^2}$ the Japanese bracket.

\subsection{Statement of our results} \label{section_choice_of_parameters}
For clarity of exposition, in the introduction, we explain the setup and summarise our results for bosons. The analysis for fermions is done similarly; the precise modifications and statements are given in Section \ref{Fermi_gas} below. 
Following \cite{Ginibre}, we also consider the regime of Boltzmann statistics; the analysis here is summarised in Remarks \ref{Remark_Boltzmann_gas_partition_function}, \ref{Remark_Boltzmann_density_matrix}, and \ref{fermions_infinite_volume} below. We emphasise that our main analysis concerns the regime of Bose statistics and that the analysis for Fermi and Boltzmann statistics follows by similar arguments.

\subsubsection{The finite-volume (bosonic) regime}
Let us first explain the precise choice of parameters.
In \eqref{Hamiltonian_H}, we set 
\begin{equation}
\label{nu}
\nu\coloneqq \frac{1}{m}\,, \qquad \lambda=\lambda_0>0\,.
\end{equation}
The regime that we are interested in analysing is $\nu \rightarrow 0$ for suitably chosen $\mu$ as a function of $\nu$ (see Assumption \ref{Choice_of_z} below). By \eqref{nu}, this corresponds to taking the mass going to infinity. 
Throughout the paper, we consider $\nu>0$ small.
For simplicity of notation, we henceforth consider the case
\begin{equation}
\label{lambda_0=1}
\lambda_0=1\,.
\end{equation}
The case of general $\lambda_0$ follows by replacing $v$ with $\lambda_0 v$.

Let us now explain the precise choice of the chemical potential $\mu$ as a function of $\nu$.

\begin{assumption}
\label{Choice_of_z}
Let $\zeta>0$ be given. We choose $\mu \equiv \mu(\nu,\zeta)$ such that 
\begin{equation*}
\ee^{\mu({\nu,\zeta})}\,\nu^{-d/2}=\zeta\,,
\end{equation*}
i.e.\
\begin{equation}
\label{mu_choice}
\mu(\nu,\zeta)=\log \zeta+\frac{d}{2}\log \nu\,.
\end{equation}
For the remainder of the paper, we implicitly assume that $\nu>0$ is chosen sufficiently small so that the chemical potential~\eqref{mu_choice} is negative, i.e.\
\begin{equation}
\label{nu_0}
\nu \in (0,\nu_0]\,,
\end{equation}
for a suitable $\nu_0 \equiv \nu_0(d,\zeta) \in (0,1]$.
\end{assumption}

We henceforth fix $\zeta>0$ in Assumption \ref{Choice_of_z}. With $\mu$ chosen as in \eqref{mu_choice}, the parameter
\begin{equation}
\label{fugacity}
z\equiv z({\nu,\zeta})\coloneqq \ee^{\frac{\mu({\nu,\zeta})}{\nu}}
\end{equation}
then satisfies
\begin{equation}
\label{fugacity_condition}
z^{\nu}\,\nu^{-d/2}=\zeta\,.
\end{equation}
In the rest of the paper, we suppress the $\nu$ and $\zeta$ dependence of $z$ for simplicity of notation.
With the above choice of parameters, we rewrite the $n$-body Hamiltonian \eqref{Hamiltonian_H} as 
\begin{equation}
\label{n_body_Hamiltonian}
H_n^{\nu}\coloneqq -\frac{\nu}{2} \sum_{i=1}^{n} \Delta_i +\sum_{1 \leq i < j \leq n} v(x_i-x_j)
\end{equation}
and the grand-canonical ensemble \eqref{grand_canonical_ensemble_1} as\footnote{In \eqref{n_body_Hamiltonian}--\eqref{reduced_p_particle_density_matrix}, and throughout the paper, we adopt the convention that the $\nu,\zeta$ superscripts in operators and functions such as $H_n^{\nu}$, $\rho_n^{\nu,\zeta}$, $\Xi^{\nu,\zeta}$, $\Xi^{\nu,\zeta}_{(0)}$, $\mathcal{Z}^{\nu,\zeta}$, $\Gamma_p^{\nu,\zeta},\ldots$  are parameters (and not powers), whereas expressions such as $z^{n\nu}$ are always powers of $z$.}
\begin{equation}
\label{grand_canonical_ensemble}
\rho_n \equiv \rho_n^{\nu,\zeta}\coloneqq \frac{1}{\Xi}\, \ee^{-H_{n}^{\nu}}\,z^{n \nu} \,, \qquad \Xi \equiv \Xi^{\nu,\zeta}\coloneqq 1 + \sum_{n=1}^{\infty} \mathrm{Tr}_{\mathcal{H}_n} \bigl(\ee^{-H_{n}^{\nu}}\,z^{n \nu} \bigr) \,.
\end{equation}
Note that \eqref{n_body_Hamiltonian} is a densely-defined operator on \eqref{n_particle_space}.
Furthermore, we  define the \textbf{\emph{free grand-canonical partition function}} as
\begin{equation}
\label{free_grand_canonical_partition_function}
\Xi_{(0)} \equiv \Xi^{\nu,\zeta}_{(0)} \coloneqq \Xi^{\nu,\zeta} \big|_{v=0}\,,
\end{equation}
i.e.\ we set the interaction potential equal to zero in \eqref{grand_canonical_ensemble} (the latter in turn implicitly depends on $v$ through \eqref{n_body_Hamiltonian}).
Using \eqref{grand_canonical_ensemble}--\eqref{free_grand_canonical_partition_function}, the \emph{quantum relative partition function} is defined by
\begin{equation}
\label{quantum_relative_partition_function}
\mathcal{Z} \equiv \mathcal{Z}^{\nu,\zeta}\coloneqq \frac{\Xi^{\nu,\zeta}}{\Xi^{\nu,\zeta}_{(0)}}\,.
\end{equation} 
For a given $p \in \N^*$, recalling \eqref{grand_canonical_ensemble}, the \textbf{\emph{reduced $p$-particle density matrix}} of the grand-canonical ensemble is defined as
\begin{equation}
\label{reduced_p_particle_density_matrix}
    \Gamma_p^{\nu,\zeta}\coloneqq \sum_{n=0}^{\infty} \frac{(n+p)!}{n!}\, \mathrm{Tr}_{p+1,\ldots,n+p}(\rho^{\nu,\zeta}_{n+p})\,,
\end{equation}
where $\mathrm{Tr}_{p+1,\ldots,n+p}$ denotes the partial trace over the coordinates $x_{p+1},\ldots, x_{n+p}$, and $\rho^{\nu,\zeta}_{n+p}$ is defined as in \eqref{grand_canonical_ensemble}. Moreover, for fixed $\vec{x},\vec{y} \in \Lambda^p$, we denote by $\Gamma_p^{\nu,\zeta}(\vec{x},\vec{y})$ the operator kernel of $\Gamma_p^{\nu,\zeta}$. 

In order to precisely state our results, we make the following assumptions about the interaction potential. Throughout the paper, $|\cdot|_{L}$ denotes the periodic Euclidean norm on $\Lambda \equiv \Lambda_L$ which is given by 
\begin{equation}
\label{periodic_Euclidean_norm}
|u|_{L}\coloneqq \mathop{\mathrm{min}}_{m \in \Z^d} |u-Lm|\,,
\end{equation}
where $|\cdot|$ on the right-hand side of \eqref{periodic_Euclidean_norm} is the Euclidean norm on $\R^d$.

\begin{assumption}
\label{Assumption_on_v}
Suppose that $v \colon \Lambda \rightarrow \R$ is even and that it satisfies the following conditions.
\begin{itemize}
\item[(i)] $v$ is a \emph{stable potential}, i.e.\ there exists $B\geq 0$ such that for every $n \in \N^*$ and every $x_1,\ldots,x_n \in \Lambda$, we have
\begin{equation}
\label{stable_potential_condition}
    \sum_{1\leq i<j\leq n} v(x_i-x_j)\geq -Bn\,.
\end{equation}
\item[(ii)] $v$ is $\alpha$-H\"{o}lder continuous for some fixed $\alpha \in (0,1]$, i.e.\  there exists $C>0$ such that for all $x,y \in \Lambda$, we have
\begin{equation*}
|v(x)-v(y)| \leq C |x-y|_{L}^{\alpha}\,.
\end{equation*}
\end{itemize}
\end{assumption}

We can now state our result on the partition function.

\begin{theorem}
\label{partition_function_result}
Let $\mu \equiv \mu (\nu,\zeta)$ be as in Assumption \ref{Choice_of_z} and let $z \equiv z(\nu,\zeta)$ be as in \eqref{fugacity}. Let $v$ be as in Assumption \ref{Assumption_on_v}. Then the grand-canonical partition function $\Xi^{\nu,\zeta}$  defined in \eqref{grand_canonical_ensemble} satisfies
\begin{equation}
\label{partition_function_result_1}
\lim_{\nu \rightarrow 0} \Xi^{\nu,\zeta}=\Xi^{0,\zeta}\,, 
\end{equation}
where in \eqref{partition_function_result_1}, 
\begin{equation}
\label{partition_function_result_2}
\Xi^{0,\zeta}\coloneqq \sum_{n=0}^{\infty} \frac{1}{n!}\, \int_{\Lambda^n} \dd x_1\,\cdots\, \dd x_n\, (2\pi)^{-\frac{nd}{2}}\, \zeta^n\,\mathrm{exp} \Biggl(
- \sum_{1 \leq i<j \leq n} v(x_i-x_j)\Biggr)\
\end{equation}
is the grand-canonical partition function of a gas of classical interacting particles.
Moreover, we can estimate the rate of convergence in \eqref{partition_function_result_1} as
\begin{equation}
\label{partition_function_result_3}
\bigl|\Xi^{\nu,\zeta}- \Xi^{0,\zeta} \bigr| \lesssim_{d,\zeta,v,\alpha,B} \nu^{\alpha/2}\,\exp \left [C_0(d,\zeta,B) L^d \right],
\end{equation}
for a suitable $C_0(d,\zeta,B)>0$.
\end{theorem}

From Theorem \ref{partition_function_result}, we immediately deduce the following result by considering quotients.
\begin{corollary}
\label{relative_partition_function_convergence}
With the same assumptions as in Theorem \ref{partition_function_result} and with the quantum relative partition function defined as in \eqref{quantum_relative_partition_function}, we have
\begin{equation*}
\lim_{\nu \rightarrow 0} \mathcal{Z}^{\nu,\zeta}=\mathcal{Z}^{0,\zeta}\,, 
\end{equation*}
where
\begin{equation}
\label{relative_partition_function_convergence_1}
\mathcal{Z}^{0,\zeta}\coloneqq \frac{\Xi^{0,\zeta}}{\Xi^{0,\zeta}_{(0)}}\,,\qquad \Xi^{0,\zeta}_{(0)}\coloneqq \Xi^{0,\zeta}\bigl|_{v=0}\,.
\end{equation}
In particular, following~\eqref{partition_function_result_3} we have
\begin{equation}
	\left| \mathcal{Z}^{\nu, \zeta} - \mathcal{Z}^{0, \zeta} \right| \lesssim_{d, \zeta, v, \alpha,B,L} \nu^{\alpha / 2}\,.
\end{equation}
\end{corollary}

In addition, we have the analogous result for reduced $p$-particle density matrices. Let us first introduce the following notation. Let $n,p \in \N^*$ and $\widetilde{\vec{x}} \in \Lambda^n, \vec{x}\in \Lambda^p$ be given. We denote by $\widetilde{\vec{x}}\vec{x}\in \Lambda^{n+p}$ the concatenation of $\widetilde{\vec{{x}}}$ and $\vec{x}$, i.e.\
\begin{equation}
\label{concatenation}
\widetilde{\vec{x}}\vec{x}\coloneqq (\widetilde{x}_1,\ldots,\widetilde{x}_n,x_1,\ldots , x_p)\,.
\end{equation}
Furthermore, for $v$ as in Assumption \ref{Assumption_on_v},  we define $\mathcal{V} \colon \Lambda^n\to [0,\infty)$ as 
\begin{equation}
\label{interaction_function_points}
\mathcal{V}(\vec{x})\coloneqq \sum_{1 \leq i<j \leq n} v(x_i-x_j)\,.    
\end{equation}
From now onwards, we use the convention
\begin{equation}
\label{pi_vec_y}
\pi \vec{y}\coloneqq (y_{\pi(1)},\ldots, y_{\pi(p)})
\end{equation}
for $p \in \N^*, \vec{y} \in \Lambda^p$ and $\pi \in S_p$.
\begin{theorem}
\label{density_matrix_result}
Let $\mu \equiv \mu (\nu,\zeta)$ be as in Assumption \ref{Choice_of_z} and let $z \equiv z(\nu,\zeta)$ be as in \eqref{fugacity}. Let $v$ be as in Assumption \ref{Assumption_on_v}.  Let $p \in \N^*$ be fixed. The kernel of the reduced $p$-particle density matrix \eqref{reduced_p_particle_density_matrix} satisfies 
    \begin{equation}
    \label{density_matrix_result_1}
        \lim_{\nu \to 0} \Gamma_p^{\nu,\zeta}(\vec{x},\vec{y}) = \Gamma_p^{0,\zeta}(\vec{x},\vec{y})\,,
    \end{equation}
    for all $\vec{x}, \vec{y} \in \Lambda^p$ where in \eqref{density_matrix_result_1},
    \begin{equation}
    \label{density_matrix_result_2}
        \Gamma_p^{0,\zeta}(\vec{x},\vec{y})\coloneqq (2\pi)^{-\frac{pd}{2}}\, \zeta^p \sum_{\pi \in S_p} \delta_K(\vec{x}-\pi\vec{y})\,\frac{\Xi^{0,\zeta}(\vec{x})}{\Xi^{0,\zeta}}\,.
    \end{equation}
In \eqref{density_matrix_result_2}, 
\begin{equation}
\label{Kronecker_delta}
\delta_K(\vec{v})=
\begin{cases}
1 &\text{if } \vec{v}=\vec{0} \in \Lambda^p\\
0 &\text{otherwise}
\end{cases}
\end{equation}
denotes the Kronecker delta function on $\Lambda^p$.
Furthermore, 
 $\Xi^{0,\zeta}$ is given by \eqref{partition_function_result_2}, and we define
     \begin{equation}\label{density_matrix_result_3}
         \Xi^{0,\zeta}(\vec{x})\coloneqq \sum_{n=0}^{\infty} \frac{(2\pi)^{-\frac{nd}{2}} \zeta^n}{n!}\, \int_{\Lambda^n} \dd \widetilde{x}_1\,\cdots\, \dd \widetilde{x}_n\,  \exp\left(-\mathcal{V}(\widetilde{\vec{x}}\vec{x})\right),
     \end{equation}
where we recall \textup{\eqref{concatenation}--\eqref{interaction_function_points}}.
Moreover, by denoting
\begin{equation}
\label{Pi_{x,y}}
\Pi_{\vec{x},\vec{y}}\coloneqq \left\{\pi \in S_p,\, \vec{x}=\pi \vec{y}\right\},
\end{equation}
we can bound the rate of convergence in \eqref{density_matrix_result_1} by
\begin{multline}
\label{Theorem_1.5(ii)_rate_of_convergence}
\bigl|\Gamma_p^{\nu,\zeta}(\vec{x},\vec{y})-\Gamma_p^{0,\zeta}(\vec{x},\vec{y})\bigr| 
\\
\lesssim_{p,d,\zeta,v,\alpha,B,L} \Biggl[\nu^{\alpha/2}+\sum_{\tilde{\pi} \in S_p \setminus \Pi_{\vec{x},\vec{y}}} \prod_{i=1}^{p} \Biggl\{\nu^{d/2}+ 
\exp \biggl(-\frac{|y_{\tilde{\pi}(i)}-x_i|_{L}^2}{2\nu}\biggr)\Biggr\}\Biggr]\,.
\end{multline}
\end{theorem}

\begin{remark}
\label{Theorem_1.5(ii)_rate_of_convergence_remark}
We note that by \eqref{Pi_{x,y}}, the right-hand side of \eqref{Theorem_1.5(ii)_rate_of_convergence} indeed converges to zero as $\nu \rightarrow 0$ since for all $\tilde{\pi} \in S_p \setminus \Pi_{\vec{x},\vec{y}}$, we have that $y_{\tilde{\pi}(i)} \neq x_i$ for some $1 \leq i \leq p$.
\end{remark}

\begin{remark}
\label{tuning_the_particle_density}
When studying the mean-field limit in the continuum in dimensions $d>1$, it is necessary to tune the particle density by means of a renormalisation procedure; see \cite{FKSS_2020_1}. This is not necessary in the context of the large-mass limit (see Theorems \ref{partition_function_result} and \ref{density_matrix_result} above), but it can be done in principle. We explain how tuning the particle density can be implemented in the context of the large-mass limit in Appendix \ref{tuning_particle_density} below. Here, the method is explained in the context of the relative partition function, i.e.\ in the convergence given by Corollary \ref{relative_partition_function_convergence}. We omit the similar analysis for reduced $p$-particle density matrices.
\end{remark}

\begin{remark}
\label{Fermi_gas_remark}
We prove analogues of Theorems \ref{partition_function_result} and \ref{density_matrix_result} in the fermionic regime by using similar methods. These are Theorems \ref{partition_function_result_Fermi} and \ref{density_matrix_result_Fermi} in Section \ref{Fermi_gas} below. The suitable modifications are explained in Section \ref{Fermi_gas}. See also Remarks \ref{Remark_Boltzmann_gas_partition_function} and \ref{Remark_Boltzmann_density_matrix} for the regime of Boltzmann statistics.
\end{remark}

\subsubsection{The infinite-volume (bosonic) regime}
\label{The_infinite volume_(bosonic)_regime}

We now study the infinite-volume regime and prove a large-mass limit for the reduced density matrices in Theorem \ref{Infinite_volume_Theorem_2}. We also prove a large-mass limit for the free energy in Theorem \ref{free_energy_infinite_volume}.
In the analysis, we explicitly include the sidelength $L$ of the box $\Lambda \equiv \Lambda_L$ in the notation.
In all of the other relevant quantities, we \emph{indicate the $L$-dependence with a superscript}.
Our choice of interaction potential will be slightly different than in a finite volume; see Assumption \ref{Assumption_on_v_2} below. 
In particular, we assume that the interactions are pointwise nonnegative. 
Theorems \ref{Infinite_volume_Theorem_2} and \ref{free_energy_infinite_volume} are proved under the assumption of a small interaction potential or a large chemical potential; this is explicitly stated in \eqref{smallness_v_kappa} below.

When studying the infinite-volume limit, we work with an interaction potential constructed as follows.
\begin{assumption}
\label{Assumption_on_v_2}
Suppose that $v \colon \R^d \rightarrow \R$ is even and that it satisfies the following conditions.
\begin{itemize}
\item[(i)] $v$ is differentiable.
\item[(ii)] $v \geq 0$ pointwise.
\item[(iii)] There exists $\varepsilon_0>0$ and $C>0$ such that for all $x \in \R^d$ we have
\begin{equation}
\label{v_assumption_iii}
|v(x)|+ \|\nabla v(x)\| \leq \frac{C}{\langle x \rangle^{d+\varepsilon_0}}\,.
\end{equation} 
\end{itemize}
\end{assumption}
With $v$ given as in Assumption \ref{Assumption_on_v_2}, on $\Lambda_L$ we consider the interaction potential $v^L$ given by
\begin{equation}
\label{v^L}
v^L(x)\coloneqq \sum_{k \in \Z^d} v(x+Lk)\,.
\end{equation}
We note that \eqref{v^L} defines a function on $\Lambda_L$ which satisfies Assumption \ref{Assumption_on_v} with $\alpha=1$; see Lemma \ref{v^L_lemma} (i) below.

\begin{remark}
\label{class_of_v}
The class of $v$ satisfying Assumption \ref{Assumption_on_v_2} contains the even and nonnegative functions in the Schwarz class $\mathscr{S}(\R^d)$.
\end{remark}

In the infinite-volume setting, we assume that the quantity $\zeta$ in Assumption \ref{Choice_of_z} satisfies
\begin{equation}
\label{zeta_infinite_volume}
\zeta \in \left(0,\frac{1}{\ee}\right).
\end{equation}
With $\zeta$ as in \eqref{zeta_infinite_volume}, we define
\begin{equation}
\label{kappa_definition}
\kappa \coloneqq -\log \zeta>1\,.
\end{equation}


We study the free energy in the infinite-volume limit.
Recalling \eqref{grand_canonical_ensemble}, the \emph{free energy (or specific Gibbs potential)} of the interacting Bose gas in finite volume is defined as
\begin{equation}
\label{specific_Gibbs_potential}
g^{\nu,\zeta,L} \coloneqq \frac{1}{|\Lambda_L|}\,\log \Xi^{\nu,\zeta,L}\,.
\end{equation}
We are interested in the limit of \eqref{specific_Gibbs_potential} as $L \rightarrow \infty$. To this end, we prove the following result.

\begin{theorem}[Convergence of the free energy in the infinite volume]
\label{free_energy_infinite_volume}
Suppose that \eqref{smallness_v_kappa} holds. We have the following results.
\begin{itemize}
\item[(i)] Recalling the definition \eqref{specific_Gibbs_potential}, the limit 
\begin{equation}
\label{free_energy_infinite_volume_i}
g^{\nu,\zeta,\infty} \coloneqq \lim_{L \rightarrow \infty} g^{\nu,\zeta,L}
\end{equation}
exists uniformly in $\nu \in (0,\nu_0]$ for $\nu_0$ as in \eqref{nu_0}.
\item[(ii)] We define the classical free energy (or specific Gibbs potential) as 
\begin{equation}
\label{classical_specific_Gibbs_potential}
g^{\mathrm{cl},\zeta,L} \coloneqq \frac{1}{|\Lambda_L|}\,\log \Xi^{0,\zeta,L}\,,
\end{equation}
where as in \eqref{partition_function_result_2}, in \eqref{classical_specific_Gibbs_potential} we take
\begin{equation*}
\Xi^{0,\zeta}\coloneqq \sum_{n=0}^{\infty} \frac{1}{n!}\, \int_{\Lambda_L^n} \dd x_1\,\cdots\, \dd x_n\, (2\pi)^{-\frac{nd}{2}}\, \zeta^n\,\mathrm{exp} \Biggl(
- \sum_{1 \leq i<j \leq n} v^L(x_i-x_j)\Biggr)\,.
\end{equation*}
We then have
\begin{equation}
\label{free_energy_infinite_volume_ii}
\lim_{\nu \rightarrow 0} \lim_{L \rightarrow \infty} g^{\nu,\zeta,L} = \lim_{L \rightarrow \infty} g^{\mathrm{cl},\zeta,L}\,.
\end{equation}
All the limits in \eqref{free_energy_infinite_volume_ii} exist and are finite.
\end{itemize}
\end{theorem}

We can prove an analogue of Theorem \ref{density_matrix_result} for reduced density matrices in the infinite-volume limit. 
In order to precisely state the result, we need to introduce some notation and conventions. Let $L_0 \geq 1$ and $p \in \N^*$ be given. Given $L \geq L_0$ and an operator $A$ acting on $p$ variables in $\Lambda_{L,d}$ with operator kernel 
defined on $\Lambda_L^p \times \Lambda_L^p$ (e.g.\ $\Gamma^{\nu,\zeta,p}(\vec{x},\vec{y})$ for $\vec{x},\vec{y} \in \Lambda_L^p$), we define the new operator $\mathbf{P}_{L_{0},p}\, A$ through its kernel
\begin{equation}
\label{1.37*}
(\mathbf{P}_{L_{0},p}\, A)(\vec{x},\vec{y}) \coloneqq A(\vec{x},\vec{y})\,\mathbf{1}_{\vec{x} \in \Lambda_{L_0}}\,\mathbf{1}_{\vec{y} \in \Lambda_{L_0}}\,.
\end{equation}
In other words, we project the kernel to $\Lambda_{L_0}^p \times \Lambda_{L_0}^p$.
The above construction is also applied when $L=\infty$, in which case we set $\Lambda_{\infty,d}=\R^d$. By using \eqref{1.37*}, we define
\begin{equation}
\label{L_0_seminorm}
\|A\|_{L_0,p} \coloneqq \|\mathbf{P}_{L_{0},p}\, A\|_{L^{\infty}_{\vec{x}}L^{1}_{\vec{y}}}
\equiv \left\|  \left\| (\mathbf{P}_{L_{0},p}\, A)(\vec{x},\vec{y}) \right\|_{L^1_{\vec{y}}} \right\|_{L^{\infty}_{\vec{x}}}\,.
\end{equation}
Note that $\|\cdot\|_{L_0,p}$ is a seminorm.

We can also prove the following convergence result for the reduced density matrices.
\begin{theorem}[Convergence of reduced density matrices in the infinite volume]
\label{Infinite_volume_Theorem_2}
Suppose that 
\begin{equation}
\label{smallness_v_kappa}
\frac{1}{\kappa-1}\,(\|v\|_{L^1(\R^d)}+1) \leq c_d
\end{equation}
holds for some constant $c_d>0$ small enough depending on $d$.
Then, recalling \eqref{reduced_p_particle_density_matrix} and \textup{\eqref{density_matrix_result_2}--\eqref{density_matrix_result_3}}, we have that for all $p \in \N^*$ and $L_0 \geq 1$
\begin{equation}
\label{Infinite_volume_Theorem_2_claim}
\lim_{\nu \rightarrow 0} \lim_{L \rightarrow \infty} \Gamma^{\nu,\zeta,L}_p=0
\end{equation}

with respect to \eqref{L_0_seminorm}.
\end{theorem}

\begin{remark}
\label{v_kappa_remark}
The assumption \eqref{smallness_v_kappa} tells us that we have to consider $\|v\|_{L^1(\R^d)}$ small or $\kappa$ large.
\end{remark}

\begin{remark}
\label{non_uniformity_in_L_0_introduction}
The convergence that we obtain in Theorem \ref{Infinite_volume_Theorem_2} is not shown to be uniform in $L_0$, as it relies on \eqref{Theorem_1.5(ii)_rate_of_convergence}; see also Remark \ref{non_uniformity_in_L_0} below.
\end{remark}

\begin{remark}
\label{zero_limit_Gamma}
The limit in \eqref{Infinite_volume_Theorem_2_claim} equals zero as a result of $\|\cdot\|_{L_0,p}$ given in \eqref{L_0_seminorm} and the fact that by 
\eqref{density_matrix_result_2}--\eqref{Kronecker_delta}, we have that $\Gamma_p^{0,\zeta}(\vec{x},\vec{y})=0$ unless $\vec{x}=\pi \vec{y}$ for some $\pi \in S_p$. This is different from the corresponding result on the lattice in \cite[Theorem 1.8 (ii)]{FKSS_2023}.
We refer the reader to \eqref{convergence_item_2}--\eqref{convergence_item_2**} below for details. 
\end{remark}

\subsection{Loop ensembles and the Ginibre representation}
\label{Section_1.3}
Given $T > 0$ and $x,y \in \Lambda$, let $\Omega^{T}_{y,x}$ denote the set of continuous paths $\omega \colon [0,T] \rightarrow \Lambda$ satisfying $\omega(0)=x$ and $\omega(T)=y$. Furthermore, we say that $\omega$ is a \textbf{\emph{loop}} if $\omega(T)=\omega(0)$. We refer to $T \equiv T(\omega)$ as the \textbf{\emph{duration}} of $\omega$. We abbreviate 
\begin{equation*}
\Omega^T\coloneqq \bigcup_{x,y \in \Lambda} \Omega^T_{y,x}\,,\qquad \Omega\coloneqq \bigcup_{T \geq 0}\Omega^T\,.
\end{equation*}
In particular, $\Omega^T$ denotes the space of continuous paths $\omega \colon [0,T] \rightarrow \Lambda$.
Furthermore, given $n \in \N^*$, $\vec{T}=(T_1,\ldots,T_n) \in (0,\infty)^n$, $\vec{x}, \vec{y} \in \Lambda^n$, we let
\begin{equation}
\label{Omega_vec_T}
\Omega^{\vec{T}}\coloneqq \prod_{i=1}^{n} \Omega^{T_i}\,,\qquad \Omega^{\vec{T}}_{\vec{y},\vec{x}}\coloneqq \prod_{i=1}^{n}  \Omega^{T_i}_{y_i,x_i}\,.
\end{equation}
For $t>0$, the heat kernel on $\Lambda \equiv \Lambda_{L,d}$ is given by\footnote{Note that in \eqref{heat_kernel} and in similar expressions from here onwards, the $t$ denotes time-dependence, and it is not a power.}
\begin{equation}
\label{heat_kernel}
\psi^t(x)\coloneqq (\ee^{t\Delta/2})(x,0)=\sum_{n \in \Z^d} \frac{1}{(2\pi t)^{d/2}}\, \ee^{-\frac{|x-Ln|^2}{2t}}\,.
\end{equation}
We have that 
\begin{equation}
\label{heat_kernel_integral_1}
\int_{\Lambda} \dd x\, \psi^t(x)=1\,.
\end{equation}
Given $x \in \Lambda$, let $\mathbb{P}^T_{x}$ denote the law on $\Omega^T$ of the standard Brownian motion on $\Lambda$ equal to $x$ at time $0$, with periodic boundary conditions. 
More precisely, we characterise $\mathbb{P}^T_{x}$ through its finite-dimensional distributions as follows.
For all $n \in \N^*$, $f \colon \Lambda^n \rightarrow \mathbb{C}$ continuous and $0<t_1<\ldots<t_n<T$, we have
\begin{multline}
\label{P^T}
\int \mathbb{P}^T_{x} (\dd \omega) f(\omega(t_1),\ldots,\omega(t_n))
\\
=\int_{\Lambda^n} \dd x_1\,\cdots \,\dd x_n\,\psi^{t_1}(x_1-x)\, \psi^{t_2-t_1}(x_2-x_1)\,\cdots\,\psi^{t_n-t_{n-1}}(x_n-x_{n-1})\,f(x_1,\ldots,x_n)\,.
\end{multline}
Given $x,y \in \Lambda$, let $\mathbb{P}^T_{y,x}$ denote the law of the Brownian bridge equal to $x$ at time $0$ and equal to $y$ at time $T$, with periodic boundary conditions on $\Lambda$. As in \eqref{P^T}, the \emph{Wiener measure}
\begin{equation}
\label{Wiener_measure}
\mathbb{W}^T_{y,x}(\dd \omega)\coloneqq \psi^T(y-x)\,\mathbb{P}^T_{y,x}(\dd \omega)
\end{equation}
is characterised by its finite-dimensional distribution as 
\begin{multline}
\label{W^T}
\int \mathbb{W}^T_{y,x} (\dd \omega) f(\omega(t_1),\ldots,\omega(t_n))
=\int_{\Lambda^n} \dd x_1\,\cdots\, \dd x_n\,\psi^{t_1}(x_1-x)\, \psi^{t_2-t_1}(x_2-x_1)\,\cdots\,
\\
\times \cdots\,
\psi^{t_n-t_{n-1}}(x_n-x_{n-1})\,\psi^{T-t_n}(y-x_n)\,
f(x_1,\ldots,x_n)\,.
\end{multline}
In particular, we can write
\begin{equation*}
\mathbb{W}^T_{y,x}(\dd \omega)=\delta (\omega(T)-y)\,\mathbb{P}^T_{x}(\dd \omega)\,.
\end{equation*}
By \eqref{W^T} with $f=1$, we also have that for all $x,y \in \Lambda$ and $T>0$ 
\begin{equation}
\label{heat_kernel_integral}
\int \mathbb{W}_{y, x}^{T}(\dd \omega)=\psi^{T}(y-x)\,.
\end{equation}
Throughout the paper, we will use several quantitative properties of the heat kernel \eqref{heat_kernel}, which we recall for completeness in Appendix \ref{appendix_heat_kernel_estimates} below; see also \cite[Lemmas 2.2--2.3]{FKSS_2020_1} and \cite[Lemma 5.4]{FKSS_2022} for proofs of these results. A further quantitative property of the heat kernel is given in Lemma \ref{Riemann_sums_estimate} below.

\begin{lemma}[Estimates on the heat kernel]
\label{heat_kernel_estimates}
The following estimates hold for $t>0$.
\begin{itemize}
\item[(i)] For all $x \in \Lambda \equiv \Lambda_L$, we have
\begin{equation}
\label{heat_kernel_estimates_1}
0 \leq \psi^t(x) \lesssim_{d} \frac{1}{L^d}+\frac{1}{t^{d/2}}\,.
\end{equation}
\item[(ii)] We define
\begin{equation}
\label{heat_kernel_estimates_tilde_A}
\eta^t\coloneqq \sum_{n \in \Z^d \setminus \{0\}} \frac{1}{(2\pi t)^{d/2}} \,\ee^{-\frac{|Ln|^2}{2t}}=\psi^t(0)-\frac{1}{(2\pi t)^{d/2}}\,.
\end{equation}
Then, we have that, uniformly in $t > 0$,
\begin{equation}
\label{heat_kernel_estimates_tilde}
0 \leq \eta^t \lesssim_{d} \frac{1}{L^d}\,.
\end{equation}
\item[(iii)] For all $x,y \in \Lambda$, and for all $0 \leq s_1 \leq s_2 \leq t$, we have
\begin{equation}
\label{FKSS_2022_estimate}
\int \mathbb{W}^{t}_{y,x}(\dd \omega)\,|\omega(s_1)-\omega(s_2)|_{L} \lesssim_{d} \biggl(\frac{1}{L^d}+\frac{1}{t^{d/2}}\biggr)\,\biggl((s_2-s_1)^{1/2}+|x-y|_{L} \,\frac{s_2-s_1}{t}\biggr)\,.
\end{equation}
\end{itemize}
\end{lemma}

We henceforth fix $\nu, \zeta > 0$, and consider $z \equiv z(\nu, \zeta)$ as in \eqref{fugacity}.
In what follows, we work with the measure defined on loops of duration $T$ given by
\begin{equation}
\label{W^T_definition}
\mathbb{W}^T(\dd \omega) \coloneqq \int _{\Lambda} \dd x \, \mathbb{W}^T_{x,x}(\dd \omega)\,.
\end{equation}
Recalling \eqref{fugacity} and using \eqref{W^T_definition}, the \emph{single-loop measure} $\mathbb{L}^{\nu,\zeta}(\dd \omega)$ is defined as\footnote{Our notational choices in \eqref{fugacity}--\eqref{fugacity_condition} were taken so that the single loop measure \eqref{single_loop_measure} has a simple form. The latter is a crucial object in our analysis.}
\begin{equation}
\label{single_loop_measure}
\mathbb{L}^{\nu,\zeta}(\dd \omega)\coloneqq \nu \sum_{T \in \nu \N^*} \frac{z^T}{T} \, \mathbb{W}^T(\dd \omega)\,.
\end{equation}
Given $\omega \in \Omega^{T(\omega)}, \tilde{\omega} \in \Omega^{T(\tilde{\omega})}$ with $T(\omega), T(\tilde{\omega}) \in \nu \N^*$, we define their \emph{interaction} by 
\begin{equation}
\label{loop_interaction}
\mathcal{V}^{\nu}(\omega,\tilde{\omega})\coloneqq \frac{1}{\nu} \sum_{s \in \nu \N} \vec{1}_{s<T(\omega)} \sum_{\tilde{s} \in \nu \N} \vec{1}_{\tilde{s}<T(\tilde{\omega})} \int_{0}^{\nu}\dd t\, v\bigl(\omega(s+t)-\tilde{\omega}(\tilde{s}+t)\bigr)\,.
\end{equation}
Moreover, given $\omega \in \Omega^{T(\omega)}$ with $T(\omega) \in \nu \N$, we define its \emph{self-interaction} by
\begin{equation}
\label{self_interaction}
\mathcal{V}^{\nu}(\omega)\coloneqq \frac{1}{\nu} \sum_{s,s' \in \nu \N} \vec{1}_{s<s'<T(\omega)} \int_{0}^{\nu}\dd t\, v\bigl(\omega(s+t)-\omega(s'+t)\bigr)\,.
\end{equation}
Let us note that, by \eqref{self_interaction} we have 
\begin{equation}
\label{self_interaction_duration_nu}
\mathcal{V}^{\nu}(\omega)=0\,\, \text{ if }\,\, T(\omega)=\nu\,.
\end{equation}
Recalling \eqref{Omega_vec_T}, we consider $n \in \N^*$ and $\vec{\omega}=(\omega_1,\ldots,\omega_n) \in \Omega^{\vec{T}}$ such that $T_i=T(\omega_i) \in \nu \N^*$ for $i = 1, \ldots, n$. Using \eqref{loop_interaction}--\eqref{self_interaction}, we define for such $\vec{\omega}$ the interaction
\begin{equation}
\label{loop_interaction_2}
\mathcal{V}^{\nu}(\vec{\omega})\coloneqq \sum_{1 \leq i<j \leq n} \mathcal{V}^{\nu}(\omega_i,\omega_j)+\sum_{i=1}^{n} \mathcal{V}^{\nu}(\omega_i)\,.
\end{equation}
We define for $n \in \N^{*}, \vec{T}=(T_1,\ldots,T_n) \in (0,\infty)^n$ the following positive measure on  $\Omega^{\vec{T}}$.
\begin{equation}
\label{W_product}
\mathbb{W}^{\vec{T}}_{\vec{y},\vec{x}} (\dd \vec{\omega})\coloneqq \prod_{i=1}^{n}  \mathbb{W}^{T_i}_{y_i,x_i} (\dd \omega_i)\,.
\end{equation}  
The starting point of our analysis is the Ginibre representation \cite{Ginibre1,Ginibre2,Ginibre3,Ginibre}.

\begin{lemma}[Ginibre representation]
\label{Ginibre_representation} 
Let $z \equiv z(\nu,\zeta)$ be given by \eqref{fugacity}. Recalling the single-loop measure given by \eqref{single_loop_measure} and interaction given by \eqref{loop_interaction}, the following representations hold. 
\begin{itemize}
    \item [(i)] The grand-canonical partition function \eqref{grand_canonical_ensemble} can be written as
\begin{equation}
\label{Ginibre_representation_1}
\Xi^{\nu,\zeta}= 1 + \sum_{n=1}^{\infty} \frac{1}{n!} \int \mathbb{L}^{\nu,\zeta}(\dd \omega_1) \cdots\, \mathbb{L}^{\nu,\zeta}(\dd \omega_n) \,\exp\bigl(-\mathcal{V}^{\nu}(\vec{\omega})\bigr)\,.
\end{equation}
\item [(ii)] For $\vec{\omega}=(\omega_1,\ldots,\omega_p)$, where $\omega_i \in \Omega^{T(\omega_i)}$ and $T(\omega_i) \in \nu \N^*$ for $i =1,\ldots,p$, we define
\begin{equation}\label{Ginibre_representation_3}
    \Xi^{\nu,\zeta}(\vec{\omega})\coloneqq \sum_{n=0}^{\infty} \frac{1}{n!} \int \mathbb{L}^{\nu,\zeta}(\dd \widetilde{\omega}_1) \cdots\, \mathbb{L}^{\nu,\zeta}(\dd \widetilde{\omega}_n) \,\exp\bigl(-\mathcal{V}^{\nu}(\vec{\widetilde{\omega}}\vec{\omega})\bigr)\,,
\end{equation}
where in analogy with \eqref{concatenation}, we define the concatenation of the vectors of paths given by $\vec{\widetilde{\omega}}=(\widetilde{\omega}_1,\ldots,\widetilde{\omega}_n)$ and  $\vec{\omega}=(\omega_1,\ldots,\omega_p)$ as
\begin{equation}
\label{concatenation_omega}
\vec{\widetilde{\omega}}\vec{\omega}\coloneqq (\widetilde{\omega}_1,\ldots,\widetilde{\omega}_n,\omega_1,\ldots,\omega_p)\,.
\end{equation}
The kernel of the reduced $p$-particle density matrix \eqref{reduced_p_particle_density_matrix} can be written as
\begin{equation}
\label{Ginibre_representation_4}
\Gamma_p^{\nu,\zeta}(\vec{x},\vec{y})=\sum_{\pi \in S_p}\sum_{\vec{k} \in (\N^{*})^p} \prod_{i=1}^p z^{k_i \nu} 
\int \mathbb{W}^{\nu \vec{k}}_{\pi \vec{y}, \vec{x}}(\dd \vec{\omega}) \,\frac{\Xi^{\nu,\zeta}(\vec{\omega})}{\Xi^{\nu,\zeta}}\,, \quad \vec{x},\vec{y} \in \Lambda^p\,. 
\end{equation}
\end{itemize}
Note that by \eqref{Ginibre_representation_3}, we have $\Xi^{\nu,\zeta} \equiv \Xi^{\nu,\zeta}(\emptyset),$ where $\emptyset$ denotes the empty vector. In \eqref{Ginibre_representation_4}, we also recall \eqref{W_product}.
\end{lemma}
The result of Lemma \ref{Ginibre_representation} follows from the analysis in \cite[Section 2]{Ginibre}. In particular, the stability condition given by Assumption \ref{Assumption_on_v} (i) corresponds to condition B) in \cite[(2.4)]{Ginibre}. The analysis there is done in this general setting. 
We refer the reader to \cite[Appendix A.2]{FKSS_2023} for a summary of the proof of Lemma \ref{Ginibre_representation}, under slightly different conventions in \eqref{Hamiltonian_H}, \eqref{n_body_Hamiltonian}, \eqref{self_interaction}
 and \eqref{loop_interaction_2} above. Furthermore, this is done for pointwise nonnegative potentials. For completeness and for the reader's convenience, in Appendix \ref{Ginibre_representation_appendix} below, we explain how to reconcile the conventions from \cite{FKSS_2023} and to obtain Lemma \ref{Ginibre_representation} from the analysis summarised in \cite[Appendix A.2]{FKSS_2023}. For a more general discussion on loop ensembles, we refer the reader to \cite[Section 1]{FKSS_2023}.

\subsubsection{The infinite-volume limit of the free energy for Dirichlet boundary conditions}

In addition to the results stated in Section \ref{The_infinite volume_(bosonic)_regime}, we revisit the subadditivity argument from \cite{Ginibre} used to show the existence of the infinite-volume limit of the free energy in our context when one imposes Dirichlet boundary conditions.  We introduce some more terminology and notation.
Our starting point is the Ginibre representation \eqref{Ginibre_representation_1} from Lemma \ref{Ginibre_representation} (i) above.

Let us define
\begin{equation}
\label{I^{L}}
 \mathscr{I}_L \coloneqq \biggl\{\omega \in \bigcup_{T > 0} \Omega^{\infty,T}\,,\,\omega(t) \in \Lambda_{L}\,\, \forall\, t \in [0,T(\omega)]\biggr\}\,.
\end{equation}
With $\mathscr{I}_L$ given by \eqref{I^{L}}, we modify the definition of the loop measure \eqref{single_loop_measure} and define
\begin{equation}
\label{single_loop_measure_bar}
\overline{\mathbb{L}}^{\,\nu,\zeta,L}(\dd \omega)\coloneqq \nu \sum_{T \in \nu \N^*} \frac{z^T}{T} \, \int_{\Lambda_L} \dd x\,\int \mathbb{W}^{\infty,T}_{x,x}(\dd \omega)\,\mathbf{1}_{\mathscr{I}_L}(\omega)\,.
\end{equation}
In \eqref{single_loop_measure_bar}, we also recalled \eqref{W^T_definition}, but now on $\Lambda=\R^d$. In other words, in \eqref{single_loop_measure_bar}, we integrate with respect to $\mathbb{W}^{\infty,T}_{x,x}(\dd \omega)$ all loops that lie completely in $\Lambda_L$.
Using \eqref{single_loop_measure_bar}, we modify the Ginibre representation \eqref{Ginibre_representation_1} of \eqref{grand_canonical_ensemble} and define
\begin{equation}
\label{Ginibre_representation_1_bar}
\overline{\Xi}^{\,\nu,\zeta,L} \coloneqq  1 + \sum_{n=1}^{\infty} \frac{1}{n!} \int \overline{\mathbb{L}}^{\,\nu,\zeta,L}(\dd \omega_1) \cdots\, \overline{\mathbb{L}}^{\,\nu,\zeta,L}(\dd \omega_n) \,\exp\bigl(-\mathcal{V}^{\nu,\infty}(\vec{\omega})\bigr)\,.
\end{equation}
We emphasise that the bar in \eqref{single_loop_measure_bar}--\eqref{Ginibre_representation_1_bar} does not denote a complex conjugate.
In \eqref{Ginibre_representation_1_bar} and onwards, we use $\mathcal{V}^{\nu,\infty}(\vec{\omega})$ to denote \eqref{loop_interaction_2} with interaction potential $v^{\infty} \equiv v$ in \eqref{loop_interaction}--\eqref{self_interaction}. This is in contrast with $\mathcal{V}^{\nu,L}(\vec{\omega})$ where the interaction potential in \eqref{loop_interaction}--\eqref{self_interaction} is given by \eqref{v^L}.

\begin{remark}
\label{Dirichlet_boundary_condition_remark}
Let us note that \eqref{Ginibre_representation_1_bar} corresponds to the Ginibre representation of the grand-canonical partition function \eqref{grand_canonical_ensemble}, where in \eqref{n_body_Hamiltonian}, we consider \emph{Dirichlet boundary conditions} instead of periodic boundary conditions. Here, one considers $d \geq 2$. This is the context of the analysis in \cite{Ginibre}; see Appendix \ref{Dirichlet_boundary_conditions} below for a precise summary.
\end{remark}

\begin{proposition}
\label{Infinite_volume_Theorem_1}
Suppose that \eqref{smallness_v_kappa} holds. Then, we have the following.
\begin{itemize}
\item[(i)] With $ \overline{\Xi}^{\,\nu,\zeta,L}$ defined as in \eqref{Ginibre_representation_1_bar}, the limit 
\begin{equation}
\label{Theorem_1.16_ii}
\lim_{L \rightarrow \infty} \frac{\log \overline{\Xi}^{\,\nu,\zeta,L}}{|\Lambda_L|}
\end{equation}
exists and is finite.
\item[(ii)] We have
\begin{equation}
\label{Infinite_volume_Theorem_1_Star_1}
\lim_{L \rightarrow \infty} \left(\frac{\log \Xi^{\nu,\zeta,L}}{|\Lambda_L|}-\frac{\log \overline{\Xi}^{\,\nu,\zeta,L}}{|\Lambda_L|}\right)=0\,,
\end{equation}
uniformly in $\nu \in (0,\nu_0]$.
In particular by \eqref{Infinite_volume_Theorem_1_Star_1} and Theorem \ref{free_energy_infinite_volume} (i), the limit \eqref{Theorem_1.16_ii} is uniform in $\nu \in (0,\nu_0]$ and equals \eqref{free_energy_infinite_volume_i}.
\end{itemize}
\end{proposition}

We give the proof of Proposition \ref{Infinite_volume_Theorem_1} in Appendix \ref{Dirichlet_boundary_conditions} below. The proof of Proposition \ref{Infinite_volume_Theorem_1} (i) is based on a variant of Fekete's subadditivity lemma. This is possible because the interaction potential on $\Lambda_L$ is given by $v$ instead of \eqref{v^L}.  Proposition  \ref{Infinite_volume_Theorem_1} (ii) allows us to connect it with our original free energy \eqref{specific_Gibbs_potential} as stated above. In particular, this result tells us that in the infinite-volume limit, it does not matter whether we consider periodic or Dirichlet boundary conditions in the finite volume. For completeness, we present the details in Appendix \ref{Dirichlet_boundary_conditions}.

\subsection{Main ideas of the proof}

Let us illustrate the heuristic ideas of our method on the proof of the convergence \eqref{partition_function_result_1} in Theorem \ref{partition_function_result}. We simplify the analysis and consider $v \geq 0$, instead of $v$ of stable type as in Assumption \ref{Assumption_on_v} above. The latter generalisation requires additional considerations, but does not change the argument in a substantial way. For simplicity, we set $L=1$.
In what follows, we use the shorthand $A \thickapprox B$ for quantities $A$ and $B$ depending on $\nu$ to mean that $A-B \rightarrow 0$ as $\nu \rightarrow 0$.
Starting from the Ginibre representation \eqref{Ginibre_representation_1} in Lemma \ref{Ginibre_representation} (i) and recalling \eqref{single_loop_measure},  we write
\begin{equation}
\label{Heuristic_1}
\Xi^{\nu,\zeta} = \sum_{n=0}^{\infty} \frac{1}{n!} \left[\prod_{k=1}^{n} \nu \sum_{T_k \in \nu \N^*} \frac{z^{T_k}}{T_k}\,\int_{\Lambda} \dd x_k \int \mathbb{W}^{T_k}_{x_k,x_k}(\dd \omega_k)\right]\, \exp\bigl(-\mathcal{V}^{\nu}(\vec{\omega})\bigr) \,.
\end{equation}
We interpret 
\begin{equation}
\label{Heuristic_1B}
\left[\prod_{k=1}^{n} \nu \sum_{T_k \in \nu \N^*} \frac{z^{T_k}}{T_k}\,\int_{\Lambda} \dd x_k \int \mathbb{W}^{T_k}_{x_k,x_k}(\dd \omega_k)\right]\, \exp\bigl(-\mathcal{V}^{\nu}(\vec{\omega})\bigr) 
\end{equation}
as a (formal) Riemann sum in $\R^n$ of mesh size $\nu$. It does not correspond to a convergent integral though (even when $v=0$). Using \eqref{heat_kernel_integral}, Lemma \ref{heat_kernel_estimates} (i) and $v \geq 0$, we get that

\begin{equation}
\label{Heuristic_2}
\eqref{Heuristic_1B} \lesssim \sum_{n=0}^{\infty} \frac{1}{n!} \left[\prod_{k=1}^{n} \nu \sum_{T_k \in \nu \N^*} z^{T_k}\,\left(1+\frac{1}{T_k^{1+d/2}}\right) \right].
\end{equation}
By Lemma  \ref{heat_kernel_estimates} (ii), we have that the opposite inequality in \eqref{Heuristic_2} holds when $v=0$.
We note that the function 
\begin{equation*}
t \in (0,\infty) \mapsto 1+\frac{1}{t^{1+d/2}}
\end{equation*}
is not integrable near the origin. Therefore, (at least when $v=0$) \eqref{Heuristic_1B} does not correspond to the Riemann sum of a convergent integral. We expect the main contribution to be the one coming from the most singular terms, i.e.\ when $T_k=\nu$ for all $k=1,\ldots,n$. This is made precise in Proposition \ref{Proposition_1} below, where it is shown that 
\begin{equation}
\label{Heuristic_1C}
\Xi^{\nu,\zeta} \thickapprox \widehat{\Xi}^{\nu,\zeta}\coloneqq  \sum_{n=0}^{\infty} \frac{1}{n!} \left[\prod_{k=1}^{n} z^{\nu}\,\int_{\Lambda} \dd x_k \int \mathbb{W}^{\nu}_{x_k,x_k}(\dd \omega_k)\right]\, \exp\bigl(-\mathcal{V}^{\nu}(\vec{\omega})\bigr) \,,
\end{equation}
together with a bound on the rate of convergence as $\nu \rightarrow 0$.

The next step is to note that for $\vec{\omega}=(\omega_1,\ldots,\omega_n)$ with $\omega_k \in \Omega^{\nu}_{x_k,x_k}$ for $k=1,\ldots,n$, we have
\begin{equation}
\label{Heuristic_2B}
\mathcal{V}^{\nu} (\omega_i,\omega_j) \thickapprox v(x_i-x_j)\,,
\end{equation}
for all $1 \leq i<j \leq n$. This follows formally from \eqref{loop_interaction} by considering averages in $t$. In particular, recalling \eqref{loop_interaction_2}, substituting \eqref{Heuristic_2B} into \eqref{Heuristic_1C}, and using  \eqref{heat_kernel_integral}, we get 

\begin{equation}
\label{Heuristic_1D}
\widehat{\Xi}^{\nu,\zeta} \thickapprox \sum_{n=0}^{\infty} \frac{1}{n!} \left[\prod_{k=1}^{n} z^{\nu}\,\psi^{\nu}(0)\,\int_{\Lambda} \dd x_k \right]\, \exp \left(-\sum_{1 \leq i < j \leq n} v(x_i-x_j)\right).
\end{equation}
We rigorously prove \eqref{Heuristic_1D} in \eqref{Proposition_2_5} of the proof of Proposition \ref{Proposition_2}  below. Here, we also prove a bound on the rate of convergence in \eqref{Heuristic_1D} as $\nu \rightarrow 0$. 
The main idea in this step is to use quantitative estimates on the displacement of Brownian motion given in Lemmas \ref{Lemma_2}--\ref{Lemma_3} below. Heuristically, we are replacing Brownian loops of duration $\nu$ by their basepoints.

Using Lemma \ref{heat_kernel_estimates} (ii), and recalling \eqref{fugacity_condition}, we get that 
\begin{multline}
\label{Heuristic_3}
\eqref{Heuristic_1D}
\thickapprox \sum_{n=0}^{\infty} \frac{1}{n!}\,\int_{\Lambda^n} \dd x_1\,\cdots\, \dd x_n\ \,(2\pi)^{-\frac{nd}{2}}\, \left(z^{\nu} \nu^{-d/2}\right)^n\, \exp \left(-\sum_{1 \leq i < j \leq n} v(x_i-x_j)\right) 
\\
=\sum_{n=0}^{\infty} \frac{1}{n!}\, \int_{\Lambda^n} \dd x_1\,\cdots\, \dd x_n\, (2\pi)^{-\frac{nd}{2}}\, \zeta^n \,\mathrm{exp} \Biggl(
- \sum_{1 \leq i<j \leq n} v(x_i-x_j)\Biggr) = \Xi^{0,\zeta}\,.
\end{multline}
The step \eqref{Heuristic_3} is done precisely in \eqref{Proposition_2_6} in the proof of Proposition \ref{Proposition_2} below. For the convergence in \eqref{Heuristic_3}, we also give an explicit bound on the rate.
We note that \eqref{fugacity_condition} was crucially used in \eqref{Heuristic_3}.
Putting everything together, we get the quantitative estimate \eqref{partition_function_result_3} in Theorem \ref{partition_function_result}, and in particular we get \eqref{partition_function_result_1}.

The analysis of the reduced $p$-particle density matrices in Theorem \ref{density_matrix_result} is done in a similar way, except that we now start from the Ginibre representation \eqref{Ginibre_representation_4} in Lemma \ref{Ginibre_representation} (ii). This requires us to eliminate both loops and open paths of duration greater than $\nu$ (similarly to \eqref{Heuristic_1C} above). These steps are made precise in Propositions \ref{Proposition_4_3_A} and \ref{Proposition_4_3} respectively. 

Let us comment on the conditions on $v$ given in Assumption \ref{Assumption_on_v}.
Since in Assumption \ref{Assumption_on_v} (i) we are in general working with $v$ which are \emph{stable} in the sense of Assumption \ref{Assumption_on_v} (i) (and not pointwise nonnegative), some additional care is needed to estimate the exponential term $\exp\bigl(-\mathcal{V}^{\nu}(\vec{\omega})\bigr)$ that occurs in the Ginibre representation. A lower bound on $\mathcal{V}(\vec{\omega})$ that is used to this end is shown in Lemma \ref{lemma_loop_interactions} below. The quantitative bounds on the rates of convergence that we obtain are obtained using the H\"{o}lder continuity of $v$ given in Assumption \ref{Assumption_on_v} (ii) (see Proposition \ref{Proposition_2} below).

When studying the infinite-volume limit, we use the method of \emph{cluster expansions}. We follow the setup of \cite{Ueltschi_2004}, which was also used on the lattice in \cite[Section 5]{FKSS_2023}. The setup of the cluster expansion is given in Proposition \ref{gamma_p_cluster_expansion} below. We study the cluster expansion by graphical methods. 
The starting point is Lemma \ref{tree_bound_estimate} which gives us a tree bound for the terms appearing in the expansion.
We then integrate over the trees using the integration algorithm obtained from Lemmas \ref{Lemma_5.16*} and \ref{integration_algorithm} below. As a result, we obtain uniform bounds in the volume on the free energy and reduced density matrices in Proposition \ref{Proposition_5.3}. In order to pass to the infinite-volume limit, we see that long paths and large-mass Brownian paths in the Ginibre representation give negligible contributions in the limit; this is explained precisely in Lemmas \ref{Lemma_5.20*} and \ref{Lemma_5.21*} below. As we are working with periodic boundary conditions, the proof of the latter results requires an additional step connecting Brownian motion with periodic boundary conditions with Brownian motion on all of $\R^d$. This is explained in Lemma \ref{Wiener_measure_pi_L}.

\subsection{Previously known results} Most of the existing literature has been devoted to the study of the \emph{mean-field limit}. The scaling in this setting is chosen differently. In particular, one takes $\lambda=\nu^2$ in \eqref{nu} and the chemical potential is fixed (i.e.\ one does not need to tune it as in \eqref{fugacity_condition} above). 
Instead of obtaining a classical particle theory as in the large-mass limit, in this setting one obtains a classical field theory, which is a nonlocal or local $\Phi^4$ theory.
When one is working in the continuum in dimensions $d>1$, the expectation of the rescaled particle number in the grand-canonical ensemble blows up as $\nu \rightarrow 0$. In order to obtain the suitable limit, one needs to apply a renormalisation procedure both on the quantum and classical level. 

In the mean-field limit, quantum Gibbs states give rise to classical Gibbs states. The latter have been studied in the context of Gibbs measures for nonlinear Schr\"{o}dinger equations, following the pioneering results of Bourgain in the 1990s \cite{Bourgain_1994,Bourgain_1996,Bourgain_1997}.
The first result in the study of the mean-field limit was the work of Lewin, Nam, and Rougerie \cite{LNR15} using variational methods and the quantum de Finetti theorem. Soon afterwards, an expansion method combined with graphical methods and Borel resummation was used by Fr\"{o}hlich, Knowles, Schlein, and the third author in \cite{FKSS17}. In \cite{FKSS17}, one could study translation-invariant interactions when $d=2,3$, albeit with a suitably modified quantum Gibbs state. The latter modification was removed later in the independent works \cite{LNR21} and \cite{FKSS_2020_1}. The methods of \cite{LNR21} build on those from \cite{LNR15} in a nontrivial way. The methods of \cite{FKSS_2020_1} are based on a functional integral approach. In \cite{FKSS_2022}, the authors could obtain the $\Phi^4_2$ theory in the mean-field limit. In subsequent work \cite{Nam_Zhu_Zhu} the $\Phi^4_3$ was obtained in the mean-field limit. We also refer the reader to the related works \cite{AFP24,AR21,CKRT,Dinh_Rougerie,FKSS18,FKSS_2020_2,Jougla_Rougerie,LNR18,LNR18b,LNR19,Nam_Zhu_Zhu_2,RS22,RS23,Sohinger_2019}.

The large-mass limit was studied on the lattice in \cite{FKSS_2023}. Here, the rescaling of the chemical potential is simpler than that given by \eqref{fugacity_condition}, as the heat kernel is bounded. The methods in \cite{FKSS_2023} allow one to study the mean-field limit as well. The setup of both limits is given in terms of convergence of loop ensembles. We also refer the reader to \cite[Section 3]{Knowles} and \cite{Salmhofer} for related results.

\subsection{Organisation of the paper}
\label{Organisation_of_the_paper}

We prove Theorem \ref{partition_function_result} in Section \ref{Proof_of_Theorem_1.3}. In Section \ref{section_density_matrix_Bose}, we prove Theorem \ref{density_matrix_result}. Both of these sections are devoted to the study of bosons in a finite volume. We prove analogous results for fermions in Theorems \ref{partition_function_result_Fermi} and \ref{density_matrix_result_Fermi} in Section \ref{Fermi_gas}. Section \ref{section_infinite_volume} is devoted to the study of the infinite-volum limit. Here, we prove Theorems \ref{Infinite_volume_Theorem_2} and \ref{free_energy_infinite_volume}. In Appendix \ref{appendix_heat_kernel_estimates}, we record estimates on the heat kernel that we use in the analysis. In particular, we prove Lemmas \ref{heat_kernel_estimates} and \ref{Riemann_sums_estimate}. In Appendix \ref{Ginibre_representation_appendix}, we comment on the difference on the conventions that we are using for the Ginibre representation in comparison to \cite{FKSS_2023}. In Appendix \ref{tuning_particle_density}, we explain how one can tune the particle density in the context of the large-mass limit of bosons. Finally, in Appendix \ref{Dirichlet_boundary_conditions}, we justify Remark \ref{Dirichlet_boundary_condition_remark}  and prove Proposition \ref{Infinite_volume_Theorem_1} stated above.

\subsubsection*{Acknowledgements} The authors would like to thank Zied Ammari, J\"{u}rg Fr\"{o}hlich, Antti Knowles, Rapha\"{e}l Lefevere, Tommaso Rosati, Benjamin Schlein, Roger Tribe, Daniel Ueltschi, and Oleg Zaboronski for helpful discussions during various stages of the project. 
They would also like to acknowledge the late Stefan Adams, who was the PhD supervisor of S.G., for sharing his insights. 
The authors would like to thank the anonymous referee for their careful reading of our manuscript and for their helpful comments.
G.S.\ is supported by the Warwick Mathematics Institute Centre for Doctoral Training, and gratefully acknowledges funding from the University of Warwick. S.G. was supported by the Warwick Mathematics Institute Centre for Doctoral Training, and gratefully acknowledges funding from the UK Engineering and Physical Sciences Research Council (Grant number: EP/W523793/1).
No Generative AI tools were used during the preparation of this work.

\section{Convergence of partition functions in finite volume: proof of Theorem \ref{partition_function_result}}
\label{Proof_of_Theorem_1.3}
In this section, we prove Theorem \ref{partition_function_result} on the convergence of the grand-canonical partition function of the interacting Bose gas to that of a classical gas in finite volume.  
Before proceeding with the proof, we note the following lower bound on $\mathcal{V}^\nu (\vec{\omega})$ defined in \eqref{loop_interaction_2}, which will be useful throughout the paper.
\begin{lemma}[Lower bound on $\mathcal{V}^{\nu}(\vec{\omega})$]
\label{lemma_loop_interactions}
	Recall \eqref{Omega_vec_T}. Let $\vec{\omega} = (\omega_1, \ldots, \omega_n) \in \Omega^{\vec{T}}$ be such that  for all $i=1,\ldots, n$, we have that $T_i=T(\omega_i) \in \nu \N^*$. 
	Then, with notation as in \eqref{loop_interaction_2}, we have
    \begin{equation}\label{loop_interaction_3}
        \mathcal{V}^{\nu}(\vec{\omega})\geq-\frac{B}{\nu} \sum_{i=1}^nT(\omega_i)\,,
    \end{equation}
    where $B$ is given by Assumption \ref{Assumption_on_v} (i).
\end{lemma}
\begin{proof}
    For $i=1,\ldots,n$, let 
    \begin{equation*}
    k_i \coloneqq  \frac{T(\omega_i)}{\nu} \in \N^*\,.
    \end{equation*}
     In view of \eqref{loop_interaction}$-$\eqref{loop_interaction_2}, we have 
    \begin{multline}
    \label{loop_interaction_4}
        \mathcal{V}^{\nu}(\vec{\omega})=\frac{1}{\nu}\int_0^\nu \dd t\Bigg(\sum_{1\leq i<j\leq n}\sum_{\ell_i=0}^{k_i-1}\sum_{\ell_j=0}^{k_j-1}  v(\omega_i(t+\ell_i\nu)-\omega_j(t+\ell_j\nu))\\
        + \sum_{i=1}^n \sum_{0\leq\ell_i<\tilde{\ell}_i\leq k_i-1}  v(\omega_i(t+\ell_i\nu)-\omega_i(t+\tilde{\ell}_i\nu))\Bigg).
    \end{multline}
 In order to apply \eqref{stable_potential_condition}, we relabel the points on the paths $\omega_i$ occurring in \eqref{loop_interaction_4}. 
    To this end, we first write
    \begin{equation*}
    K\coloneqq\sum_{i=1}^n k_i=\sum_{i=1}^{n} \frac{T(\omega_i)}{\nu}\,.
   \end{equation*}
   Note that $K$ is the total number of segments of duration $\nu$ in the paths $\omega_1,\ldots,\omega_n$ which occur as components of $\vec{\omega}$.
   Let us consider $\ell \in [K]=\{1, \ldots, K\}$ and $t \in [0, \nu]$ arbitrary. 
   By construction, there exists a unique $m \in \{0,1,\ldots,n-1\}$ 
  such that 
   \begin{equation}
   \label{m_choice_decomposition}
   \sum_{i = 0}^{m} k_{i} < \ell \leq \sum_{i = 0}^{m + 1} k_{i}\,,
  \end{equation} 
  where by convention, we take $k_{0} \coloneqq 0$. 
  Using \eqref{m_choice_decomposition}, we can write 
   \begin{equation*}
   \ell = \sum_{i = 0}^{m} k_{i} + r\,,
   \end{equation*}
   for some $r \in \{1, \ldots, k_{m+1}\}$. We now define
    \begin{equation}\label{loop_interaction_5}
	    x^{(t)}_{\ell} \coloneqq \omega_{m+1}(t+(r-1)\nu)\,.
    \end{equation}
    Our labelling of the points is given by \eqref{loop_interaction_5}. For illustration, we have
\begin{align*}
 &x_1^{(t)}=\omega_1(t),\,\, x_2^{(t)}=\omega_1(t+\nu)\,,\,\ldots\,,\,x_{k_1}^{(t)}=\omega_1(t+(k_1-1)\nu),\\
 &x_{k_1+1}^{(t)}=\omega_2(t),\,\,x_{k_1+2}^{(t)}=\omega_2(t+\nu),\,\, \ldots, \,\, x_{k_1+k_2}^{(t)}=\omega_2(t+(k_2-1)\nu)\,,\,
\ldots
\\
&x_{k_1+\cdots+k_{n-1}+1}^{(t)}=\omega_n(t)\,,x_{k_1+\cdots+k_{n-1}+2}^{(t)}=\omega_n(t+\nu)\,,\ldots\,,x_{K}^{(t)}=\omega_n(t+(k_{n}-1)\nu)\,.
\end{align*}
 Using the labelling \eqref{loop_interaction_5} and integrating over $t \in [0,\nu]$ in \eqref{loop_interaction_4}, we have 
    \begin{equation}\label{loop_interaction_6}
        \mathcal{V}^{\nu}(\vec{\omega})=\frac{1}{\nu}\int_0^\nu \dd t\sum_{1\leq i<j\leq K}v(x^{(t)}_i-x^{(t)}_j)\geq -BK\,.
    \end{equation}
    In  \eqref{loop_interaction_6}, we used Assumption \ref{Assumption_on_v} (i) for each $t \in [0,\nu]$ to deduce the lower bound. 
\end{proof}

Let us now introduce some more notation. We first rewrite \eqref{single_loop_measure} as
\begin{equation}
\label{single_loop_measure_2}
\mathbb{L}^{\nu,\zeta}(\dd \omega)= \sum_{k=1}^{\infty} \frac{z^{k\nu}}{k} \, \mathbb{W}^{k \nu}(\dd \omega)\,.
\end{equation}
By using \eqref{single_loop_measure_2}, we can rewrite \eqref{Ginibre_representation_1} as
\begin{equation}
\label{Ginibre_representation_2}
\Xi^{\nu,\zeta}=\sum_{n=0}^{\infty} \frac{1}{n!} \sum_{\vec{k} \in (\N^{*})^n} \prod_{i=1}^{n} \frac{z^{k_i\nu}}{k_i} \int \mathbb{W}^{\nu \vec{k}} (\dd \vec{\omega}) \, \exp \bigl(-\mathcal{V}^{\nu}(\vec{\omega})\bigr)\,.
\end{equation}
In \eqref{Ginibre_representation_2}, we define for $n \in \N^{*}, \vec{T}=(T_1,\ldots,T_n) \in (0,\infty)^n$ the following positive measure on 
$\Omega^{\vec{T}}$.
\begin{equation}
\label{W_product_2}
\mathbb{W}^{\vec{T}} (\dd \vec{\omega})\coloneqq \prod_{i=1}^{n}  \mathbb{W}^{T_i} (\dd \omega_i)\,.
\end{equation}
Here, we recall \eqref{W^T_definition}.
Given $n \in \N^*$, we let 
\begin{equation}
\label{vec_1_n}
\vec{1}_n\coloneqq (1,1,\ldots,1) \in (\mathbb{N}^*)^n\,.
\end{equation}

We prove Theorem \ref{partition_function_result} in two steps. The first step is the following proposition, which asserts that the contribution to the right-hand side of \eqref{Ginibre_representation_2} of families of loops having the form $\vec{\omega}=(\omega_1,\ldots,\omega_n)$ with $\omega_i \in \Omega^{k\nu}$ for some $k\geq 2$ and $1\leq i \leq n$ vanishes in the limit as $\nu \rightarrow 0$.
\begin{proposition}[The contribution of loops of duration greater than $\nu$ to \eqref{Ginibre_representation_2} vanishes as $\nu \rightarrow 0$]
\label{Proposition_1}
The partition function \eqref{Ginibre_representation_2} satisfies
\begin{equation}
\label{Proposition_1_1}
\lim_{\nu \rightarrow 0} \bigl(\Xi^{\nu,\zeta}- \widehat{\Xi}^{\nu,\zeta} \bigr)=0\,,
\end{equation}
where
\begin{equation}
\label{Xi_hat}
\widehat{\Xi}^{\nu,\zeta}\coloneqq \sum_{n=0}^{\infty} \frac{z^{n\nu}}{n!}  \int \mathbb{W}^{\nu \vec{1}_n} (\dd \vec{\omega}) \, \exp \bigl(-\mathcal{V}^{\nu}(\vec{\omega})\bigr)\,,
\end{equation}
with notation as in \eqref{loop_interaction}--\eqref{loop_interaction_2}.
Moreover, we have
\begin{equation}
\label{Proposition_2_1_convergence_rate}
 \bigl|\Xi^{\nu,\zeta}- \widehat{\Xi}^{\nu,\zeta} \bigr| \lesssim_{d,\zeta,B} \nu^{d/2}\,\exp \bigl[C_1(d, \zeta,B) L^d\bigr]\,,
\end{equation}
for some $C_1(d,\zeta,B)>0$.
\end{proposition}

\begin{proof}
Using \eqref{Ginibre_representation_2}, \eqref{Xi_hat}, and recalling \eqref{vec_1_n}, we have
\begin{equation}
\label{Proposition_2_1_1}
\Xi^{\nu,\zeta} - \widehat{\Xi}^{\nu,\zeta}  = \sum_{n = 1}^{\infty} \frac{1}{n!} \sum_{\substack{\vec{k} \in (\N^{*})^n \\ \vec{k} \neq \vec{1}_{n}}} \prod_{i = 1}^{n} \frac{z^{k_{i}\nu}}{k_{i}} \int \mathbb{W}^{\nu \vec{k}} (\dd \vec{\omega}) \, \exp \left(-\mathcal{V}^{\nu}(\vec{\omega})\right).
\end{equation}
Applying Lemma~\ref{lemma_loop_interactions} in \eqref{Proposition_2_1_1}, it follows that
\begin{equation}
\label{Proposition_2_1_2}
| \Xi^{\nu,\zeta} - \widehat{\Xi}^{\nu,\zeta} | \leq \sum_{n = 1}^{\infty} \frac{1}{n!} \sum_{\substack{\vec{k} \in (\N^{*})^n \\ \vec{k} \neq \vec{1}_{n}}} \prod_{i = 1}^{n} \frac{z^{k_{i}\nu}}{k_{i}} \int \mathbb{W}^{\nu \vec{k}} (\dd \vec{\omega})\exp\!\left(B\sum_{i=1}^nk_i\right).
\end{equation}
We recall \eqref{fugacity_condition} and define
\begin{equation}
\label{definition_of_z_tilde}
\tilde{z} \coloneqq z\,\ee^{B/\nu}=(\zeta \nu^{d/2} \mathrm{e}^B)^{1/\nu}\,,
\end{equation}
and observe that for all $k \in \N^*$, we have
\begin{equation}
\label{def_of_z_tilde_consequence}
\tilde{z}^{k \nu} = z^{k \nu}\mathrm{e}^{Bk}\,. 
\end{equation}
Using \eqref{def_of_z_tilde_consequence}, we can rewrite the right-hand side of \eqref{Proposition_2_1_2} as 
\begin{equation}
\label{(2.16)}
	\sum_{n = 1}^{\infty} \frac{1}{n!} \sum_{\substack{\vec{k} \in (\N^{*})^n \\ \vec{k} \neq \vec{1}_{n}}} \prod_{i = 1}^{n} \frac{\tilde{z}^{k_{i}\nu}}{k_{i}} \int \mathbb{W}^{\nu \vec{k}} (\dd \vec{\omega})\,.
\end{equation}
Recall \eqref{W_product_2} and note that there is at least one factor in \eqref{(2.16)} where $k_{i} \geq 2$. This allows us to write
\begin{multline}
	\eqref{(2.16)} \leq \sum_{n = 1}^{\infty} \frac{n}{n!} \left( \sum_{\ell=2}^{\infty} \frac{\tilde{z}^{\ell \nu}}{\ell} \int_{\Lambda} \dd x \,\int \mathbb{W}_{x, x}^{\ell \nu}(\dd \omega) \right) \\
					    \mathrel{\phantom{=}} \times \left( \tilde{z}^{\nu} \int_{\Lambda} \dd x\,\int \mathbb{W}_{x, x}^{\nu} (\dd \omega) + \sum_{\ell=2}^{\infty} \frac{\tilde{z}^{\ell \nu}}{\ell} \int_{\Lambda} \dd x \int \mathbb{W}_{x, x}^{\ell \nu}(\dd \omega) \right)^{n-1}\,. \label{fin-vol-telescoping-bound}
\end{multline}
We now proceed to bound the right-hand side of \eqref{fin-vol-telescoping-bound}. 
Using \eqref{heat_kernel_integral} and Lemma \ref{heat_kernel_estimates} (i), followed by \eqref{fugacity_condition} and \eqref{def_of_z_tilde_consequence}, we have that
\begin{multline}
\label{fin-vol-two-sums-split_A}
\sum_{\ell=2}^{\infty} \frac{\tilde{z}^{\ell \nu}}{\ell} \int_{\Lambda} \dd x \int \mathbb{W}_{x, x}^{\ell \nu}(\dd \omega) = \sum_{\ell=2}^{\infty} \frac{\tilde{z}^{\ell \nu}}{\ell} \psi^{\ell \nu}(0) L^d \lesssim_d \sum_{\ell=2}^{\infty} \frac{\tilde{z}^{\ell \nu}}{\ell} \left(\frac{1}{L^d} + \frac{1}{(\ell \nu)^{d/2}} \right) L^d \\= \sum_{\ell=2}^{\infty} \frac{(\zeta \mathrm{e}^B)^{\ell} \,\nu^{d \ell/2}}{\ell} + \sum_{\ell=2}^{\infty} \frac{1}{\ell^{1 + d/2}} \,(\zeta \mathrm{e}^B)^{\ell} \,\nu^{d(\ell - 1)/2}\,L^d\,. 
\end{multline}
Since $\ell \geq 2$ in both sums on the right-hand side of \eqref{fin-vol-two-sums-split_A}, this expression is
\begin{equation}
\label{fin-vol-two-sums-split}
\leq \sum_{\ell=2}^{\infty} \left( \zeta \mathrm{e}^B \nu^{d/2} \right)^{\ell} + \sum_{\ell=2}^{\infty} \frac{1}{\ell^{1 + d/2}} \left( \zeta\mathrm{e}^B \nu^{d/4} \right)^{\ell}\,L^d
\eqqcolon (*) + (**)\,.
\end{equation}
We choose $\nu>0$ small enough so that 
\begin{equation}
\label{nu_small}
\zeta \mathrm{e}^B\nu^{d/2} \leq 1/2 \quad \text{and} \quad \zeta \mathrm{e}^B \nu^{d/4} \leq 1\,.
\end{equation}
Recalling \eqref{fin-vol-two-sums-split} and using \eqref{nu_small}, we have
\begin{equation}
	(*) = \sum_{\ell=2}^{\infty} \left(\zeta\mathrm{e}^B \nu^{d/2} \right)^{\ell} 
		 = \frac{\left( \zeta\mathrm{e}^B \nu^{d/2} \right)^{2}}{1 - \zeta\mathrm{e}^B \nu^{d/2}} 
		 \leq 2 \zeta^{2}\mathrm{e}^{2B} \nu^{d} \label{fin-vol-star1}
\end{equation}
and
\begin{equation}
	(**) \leq \zeta^{2}\mathrm{e}^{2B} \nu^{d/2} \,\sum_{\ell=2}^{\infty} \frac{1}{\ell^{1 + d/2}}\,L^d
		 \lesssim_{d}\zeta^{2} \mathrm{e}^{2B} \nu^{d/2}\,L^d\,. \label{fin-vol-star2}
\end{equation}
Substituting \eqref{fin-vol-star1}--\eqref{fin-vol-star2} into \eqref{fin-vol-two-sums-split_A}--\eqref{fin-vol-two-sums-split} gives us
\begin{align}
	\sum_{\ell=2}^{\infty} \frac{\tilde{z}^{\ell \nu}}{\ell} \int_{\Lambda} \dd x \int \mathbb{W}_{x, x}^{\ell \nu}(\dd \omega) \lesssim_{d} \zeta^{2} \mathrm{e}^{2B}\nu^{d} + \zeta^{2}\mathrm{e}^{2B} \nu^{d/2}\,L^d \leq C(d,\zeta,B)\, \nu^{d/2}\,L^d\,. \label{fin-vol-bound1}
\end{align}
By analogous arguments, we have 
\begin{equation}
\tilde{z}^{\nu} \int_{\Lambda}  \dd x\, \int \mathbb{W}_{x, x}^{\nu} (\dd \omega) = \tilde{z}^{\nu} \psi^{\nu}(0) L^d \lesssim_d \tilde{z}^{\nu} \left( \frac{1}{L^d} + \frac{1}{\nu^{d/2}} \right) L^d= \zeta\mathrm{e}^B \bigl(\nu^{d/2} + L^d \bigr)\,. \label{fin-vol-bound2}
\end{equation}
Substituting \eqref{fin-vol-bound1}--\eqref{fin-vol-bound2} into~\eqref{fin-vol-telescoping-bound} gives us
\begin{align}
\label{zeta-zeta_hat_1}
		| \Xi^{\nu,\zeta} - \widehat{\Xi}^{\nu,\zeta} | & \leq \sum_{n = 1}^{\infty} \frac{1}{(n - 1)!}\, C(d, \zeta, B) \,\nu^{d/2}\, L^d\,\left[ C(d, \zeta,B) \nu^{d/2} L^d + \zeta \mathrm{e}^B L^d +\zeta \ee^{B} \,\nu^{d/2} \right]^{n - 1} \\
					    & 
					    \label{zeta-zeta_hat_1B}
					    \leq  C(d, \zeta,B) \,\nu^{d/2} \,L^d\,\exp \left[ \tilde{C}(d, \zeta,B) L^d +\zeta \ee^{B} \right] 
\\
&\lesssim_{d,\zeta,B} \nu^{d/2}\,L^d\,\exp \left[\tilde{C}(d, \zeta,B)  L^d\right]
\,. \label{zeta-zeta_hat}
\end{align}
In \eqref{zeta-zeta_hat_1B}--\eqref{zeta-zeta_hat}, we took $\tilde{C}(d, \zeta,B)\coloneqq  C(d, \zeta,B) + \zeta\mathrm{e}^B$, following \eqref{fin-vol-star1}--\eqref{fin-vol-star2}.
For fixed $\zeta > 0, B\geq 0$, and $L \geq 1$, the right-hand side of~\eqref{zeta-zeta_hat} tends to $0$ as $\nu \to 0$. We thus obtain \eqref{Proposition_1_1}. Furthermore, we obtain the convergence rate given by \eqref{Proposition_2_1_convergence_rate} by taking $C_1(d,\zeta,B)>\tilde{C}(d,\zeta,B)$. This concludes the proof.

\end{proof}

The second step in proving Theorem \ref{partition_function_result} consists in noting that the loop components of $\vec{\omega}$ in \eqref{Xi_hat} collapse to points when $\nu \rightarrow 0$, thus giving \eqref{partition_function_result_2} in the limit. We show this in a quantitative way in the following proposition.

\begin{proposition}[Letting $\nu \rightarrow 0$ in \eqref{Xi_hat}, the loops collapse to points]
\label{Proposition_2}
Let $\alpha \in (0,1]$ be given and let $v$ be as in Assumption \ref{Assumption_on_v}.
Recalling \eqref{partition_function_result_2} and \eqref{Xi_hat}, we have that 
\begin{equation}
\label{Proposition_2_1}
\lim_{\nu \rightarrow 0} \widehat{\Xi}^{\nu,\zeta}=\Xi^{0,\zeta}\,.
\end{equation}
Moreover, we have
\begin{equation}
\label{Proposition_2_convergence_rate}
\bigl|\widehat{\Xi}^{\nu,\zeta}- \Xi^{0,\zeta} \bigr| \lesssim_{d,\zeta,v,\alpha,B} \nu^{\alpha/2}\, \exp\left[C_2(d,\zeta,B) L^d\right],
\end{equation}
for some $C_2(d,\zeta,B)>0$.
\end{proposition}

Unlike the proof of Proposition \ref{Proposition_1}, which relied only on Assumption \ref{Assumption_on_v} (i) for the interaction potential $v$ (i.e.\ on its stability), the proof of Proposition \ref{Proposition_2} additionally uses the H\"{o}lder continuity properties of $v$ given by Assumption \ref{Assumption_on_v} (ii). Before proceeding with the proof of Proposition \ref{Proposition_2}, we note two facts about Brownian motion which are given by Lemmas \ref{Lemma_2} and \ref{Lemma_3} below. We first observe the following quantitative estimate for Brownian motion on $\Lambda$.
\begin{lemma} 
\label{Lemma_2}
Let $\alpha \in (0,1]$ be given. Let the function $g \colon \Lambda \times \Lambda \rightarrow \mathbb{C}$ be $\alpha$-H\"{o}lder continuous in the sense that there exists $C>0$ such that for all $(p,q), (\tilde{p},\tilde{q}) \in \Lambda \times \Lambda$, we have 
\begin{equation}
\label{g_Holder_continuous}
|g(p,q)-g(\tilde{p},\tilde{q})| \leq C\bigl(|p-\tilde{p}|_{L}^{\alpha}+|q-\tilde{q}|_{L}^{\alpha}\bigr)\,.
\end{equation}
Here, we recall \eqref{periodic_Euclidean_norm}. 
Then, uniformly in $x,y \in \Lambda$ and $t \in [0,\nu]$, we have that 
\begin{multline}
\label{Lemma_2_1}
z^{2\nu} \int \mathbb{W}^{\nu}_{x,x}(\dd \omega_1)\,\int \mathbb{W}^{\nu}_{y,y}(\dd \omega_2)\,g \bigl(\omega_1(t),\omega_2(t)\bigr)
\\
=
g(x,y)\, z^{2\nu} \int \mathbb{W}^{\nu}_{x,x}(\dd \omega_1)\,\int \mathbb{W}^{\nu}_{y,y}(\dd \omega_2)+\mathcal{O}_{d,\zeta,g,\alpha}\bigl(\nu^{\alpha/2}\bigr)\,,
\end{multline}
where the $g$ dependence of the error term in \eqref{Lemma_2_1} is given in terms of the $\alpha$-H\"{o}lder constant $C$ in \eqref{g_Holder_continuous}.
\end{lemma}

\begin{proof}
Fix $x,y\in \Lambda$ and $t \in [0,\nu]$. By \eqref{g_Holder_continuous}, we have 
\begin{multline}
\label{lipschitz_g_1}
    \biggl|z^{2\nu} \int \mathbb{W}^{\nu}_{x,x}(\dd \omega_1)\,\int \mathbb{W}^{\nu}_{y,y}(\dd \omega_2)\,\Bigl[g \bigl(\omega_1(t),\omega_2(t)\bigr)-g(x,y)\Bigr]\biggr| \lesssim_g
    \\
 z^{2\nu} \int \mathbb{W}^{\nu}_{x,x}(\dd \omega_1)\,\int \mathbb{W}^{\nu}_{y,y}(\dd \omega_2)\,|\omega_1(t)-x|_L^{\alpha}
 +  z^{2\nu} \int \mathbb{W}^{\nu}_{x,x}(\dd \omega_1)\,\int \mathbb{W}^{\nu}_{y,y}(\dd \omega_2)\,|\omega_2(t)-y|_L^{\alpha}\,.
\end{multline}
Let us analyse the two terms on the right-hand side of \eqref{lipschitz_g_1} separately. For the first term, we have
\begin{multline}\label{lipschitz_g_2}
z^{2\nu} \int \mathbb{W}^{\nu}_{x,x}(\dd \omega_1)\,\int \mathbb{W}^{\nu}_{y,y}(\dd \omega_2)\,\left|\omega_1(t)-x\right|_L^{\alpha}
        =\biggl(z^{\nu} \int \mathbb{W}^{\nu}_{x,x}(\dd \omega_1)\,\left|\omega_1(t)-x\right|_L^{\alpha}\biggr)\,
\\        
\times \biggl(z^{\nu}\int \mathbb{W}^{\nu}_{y,y}(\dd \omega_2)\biggr)\,.
\end{multline}
We estimate the second factor on the right-hand side of \eqref{lipschitz_g_2}, by noting that by \eqref{heat_kernel_integral}, Lemma \ref{heat_kernel_estimates} (i), and \eqref{fugacity_condition}, we have
\begin{equation}
\label{lipschitz_g_2b}
z^{\nu}\int \mathbb{W}^{\nu}_{y,y}(\dd \omega_2)\lesssim_d z^{\nu}\left(\frac{1}{L^d}+\frac{1}{\nu^{d/2}}\right)\lesssim_{d} \zeta\,,
\end{equation}
where we recalled that $L \geq 1$ by assumption.
In order to estimate the first factor on the right-hand side of \eqref{lipschitz_g_2}, we use H\"{o}lder's inequality, followed by Lemma \ref{heat_kernel_estimates} (iii), \eqref{fugacity_condition}, \eqref{lipschitz_g_2b} (with $y,\omega_2$ replaced with $x,\omega_1$ respectively), and the condition $t \in [0,\nu]$ to deduce that
\begin{multline}
\label{lipschitz_g_2c}
z^{\nu} \int \mathbb{W}^{\nu}_{x,x}(\dd \omega_1)\,\left|\omega_1(t)-x\right|_L^{\alpha} \leq \biggl(z^{\nu} \int \mathbb{W}^{\nu}_{x,x}(\dd \omega_1)\,\left|\omega_1(t)-x\right|_L \biggr)^{\alpha}\, \biggl(z^{\nu}\int \mathbb{W}^{\nu}_{x,x}(\dd \omega_1)\biggr)^{1-\alpha}
\\
\lesssim_{d,\zeta,\alpha}
\biggl[z^{\nu}\left(\frac{1}{L^d}+\frac{1}{\nu^{d/2}}\right)\sqrt{t}\,\biggr]^{\alpha} \lesssim_{d,\zeta,\alpha} \nu^{\alpha/2}\,.
\end{multline}
From \eqref{lipschitz_g_2}--\eqref{lipschitz_g_2c}, 
we obtain 
\begin{equation}
\label{lipschitz_g_3}
 z^{2\nu} \int \mathbb{W}^{\nu}_{x,x}(\dd \omega_1)\,\int \mathbb{W}^{\nu}_{y,y}(\dd \omega_2)\,|\omega_1(t)-x|_L^{\alpha}\lesssim_{d,\zeta,\alpha} \nu^{\alpha/2}\,.
\end{equation}
By symmetry, we have
\begin{equation}
\label{lipschitz_g_4}
z^{2\nu} \int \mathbb{W}^{\nu}_{x,x}(\dd \omega_1)\,\int \mathbb{W}^{\nu}_{y,y}(\dd \omega_2)\,|\omega_2(t)-y|_L^{\alpha}\lesssim_{d,\zeta,\alpha} \nu^{\alpha/2}\,.
 \end{equation}
Note that \eqref{lipschitz_g_3}--\eqref{lipschitz_g_4} hold uniformly in $x,y \in \Lambda$ and $t \in [0,\nu]$.
Using \eqref{lipschitz_g_3}--\eqref{lipschitz_g_4} in \eqref{lipschitz_g_1}, we obtain the claim.
\end{proof}

The following consequence of Lemma \ref{Lemma_2} will be useful in the proof of Proposition \ref{Proposition_2}.

\begin{lemma}
\label{Lemma_3} 
Let $\alpha \in (0,1]$ be given. For $v$ as in Assumption \ref{Assumption_on_v}, we have
\begin{equation}
\label{Lemma_3_1}
z^{2\nu} \int \mathbb{W}^{\nu}_{x,x}(\dd \omega_1)\,\int \mathbb{W}^{\nu}_{y,y}(\dd \omega_2) \, \frac{1}{\nu} \int_{0}^{\nu} \dd t \left| v(\omega_1(t)-\omega_2(t)) - v(x-y) \right| = \mathcal{O}_{d,\zeta,v,\alpha}(\nu^{\alpha/2})
\end{equation}
uniformly in $x, y \in \Lambda$.
\end{lemma}

\begin{proof}
By Fubini's theorem we have
\begin{multline}
\label{Lemma_3_2}
        z^{2\nu} \int \mathbb{W}^{\nu}_{x,x}(\dd \omega_1)\,\int \mathbb{W}^{\nu}_{y,y}(\dd \omega_2) \,\frac{1}{\nu} \int_0^{\nu} \dd t \,\bigl|v(\omega_1(t)-\omega_2(t))- v(x-y)\bigr|\\
        = \frac{1}{\nu} \int_0^{\nu} \dd t\,\left( z^{2\nu} \int \mathbb{W}^{\nu}_{x,x}(\dd \omega_1)\,\int \mathbb{W}^{\nu}_{y,y}(\dd \omega_2) \, \,\bigl|v(\omega_1(t)-\omega_2(t))- v(x-y)\bigr|\right).
\end{multline}
We define $g_{x,y} \colon\Lambda\times\Lambda\to [0,\infty)$ as 
\begin{equation}
\label{function_g}
g_{x,y}(p,q)\coloneqq |v(p-q)-v(x-y)|\,.
\end{equation}
By Assumption \ref{Assumption_on_v} (ii), we obtain that for all $(p,q), (\tilde{p},\tilde{q}) \in \Lambda \times \Lambda$
\begin{multline}
\label{function_g_2}
|g_{x,y}(p,q)-g_{x,y}(\tilde{p},\tilde{q})|\leq |v(p-q)-v(\tilde{p}-\tilde{q})| \lesssim_v |(p-q)-(\tilde{p}-\tilde{q})|_{L}^{\alpha}
\\
\leq 
\bigl(|p-\tilde{p}|_{L}+ |q-\tilde{q}|_{L}\bigr)^{\alpha}
\lesssim_{\alpha} |p-\tilde{p}|_{L}^{\alpha}+ |q-\tilde{q}|_{L}^{\alpha}\,.
\end{multline}
From \eqref{function_g_2}, it follows that the function \eqref{function_g} is $\alpha$-H\"{o}lder continuous in the sense of \eqref{g_Holder_continuous} above, uniformly in $x,y \in \Lambda$. (By the latter, we mean that the H\"{o}lder constant $C>0$ in \eqref{g_Holder_continuous} can be chosen uniformly in $x,y \in \Lambda$.)

Using \eqref{function_g}--\eqref{function_g_2}, and Lemma \ref{Lemma_2}, we write
\begin{align}
\eqref{Lemma_3_2} & =\frac{1}{\nu} \int_0^{\nu} \dd t\,\left( z^{2\nu} \int \mathbb{W}^{\nu}_{x,x}(\dd \omega_1)\,\int \mathbb{W}^{\nu}_{y,y}(\dd \omega_2) \, \,\bigl|g_{x,y}(\omega_1(t),\omega_2(t))\bigr|\right)\\
		  & =\frac{1}{\nu} \int_0^{\nu} \dd t\,\left( z^{2\nu} \int \mathbb{W}^{\nu}_{x,x}(\dd \omega_1)\,\int \mathbb{W}^{\nu}_{y,y}(\dd \omega_2) \, \,\left|g_{x,y}(x,y)\right|\right)+\mathcal{O}_{d,\zeta,v,\alpha}(\nu^{\alpha/2})\,. \label{Lemma_3_3}
\end{align}
Observe that $g_{x,y}(x,y)=0$ by~\eqref{function_g_2}, and hence we obtain \eqref{Lemma_3_1} from \eqref{Lemma_3_3}.

\end{proof}

We now have the necessary tools to prove Proposition \ref{Proposition_2}.

\begin{proof}[Proof of Proposition \ref{Proposition_2}]
We define 
\begin{equation}
\label{Proposition_2_2}
    \widetilde{\Xi}^{\nu,\zeta}\coloneqq \sum_{n=0}^{\infty} \frac{z^{n\nu}}{n!} \int_{\Lambda^n} \dd x_1\,\cdots \,\dd x_n\, \int \mathbb{W}^{\nu}_{x_1,x_1}(\dd \omega_1)\,\cdots\,\mathbb{W}^{\nu}_{x_n,x_n}(\dd \omega_n)\,\exp\Biggl(-\sum_{1\leq i<j \leq n} v(x_i-x_j)\Biggr)\,.
\end{equation}
We note that $\widetilde{\Xi}^{\nu,\zeta}$ defined in \eqref{Proposition_2_2} differs from $\widehat{\Xi}^{\nu,\zeta}$ defined in \eqref{Xi_hat} in the exponential term. Namely, the interaction $\mathcal{V}^{\nu}(\vec{\omega})$ of the loops $\omega_i \in \Omega^{\nu}_{x_i,x_i}$ is replaced by the interaction $\mathcal{V}(\vec{x})$ of their base points $x_i$, given by \eqref{interaction_function_points}. 
Let us first show that
\begin{equation}
\label{Proposition_2_3_a}
\widehat{\Xi}^{\nu,\zeta}=\widetilde{\Xi}^{\nu,\zeta}+\mathcal{O}_{d,\zeta,v,\alpha,B}\left(\nu^{\alpha/2}\,\exp(C_{2,1}(d,\zeta,B)\,L^d)\right),
\end{equation}
for suitable $C_{2,1}(d,\zeta,B)>0$.
We start by observing that for $a, b \in \R$ with $a,b \leq M$ we have 
\begin{equation}\label{MVT_exponential}
  |\ee^{a} - \ee^{b}| \leq \ee^{M} |a-b|\,.
\end{equation}
The estimate \eqref{MVT_exponential} is a consequence of the mean value theorem. Observe that Assumption~\ref{Assumption_on_v} (i) implies that for all $(\omega_1,\ldots,\omega_n)$ with $\omega_i \in \Omega^{\nu}$ for all $i=1,\ldots,n$ and for all $t \in [0,\nu]$, we have 
\begin{equation}
\label{assumption_v_i_on_loops_A}
    \sum_{1\leq i<j\leq n}v(\omega_i(t)-\omega_j(t))\ge -Bn\,,
\end{equation}
and therefore by taking averages in $t \in [0,\nu]$ in \eqref{assumption_v_i_on_loops_A}, we have 
\begin{equation}
\label{assumption_v_i_on_loops}
    \sum_{1\leq i<j\leq n}\frac{1}{\nu}\int_0^\nu \dd t\, v(\omega_i(t)-\omega_j(t))\ge -Bn\,.
\end{equation}
Combining \eqref{stable_potential_condition}, \eqref{assumption_v_i_on_loops}, and \eqref{MVT_exponential} with $M=Bn$, and using the triangle inequality, we obtain
\begin{multline}\label{telescoping1}
    \left|\exp\Biggl(-\sum_{1\leq i<j \leq n} \frac{1}{\nu}\int_0^\nu \dd t\, v(\omega_i(t)-\omega_j(t))\Biggr)-\exp\Biggl(-\sum_{1\leq i<j \leq n} v(x_i-x_j)\Biggr)\right|\\
    \leq \ee^{Bn} \sum_{1\leq i<j \leq n} \frac{1}{\nu} \int_{0}^{\nu} \dd t \left| v(\omega_i(t)-\omega_j(t)) - v(x_i-x_j) \right|.
\end{multline}
Consequently,
\begin{multline}
\label{Proposition_2_3}
|\widehat{\Xi}^{\nu,\zeta}-\widetilde{\Xi}^{\nu,\zeta}|\leq \sum_{n=0}^{\infty} \frac{z^{n \nu}}{n!} \,\ee^{Bn} \sum_{1 \leq i<j \leq n}\int_{\Lambda^n}\, \dd x_1\,\cdots \,\dd x_n\, \int \mathbb{W}^{\nu}_{x_1,x_1}(\dd \omega_1)\,\cdots\,\mathbb{W}^{\nu}_{x_n,x_n}(\dd \omega_n)\\
        \times \frac{1}{\nu} \int_{0}^{\nu} \dd t \left| v(\omega_i(t)-\omega_j(t)) - v(x_i-x_j) \right|.
\end{multline}
Fix $n\in \mathbb{N}^*$ and $i,j \in \{1,\dots,n\}$ with $i<j$. For any $k\neq i,j$, as in \eqref{lipschitz_g_2b}, we have that
\begin{equation}
\label{Proposition_2_4_a}
z^{\nu}\int_\Lambda \dd x_k\int \mathbb{W}_{x_k,x_k}^{\nu}(\dd \omega_k) \lesssim_{d,\zeta} L^d\,.
\end{equation}
Moreover, according to Lemma~\ref{Lemma_3}, we have
\begin{multline}
\label{Proposition_2_4}
z^{2\nu} \int_{\Lambda^2} \dd x_i\, \dd x_j\, \int \mathbb{W}^{\nu}_{x_i,x_i}(\dd \omega_i)\,\int \mathbb{W}^{\nu}_{x_j,x_j}(\dd \omega_j)\, \frac{1}{\nu} \int_{0}^{\nu} \dd t \left| v(\omega_i(t)-\omega_j(t)) - v(x_i-x_j) \right|
\\
\lesssim_{d,\zeta,v,\alpha} \nu^{\alpha/2}\,L^{2d}\,.
\end{multline}
Substituting \eqref{Proposition_2_4_a}--\eqref{Proposition_2_4} into \eqref{Proposition_2_3}, we obtain for suitable $C_{d,\zeta}>0$ the following estimate.
\begin{equation}
\label{Proposition_2_5}
\bigl|\widehat{\Xi}^{\nu,\zeta}-\widetilde{\Xi}^{\nu,\zeta}\bigr|\lesssim_{v,\alpha} \sum_{n=0}^{\infty} \frac{1}{n!} \, \ee^{Bn} \,C_{d,\zeta}^{n-2}\,n(n-1)\, \nu^{\alpha/2}\,(L^d)^{n}\lesssim_{d,\zeta,v,\alpha,B} \nu^{\alpha/2}\,\exp \left[\ee^B\,C_{d,\zeta}\,L^d\right].
\end{equation}
We hence deduce \eqref{Proposition_2_3_a} with $C_{2,1}(d,\zeta,B)=\ee^{2B}\,C_{d,\zeta}$ from \eqref{Proposition_2_5}.

Next, we show that
\begin{equation}
\label{Proposition_2_convergence_rate*} 
\widetilde{\Xi}^{\nu,\zeta}=\Xi^{0,\zeta}+\mathcal{O}_{d,\zeta,B} \left(\nu^{d/2}\exp(C_{2,2}(d,\zeta,B) L^{d})\right),
\end{equation}
for suitable $C_{2,2}(d,\zeta,B)>0$.
We start by using \eqref{heat_kernel_integral} to rewrite \eqref{Proposition_2_2} as 
\begin{equation}
\label{Proposition_2_2_B}
    \widetilde{\Xi}^{\nu,\zeta}=\sum_{n=0}^{\infty} \frac{(z^{\nu} \psi^{\nu}(0))^n}{n!} \int_{\Lambda^n} \dd x_1\,\cdots \,\dd x_n\,\exp\Biggl(-\sum_{1\leq i<j \leq n} v(x_i-x_j)\Biggr)\,.
\end{equation}
Recalling~\eqref{heat_kernel_estimates_tilde_A}, followed by~\eqref{fugacity_condition}, we can write
\begin{equation}
\label{eta_1}
z^{\nu} \psi^{\nu}(0) = z^{\nu} \eta^{\nu} + \frac{\zeta}{(2\pi)^{d/2}} \,.
\end{equation}
Using \eqref{partition_function_result_2}, \eqref{Proposition_2_2_B}--\eqref{eta_1}, and Assumption \ref{Assumption_on_v} (i), we have
\begin{equation}
|\widetilde{\Xi}^{\nu,\zeta}-\Xi^{0,\zeta}|
\leq \sum_{n=0}^{\infty}\, \frac{\ee^{Bn}L^{dn}}{n!} \,\left|\left(z^{\nu}\eta^{\nu} + \frac{\zeta}{(2\pi)^{d/2}}\right)^n-\left(\frac{\zeta}{(2\pi)^{d/2}}\right)^n\right|\,.
\end{equation}
We now use the mean-value theorem for the function $x \mapsto x^{n}$ in the $n$-th term of the sum above, for every $n \in \N^*$. Note that each derivative is increasing on the nonnegative real line. We therefore get
\begin{align}
|\widetilde{\Xi}^{\nu,\zeta}-\Xi^{0,\zeta}|
& \leq \sum_{n=0}^{\infty}\, \frac{\ee^{Bn}L^{dn}}{n!} \, n \, z^{\nu} \eta^{\nu} \left(z^{\nu}\eta^{\nu} + \frac{\zeta}{(2\pi)^{d/2}}\right)^{n-1}
\\
& = \ee^{B} L^{d} z^{\nu} \eta^{\nu} \exp\left(\ee^{B} L^{d} \left( z^{\nu} \eta^{\nu} + \frac{\zeta}{(2\pi)^{d/2}} \right) \right) \,.
\label{Proposition_2_6}
\end{align}
By Lemma~\ref{heat_kernel_estimates}~(ii) we have that $L^{d} \eta^{\nu} \leq C(d)$ for some constant $C(d) > 0$ uniformly in $\nu$. Applying this and \eqref{fugacity_condition} in \eqref{Proposition_2_6} we get
\begin{equation}
	|\widetilde{\Xi}^{\nu,\zeta}-\Xi^{0,\zeta}| \leq C(d) \ee^{B} \zeta \nu^{d/2} \exp\!\left(\!C(d)\ee^{B}\zeta\nu^{d/2} + \ee^{B}L^{d}\frac{\zeta}{(2\pi)^{d/2}}\!\right) \lesssim_{d, \zeta, B} \nu^{d/2} \exp(C_{2,2}(d,\zeta,B) L^{d}) \,.
\end{equation}

Let us note that \eqref{Proposition_2_convergence_rate} now follows from \eqref{Proposition_2_3_a} and \eqref{Proposition_2_convergence_rate*} with 
\begin{equation*}
C_2(d,\zeta,B)=\max\left\{C_{2,1}(d,\zeta,B),C_{2,2}(d,\zeta,B)\right\}.
\end{equation*}
In particular, claim \eqref{Proposition_2_1} follows.


\end{proof}

We can now prove Theorem \ref{partition_function_result}.
\begin{proof}[Proof of Theorem \ref{partition_function_result}]
The result follows immediately from Proposition \ref{Proposition_1} and Proposition \ref{Proposition_2}. The bound on the rate of convergence given by \eqref{partition_function_result_3} is deduced from \eqref{Proposition_2_1_convergence_rate} and \eqref{Proposition_2_convergence_rate} with 
\begin{equation*}
C_0(d,\zeta,B)=\max\left\{C_1(d,\zeta,B),C_2(d,\zeta,B)\right\}.
\end{equation*}
\end{proof}

\begin{remark}[The grand-canonical partition function of the Boltzmann gas]
	\label{Remark_Boltzmann_gas_partition_function}
	In Boltzmann statistics, the $n$-particle Hamiltonian acts on the entire $L^2(\Lambda^n)$, not just the subspace of symmetric functions. In particular, one replaces $P_{n}^{+}$ in \eqref{n_particle_space} with a multiplication operator $M_{n} \colon f \mapsto \frac{1}{n!} f$. See \cite[Section 2]{Ginibre} for more details. 
When deriving the Ginibre representation of the Bose gas partition function in \eqref{Ginibre_representation_1}, permutations in \eqref{symmetric_projector} are decomposed into cycles. In the Boltzmann ensemble, we only have the identity permutation, which decomposes into cycles of length $1$. Therefore, the grand-canonical partition function of a Boltzmann gas is
\begin{equation}
\label{Boltzmann_partition_function}
\Xi^{\nu,\zeta;\Bolt}=\widehat{\Xi}^{\nu,\zeta}\,,
\end{equation}
where $\widehat{\Xi}^{\nu,\zeta}$ is defined in~\eqref{Xi_hat}. In particular, the analogue of Proposition \ref{Proposition_1} automatically holds for the Boltzmann statistics. We then argue as in the Bose gas regime to deduce that the analogue of Theorem \ref{partition_function_result} holds for Boltzmann gases:
	\begin{equation*}
		\label{partition_function_result_3_Boltzmann}
		\bigl|\Xi^{\nu,\zeta;\Bolt}- \Xi^{0,\zeta} \bigr| \lesssim_{d,\zeta,v,\alpha,B} \nu^{\alpha/2}\,\exp \left [C(d,\zeta,B) L^d \right]
	\end{equation*}
	for some $C(d,\zeta,B)>0$, and consequently,
	\begin{equation*}
		\lim_{\nu \rightarrow 0} \Xi^{\nu,\zeta;\Bolt}=\Xi^{0,\zeta}\,. 
	\end{equation*}
\end{remark}

\section{Convergence of $p$-particle reduced density matrices in finite volume: proof of Theorem~\ref{density_matrix_result} }
\label{section_density_matrix_Bose}

In this section, we prove Theorem \ref{density_matrix_result} on the convergence of the reduced $p$-particle density matrices in finite volume. The starting point is the Ginibre representation from Lemma \ref{Ginibre_representation}  (ii). Before proceeding with the analysis, let us first reformulate the latter result.
We henceforth fix $p \in \N^*$. 
Using \eqref{W_product}, we rewrite the Ginibre representation of the kernel of the reduced $p$-particle density matrix 
given by \eqref{Ginibre_representation_4} as 
\begin{equation}
\label{ploopkernel}
    \Gamma_p^{\nu,\zeta}(\vec{x},\vec{y})
    =\sum_{\pi \in S_p}\sum_{\vec{k}\in (\N^*)^p}\prod_{i=1}^p z^{k_i \nu}\int \mathbb{W}^{\nu \vec{k}}_{\pi \vec{y},\vec{x}}(\dd \vec{\omega})\,\gamma_p^{\nu,\zeta}\,(\vec{\omega})\,,
\end{equation}
where 
    \begin{equation}\label{p_loop_correlation_function}
        \gamma_p^{\nu,\zeta}\,(\vec{\omega})\coloneqq \frac{1}{\Xi^{\nu,\zeta}}\sum_{n=0}^{\infty} \frac{1}{n!} \sum_{\vec{k} \in (\N^{*})^n} \prod_{i=1}^{n} \frac{z^{k_i\nu}}{k_i} \int \mathbb{W}^{\nu \vec{k}} (\dd \vec{\widetilde{\omega}}) \, \exp \bigl(-\mathcal{V}^{\nu}(\vec{\widetilde{\omega}}\vec{\omega})\bigr)\,.
    \end{equation}
In \eqref{ploopkernel}, we recall the convention \eqref{pi_vec_y}.    
In \eqref{p_loop_correlation_function}, we recall \eqref{grand_canonical_ensemble_1}, \eqref{loop_interaction_2}, and \eqref{concatenation_omega}. We refer to $\gamma_p^{\nu,\zeta}(\vec{\omega})$ defined in \eqref{p_loop_correlation_function} as the \textbf{\emph{$p$-loop correlation function}}. Let us note that, on the right-hand side of \eqref{ploopkernel}, for $i=1,\ldots,p$ and fixed $\pi \in S_p$, the $\omega_i$ are open paths of duration $k_i \nu$ joining $x_i$ to $y_{\pi(i)}$.

Let us now proceed with the proof Theorem~\ref{density_matrix_result}. We follow the pattern of the proof of Theorem~\ref{partition_function_result}. As in the analysis of the grand-canonical partition function $\Xi^{\nu,\zeta}$ (given in Proposition \ref{Proposition_1}), we eliminate from the Ginibre representation of $\Gamma_p^{\nu,\zeta}(\vec{x},\vec{y})$ loops and open paths where at least one component has duration greater than $\nu$. This is achieved in two steps: In Proposition \ref{Proposition_4_3_A} below, we eliminate closed loops $\tilde{\omega}_i$ of duration greater than $\nu$; then, in Proposition \ref{Proposition_4_3} below, we eliminate open paths $\omega_i$ of duration greater than $\nu$. 

Let us now give the precise details of the analysis.
The first step is to show that only closed paths of duration $\nu$ give the leading contribution in \eqref{ploopkernel}--\eqref{p_loop_correlation_function} as $\nu \rightarrow 0$.
 \begin{proposition}[Eliminating loops of duration greater than $\nu$]
 \label{Proposition_4_3_A}
 Let $p \in \N^*$ be given. For $\vec{\omega}=(\omega_1,\ldots,\omega_p) \in \Omega^{\vec{T}}$, such that $T_i=T(\omega_i) \in \nu \N^{*}$ for all $i \in \{1,\ldots,p\}$, we define
    \begin{equation}\label{gamma_hat_definition}
        \widehat{\gamma}_p^{\nu,\zeta}(\vec{\omega})\coloneqq \frac{1}{\Xi^{\nu,\zeta}}\sum_{n=0}^{\infty} \frac{z^{n \nu}}{n!}  \int \mathbb{W}^{\nu \vec{1}_n} (\dd \vec{\widetilde{\omega}}) \, \exp \bigl(-\mathcal{V}^{\nu}(\vec{\widetilde{\omega}}\vec{\omega})\bigr)\,.
    \end{equation}
    Using \eqref{gamma_hat_definition}, we let
     \begin{equation}
     \label{Gamma_hat_nu_zeta}
         \widehat{\Gamma}_p^{\nu,\zeta}(\vec{x},\vec{y})\coloneqq \sum_{\pi \in S_p}\sum_{\vec{k}\in (\N^{*})^p}\prod_{i=1}^p z^{k_i \nu}\int \mathbb{W}^{\nu \vec{k}}_{\pi \vec{y}, \vec{x}}(\dd \vec{\omega})\,\widehat{\gamma}_p^{\nu,\zeta}\,(\vec{\omega})\,.
     \end{equation}
     Then, 
     \begin{equation}\label{Gamma_minus_Gamma_hat_estimate}
         |\Gamma_p^{\nu,\zeta}(\vec{x},\vec{y})-\widehat{\Gamma}_p^{\nu,\zeta}(\vec{x},\vec{y})|\lesssim_{p,d,\zeta,v,\alpha,B,L} \nu^{d/2}\,.
     \end{equation}
 \end{proposition}
 
 For the second step in the analysis, we eliminate open paths from $\vec{x}$ to $\pi\vec{y}$ with at least one component of duration greater than $\nu$. The interaction $\mathcal{V}^{\nu}(\vec{\widetilde{\omega}}\vec{\omega})$ includes vectors of closed loops $\vec{\widetilde{\omega}}$, and vectors of open paths $\vec{\omega}$ linking $\vec{x}$ to $\pi \vec{y}$. In order to explicitly analyse the open path dependence, we use \eqref{gamma_hat_definition} and \eqref{Gamma_hat_nu_zeta} to write
\begin{equation}\label{Gamma_hat_expansion}
    \widehat{\Gamma}_{p}^{\nu,\zeta}(\vec{x},\vec{y})=
    \frac{1}{\Xi^{\nu,\zeta}}\sum _{\pi \in S_p}\sum_{\vec{k}\in (\N^{*})^p}\sum_{n=0}^\infty\frac{z^{n\nu}}{n!}\prod_{i=1}^p z^{k_i \nu}
\int \mathbb{W}^{\nu \vec{1}_n} (\dd \widetilde{\vec{\omega}})\,
\int \mathbb{W}^{\nu \vec{k}}_{\pi \vec{y},\vec{x}}(\dd \vec{\omega})
\, \exp \bigl(-\mathcal{V}^{\nu}(\vec{\widetilde{\omega}}\vec{\omega})\bigr)\,.
\end{equation}
The next step is to show that only open paths of duration $\nu$ give the leading contribution in \eqref{Gamma_hat_expansion} as $\nu \rightarrow 0$.
\begin{proposition}[Eliminating open paths of duration greater than $\nu$]
\label{Proposition_4_3}
    Let $p \in \N^*$ be given. Define 
    \begin{equation}\label{Gamma_check_definition}
    \widecheck{\Gamma}_{p}^{\nu,\zeta}(\vec{x},\vec{y})\coloneqq 
    \frac{1}{\Xi^{\nu,\zeta}}\sum_{\pi \in S_p}\sum_{n=0}^\infty\frac{z^{(n+p) \nu}}{n!}
    \int \mathbb{W}^{\nu \vec{1}_n}(\dd \vec{\widetilde{\omega}}) 
    \int \mathbb{W}^{\nu \vec{1}_p}_{\pi \vec{y}, \vec{x}}(\dd \vec{\omega}) 
   \exp \bigl(-\mathcal{V}^{\nu}(\vec{\widetilde{\omega}}\vec{\omega}))\,.
\end{equation}
Then 
\begin{equation}\label{proposition_4_3_1}
    |\widehat{\Gamma}_{p}^{\nu,\zeta}(\vec{x},\vec{y})-\widecheck{{\Gamma}}_{p}^{\nu,\zeta}(\vec{x},\vec{y})|\lesssim_{p,d,\zeta,v,\alpha,B,L} \nu^{d/2}\,.
\end{equation}
\end{proposition}

We now prove Propositions \ref{Proposition_4_3_A} and \ref{Proposition_4_3}.
Before proceeding to the proof of Proposition \ref{Proposition_4_3_A}, we note the following related quantitative result comparing the $p$-loop correlation functions \eqref{p_loop_correlation_function} with the quantity defined in \eqref{gamma_hat_definition}.


\begin{lemma}
\label{Proposition4_1}
 Let $p \in \N^*$ be given. Given $\vec{\omega}=(\omega_1,\ldots,\omega_p) \in \Omega^{\vec{T}}$, such that $T_i=T(\omega_i) \in \nu \N^{*}$ for all $i \in \{1,\ldots,p\}$, we have
    \begin{equation}
    \label{Proposition4_1_bound}
        |\gamma_p^{\nu,\zeta}(\vec{\omega})-\widehat{\gamma}_p^{\nu,\zeta}(\vec{\omega})|\lesssim_{d,\zeta,v,\alpha,B,L}\nu^{d/2}\exp\left(\frac{B}{\nu}\sum_{i=1}^p T(\omega_i)\right )\,.
    \end{equation}
Here, we recall \textup{\eqref{p_loop_correlation_function}--\eqref{gamma_hat_definition}}.
\end{lemma}

\begin{proof}[Proof of Lemma \ref{Proposition4_1}]
    We first observe that by Theorem~\ref{partition_function_result}, the quantities $\Xi^{\nu,\zeta}$ and $\Xi^{0,\zeta}$ given by \eqref{grand_canonical_ensemble} and \eqref{partition_function_result_2} respectively satisfy
\begin{equation}
\label{Xi_bounds}
\Xi^{\nu,\zeta} \sim_{d,\zeta,v,\alpha,B,L} 1
\,,\qquad \Xi^{0,\zeta} \sim_{d,\zeta,B,L} 1\,.
\end{equation}
Furthermore, by \eqref{p_loop_correlation_function}--\eqref{gamma_hat_definition}, we have
    \begin{equation}\label{proposition_4_1_1}
    \Xi^{\nu,\zeta}(\gamma_p^{\nu,\zeta}(\vec{\omega})-\widehat{\gamma}_p^{\nu,\zeta}(\vec{\omega}))=
    \sum_{n = 1}^{\infty} \frac{1}{n!} \sum_{\substack{\vec{k} \in (\N^{*})^n \\ \vec{k} \neq \vec{1}_n}} \prod_{i = 1}^{n} \frac{z^{k_{i}\nu}}{k_{i}} 
    \int \mathbb{W}^{\nu \vec{k}} (\dd \vec{\widetilde{\omega}}) \, \exp \left(-\mathcal{V}^{\nu}(\vec{\widetilde{\omega}}\vec{\omega})\right) \,.
    \end{equation}
    By Lemma \ref{lemma_loop_interactions}, we have
    \begin{equation}\label{stability_condition_on_correlation_function_1}
        \exp \left(-\mathcal{V}^{\nu}(\vec{\widetilde{\omega}}\vec{\omega})\right)\leq \exp\left(\frac{B}{\nu}\sum _{i=1}^nT(\widetilde{\omega}_i)+\frac{B}{\nu}\sum_{i=1}^p T(\omega_i)\right).
    \end{equation}
    Applying \eqref{stability_condition_on_correlation_function_1} in \eqref{proposition_4_1_1}, we obtain
    \begin{multline}
    \label{proposition_4_1_2}
        \Xi^{\nu,\zeta}(\gamma_p^{\nu,\zeta}(\vec{\omega})-\widehat{\gamma}_p^{\nu,\zeta}(\vec{\omega}))
        \\
        \leq \exp\left(\frac{B}{\nu}\sum_{i=1}^p T(\omega_i)\right) \sum_{n = 1}^{\infty} \frac{1}{n!} \sum_{\substack{\vec{k} \in (\N^{*})^n \\ \vec{k} \neq \vec{1}_n}} \prod_{i = 1}^{n} \frac{z^{k_{i}\nu}}{k_{i}} 
    \int \mathbb{W}^{\nu \vec{k}} (\dd \vec{\widetilde{\omega}}) \, \exp \left(\sum_{i=1}^{n} Bk_i\right). 
    \end{multline}
Recalling the analysis \eqref{Proposition_2_1_2}--\eqref{zeta-zeta_hat} in the proof of Proposition \ref{Proposition_1}, we have 
    \begin{equation}
    \label{proposition_4_1_3}
        \sum_{n = 1}^{\infty} \frac{1}{n!} \sum_{\substack{\vec{k} \in (\N^{*})^n \\ \vec{k} \neq \vec{1}_n}} \prod_{i = 1}^{n} \frac{z^{k_{i}\nu}}{k_{i}} 
    \int \mathbb{W}^{\nu \vec{k}} (\dd \vec{\widetilde{\omega}})\, \, \exp{\left(\sum_{i=1}^{n} Bk_i\right)}\lesssim_{d,\zeta ,B,L} \nu^{d/2}\,,
    \end{equation}
    and \eqref{Proposition4_1_bound} follows from \eqref{proposition_4_1_2}--\eqref{proposition_4_1_3} by recalling \eqref{Xi_bounds}.
\end{proof}

We now apply Lemma \ref{Proposition4_1} in \eqref{ploopkernel} and \eqref{Gamma_hat_nu_zeta} to prove Proposition \ref{Proposition_4_3_A}.

 \begin{proof}[Proof of Proposition \ref{Proposition_4_3_A}]
     For $\pi \in S_p$, we define
     \begin{equation}
     \label{Gamma_pi_def}
         \Gamma_{p,\pi}^{\nu,\zeta}(\vec{x},\vec{y})\coloneqq \sum_{\vec{k}\in (\N^{*})^p}\prod_{i=1}^p z^{k_i \nu}\int
         \mathbb{W}^{\nu \vec{k}}_{\pi \vec{y}, \vec{x}} (\dd \vec{\omega})\,\gamma_p^{\nu,\zeta}\,(\vec{\omega})\,,
     \end{equation}
     and 
     \begin{equation}
     \label{hat_Gamma_pi_def}
         \widehat{\Gamma}_{p,\pi}^{\nu,\zeta}(\vec{x},\vec{y})\coloneqq \sum_{\vec{k}\in (\N^{*})^p}\prod_{i=1}^p z^{k_i \nu}\int \mathbb{W}^{\nu \vec{k}}_{\pi \vec{y}, \vec{x}} (\dd \vec{\omega})\,\widehat{\gamma}_p^{\nu,\zeta}\,(\vec{\omega})\,,
     \end{equation}
where we recall \eqref{p_loop_correlation_function}--\eqref{gamma_hat_definition}.
Note that by \eqref{ploopkernel}, \eqref{Gamma_hat_nu_zeta}, and \eqref{Gamma_pi_def}--\eqref{hat_Gamma_pi_def}, we can write 
\begin{equation}
\label{permutation_sum}
\Gamma_p^{\nu,\zeta}(\vec{x},\vec{y})=\sum_{\pi \in S_p} \Gamma_{p,\pi}^{\nu,\zeta}(\vec{x},\vec{y})\,,\qquad 
 \widehat{\Gamma}_p^{\nu,\zeta}(\vec{x},\vec{y})=\sum_{\pi \in S_p} \widehat{\Gamma}_{p,\pi}^{\nu,\zeta}(\vec{x},\vec{y})\,.
\end{equation}
By \eqref{Gamma_pi_def}--\eqref{permutation_sum}, we can deduce \eqref{Gamma_minus_Gamma_hat_estimate} provided that we show     
     \begin{equation}\label{proposition4_2_1}
         |\Gamma_{p, \pi}^{\nu,\zeta}(\vec{x},\vec{y})-\widehat{\Gamma}_{p,\pi}^{\nu,\zeta}(\vec{x},\vec{y})|\lesssim_{p,d,\zeta,v,\alpha,B,L} \nu^{d/2}
     \end{equation}
     for a fixed $\pi \in S_p$.

 We henceforth consider $\pi \in S_p$ fixed. Using \eqref{Gamma_pi_def}--\eqref{hat_Gamma_pi_def} and Lemma~\ref{Proposition4_1}, we have
     \begin{multline}\label{proposition4_2_2}
         |\Gamma_{p, \pi}^{\nu,\zeta}(\vec{x},\vec{y})-\widehat{\Gamma}_{p,\pi}^{\nu,\zeta}(\vec{x},\vec{y})|\lesssim_{d,\zeta,v,\alpha,B,L} \nu^{d/2}\sum_{\vec{k}\in (\N^*)^p}\prod_{i=1}^p z^{k_i \nu}\int  \mathbb{W}^{\nu \vec{k}}_{\pi \vec{y}, \vec{x}} (\dd \vec{\omega})\,\exp\left(\frac{B}{\nu}\sum_{i=1}^p T(\omega_i)\right)\\
         =\nu^{d/2}\sum_{\vec{k}\in (\N^*)^p}\prod_{i=1}^p z^{k_i \nu}\ee ^{B k_i}\int \mathbb{W}^{\nu \vec{k}}_{\pi \vec{y}, \vec{x}} (\dd \vec{\omega})=\nu^{d/2}\sum_{\vec{k}\in (\N^*)^p}\prod_{i=1}^p z^{k_i \nu}\ee ^{B k_i}\psi^{k_i\nu}(y_{\pi(i)}-x_i)\,,
     \end{multline}
     where above, we recalled \eqref{heat_kernel_integral}. Consider $i \in \{1,\ldots,p\}$ fixed. Using Lemma \ref{heat_kernel_estimates} (i) along with \eqref{fugacity_condition}, we estimate
     \begin{multline}\label{proposition4_2_3}
         \sum_{k=1}^{\infty}z^{k \nu}\ee^{Bk} \psi^{k \nu}(y_{\pi(i)}-x_i)\lesssim_{d} \sum_{k=1}^{\infty}z^{k \nu}\ee^{Bk}\left(\frac{1}{L^d}+\frac{1}{(k \nu)^{d/2}}\right) 
\\
\leq \sum_{k=1}^{\infty}\left( ( \zeta \ee^B \nu^{d/2})^{k}+\frac{(\zeta \ee^B)^{k}\nu^{d(k-1)/2}}{ k^{d/2}} \right)
\leq \sum_{k=1}^{\infty} \bigl(\zeta \ee^{B} \nu^{d/2}\bigr)^{k}+\zeta \ee^{B}+\sum_{k=2}^{\infty} \bigl(\zeta \ee^{B} \nu^{d/4}\bigr)^{k}\,.
     \end{multline}
Choosing $\nu$ sufficiently small such that $\zeta \ee^B\nu^{d/4}<\frac{1}{2}$, it follows that 
     \begin{equation}\label{proposition4_2_4}
        \eqref{proposition4_2_3} \lesssim_{d,\zeta,B} \nu^{d/2} +\zeta \ee^B\,.
     \end{equation}
     Substituting \eqref{proposition4_2_4} into \eqref{proposition4_2_2}, we obtain
     \begin{equation}\label{proposition4_2_5}
         |\Gamma_{p, \pi}^{\nu,\zeta}(\vec{x},\vec{y})-\widehat{\Gamma}_{p,\pi}^{\nu,\zeta}(\vec{x},\vec{y})|\lesssim_{p,d,\zeta,B,L} \nu^{d/2}(\nu^{d/2} +\zeta \ee^B)^p\,,
     \end{equation}
     and \eqref{proposition4_2_1} follows.
     \end{proof}

\begin{proof}[Proof of Proposition \ref{Proposition_4_3}]
Let $\pi \in S_p$ be given. 
Rewriting \eqref{hat_Gamma_pi_def} in analogy with \eqref{Gamma_hat_expansion}, we have
\begin{equation*}
	\widehat{\Gamma}_{p,\pi}^{\nu,\zeta}(\vec{x},\vec{y})=\frac{1}{\Xi^{\nu,\zeta}}\sum_{\vec{k}\in (\N^{*})^p}\sum_{n=0}^\infty\frac{z^{n\nu}}{n!}\prod_{i=1}^p z^{k_i \nu}\int \mathbb{W}^{\nu \vec{1}_n}(\dd \vec{\widetilde{\omega}})\,
\int \mathbb{W}^{\nu \vec{k}}_{\pi \vec{y}, \vec{x}} (\dd \vec{\omega}) \,\exp \bigl(-\mathcal{V}^{\nu}(\vec{\widetilde{\omega}}\vec{\omega})\bigr)\,.
\end{equation*}
Furthermore, let us  define $\widecheck{\Gamma}_{p,\pi}^{\nu,\zeta}(\vec{x},\vec{y})$ as
\begin{equation}\label{Gamma_check_pi_definition}
    \widecheck{\Gamma}_{p,\pi}^{\nu,\zeta}(\vec{x},\vec{y})\coloneqq \frac{1}{\Xi^{\nu,\zeta}}\sum_{n=0}^\infty\frac{z^{(n+p) \nu}}{n!} \int \mathbb{W}^{\nu \vec{1}_n}(\dd \vec{\widetilde{\omega}})\,
\int \mathbb{W}^{\nu \vec{1}_p}_{\pi \vec{y}, \vec{x}} (\dd \vec{\omega}) \,\exp \bigl(-\mathcal{V}^{\nu}(\vec{\widetilde{\omega}}\vec{\omega})\bigr)\,.
\end{equation}
It suffices to show that
 \begin{equation}\label{Gamma_hat_pi_Gamma_check_estimation}
     |\widehat{\Gamma}_{p,\pi}^{\nu,\zeta}(\vec{x},\vec{y})-\widecheck{{\Gamma}}_{p,\pi}^{\nu,\zeta}(\vec{x},\vec{y})|\lesssim_{p,d,\zeta,B,L} \nu^{d/2}\,,
 \end{equation}
 and \eqref{proposition_4_3_1} follows by summation over $\pi \in S_p$.

By \eqref{Gamma_check_pi_definition}--\eqref{Gamma_hat_pi_Gamma_check_estimation} and 
Lemma \ref{lemma_loop_interactions}, we have that
\begin{align}
	|\widehat{\Gamma}_{p,\pi}^{\nu,\zeta} (\vec{x},\vec{y})&-\widecheck{\Gamma}_{p,\pi}^{\nu,\zeta}(\vec{x},\vec{y})|
	\nonumber \\
    & =\frac{1}{\Xi^{\nu,\zeta}}\sum_{n=1}^\infty\sum_{\substack{\vec{k} \in (\N^{*})^p \\ \vec{k} \neq \vec{1}_{p}}}\frac{z^{n\nu}}{n!}\prod_{i=1}^p z^{k_i \nu} \int \mathbb{W}^{\nu \vec{1}_n}(\dd \vec{\widetilde{\omega}})\,
\int \mathbb{W}^{\nu \vec{k}}_{\pi \vec{y}, \vec{x}} (\dd \vec{\omega}) \,\exp \bigl(-\mathcal{V}^{\nu}(\vec{\widetilde{\omega}}\vec{\omega})\bigr)
\label{proposition_4_3_2_A} \\
    & \leq \frac{1}{\Xi^{\nu,\zeta}}\sum_{n=1}^\infty\sum_{\substack{\vec{k} \in (\N^{*})^p \\ \vec{k} \neq \vec{1}_{p}}}\frac{z^{n\nu}}{n!}\prod_{i=1}^p z^{k_i \nu}    \int \mathbb{W}^{\nu \vec{1}_n}(\dd \vec{\widetilde{\omega}})\,
\int \mathbb{W}^{\nu \vec{k}}_{\pi \vec{y}, \vec{x}} (\dd \vec{\omega})
    \, \exp\left[B\left(\sum_{i=1}^p k_i+n\right)\right]
    \nonumber \\
    & =\frac{1}{\Xi^{\nu,\zeta}}\sum_{n=1}^\infty\sum_{\substack{\vec{k} \in (\N^{*})^p \\ \vec{k} \neq \vec{1}_{p}}}\frac{z^{n\nu}\ee^{Bn}}{n!}\prod_{i=1}^p z^{k_i \nu}\ee^{B k_i}
\int \mathbb{W}^{\nu \vec{1}_n}(\dd \vec{\widetilde{\omega}})\,
\int \mathbb{W}^{\nu \vec{k}}_{\pi \vec{y}, \vec{x}} (\dd \vec{\omega})\,.
\label{proposition_4_3_2}
\end{align}
Let us now integrate in $\vec{\widetilde{\omega}}$ in \eqref{proposition_4_3_2}. In order to do this, we first recall
\eqref{heat_kernel_integral}, \eqref{W^T_definition}, \eqref{W_product_2}, and then apply Lemma  \ref{heat_kernel_estimates} (i) together with \eqref{fugacity_condition} to deduce that 
    \begin{multline}
     \label{proposition_4_3_2_A_bound}
    \eqref{proposition_4_3_2}
    \lesssim_{d,\zeta} \sum_{n=0}^\infty\frac{(\zeta \ee^B)^{n}}{n!}L^{dn}\sum_{\substack{\vec{k} \in (\N^{*})^p \\ \vec{k} \neq \vec{1}_{p}}}\prod_{i=1}^p z^{k_i \nu}\ee^{B k_i}\int \mathbb{W}^{\nu \vec{k}}_{\pi \vec{y}, \vec{x}} (\dd \vec{\omega})\\
    \leq p\, \mathrm{exp} \left(\zeta\ee^{B} L^d \right) \sum_{k_1=2}^{\infty}\sum_{k_2=1}^{\infty}\cdots\sum_{k_p=1}^{\infty}\prod_{i=1}^p z^{k_i \nu}\ee^{B k_i}\int\prod_{i=1}^p \mathbb{W}^{k_i \nu}_{y_{\pi(i)},x_i}(\dd \omega_i)\,. 
\end{multline}
In \eqref{proposition_4_3_2}, we also recalled \eqref{W_product}.
Arguing as in \eqref{proposition4_2_4}, we have that 
\begin{equation}
\label{proposition_4_3_2_B_bound}
\sum_{k_2=1}^{\infty}\cdots\sum_{k_p=1}^{\infty}\prod_{i=2}^p z^{k_i \nu}\ee^{B k_i}\int\prod_{i=2}^{p} \mathbb{W}^{k_i \nu}_{y_{\pi(i)},x_i}(\dd \omega_i)\lesssim_{d,\zeta,B}(\nu^{d/2}+\zeta\ee^B)^{p-1}\,.
\end{equation}
Substituting \eqref{proposition_4_3_2_B_bound} into \eqref{proposition_4_3_2_A_bound}, and recalling \eqref{proposition_4_3_2}, we get
\begin{equation}\label{proposition_4_3_3}
    \left|\widehat{\Gamma}_{p,\pi}^{\nu,\zeta}(\vec{x},\vec{y})-\widecheck{\Gamma}_{p,\pi}^{\nu,\zeta}(\vec{x},\vec{y})\right|
    \lesssim_{p,d,\zeta,B,L} \sum_{k_1=2}^{\infty} z^{k_1 \nu}\ee^{k_1 B}\int\mathbb{W}^{k_1 \nu}_{y_{\pi(1)},x_1}(\dd \omega_1)\,.
\end{equation}
Note that we are taking the sum on right hand side of \eqref{proposition_4_3_3} over all indices $k_1\geq 2$, and therefore, arguing as in \eqref{fin-vol-two-sums-split}--\eqref{fin-vol-star2}, we have   
\begin{equation}
\label{proposition_4_3_3_B}
    \sum_{k_1=2}^{\infty} z^{k_1 \nu}\ee^{k_1 B}\int\mathbb{W}^{k_1 \nu}_{y_{\pi(1)},x_1}(\dd \omega_1)\lesssim_{d,\zeta,B} \nu^{d/2}\,,
\end{equation}
which concludes the proof. 
\end{proof}

Having proved Propositions \ref{Proposition_4_3_A}--\ref{Proposition_4_3} above, the final step of the proof of Theorem~\ref{density_matrix_result}
reduces to comparing  \eqref{Gamma_check_definition} and \eqref{density_matrix_result_2}.
In order to proceed with this step, we will need the following refinement of Lemma \ref{heat_kernel_estimates} (i), which is shown in Appendix \ref{appendix_heat_kernel_estimates} below.
\begin{lemma}
\label{Riemann_sums_estimate}
Let $x \in \Lambda $ be fixed. Recalling \eqref{periodic_Euclidean_norm}, we have that 
\begin{equation}
\label{Riemann_sums_estimate_1}
\psi^t(x) \lesssim_{d}\frac{1}{L^d}+\frac{1}{t^{d/2}}\,\ee^{-\frac{|x|_{L}^2}{2t}} 
\end{equation}
for all $t>0$. 
In particular, if we take $n_0 \equiv n_0(x) \in \Z^d$ to be such that 
\begin{equation}
\label{n_0_choice}
|x|_{L}=|x-L n_0|\,,
\end{equation}
we have
\begin{equation}
\label{Riemann_sums_estimate_1B}
\psi^t(x) \lesssim_{d}\frac{1}{L^d}+\frac{1}{t^{d/2}}\,\ee^{-\frac{|x-Ln_0|^2}{2t}}\,.
\end{equation}
\end{lemma}

We now have the necessary tools to prove Theorem~\ref{density_matrix_result}.
\begin{proof}[Proof of Theorem~\ref{density_matrix_result}]

Let us define
\begin{equation}
\label{diagonal_contribution}
\Delta_p^{0,\zeta}(\vec{x})\coloneqq (2\pi)^{-\frac{pd}{2}}\,\zeta^p\,\frac{\Xi^{0,\zeta}(\vec{x})}{\Xi^{0,\zeta}}\,.
\end{equation}
Using \eqref{density_matrix_result_2}--\eqref{Kronecker_delta}, \eqref{diagonal_contribution}, and recalling \eqref{Pi_{x,y}}, the claim \eqref{density_matrix_result_1} is equivalent to showing
\begin{equation}
\label{Theorem_1.5(ii)}
\lim_{\nu \rightarrow 0} \Gamma_p^{\nu,\zeta}(\vec{x},\vec{y})=|\Pi_{\vec{x},\vec{y}}|\,\Delta_p^{0,\zeta}(\vec{x})\,.
\end{equation}
Moreover, \eqref{Theorem_1.5(ii)_rate_of_convergence} is equivalent to showing
\begin{multline}
\label{Theorem_1.5(ii)_rate_of_convergence_rewritten}
\Bigl|\Gamma_p^{\nu,\zeta}(\vec{x},\vec{y})-|\Pi_{\vec{x},\vec{y}}|\,\Delta_p^{0,\zeta}(\vec{x})
\Bigr| 
\\
\lesssim_{p,d,\zeta,v,\alpha,B,L} \Biggl[\nu^{\alpha/2}+\sum_{\tilde{\pi} \in S_p \setminus \Pi_{\vec{x},\vec{y}}} \prod_{i=1}^{p} \Biggl\{\nu^{d/2}+ 
\exp \biggl(-\frac{|y_{\tilde{\pi}(i)}-x_i|_{L}^2}{2\nu}\biggr)\Biggr\}\Biggr]\,.
\end{multline}

In order to prove Theorem~\ref{density_matrix_result}, we consider two cases for the set $\Pi_{\vec{x},\vec{y}}$ defined in \eqref{Pi_{x,y}}.\\
   \textbf{Case 1}: $\Pi_{\vec{x},\vec{y}}=\emptyset$.
   

 We recall \eqref{Gamma_check_definition}, the first bound in \eqref{Xi_bounds}, and use 
 Assumption \ref{Assumption_on_v} (i) combined with Lemma \ref{lemma_loop_interactions} to estimate
   \begin{equation}\label{check_Gamma_estimation_case_1}
      0 \leq \widecheck{\Gamma}_{p}^{\nu,\zeta}(\vec{x},\vec{y})\lesssim_{d,\zeta,v,B,L} \sum_{\pi \in S_p} \sum_{n=0}^\infty\frac{z^{(n+p) \nu}\ee^{(n+p)B}}{n!}
\int \mathbb{W}^{\nu \vec{1}_n} (\dd \vec{\tilde{\omega}}) \, \int \mathbb{W}^{\nu \vec{1}_p}_{\pi \vec{y}, \vec{x}} (\dd \vec{\omega})\,.  
 \end{equation}
We note the following bound for all $\pi \in S_p$ and for all $i=1,\ldots,p$. 
\begin{multline}
\label{check_Gamma_estimation_case_1_A}
 z^{\nu} \int \mathbb{W}^{\nu}_{y_{\pi(i)},x_i}(\dd \omega)=z^{\nu} \psi^{\nu}(y_{\pi(i)}-x_i) \lesssim_d z^{\nu} \left(
\frac{1}{L^d}+\frac{1}{\nu^{d/2}}\,\ee^{-\frac{|y_{\pi(i)}-x_i|_{L}^2}{2\nu}} \right)
\\
=\zeta \nu^{d/2} L^{-d}+\zeta\,\ee^{-\frac{|y_{\pi(i)}-x_i|_{L}^2}{2\nu}}\,.
\end{multline}
In order to deduce \eqref{check_Gamma_estimation_case_1_A}, we used \eqref{heat_kernel_integral}, Lemma \ref{Riemann_sums_estimate}, followed by \eqref{fugacity_condition}.
We now recall \eqref{W_product}, \eqref{W_product_2}, and apply \eqref{lipschitz_g_2b}, \eqref{check_Gamma_estimation_case_1_A} to deduce that
\begin{multline}
\label{check_Gamma_estimation_case_1_B}
\eqref{check_Gamma_estimation_case_1} \leq \sum_{\pi \in S_p} \sum_{n=0}^{\infty} \frac{\ee^{(n+p)B}\,C_d^{n+p}\,L^{dn}\,\zeta^{n+p}}{n!}\, \prod_{i=1}^{p} \left\{\nu^{d/2}+\exp\left({-\frac{|x_i-y_{\pi(i)}|^2_L}{2\nu}}\right)\right\}
\\
\lesssim_{p,d,\zeta,B,L} \sum_{\pi \in S_{p}} \prod_{i=1}^p \left\{\nu^{d/2}+\exp\left({-\frac{|x_i-y_{\pi(i)}|^2_L}{2\nu}}\right)\right\}.
\end{multline}
We hence deduce \eqref{Theorem_1.5(ii)}--\eqref{Theorem_1.5(ii)_rate_of_convergence_rewritten} by combining \eqref{check_Gamma_estimation_case_1_B} with Propositions \ref{Proposition_4_3_A} and \ref{Proposition_4_3}.
   

\textbf{Case 2}: $\Pi_{\vec{x},\vec{y}} \neq \emptyset$. 

We recall \eqref{Gamma_pi_def} and note that, by the arguments as in Case 1 (in particular by using also Propositions \ref{Proposition_4_3_A} and \ref{Proposition_4_3}), we have

\begin{equation}
\label{x=pi(y)_1}
0 \leq \sum_{\tilde{\pi} \in S_p \setminus \Pi_{\vec{x},\vec{y}}} \Gamma_{p,\tilde{\pi}}^{\nu,\zeta}(\vec{x},\vec{y}) \lesssim_{p,d,\zeta,v,\alpha,B,L} \Biggl[\nu^{d/2}+\sum_{\tilde{\pi} \in S_p \setminus \Pi_{\vec{x},\vec{y}}} \prod_{i=1}^{p} \Biggl\{\nu^{d/2}+ 
\exp \biggl(-\frac{|y_{\tilde{\pi}(i)}-x_i|_{L}^2}{2\nu}\biggr)\Biggr\}\Biggr]\,.
\end{equation}
Furthermore, we recall \eqref{Pi_{x,y}}, \eqref{Gamma_check_pi_definition}, \eqref{diagonal_contribution}, and argue as in the proof of Proposition \ref{Proposition_2} to deduce that for all $\pi \in \Pi_{\vec{x},\vec{y}}$, we have
\begin{equation}
\label{x=pi(y)_2}
\bigl|\widecheck{\Gamma}_{p,\pi}^{\nu,\zeta}(\vec{x},\vec{y})- \Delta_p^{0,\zeta}(\vec{x}) \bigr| \lesssim_{p,d,\zeta,v,\alpha,B,L} \nu^{\alpha/2}\,.
\end{equation}
In order to deduce \eqref{x=pi(y)_2}, we also recalled Lemma \ref{Lemma_3} and the fact that \eqref{Lemma_3_1} holds uniformly in $x, y \in \Lambda$. Therefore, the proof of Proposition \ref{Proposition_2} goes through without the need to integrate in $x_1,\ldots,x_p$. Combining \eqref{proposition4_2_1}, \eqref{Gamma_hat_pi_Gamma_check_estimation}, and \eqref{x=pi(y)_2}, we obtain that for all $\pi \in \Pi_{\vec{x},\vec{y}}$
\begin{equation}
\label{x=pi(y)_3}
\bigl|\Gamma_{p,\pi}^{\nu,\zeta}(\vec{x},\vec{y})- \Delta_p^{0,\zeta}(\vec{x}) \bigr| \lesssim_{p,d,\zeta,v,\alpha,B,L} \nu^{\alpha/2}\,.
\end{equation}
We hence conclude \eqref{Theorem_1.5(ii)}--\eqref{Theorem_1.5(ii)_rate_of_convergence_rewritten} in this case from \eqref{x=pi(y)_1} and \eqref{x=pi(y)_3} after summing over all permutations in $S_p$.

\end{proof}

\begin{remark}[The reduced $p$-particle density matrices of the Boltzmann gas]
	\label{Remark_Boltzmann_density_matrix}
	Let us recall Remark \ref{Remark_Boltzmann_gas_partition_function} concerning the partition function of the Boltzmann gas. Replacing $P_{n}^{+}$ in \eqref{n_particle_space} with $M_{n} \colon f \mapsto \frac{1}{n!} f$, the Boltzmann reduced $p$-particle matrix is defined as in \eqref{reduced_p_particle_density_matrix}. To derive the Ginibre representation of its kernel, recall the Ginibre representation of its Bose counterpart in \eqref{Ginibre_representation_4}, and recall that we only have to consider cycles of length $1$ in the Boltzmann case. We therefore have
	\begin{equation}
	\label{Boltzmann_reduced_density_matrix}
		\Gamma_{p}^{\nu,\zeta;\Bolt}(\vec{x},\vec{y}) = \frac{\Xi^{\nu,\zeta}}{\widehat{\Xi}^{\nu,\zeta}} \, \widecheck{\Gamma}_{p,\id}^{\nu,\zeta}(\vec{x},\vec{y})\,,
	\end{equation}
	where $\widecheck{\Gamma}_{p,\id}^{\nu,\zeta}(\vec{x},\vec{y})$ is defined in \eqref{Gamma_check_pi_definition}. Note that the factor $\Xi^{\nu,\zeta} / \widehat{\Xi}^{\nu,\zeta}$ ensures that \eqref{Boltzmann_reduced_density_matrix} is properly normalised.

	We have the following analogue of Theorem \ref{density_matrix_result}. Under the same assumptions as in Theorem \ref{density_matrix_result}, for all $\vec{x}, \vec{y} \in \Lambda^p$,
	\begin{gather}
		\bigl|\Gamma_{p}^{\nu,\zeta;\Bolt}(\vec{x},\vec{y}) - \Gamma_{p}^{0,\zeta;\Bolt}(\vec{x},\vec{y})\bigr| \lesssim_{p,d,\zeta,v,B,L} \prod_{i=1}^p \left\{\nu^{d/2}+\exp\left({-\frac{|x_i-y_i|^2_L}{2\nu}}\right)\right\},
		\label{density_matrix_result_2_Boltzmann}
		\shortintertext{where}
		\Gamma_p^{0,\zeta;\Bolt}(\vec{x}, \vec{y}) \coloneqq \zeta^p(2\pi)^{-\frac{pd}{2}} \, \delta_{K}(\vec{x}-\vec{y}) \, \frac{\Xi^{0,\zeta}(\vec{x})}{\Xi^{0,\zeta}} \,,
		\label{density_matrix_result_1_Boltzmann}
	\end{gather}
	and $\Xi^{0,\zeta}$, $\Xi^{0,\zeta}(\vec{x})$, and $\delta_{K}$ are defined in \eqref{partition_function_result_2}, \eqref{density_matrix_result_3}, and \eqref{Kronecker_delta}, respectively. In particular,
	\begin{equation*}
		\lim_{\nu\to 0} \Gamma_{p}^{\nu,\zeta;\Bolt}(\vec{x},\vec{y}) = \Gamma_{p}^{0,\zeta;\Bolt}\,.
	\end{equation*}

	To prove \eqref{density_matrix_result_2_Boltzmann}, we consider two cases. If $\vec{x} \neq \vec{y}$ we have by the same reasoning as in Case 1 of the previous proof that
	\begin{equation}
		\label{Boltzmann_Gamma_case_1}
		\bigl| \Gamma_{p}^{\nu,\zeta;\Bolt}(\vec{x},\vec{y}) \bigr| \lesssim_{p,d,\zeta, v, B,L} \prod_{i=1}^p \left\{\nu^{d/2}+\exp\left({-\frac{|x_i-y_i|^2_L}{2\nu}}\right)\right\}.
	\end{equation}
	If, on the other hand, $\vec{x} = \vec{y}$, we have
	\begin{align}
		\bigl| \Gamma_{p}^{\nu,\zeta;\Bolt}(\vec{x},\vec{x}) - \Delta_p^{0,\zeta}(\vec{x}) \bigr| & = 
		\left| \frac{\Xi^{\nu,\zeta}}{\widehat{\Xi}^{\nu,\zeta}} \, \widecheck{\Gamma}_{p,\id}^{\nu,\zeta}(\vec{x},\vec{x}) - \Delta_p^{0,\zeta}(\vec{x}) \right| 
		\notag \\
		& = \left|\frac{\Xi^{\nu,\zeta}}{\widehat{\Xi}^{\nu,\zeta}} \left( \widecheck{\Gamma}_{p,\id}^{\nu,\zeta}(\vec{x},\vec{x})- \Delta_p^{0,\zeta}(\vec{x}) \right) 
		+ \left(\frac{\Xi^{\nu,\zeta}}{\widehat{\Xi}^{\nu,\zeta}} - 1 \right) \Delta_p^{0,\zeta}(\vec{x}) \right| 
		\notag \\
		& \leq \left|\frac{\Xi^{\nu,\zeta}}{\widehat{\Xi}^{\nu,\zeta}} \right| \, \left| \widecheck{\Gamma}_{p,\id}^{\nu,\zeta}(\vec{x},\vec{x})
	        - \Delta_p^{0,\zeta}(\vec{x}) \right| + \left|\frac{\Xi^{\nu,\zeta}}{\widehat{\Xi}^{\nu,\zeta}} - 1 \right| \, \bigl|\Delta_p^{0,\zeta}(\vec{x})\bigr| 
	        \notag \\
		& \lesssim_{p,d,\zeta,v,\alpha,B,L} \nu^{\alpha/2}\,,
		\label{Boltzmann_Gamma_case_2}
	\end{align}
	where $\Delta_{p}^{0,\zeta}(\vec{x})$ is defined in \eqref{diagonal_contribution}. To get to the last line, we recalled \eqref{Proposition_2_1_convergence_rate}, \eqref{Xi_bounds}, and \eqref{x=pi(y)_2}, we used $\Delta_{p}^{0,\zeta}(\vec{x}) \lesssim_{p,d,\zeta,v,B,L} 1$ for all $\vec{x} \in \Lambda^{p}$, and we recalled \eqref{Xi_hat}, which implies $\widehat{\Xi}^{\nu,\zeta} \geq 1$. We now deduce \eqref{density_matrix_result_1_Boltzmann} from \eqref{Boltzmann_Gamma_case_1} and \eqref{Boltzmann_Gamma_case_2}.
\end{remark}

\section{Convergence in finite volume for a Fermi gas}
\label{Fermi_gas}

In this section we show the results of the previous sections for Fermi gases in a finite volume $\Lambda \equiv \Lambda_L$. The analysis will be done under the same assumptions as in the previous sections. Most of the analysis is the same as for Bose gases, so we will only point out the differences, and not go into the details. For clarity, throughout this section we use the superscripts $\FF$ for objects that relate to the Fermi gas, and $\BB$ for objects that relate to the Bose gas.

Recalling the setup from the introduction, we work with the $n$-body fermionic Hamiltonian \eqref{Hamiltonian_H} acting on the $n$-particle fermionic Hilbert space \eqref{n_particle_space_Fermi}. Here, we also recall \eqref{antisymmetric_projector}.
Note that replacing $P_{n}^{+}$ with $P_{n}^{-}$ is the only difference in the setup from Bose gases. We define the grand-canonical ensemble of a Fermi gas as a sequence of operators $(\rho_n^{\FF})_{n \in \mathbb{N}}$, where
\begin{equation}
\label{grand_canonical_ensemble_1_Fermi}
\rho_n^{\FF}\coloneqq \frac{1}{\Xi^{\FF}} \: \ee^{-\beta(\mathbb{H}_{n}^{\FF}-\mu n)}\,,\qquad \Xi^{\FF}\coloneqq 1+\sum_{n=1}^{\infty} \mathrm{Tr}_{\mathcal{H}_n^{\FF}} \ee^{-\beta(\mathbb{H}_{n}^{\FF}-\mu n)}\,.
\end{equation}
Here, $\beta>0$ denotes the inverse temperature, and $\mu<0$ denotes the chemical potential. The quantity $\Xi^{\FF}$ is the Fermi gas grand-canonical partition function. We now set $\beta = 1$ for simplicity.

Setting $\nu, \zeta, z > 0$ as in Section \ref{section_choice_of_parameters}, we rewrite the fermionic $n$-body Hamiltonian in \eqref{Hamiltonian_H} as
\begin{equation}
\label{n_body_Hamiltonian_Fermi}
H_n^{\nu; \FF}\coloneqq -\frac{\nu}{2} \sum_{i=1}^{n} \Delta_i +\sum_{1 \leq i < j \leq n} v(x_i-x_j)
\end{equation}
and the fermionic grand-canonical ensemble \eqref{grand_canonical_ensemble_1_Fermi} as
\begin{equation}
\label{grand_canonical_ensemble_Fermi}
\rho_n^{\FF} \equiv \rho_n^{\nu,\zeta; \FF}\coloneqq \frac{1}{\Xi^{\FF}}\, \ee^{-H_{n}^{\nu; \FF}}\,z^{n \nu} \,, \qquad \Xi^{\FF} \equiv \Xi^{\nu,\zeta; \FF}\coloneqq 1 + \sum_{n=1}^{\infty} \mathrm{Tr}_{\mathcal{H}_n^{\FF}} \bigl(\ee^{-H_{n}^{\nu; \FF}}\,z^{n \nu} \bigr) \,.
\end{equation}
Note that \eqref{n_body_Hamiltonian_Fermi} is a densely-defined operator on \eqref{n_particle_space_Fermi}. 
Furthermore, we  define the fermionic free grand-canonical partition function as
\begin{equation}
\label{free_grand_canonical_partition_function_Fermi}
\Xi_{(0)}^{\FF} \equiv \Xi^{\nu,\zeta; \FF}_{(0)} \coloneqq \Xi^{\nu,\zeta; \FF} \big|_{v=0}\,.
\end{equation}
For a given $p \in \N^*$ we define the reduced $p$-particle density matrix of the fermionic grand-canonical ensemble as
\begin{equation}
\label{reduced_p_particle_density_matrix_Fermi}
	\Gamma_p^{\nu,\zeta; \FF} \coloneqq \sum_{n=0}^{\infty} \frac{(n+p)!}{n!}\, \mathrm{Tr}_{p+1,\ldots,n+p}(\rho^{\nu,\zeta; \FF}_{n+p})\,, \\
\end{equation}
where $\Tr_{p+1, \ldots, n+p}$ denotes the partial trace over the coordinates $x_{p+1}, \ldots, x_{n+p}$ in the fermionic Hilbert space \eqref{n_particle_space_Fermi}.

The following two theorems are analogues of Theorems \ref{partition_function_result} and  \ref{density_matrix_result} respectively.

\begin{theorem}
\label{partition_function_result_Fermi}
Let $\mu \equiv \mu (\nu,\zeta)$ be as in Assumption \ref{Choice_of_z} and let $z \equiv z(\nu,\zeta)$ be as in \eqref{fugacity}. Let $v$ be as in Assumption \ref{Assumption_on_v}. Then the grand-canonical partition function $\Xi^{\nu,\zeta; \FF}$  defined in \eqref{grand_canonical_ensemble_Fermi} satisfies
\begin{equation}
\label{partition_function_result_1_Fermi}
\lim_{\nu \rightarrow 0} \Xi^{\nu,\zeta; \FF}=\Xi^{0,\zeta; \FF}\,, 
\end{equation}
where in \eqref{partition_function_result_1_Fermi}, $\Xi^{0,\zeta; \FF}\coloneqq \Xi^{0,\zeta}$ is the grand-canonical partition function of a gas of classical interacting particles from \eqref{partition_function_result_2}.
Moreover, we can estimate the rate of convergence in \eqref{partition_function_result_1_Fermi} as
\begin{equation}
\label{partition_function_result_3_Fermi}
\bigl|\Xi^{\nu,\zeta; \FF}- \Xi^{0,\zeta; \FF} \bigr| \lesssim_{d,\zeta,v,\alpha,B} \nu^{\alpha/2}\,\exp \left [C_0(d,\zeta,B) L^d \right],
\end{equation}
for some $C_0(d,\zeta,B)>0$.
\end{theorem}

Recalling the notation in \eqref{concatenation}--\eqref{pi_vec_y}, we have the analogous result for reduced $p$-particle density matrices.
\begin{theorem}
\label{density_matrix_result_Fermi}
Let $\mu \equiv \mu (\nu,\zeta)$ be as in Assumption \ref{Choice_of_z} and let $z \equiv z(\nu,\zeta)$ be as in \eqref{fugacity}. Let $v$ be as in Assumption \ref{Assumption_on_v}.  Let $p \in \N^*$ be fixed. The kernel of the Fermi reduced $p$-particle density matrix \eqref{reduced_p_particle_density_matrix_Fermi} satisfies 
    \begin{equation}
    \label{density_matrix_result_1_Fermi}
        \lim_{\nu \to 0} \Gamma_p^{\nu,\zeta; \FF}(\vec{x},\vec{y}) = \Gamma_p^{0,\zeta; \FF}(\vec{x},\vec{y})\,,
    \end{equation}
    for all $\vec{x}, \vec{y} \in \Lambda^p$ where in \eqref{density_matrix_result_1},
    \begin{equation}
    \label{density_matrix_result_2_Fermi}
    	\Gamma_p^{0,\zeta; \FF}(\vec{x},\vec{y})\coloneqq \zeta^p(2\pi)^{-\frac{pd}{2}}\sum_{\pi \in S_p} \sgn(\pi) \,\delta_K(\vec{x}-\pi\vec{y})\,\frac{\Xi^{0,\zeta; \FF}(\vec{x})}{\Xi^{0,\zeta; \FF}}\,.
    \end{equation}
In \eqref{density_matrix_result_2_Fermi} $\sgn(\pi)$ denotes the sign of $\pi \in S_p$, and $\delta_K$ denotes the Kronecker delta function on $\Lambda^p$ given by \eqref{Kronecker_delta} above.
    Furthermore, 
 $\Xi^{0,\zeta; \FF} \equiv \Xi^{0,\zeta}$ is given by \eqref{partition_function_result_2}, and we take
     \begin{equation*}
         \Xi^{0,\zeta; \FF}(\vec{x})\coloneqq \Xi^{0,\zeta}(\vec{x})\,,
     \end{equation*}
     where we recall \eqref{density_matrix_result_3}. 
     By recalling \eqref{Pi_{x,y}},   
we can bound the rate of convergence in \eqref{density_matrix_result_1_Fermi} by
\begin{multline}
\label{Theorem_1.5(ii)_rate_of_convergence_Fermi}
\bigl|\Gamma_p^{\nu,\zeta; \FF}(\vec{x},\vec{y})-\Gamma_p^{0,\zeta; \FF}(\vec{x},\vec{y})\bigr| 
\\
\lesssim_{p,d,\zeta,v,\alpha,B,L} \Biggl[\nu^{\alpha/2}+\sum_{\tilde{\pi} \in S_p \setminus \Pi_{\vec{x},\vec{y}}} \prod_{i=1}^{p} \Biggl\{\nu^{d/2}+ 
\exp \biggl(-\frac{|y_{\tilde{\pi}(i)}-x_i|_{L}^2}{2\nu}\biggr)\Biggr\}\Biggr]\,.
\end{multline}
\end{theorem}

We now state the fermionic analogue of the Ginibre representation given by Lemma \ref{Ginibre_representation}. Here, we follow \cite[Section 2]{Ginibre} and refer the reader to the latter work for the full details. Recall the definition of $\mathbb{L}^{\nu,\zeta; \BB}(\dd \omega)$ in \eqref{single_loop_measure}, and the Ginibre representation of $\Xi^{\nu,\zeta; \BB}$ in \eqref{Ginibre_representation_1}. To derive the analogue of \eqref{Ginibre_representation_1}, the permutations that appear in \eqref{symmetric_projector} are decomposed into disjoint cycles. We then sum over the number of cycles $n$, with notation as in \eqref{Ginibre_representation_1}. We implicitly sum over the lengths of the cycles in \eqref{single_loop_measure}, where a loop of duration $T \in \nu \N^*$ corresponds to a permutation cycle of length $T / \nu$. In the fermionic case, we replace $P_n^{+}$ with $P_n^{-}$; see \eqref{n_particle_space}, \eqref{grand_canonical_ensemble_1}, and \eqref{n_particle_space_Fermi}, \eqref{grand_canonical_ensemble_1_Fermi} respectively. 
Consequently, when decomposing a permutation into cycles, we additionally have to decompose its sign. Given a permutation $\pi$ and its decomposition $\pi=C_1 \cdots C_m$ into cycles $C_1, \ldots, C_m$, we have
\begin{equation}
\label{sgn_pi}
	\sgn(\pi) = \prod_{i = 1}^{m} \sgn(C_i) = (-1)^{\sum_{i = 1}^{m} (k_i - 1)}\,,
\end{equation}
where $k_i$ denotes the length of $C_i$. Note that by convention, the length of a transposition is $2$. This suggests the following definition of the single-loop measure for a Fermi gas:
\begin{equation}
	\label{single_loop_measure_Fermi}
	\mathbb{L}^{\nu,\zeta; \FF}(\dd \omega)\coloneqq \nu \sum_{T \in \nu \N^*} (-1)^{\frac{T}{\nu} - 1} \, \frac{z^T}{T} \, \mathbb{W}^T(\dd \omega)\,,
\end{equation}
where $z \equiv z(\nu,\zeta)$ is given by \eqref{fugacity}. We then have the following analogue of Lemma \ref{Ginibre_representation}.
\begin{lemma}[Ginibre representation of a Fermi gas]
\label{Ginibre_representation_Fermi}
Let $\mathcal{V}^{\nu}$ be the interaction given by \eqref{loop_interaction}. The following representations hold.
\begin{itemize}
    \item [(i)] The Fermi grand-canonical partition function \eqref{grand_canonical_ensemble_Fermi} can be written as
\begin{equation}
\label{Ginibre_representation_1_Fermi}
\Xi^{\nu,\zeta, \FF}= 1 + \sum_{n=1}^{\infty} \frac{1}{n!} \int \mathbb{L}^{\nu,\zeta, \FF}(\dd \omega_1) \cdots\, \mathbb{L}^{\nu,\zeta; \FF}(\dd \omega_n) \,\exp\bigl(-\mathcal{V}^{\nu}(\vec{\omega})\bigr)\,.
\end{equation}
\item [(ii)]
For $\vec{\omega}=(\omega_1,\ldots,\omega_p)$, where $\omega_i \in \Omega^{T(\omega_i)}$ and $T(\omega_i) \in \nu \N^*$ for $i =1,\ldots,p$, we define
\begin{equation}\label{Ginibre_representation_3_Fermi}
    \Xi^{\nu,\zeta; \FF}(\vec{\omega})\coloneqq \sum_{n=0}^{\infty} \frac{1}{n!} \int \mathbb{L}^{\nu,\zeta; \FF}(\dd \widetilde{\omega}_1) \cdots\, \mathbb{L}^{\nu,\zeta; \FF}(\dd \widetilde{\omega}_n) \,\exp\bigl(-\mathcal{V}^{\nu}(\vec{\widetilde{\omega}}\vec{\omega})\bigr)\,,
\end{equation}
where for $\vec{\widetilde{\omega}}=(\widetilde{\omega}_1,\ldots,\widetilde{\omega}_n)$, $\vec{\widetilde{\omega}}\vec{\omega}$ is the concatenation given by \eqref{concatenation_omega}.
Then, the kernel of the Fermi reduced $p$-particle density matrix \eqref{reduced_p_particle_density_matrix_Fermi} can be written as
    \begin{equation}\label{Ginibre_representation_4_Fermi}
	\Gamma_p^{\nu,\zeta; \FF}(\vec{x},\vec{y})=\sum_{\pi \in S_p} \sgn(\pi) \sum_{\vec{k} \in (\N^{*})^p} \prod_{i=1}^p (-1)^{k_i-1} z^{k_i \nu}
        \int \mathbb{W}^{\nu \vec{k}}_{\pi \vec{y}, \vec{x}}(\dd \vec{\omega}) \,\frac{\Xi^{\nu,\zeta; \FF}(\vec{\omega})}{\Xi^{\nu,\zeta; \FF}}\,, \quad \vec{x},\vec{y} \in \Lambda^p\,. 
    \end{equation}
\end{itemize}
\end{lemma}
We give a sketch of the proof of Lemma \ref{Ginibre_representation_Fermi}. For the full details, we refer the reader to \cite[Section 2]{Ginibre}.
\begin{proof}[Sketch of proof of Lemma \ref{Ginibre_representation_Fermi}]
The proof of Lemma \ref{Ginibre_representation_Fermi} is similar to that of Lemma \ref{Ginibre_representation}. Here, one replaces \eqref{n_particle_space}, \eqref{grand_canonical_ensemble_1}, and \eqref{grand_canonical_ensemble} with \eqref{n_particle_space_Fermi}, \eqref{grand_canonical_ensemble_1_Fermi}, and \eqref{grand_canonical_ensemble_Fermi}, respectively. Using \eqref{sgn_pi}--\eqref{single_loop_measure_Fermi}, it follows that
the single-loop measures $\mathbb{L}^{\nu,\zeta;\FF}$ in \eqref{Ginibre_representation_1_Fermi} collectively absorb the signs of all permutations that appear in $\Xi^{\nu,\zeta;\FF}$. This proves \eqref{Ginibre_representation_1_Fermi}. See \cite[Section 4.3]{Ueltschi_2004} for more details.
	
	We now justify \eqref{Ginibre_representation_4_Fermi}. We again decompose permutations appearing in \eqref{antisymmetric_projector}, but only show how the signs of the permutations are distributed amongst those factors in \eqref{Ginibre_representation_4_Fermi} that differ from those in \eqref{Ginibre_representation_4}. To see the complete combinatorics, we again refer the reader to \cite[Section 2]{Ginibre}. More specifically, the relevant decomposition is performed in \cite[Lemma 2.1]{Ginibre}. 

	Recalling \eqref{reduced_p_particle_density_matrix_Fermi}, let us fix $p \in \N^{*}$, $n \in \N_{0}$, and a permutation $\pi \in S_{n+p}$. We first extract from $\pi$ all disjoint cycles that are also disjoint from $\{1, \ldots, p\}$. Referring to \cite[Section 2]{Ginibre} and \cite[Appendix A.2]{FKSS_2020_2}, these cycles correspond to closed loops. Together with their signs, these cycles are collected into $\Xi^{\nu,\zeta; \FF}(\vec{\omega})$ in \eqref{Ginibre_representation_4_Fermi}. Observe that $\Xi^{\nu,\zeta;\FF}$ appears in \eqref{Ginibre_representation_4} due to the normalisation in \eqref{grand_canonical_ensemble_Fermi}.

	We now split the rest of $\pi$ into $p$ (not necessarily disjoint) cycles of the form
	\begin{equation}
		C_i = (a_i c_1 \ldots c_{k_i-1} b_i)
	\end{equation}
	for some $i, k_i \in \N^{*}$, $a_i, b_i \in \{1, \ldots, p\}$, and $c_1, \ldots, c_{k_i-1} \in \{p+1, \ldots, p+n\}$. If $k_i = 1$ we set $C_i = (a_i b_i)$. Note that we allow $a_i = b_i$. Referring again to \cite[Section 2]{Ginibre} and \cite[Appendix A.2]{FKSS_2020_2}, these cycles correspond to open paths. In this particular example, the endpoints are $x_{a_i}$ and $y_{b_i}$, where $\vec{x}, \vec{y} \in \Lambda^{p}$ are fixed by \eqref{Ginibre_representation_4_Fermi}. $C_i$ can be further decomposed into
	\begin{equation}
		C_i = (a_i b_i) (a_i c_1 \ldots c_{k_i-1})\,.
	\end{equation}
	We then have that $\sgn(C_i) = \sgn(a_i b_i) (-1)^{k_{i}-1}$ and
	\begin{equation}
		\sgn(C_1 \cdots C_p) = \sgn(\pi') \prod_{i = 1}^{p} (-1)^{k_i-1}\,,
	\end{equation}
	where $\pi' = \prod_{i = 1}^{p} (a_i b_i) \in S_{p}$. These are exactly the remaining factors in \eqref{Ginibre_representation_4_Fermi} that have not yet been accounted for.
\end{proof}

We are now ready to prove Theorem \ref{partition_function_result_Fermi}.

\begin{proof}[Proof of Theorem \ref{partition_function_result_Fermi}]
	We first note that substituting \eqref{single_loop_measure_Fermi} into \eqref{Ginibre_representation_1_Fermi} yields the analogue of \eqref{Ginibre_representation_2}:
\begin{equation}
\label{Ginibre_representation_2_Fermi}
\Xi^{\nu,\zeta;\FF}=\sum_{n=0}^{\infty} \frac{1}{n!} \sum_{\vec{k} \in (\N^{*})^n} \prod_{i=1}^{n} (-1)^{k_i-1} \frac{z^{k_i\nu}}{k_i} \int \mathbb{W}^{\nu \vec{k}} (\dd \vec{\omega}) \, \exp \bigl(-\mathcal{V}^{\nu}(\vec{\omega})\bigr)\,.
\end{equation}
	Analogously to \eqref{Xi_hat} we define
	\begin{equation}
		\label{Xi_hat_Fermi}
		\widehat{\Xi}^{\nu,\zeta; \FF}\coloneqq \sum_{n=0}^{\infty} \frac{z^{n\nu}}{n!} \int \mathbb{W}^{\nu \vec{1}_n} (\dd \vec{\omega}) \, \exp \bigl(-\mathcal{V}^{\nu}(\vec{\omega})\bigr)\,.
	\end{equation}
	Note that $\widehat{\Xi}^{\nu,\zeta; \FF} \equiv \widehat{\Xi}^{\nu,\zeta; \BB}$, where we recall \eqref{Xi_hat}. Indeed, we are only summing over loops of duration $\nu$ in \eqref{Xi_hat_Fermi}, which have a sign 1.
Using \eqref{Ginibre_representation_2_Fermi} and \eqref{Xi_hat_Fermi} we get the analogue of \eqref{Proposition_2_1_1}:
	\begin{equation}
		\Xi^{\nu,\zeta; \FF} - \widehat{\Xi}^{\nu,\zeta; \FF} = \sum_{n = 1}^{\infty} \frac{1}{n!} \sum_{\substack{\vec{k} \in (\N^{*})^n \\ \vec{k} \neq \vec{1}_{n}}} \prod_{i = 1}^{n} (-1)^{k_i-1} \frac{z^{k_{i}\nu}}{k_{i}} \int \mathbb{W}^{\nu \vec{k}} (\dd \vec{\omega}) \, \exp \left(-\mathcal{V}^{\nu}(\vec{\omega})\right).
	\end{equation}
	By the triangle inequality, and then \eqref{Proposition_2_1_1} we have
	\begin{equation}
	\label{triangle_inequality_consequence}
		\bigl|\Xi^{\nu,\zeta; \FF} - \widehat{\Xi}^{\nu,\zeta; \FF}\bigr| \leq \sum_{n = 1}^{\infty} \frac{1}{n!} \sum_{\substack{\vec{k} \in (\N^{*})^n \\ \vec{k} \neq \vec{1}_{n}}} \prod_{i = 1}^{n} \frac{z^{k_{i}\nu}}{k_{i}} \int \mathbb{W}^{\nu \vec{k}} (\dd \vec{\omega}) \, \exp \left(-\mathcal{V}^{\nu}(\vec{\omega})\right) = \Xi^{\nu,\zeta; \BB} - \widehat{\Xi}^{\nu,\zeta; \BB}\,.
	\end{equation}
	We also used the fact that all terms of \eqref{Proposition_2_1_1} are nonnegative. The rest of the proof follows by arguing as in Proposition \ref{Proposition_1} and Proposition \ref{Proposition_2}. When applying the latter, recall that $\Xi^{0,\zeta;\FF} \equiv \Xi^{0,\zeta;\BB}$.
\end{proof}

\begin{remark}
The key step that allowed us to reduce the proof of Theorem \ref{partition_function_result_Fermi} for fermions to the proof of \ref{partition_function_result} given for bosons was that we could take absolute values and use the triangle inequality in \eqref{triangle_inequality_consequence} above.
\end{remark}

Before proceeding to the proof of Theorem \ref{density_matrix_result_Fermi}, we note the following lemma.

\begin{lemma}
	\label{difference_partition_fun_Bose_Fermi}
	Under the same notation and assumptions as in Theorem \ref{partition_function_result} and Theorem \ref{partition_function_result_Fermi}, we have that
	\begin{equation}
		\label{difference_partition_fun_Bose_Fermi_eq}
		\left| \frac{\Xi^{\nu,\zeta;\FF}}{\Xi^{\nu,\zeta;\BB}} - 1 \right| \lesssim_{d,\zeta,v,\alpha,B} \nu^{\alpha/2} \, \exp[ C_0(d, \zeta, B) L^d]\,.
	\end{equation}
	In particular, we have that
	\begin{gather}
		\Xi^{\nu,\zeta;\FF} \sim_{d,\zeta,v,\alpha,B} \Xi^{\nu,\zeta;\BB} \sim_{d,\zeta,v,\alpha,B} 1
		\label{partition_fun_Bose_Fermi_sim}
		\shortintertext{and}
		\Xi^{\nu,\zeta;\FF} > 0 \notag
	\end{gather}
	for $\nu$ small enough. We only consider such $\nu$ in the remainder of the section.
\end{lemma}

\begin{proof}
	By adding and subtracting $\Xi^{0,\zeta;\BB}$, and noting that $\Xi^{0,\zeta;\FF} = \Xi^{0,\zeta;\BB}$, we write
	\begin{equation}
		\left| \frac{\Xi^{\nu,\zeta;\FF}}{\Xi^{\nu,\zeta;\BB}} - 1 \right|  = \frac{|\Xi^{\nu,\zeta;\FF} - \Xi^{0,\zeta;\FF} + \Xi^{0,\zeta;\BB} - \Xi^{\nu,\zeta;\BB}|}{|\Xi^{\nu,\zeta;\BB}|} \\
										    \leq |\Xi^{\nu,\zeta;\FF} - \Xi^{0,\zeta;\FF}| + |\Xi^{0,\zeta;\BB} - \Xi^{\nu,\zeta;\BB}|\,.
	\end{equation}
	We also recalled that $\Xi^{\nu,\zeta;\BB} \geq 1$ for all $\nu > 0$ by \eqref{Ginibre_representation_1}. We now conclude from Theorem \ref{partition_function_result} and Theorem \ref{partition_function_result_Fermi} that for $\nu > 0$ sufficiently small,
	\begin{equation}
		\left| \frac{\Xi^{\nu,\zeta;\FF}}{\Xi^{\nu,\zeta;\BB}} - 1 \right| \lesssim_{d,\zeta,v,\alpha,B} \nu^{\alpha/2} \, \exp[C_0(d, \zeta, B) L^d]\,.
	\end{equation}
\end{proof}

We can now prove Theorem \ref{density_matrix_result_Fermi}. Let us first explain how to adapt the definitions of the objects we studied in Section \ref{section_density_matrix_Bose} to the fermionic setting. Analogously to \eqref{ploopkernel}--\eqref{p_loop_correlation_function} we rewrite \eqref{Ginibre_representation_4_Fermi} as
	\begin{equation}
		\label{ploopkernel_Fermi}
		\Gamma_p^{\nu,\zeta;\FF}(\vec{x},\vec{y})
		=\sum_{\pi \in S_p} \sgn(\pi) \sum_{\vec{k}\in (\N^*)^p}\prod_{i=1}^p (-1)^{k_{i}-1} z^{k_i \nu}\int \mathbb{W}^{\nu \vec{k}}_{\pi \vec{y},\vec{x}}(\dd \vec{\omega})\,\gamma_p^{\nu,\zeta;\FF}\,(\vec{\omega})\,,
	\end{equation}
	where
	\begin{equation}
		\label{p_loop_correlation_function_Fermi}
		\gamma_p^{\nu,\zeta; \FF}\,(\vec{\omega})\coloneqq \frac{1}{\Xi^{\nu,\zeta;\FF}}\sum_{n=0}^{\infty} \frac{1}{n!} \sum_{\vec{k} \in (\N^{*})^n} \prod_{i=1}^{n} (-1)^{k_{i}-1} \frac{z^{k_i\nu}}{k_i} \int \mathbb{W}^{\nu \vec{k}} (\dd \vec{\widetilde{\omega}}) \, \exp \bigl(-\mathcal{V}^{\nu}(\vec{\widetilde{\omega}}\vec{\omega})\bigr)\,.
	\end{equation}
	Analogously to \eqref{gamma_hat_definition}--\eqref{Gamma_hat_nu_zeta}, and \eqref{Gamma_check_definition} we respectively define the following three objects:
	\begin{gather}
		\widehat{\gamma}_p^{\nu,\zeta;\FF}(\vec{\omega})\coloneqq \frac{1}{\Xi^{\nu,\zeta;\FF}}\sum_{n=0}^{\infty} \frac{z^{n \nu}}{n!}  \int \mathbb{W}^{\nu \vec{1}_n} (\dd \vec{\widetilde{\omega}}) \, \exp \bigl(-\mathcal{V}^{\nu}(\vec{\widetilde{\omega}}\vec{\omega})\bigr)\,,
		\label{gamma_hat_definition_Fermi} \\
		\widehat{\Gamma}_p^{\nu,\zeta;\FF}(\vec{x},\vec{y}) \coloneqq \sum_{\pi \in S_p} \sgn(\pi) \sum_{\vec{k}\in (\N^{*})^p}\prod_{i=1}^p (-1)^{k_{i}-1} z^{k_i \nu}\int \mathbb{W}^{\nu \vec{k}}_{\pi \vec{y}, \vec{x}}(\dd \vec{\omega})\,\widehat{\gamma}_p^{\nu,\zeta;\FF}\,(\vec{\omega})\,,
		\label{Gamma_hat_nu_zeta_Fermi} \\
    		\widecheck{\Gamma}_{p}^{\nu,\zeta}(\vec{x},\vec{y})\coloneqq 
    		\frac{1}{\Xi^{\nu,\zeta;\FF}}\sum_{\pi \in S_p} \sgn(\pi) \sum_{n=0}^\infty\frac{z^{(n+p) \nu}}{n!}
    		\int \mathbb{W}^{\nu \vec{1}_n}(\dd \vec{\widetilde{\omega}}) 
    		\int \mathbb{W}^{\nu \vec{1}_p}_{\pi \vec{y}, \vec{x}}(\dd \vec{\omega}) 
   		\exp \bigl(-\mathcal{V}^{\nu}(\vec{\widetilde{\omega}}\vec{\omega}))\,.
		\label{Gamma_check_definition_Fermi}
	\end{gather}
	For a fixed $\pi \in S_p$ we define, analogously to \eqref{Gamma_pi_def}--\eqref{hat_Gamma_pi_def}, and \eqref{Gamma_check_pi_definition} respectively,
	\begin{gather}
		\Gamma_{p,\pi}^{\nu,\zeta;\FF}(\vec{x},\vec{y})\coloneqq \sgn(\pi) \sum_{\vec{k}\in (\N^{*})^p}\prod_{i=1}^p (-1)^{k_{i}-1} z^{k_i \nu}\int
		\mathbb{W}^{\nu \vec{k}}_{\pi \vec{y}, \vec{x}} (\dd \vec{\omega})\,\gamma_p^{\nu,\zeta;\FF}\,(\vec{\omega})\,, 
		\label{Gamma_pi_def_Fermi} \\
		\widehat{\Gamma}_{p,\pi}^{\nu,\zeta;\FF}(\vec{x},\vec{y})\coloneqq \sgn(\pi) \sum_{\vec{k}\in (\N^{*})^p}\prod_{i=1}^p (-1)^{k_{i}-1} z^{k_i \nu}\int 
		\mathbb{W}^{\nu \vec{k}}_{\pi \vec{y}, \vec{x}} (\dd \vec{\omega})\,\widehat{\gamma}_p^{\nu,\zeta;\FF}\,(\vec{\omega})\,, 
		\label{hat_Gamma_pi_def_Fermi} \\
		\widecheck{\Gamma}_{p,\pi}^{\nu,\zeta;\FF}(\vec{x},\vec{y})\coloneqq \frac{\sgn(\pi)}{\Xi^{\nu,\zeta;\FF}}\sum_{n=0}^\infty\frac{z^{(n+p) \nu}}{n!} \int 
		\mathbb{W}^{\nu \vec{1}_n}(\dd \vec{\widetilde{\omega}})\,
		\int \mathbb{W}^{\nu \vec{1}_p}_{\pi \vec{y}, \vec{x}} (\dd \vec{\omega}) \,\exp \bigl(-\mathcal{V}^{\nu}(\vec{\widetilde{\omega}}\vec{\omega})\bigr)\,.
		\label{Gamma_check_pi_definition_Fermi}
	\end{gather}
The next lemma provides the fermionic analogues of Lemma \ref{Proposition4_1} and Propositions \ref{Proposition_4_3_A}--\ref{Proposition_4_3}.
\begin{lemma}
	\label{density_matrix_result_aux_lemma_Fermi}
	Under the assumptions of Theorem \ref{density_matrix_result_Fermi} we have the following three bounds.
	\begin{gather}
		\bigl|\gamma_p^{\nu,\zeta;\FF}(\vec{\omega})-\widehat{\gamma}_p^{\nu,\zeta;\FF}(\vec{\omega})\bigr|
		\lesssim_{d,\zeta,v,\alpha,B,L}\nu^{d/2}\exp\left(\frac{B}{\nu}\sum_{i=1}^p T(\omega_i)\right )\,,
		\label{gamma_estimate_Fermi} \\
		\bigl|\Gamma_p^{\nu,\zeta;\FF}(\vec{x},\vec{y})-\widehat{\Gamma}_p^{\nu,\zeta;\FF}(\vec{x},\vec{y})\bigr|\lesssim_{p,d,\zeta,v,\alpha,B,L} \nu^{d/2}\,,
		\label{Gamma_minus_Gamma_hat_estimate_Fermi} \\
		\bigl|\widehat{\Gamma}_{p}^{\nu,\zeta;\FF}(\vec{x},\vec{y})-\widecheck{{\Gamma}}_{p}^{\nu,\zeta;\FF}(\vec{x},\vec{y})\bigr|\lesssim_{p,d,\zeta,v,\alpha,B,L} \nu^{d/2}\,.
		\label{Proposition_4_3_Fermi}
	\end{gather}
\end{lemma}

\begin{proof}
	We start by proving \eqref{gamma_estimate_Fermi}. Observe that for any $p \in \N^{*}$ and $\omega \in \Omega^{\vec{T}}$, where $\vec{T} \in (\nu \N^{*})^p$, we have that
	\begin{equation}
		\Xi^{\nu,\zeta,\FF} \bigl(\gamma_p^{\nu,\zeta;\FF}(\vec{\omega}\bigr) - \widehat{\gamma}_p^{\nu,\zeta;\FF}(\vec{\omega})) = 
    		\sum_{n = 1}^{\infty} \frac{1}{n!} \sum_{\substack{\vec{k} \in (\N^{*})^n \\ \vec{k} \neq \vec{1}_n}} \prod_{i = 1}^{n} (-1)^{k_{i}-1} \frac{z^{k_{i}\nu}}{k_{i}} 
		\int \mathbb{W}^{\nu \vec{k}} (\dd \vec{\widetilde{\omega}}) \, \exp \left(-\mathcal{V}^{\nu}(\vec{\widetilde{\omega}}\vec{\omega})\right) \,.
	\end{equation} 
	By the triangle inequality
    	\begin{align}
		\Xi^{\nu,\zeta;\FF}\bigl|\gamma_p^{\nu,\zeta;\FF}(\vec{\omega})-\widehat{\gamma}_p^{\nu,\zeta;\FF}(\vec{\omega})\bigr| & \leq
    		\sum_{n = 1}^{\infty} \frac{1}{n!} \sum_{\substack{\vec{k} \in (\N^{*})^n \\ \vec{k} \neq \vec{1}_n}} \prod_{i = 1}^{n} \frac{z^{k_{i}\nu}}{k_{i}} 
    		\int \mathbb{W}^{\nu \vec{k}} (\dd \vec{\widetilde{\omega}}) \, \exp \left(-\mathcal{V}^{\nu}(\vec{\widetilde{\omega}}\vec{\omega})\right) \\
		& = \Xi^{\nu,\zeta,\BB} \bigl(\gamma_p^{\nu,\zeta;\BB}(\vec{\omega}) - \widehat{\gamma}_p^{\nu,\zeta;\BB}(\vec{\omega})\bigr)\,,
	\end{align}
	where we recalled \eqref{proposition_4_1_1} in the last line. We therefore have that
	\begin{equation}
		\bigl|\gamma_p^{\nu,\zeta;\FF}(\vec{\omega})-\widehat{\gamma}_p^{\nu,\zeta;\FF}(\vec{\omega})\bigr| 
		\leq \frac{\Xi^{\nu,\zeta,\BB}}{\Xi^{\nu,\zeta;\FF}} \bigl(\gamma_p^{\nu,\zeta;\BB}(\vec{\omega}) - \widehat{\gamma}_p^{\nu,\zeta;\BB}(\vec{\omega})\bigr)\,.
	\end{equation}
	The bound in \eqref{gamma_estimate_Fermi} now follows by applying Lemma \ref{Proposition4_1} and Lemma \ref{difference_partition_fun_Bose_Fermi}.

	Next we show \eqref{Gamma_minus_Gamma_hat_estimate_Fermi}. We calculate
	\begin{equation}
		\Gamma_{p, \pi}^{\nu,\zeta;\FF}(\vec{x},\vec{y})-\widehat{\Gamma}_{p,\pi}^{\nu,\zeta;\FF}(\vec{x},\vec{y}) = \sgn(\pi) \sum_{\vec{k}\in (\N^{*})^p}\prod_{i=1}^p 
		(-1)^{k_{i}-1} z^{k_i \nu}\int \mathbb{W}^{\nu \vec{k}}_{\pi \vec{y}, \vec{x}} (\dd \vec{\omega})\,
		\bigl(\gamma_p^{\nu,\zeta;\FF}\,(\vec{\omega}) - \widehat{\gamma}_p^{\nu,\zeta;\FF}\,(\vec{\omega})\bigr) \,.
	\end{equation}
	Taking the absolute value and applying the triangle inequality gives
	\begin{equation}
		\bigl|\Gamma_{p, \pi}^{\nu,\zeta;\FF}(\vec{x},\vec{y})-\widehat{\Gamma}_{p,\pi}^{\nu,\zeta;\FF}(\vec{x},\vec{y})\bigr| 
		\leq \sum_{\vec{k}\in (\N^{*})^p}\prod_{i=1}^p z^{k_i \nu}\int \mathbb{W}^{\nu \vec{k}}_{\pi \vec{y}, \vec{x}} (\dd \vec{\omega})\,
		\bigl|\gamma_p^{\nu,\zeta;\FF}\,(\vec{\omega}) - \widehat{\gamma}_p^{\nu,\zeta;\FF}\,(\vec{\omega})\bigr| \,.
	\end{equation}
	Applying \eqref{gamma_estimate_Fermi} yields
	\begin{equation}
		\bigl|\Gamma_{p, \pi}^{\nu,\zeta;\FF}(\vec{x},\vec{y})-\widehat{\Gamma}_{p,\pi}^{\nu,\zeta;\FF}(\vec{x},\vec{y})\bigr| 
		\lesssim_{d,\zeta,v,\alpha,B,L} \nu^{d/2} \sum_{\vec{k}\in (\N^{*})^p}\prod_{i=1}^p z^{k_i \nu}\int
		\mathbb{W}^{\nu \vec{k}}_{\pi \vec{y}, \vec{x}} (\dd \vec{\omega})\, \exp\left(\frac{B}{\nu}\sum_{i=1}^p T(\omega_i)\right )\,.
	\end{equation}
	Comparing this to \eqref{proposition4_2_2}, we can apply the rest of the proof of Proposition \ref{Proposition_4_3_A}, and \eqref{Gamma_minus_Gamma_hat_estimate_Fermi} follows.

	Lastly, we prove \eqref{Proposition_4_3_Fermi}. As before, we calculate
	\begin{align}
	     & \bigl|\widehat{\Gamma}_{p,\pi}^{\nu,\zeta;\FF}(\vec{x},\vec{y})-\widecheck{\Gamma}_{p,\pi}^{\nu,\zeta;\FF}(\vec{x},\vec{y})\bigr| 
	     \nonumber \\
	     & \quad = \left| \frac{\sgn(\pi)}{\Xi^{\nu,\zeta;\FF}}\sum_{n=1}^\infty\sum_{\substack{\vec{k} \in (\N^{*})^p \\ \vec{k} \neq \vec{1}_{p}}}\frac{z^{n\nu}}{n!}
	     \prod_{i=1}^p (-1)^{k_{i}-1} z^{k_i \nu} \int \mathbb{W}^{\nu \vec{1}_n}(\dd \vec{\widetilde{\omega}})\,
	     \int \mathbb{W}^{\nu \vec{k}}_{\pi \vec{y}, \vec{x}} (\dd \vec{\omega}) \,\exp \bigl(-\mathcal{V}^{\nu}(\vec{\widetilde{\omega}}\vec{\omega})\bigr) \right| 
	     \\
	     & \quad \leq \frac{1}{\Xi^{\nu,\zeta;\FF}}\sum_{n=1}^\infty\sum_{\substack{\vec{k} \in (\N^{*})^p \\ \vec{k} \neq \vec{1}_{p}}}\frac{z^{n\nu}}{n!}
	     \prod_{i=1}^p z^{k_i \nu} \int \mathbb{W}^{\nu \vec{1}_n}(\dd \vec{\widetilde{\omega}})
	     \int \mathbb{W}^{\nu \vec{k}}_{\pi \vec{y}, \vec{x}} (\dd \vec{\omega}) \,\exp \bigl(-\mathcal{V}^{\nu}(\vec{\widetilde{\omega}}\vec{\omega})\bigr) 
	     \\
	     & \quad \sim_{d,\zeta,v,\alpha,B} \frac{1}{\Xi^{\nu,\zeta;\BB}}\sum_{n=1}^\infty\sum_{\substack{\vec{k} \in (\N^{*})^p \\ \vec{k} \neq \vec{1}_{p}}}
	     \frac{z^{n\nu}}{n!} \prod_{i=1}^p z^{k_i \nu} \int \mathbb{W}^{\nu \vec{1}_n}(\dd \vec{\widetilde{\omega}})\,
	     \int \mathbb{W}^{\nu \vec{k}}_{\pi \vec{y}, \vec{x}} (\dd \vec{\omega}) \,\exp \bigl(-\mathcal{V}^{\nu}(\vec{\widetilde{\omega}}\vec{\omega})\bigr)\,,
	     \label{proposition_4_3_Fermi}
	\end{align}
	where we used \eqref{partition_fun_Bose_Fermi_sim} to get to \eqref{proposition_4_3_Fermi}. Comparing \eqref{proposition_4_3_Fermi} to \eqref{proposition_4_3_2_A}, we can apply the rest of the proof of Proposition \ref{Proposition_4_3} to derive \eqref{Proposition_4_3_Fermi}.
\end{proof}

We now have all the necessary ingredients to prove Theorem \ref{density_matrix_result_Fermi}.
\begin{proof}[Proof of Theorem \ref{density_matrix_result_Fermi}]
     We need to show the analogues of \eqref{Theorem_1.5(ii)} and \eqref{Theorem_1.5(ii)_rate_of_convergence_rewritten}. We do this by applying Lemma \ref{density_matrix_result_aux_lemma_Fermi} above. Recall the definition of $\Delta_p^{0,\zeta} \equiv \Delta_p^{0,\zeta;\BB}$ in \eqref{diagonal_contribution}, and recall that $\Xi^{0,\zeta;\FF} \equiv \Xi^{0,\zeta;\BB}$ and $\Xi^{0,\zeta;\FF}(\vec{x}) \equiv \Xi^{0,\zeta;\BB}(\vec{x})$ for every $\vec{x} \in \Lambda^p$. We therefore have that
     \begin{equation}
	     \label{diagonal_contribution_Fermi}
	     \Delta_p^{0,\zeta;\BB}(\vec{x}) = (2\pi)^{-\frac{pd}{2}}\,\zeta^p\,\frac{\Xi^{0,\zeta;\BB}(\vec{x})}{\Xi^{0,\zeta;\BB}} 
	     = (2\pi)^{-\frac{pd}{2}}\,\zeta^p\,\frac{\Xi^{0,\zeta;\FF}(\vec{x})}{\Xi^{0,\zeta;\FF}} \eqqcolon \Delta_p^{0,\zeta;\FF}(\vec{x})\,.
     \end{equation}
     For simplicity, we will henceforth denote \eqref{diagonal_contribution_Fermi} simply by $\Delta_p^{0,\zeta}(\vec{x})$.
     Using \eqref{density_matrix_result_2_Fermi}, \eqref{diagonal_contribution_Fermi}, and recalling \eqref{Pi_{x,y}}, we can rewrite \eqref{density_matrix_result_1_Fermi} as
     \begin{equation}
	     \label{Theorem_1.5(ii)_Fermi}
	     \lim_{\nu \rightarrow 0} \Gamma_p^{\nu,\zeta;\FF}(\vec{x},\vec{y})=\sum_{\pi \in \Pi_{\vec{x},\vec{y}}} \sgn(\pi) \,\Delta_p^{0,\zeta}(\vec{x}) \,,
     \end{equation}
     and \eqref{Theorem_1.5(ii)_rate_of_convergence_Fermi} as
     \begin{multline}
	     \label{Theorem_1.5(ii)_rate_of_convergence_rewritten_Fermi}
	     \Bigl|\Gamma_p^{\nu,\zeta;\FF}(\vec{x},\vec{y})-\sum_{\pi \in \Pi_{\vec{x},\vec{y}}} \sgn(\pi)\,\Delta_p^{0,\zeta}(\vec{x}) \Bigr| \\
	     \lesssim_{p,d,\zeta,v,\alpha,B,L} \Biggl[\nu^{\alpha/2}+\sum_{\tilde{\pi} \in S_p \setminus \Pi_{\vec{x},\vec{y}}} \prod_{i=1}^{p} \Biggl\{\nu^{d/2}+ 
	     \exp \biggl(-\frac{|y_{\tilde{\pi}(i)}-x_i|_{L}^2}{2\nu}\biggr)\Biggr\}\Biggr]\,.
     \end{multline}
     Just as in the proof of Theorem~\ref{density_matrix_result}, we need to consider two cases for the set $\Pi_{\vec{x},\vec{y}}$.\\
     \textbf{Case 1}: $\Pi_{\vec{x},\vec{y}}=\emptyset$.
     Observe that by applying the triangle inequality to \eqref{Gamma_check_definition_Fermi}, and using \eqref{partition_fun_Bose_Fermi_sim} we have
     \begin{equation}
     	\bigl|\widecheck{\Gamma}_p^{\nu,\zeta;\FF}(\vec{x},\vec{y})\bigr| \leq \frac{\Xi^{\nu,\zeta;\BB}}{\Xi^{\nu,\zeta;\FF}} \widecheck{\Gamma}_p^{\nu,\zeta;\BB}(\vec{x},\vec{y}) 
		\sim_{d,\zeta,v,\alpha,B} \widecheck{\Gamma}_p^{\nu,\zeta;\BB}(\vec{x},\vec{y})\,.
     \end{equation}
     We therefore deduce from \eqref{check_Gamma_estimation_case_1}--\eqref{check_Gamma_estimation_case_1_B} that
     \begin{equation}
	     \bigl|\widecheck{\Gamma}_p^{\nu,\zeta;\FF}(\vec{x},\vec{y})\bigr| 
	     \lesssim_{p,d,\zeta,v,\alpha,B,L} \sum_{\pi \in S_p} \prod_{i=1}^p \left\{\nu^{d/2}+\exp\left({-\frac{|x_i-y_{\pi(i)}|^2_L}{2\nu}}\right)\right\}.
     \end{equation}
     Combining this with Lemma \ref{density_matrix_result_aux_lemma_Fermi} yields \eqref{Theorem_1.5(ii)_Fermi}--\eqref{Theorem_1.5(ii)_rate_of_convergence_rewritten_Fermi}.

     \textbf{Case 2}: $\Pi_{\vec{x},\vec{y}} \neq \emptyset$.
     As in Case 1, we apply the triangle inequality and \eqref{partition_fun_Bose_Fermi_sim} to \eqref{Gamma_pi_def_Fermi} to get
     \begin{equation}
	     \bigl|\Gamma_{p,\tilde{\pi}}^{\nu,\zeta;\FF}(\vec{x},\vec{y})\bigr| \leq \frac{\Xi^{\nu,\zeta;\BB}}{\Xi^{\nu,\zeta;\FF}} \Gamma_{p,\tilde{\pi}}^{\nu,\zeta;\BB}(\vec{x},\vec{y})
	     \sim_{d,\zeta,v,\alpha,B} \Gamma_{p,\tilde{\pi}}^{\nu,\zeta;\BB}(\vec{x},\vec{y})
     \end{equation}
     for all $\tilde{\pi} \in S_p \setminus \Pi_{\vec{x},\vec{y}}$. By \eqref{x=pi(y)_1} we then have
     \begin{equation}
	     \label{x=pi(y)_1_Fermi}
	     \sum_{\tilde{\pi} \in S_p \setminus \Pi_{\vec{x},\vec{y}}} \bigl|\Gamma_{p,\tilde{\pi}}^{\nu,\zeta;\FF}(\vec{x},\vec{y})\bigr| 
	     \lesssim_{p,d,\zeta,v,\alpha,B,L} \Biggl[\nu^{d/2}+\sum_{\tilde{\pi} \in S_p \setminus \Pi_{\vec{x},\vec{y}}} \prod_{i=1}^{p} \Biggl\{\nu^{d/2}+ 
	     \exp \biggl(-\frac{|y_{\tilde{\pi}(i)}-x_i|_{L}^2}{2\nu}\biggr)\Biggr\}\Biggr]\,.
     \end{equation}
     Observe that for $\pi \in \Pi_{\vec{x},\vec{y}}$,
     \begin{equation}
	     \widecheck{\Gamma}_{p,\pi}^{\nu,\zeta;\FF}(\vec{x},\vec{y}) = \sgn(\pi) \, \frac{\Xi^{\nu,\zeta;\BB}}{\Xi^{\nu,\zeta;\FF}} \, \widecheck{\Gamma}_{p,\pi}^{\nu,\zeta;\BB}(\vec{x},\vec{y})\,,
     \end{equation}
     and so
     \begin{align}
	     \bigl|\widecheck{\Gamma}_{p,\pi}^{\nu,\zeta; \FF}(\vec{x},\vec{y})- \sgn(\pi) \, \Delta_p^{0,\zeta}(\vec{x}) \bigr| & = 
	     \left| \frac{\Xi^{\nu,\zeta;\BB}}{\Xi^{\nu,\zeta;\FF}} \, \widecheck{\Gamma}_{p,\pi}^{\nu,\zeta; \BB}(\vec{x},\vec{y})- \Delta_p^{0,\zeta}(\vec{x}) \right| 
	     \notag \\
	     & = \left|\frac{\Xi^{\nu,\zeta;\BB}}{\Xi^{\nu,\zeta;\FF}} \left( \widecheck{\Gamma}_{p,\pi}^{\nu,\zeta; \BB}(\vec{x},\vec{y})- \Delta_p^{0,\zeta}(\vec{x}) \right) 
		     + \left(\frac{\Xi^{\nu,\zeta;\BB}}{\Xi^{\nu,\zeta;\FF}} - 1 \right) \Delta_p^{0,\zeta}(\vec{x}) \right| 
	     \notag \\
	     & \leq \left|\frac{\Xi^{\nu,\zeta;\BB}}{\Xi^{\nu,\zeta;\FF}} \right| \, \left| \widecheck{\Gamma}_{p,\pi}^{\nu,\zeta; \BB}(\vec{x},\vec{y})
	     - \Delta_p^{0,\zeta}(\vec{x}) \right| + \left|\frac{\Xi^{\nu,\zeta;\BB}}{\Xi^{\nu,\zeta;\FF}} - 1 \right| \, \bigl|\Delta_p^{0,\zeta}(\vec{x})\bigr|\,.
     \end{align}
     We now apply Lemma \ref{difference_partition_fun_Bose_Fermi} and \eqref{x=pi(y)_2} to get
     \begin{equation}
	     \label{x=pi(y)_2_Fermi}
	     \bigl|\widecheck{\Gamma}_{p,\pi}^{\nu,\zeta; \FF}(\vec{x},\vec{y})- \sgn(\pi)\,\Delta_p^{0,\zeta}(\vec{x}) \bigr| \lesssim_{p,d,\zeta,v,\alpha,B,L} \nu^{\alpha/2}\,.
     \end{equation}
     We also used that $\Delta_p^{0,\zeta}(\vec{x}) \lesssim_{p,d,\zeta,v,B,L} 1$ for all $\vec{x} \in \Lambda^p$. Combining \eqref{Gamma_minus_Gamma_hat_estimate_Fermi}--\eqref{Proposition_4_3_Fermi}, and \eqref{x=pi(y)_2_Fermi}, we obtain that for all $\pi \in \Pi_{\vec{x},\vec{y}}$,
     \begin{equation}
	     \label{x=pi(y)_3_Fermi}
	     \bigl|\Gamma_{p,\pi}^{\nu,\zeta}(\vec{x},\vec{y})- \sgn(\pi) \Delta_p^{0,\zeta}(\vec{x}) \bigr| \lesssim_{p,d,\zeta,v,\alpha,B,L} \nu^{\alpha/2}\,.
     \end{equation}
     We hence conclude \eqref{Theorem_1.5(ii)_Fermi}--\eqref{Theorem_1.5(ii)_rate_of_convergence_rewritten_Fermi} in this case from \eqref{x=pi(y)_1_Fermi} and \eqref{x=pi(y)_3_Fermi} after summing over all permutations in $S_p$.

\end{proof}

\section{The analysis in infinite volume}
\label{section_infinite_volume}

In this section, we study the infinite-volume limit and prove Theorems \ref{free_energy_infinite_volume} and  \ref{Infinite_volume_Theorem_2}. Throughout, we work in the bosonic regime. For an analogue of Theorem \ref{Infinite_volume_Theorem_2} below for Boltzmann and Fermi statistics, see Remark \ref{fermions_infinite_volume}.

Let us first analyse the interaction potential $v^L$ on $\Lambda_L$ given by \eqref{v^L} for $v$ as in Assumption \ref{Assumption_on_v_2}.

\begin{lemma}
\label{v^L_lemma} 
Let $v$ be given by Assumption \ref{Assumption_on_v_2} and let $v^L$ be given by \eqref{v^L} above. 
Then the following results hold.
\begin{itemize}
\item[(i)]
$v^L$ is a well-defined nonnegative function on $\Lambda_L$ and it satisfies Assumption \ref{Assumption_on_v}  with $B=0$ and $\alpha=1$.
\item[(ii)] We have
\begin{equation}
\label{Lemma_4.1_iii}
\|v^L\|_{L^1(\Lambda_L)}=\|v\|_{L^1(\R^d)}<\infty\,.
\end{equation}
\item[(iii)] Furthermore, we have
\begin{equation}
\label{Subcase_2B_4A}
\|v^L\|_{L^{\infty}(\Lambda_L)}=\mathcal{O}(1)\,,\,\, \text{uniformly in $L \geq 1$\,.}
\end{equation}
\item[(iv)] We have 
\begin{equation*}
\lim_{L \rightarrow \infty} \|v-v^L\|_{L^1(\Lambda_L)}=0\,.
\end{equation*}
\item[(v)] We have
\begin{equation*}
\lim_{L \rightarrow \infty} \|v-v^L\|_{L^{\infty}(\Lambda_L)}=0\,.
\end{equation*}

\end{itemize}
\end{lemma}

\begin{proof}
We first show (i). Let us note that \eqref{v^L} is a convergent sum. Namely, we have that for fixed $x \in \Lambda_L$
\begin{equation}
\label{v^L_2}
\sum_{k \in \Z^d} v(x+Lk)=\mathop{\sum_{k \in \Z^d}}_{|k| \lesssim 1} v(x+Lk)+
\mathop{\sum_{k \in \Z^d}}_{|k| \gg 1} v(x+Lk)\,.
\end{equation}
We only need to consider the second sum on the right-hand side of \eqref{v^L_2}. It is finite because by \eqref{v_assumption_iii} we have 
\begin{equation*}
0 \leq \mathop{\sum_{k \in \Z^d}}_{|k| \gg 1} v(x+Lk) \lesssim  \mathop{\sum_{k \in \Z^d}}_{|k| \gg 1} \frac{1}{\langle x+Lk \rangle^{d+\varepsilon_0}} \lesssim \mathop{\sum_{k \in \Z^d}}_{|k| \gg 1} \frac{1}{\langle Lk \rangle^{d+\varepsilon_0}} \lesssim_{d,\varepsilon_0} 1\,.
\end{equation*}
Hence \eqref{v^L} is a convergent sum. By construction, we have 
\begin{equation*}
v^L(x+Lk)=v^L(x)
\end{equation*}
for all $k \in \Z^d$, so $v^L$ indeed defines a function on $\Lambda_L$.

By Assumption \ref{Assumption_on_v_2} (ii) and \eqref{v^L}, it follows that $v^L$ is pointwise nonnegative, hence it satisfies Assumption \ref{Assumption_on_v} (i) with $B=0$. Let us now show that $v^L$ satisfies Assumption \ref{Assumption_on_v} (ii) with $\alpha=1$. To this end, let $x, y \in \Lambda_L$ be given. By \eqref{periodic_Euclidean_norm}, there exists $m \in \Z^d$ such that 
\begin{equation}
\label{m_choice}
|x-y|_{L}=|x-y-Lm|\,.
\end{equation}
Note that by construction, we have 
\begin{equation*}
|m| \leq \sqrt{d}\,.
\end{equation*}
We can rewrite \eqref{v^L} as 
\begin{equation}
\label{v^L_rewritten}
v^L(y)=\sum_{k \in \Z^d}  v(y+Lm+Lk)\,.
\end{equation}
By \eqref{v^L}, the triangle inequality, the mean-value theorem, Assumption \ref{Assumption_on_v_2} (i), (iii), and \eqref{m_choice}--\eqref{v^L_rewritten}, we have
\begin{multline}
\label{v^L_Lipschitz}
\left|v^L(x)-v^L(y)\right| \leq \sum_{k \in \Z^d} \left| v(x+Lk)-v(y+Lm+Lk) \right| 
\\
\lesssim \mathop{\sum_{k \in \Z^d}}_{|k| \lesssim 1} |x-y|_{L}+\mathop{\sum_{k \in \Z^d}}_{|k| \gg 1} |x-y|_{L}\, \frac{1}{\langle Lk \rangle^{d+\varepsilon_0}}\lesssim_{d,\varepsilon_0} |x-y|_{L}\,.
\end{multline}
Hence, $v^L$ is indeed Lipschitz by \eqref{v^L_Lipschitz}.
Claim (i) now follows.

We now show claim (ii). By Assumption \ref{Assumption_on_v_2} (iii), it follows that $v \in L^1(\R^d)$. We deduce the identity \eqref{Lemma_4.1_iii} from \eqref{v^L} and Assumption \ref{Assumption_on_v_2} (ii). We hence obtain claim (ii).
We note that claim (iii) follows from \eqref{v^L}, since $|v(x)| \lesssim \langle x \rangle^{-d-\varepsilon_0}$ by \eqref{v_assumption_iii}.

In order to show (iv), we note that by the integrability of $v$,
\begin{equation}
\label{Lemma_5.4_(iv)_proof}
\|v-v^L\|_{L^1(\Lambda_L)}=\sum_{k \in  \mathbb{Z}^d \setminus \{0\}} \int_{\Lambda_L} \dd x\,v(x+Lk)=\|v\|_{L^1(\mathbb{R}^d \setminus \Lambda_L)} \rightarrow 0  \text{ as } L \rightarrow \infty\,.
\end{equation}
We show (v) analogously. Using \eqref{v_assumption_iii} and the fact that $|x+Lk| \gtrsim L$ for all $k \in \Z^d \setminus \{0\}$, we have
\begin{equation*}
\|v-v^L\|_{L^{\infty}(\Lambda_L)}=\sup_{x \in \Lambda_L} \sum_{k \in \Z^d \setminus \{0\}} v(x+Lk) \rightarrow 0 \text{ as } L \rightarrow \infty\,.
\end{equation*}
\end{proof}

As noted in the introduction, we emphasise the $L$-dependence of all the quantities (except $\Lambda_L$) by means of a superscript $L$. In particular, we write $\Omega^{L,T}, \Omega^{L,T}_{y,x}$ for the appropriate spaces of paths on $\Lambda_L$, defined 
as in Section \ref{Section_1.3}.
Furthermore, we write $\psi^{t} \equiv \psi^{L,t}$ for the heat kernel \eqref{heat_kernel} on $\Lambda_L$ and $\mathbb{W}_{y,x}^{L,T}(\dd \omega), \mathbb{W}^{L,T}(\dd \omega)$ for the positive measures \eqref{Wiener_measure} and \eqref{W^T_definition} respectively. 

The starting point of the analysis is to rewrite the logarithm of the quantum relative partition function \eqref{quantum_relative_partition_function} as well as the reduced $p$-particle density matrices \eqref{reduced_p_particle_density_matrix} as a power series which is amenable to a cluster expansion. This is done using the Ginibre representation from Lemma \ref{Ginibre_representation}. 
Let us introduce some notation. Let $\mathcal{V}^{\nu,L}(\omega)$ be defined as in \eqref{self_interaction} with interaction potential $v^L$ given by \eqref{v^L}. Let us also consider the following two positive measures defined on 
$\Omega^L \equiv \bigcup_{T > 0} \Omega^{L,T}$.
\begin{align}
\label{measure_mu_1}
\mu^{L}(\dd \omega) &\coloneqq  \nu \sum_{T\in \nu {\N}^*} \frac{z^T}{T} \,  \mathbb{W}^{L,T}(\dd \omega)
\,\ee^{-\mathcal{V}^{\nu,L}(\omega)}\,,
\\
\label{measure_mu_2}
\widehat{\mu}^{L}_{y, x}(\dd \omega) &\coloneqq  \sum_{T \in \nu \N^*} z^T \, \mathbb{W}^{L,T}_{y,x}(\dd \omega)\,\ee^{-\mathcal{V}^{\nu,L}(\omega)}\,.
\end{align}
Given $n \in \N^*$, and recalling \eqref{loop_interaction}, we define the \emph{Ursell function} as 
\begin{equation}
\label{Ursell_function_definition}
\varphi^{L}(\omega_{1},\dots, \omega_{n})
\coloneqq 
\frac{1}{n!} \sum_{\mathcal{G} \in {\mathfrak{G}}^c_n} \prod_{\{i,j\} \in \mathcal{G}} \xi^{L}(\omega_i, \omega_j)\,, 
\quad \xi^L(\omega, \tilde \omega)\coloneqq  \exp\bigl(-\mathcal{V}^{\nu,L}(\omega,\tilde \omega)\bigr) - 1\,.
\end{equation}
In \eqref{Ursell_function_definition}, we denote by ${\mathfrak{G}}^c_n$ the set of all connected graphs on $[n] \equiv \{1,2,\ldots,n\}$. Using Assumption \ref{Assumption_on_v_2} (ii), \eqref{v^L}, \eqref{loop_interaction}, and \eqref{Ursell_function_definition}, we have that
\begin{equation}
\label{xi_bound}
|\xi^L(\omega,\tilde{\omega})| \leq \mathcal{V}^{\nu,L}(\omega,\tilde \omega)\,.
\end{equation}
Similarly, recalling \eqref{self_interaction}, we have from \eqref{measure_mu_1}--\eqref{measure_mu_2} that 
\begin{align}
\label{measure_mu_1_bound}
\mu^{L}(\dd \omega) &\leq  \nu \sum_{T\in \nu {\N}^*} \frac{z^T}{T} \,  \mathbb{W}^{L,T}(\dd \omega) \eqqcolon \mu^{L,0}(\dd \omega)
\,,
\\
\label{measure_mu_2_bound}
\widehat{\mu}^{L}_{y, x}(\dd \omega) &\leq  \sum_{T \in \nu \N^*} z^T \, \mathbb{W}^{L,T}_{y,x}(\dd \omega) \eqqcolon \widehat{\mu}^{L,0}_{y, x}(\dd \omega)\,.
\end{align}

Given  $p \in \N$ and $\omega_1,\ldots,\omega_p \in \Omega$, using \eqref{measure_mu_1}--\eqref{Ursell_function_definition}, let us define 
\begin{equation*}
X^{L}(\omega_1,\ldots,\omega_p) \equiv X^{\nu,\zeta,L}(\omega_1,\ldots,\omega_p)
\end{equation*}
by
\begin{multline}
\label{Ursell_function_representation}
X^{L}(\omega_{1}, \dots, \omega_{p}) \coloneqq  \sum_{n \geq \max\{p,1\}} \frac{n!}{(n-p)!} \int \mu^{L}(\dd \omega_{p+1}) \cdots \mu^L(\dd \omega_n)\, \varphi^{L}(\omega_1, \dots, \omega_n)\,,
\\ 
X^L \coloneqq  X^L (\emptyset)\,.
\end{multline}
For future reference, we also define the quantity
\begin{multline}
\label{Ursell_function_representation_absolute_value}
|X|^{L}(\omega_{1}, \dots, \omega_{p}) \coloneqq  \sum_{n \geq \max\{p,1\}} \frac{n!}{(n-p)!} \int \mu^{L}(\dd \omega_{p+1}) \cdots \mu^L(\dd \omega_n)\, |\varphi^{L}(\omega_1, \dots, \omega_n)|\,,
\\
\quad  |X|^L \coloneqq  |X|^L (\emptyset)\,.
\end{multline}
Note that by the triangle inequality and the positivity of $\mu^L$, we have
\begin{equation}
\label{X^L_absolute_value_bound}
|X^L(\omega_1,\ldots,\omega_p)| \leq |X|^L(\omega_{1}, \dots, \omega_{p})\,.
\end{equation}
We henceforth suppress the $\nu$-dependence in \eqref{measure_mu_1}--\eqref{Ursell_function_representation} for simplicity.
For the remainder of the section, we recall the assumption \eqref{zeta_infinite_volume} on $\zeta$ and consider $\nu \leq 1$.

We can now set up the cluster expansion.
\begin{proposition}[The cluster expansion]
\label{gamma_p_cluster_expansion}
Suppose that \eqref{smallness_v_kappa} holds.
Recalling \eqref{measure_mu_2} and \eqref{Ursell_function_representation}, the following representations hold.
\begin{itemize}
\item[(i)] $\log  \Xi^{\nu,\zeta,L}=X^{L}$.
\item[(ii)]
Let $p \in \N^*$ and $\vec x, \vec y \in \Lambda_L^p$ be given. We have
\begin{equation}
\label{gamma_p_cluster_expansion_claim}
\Gamma^{\nu,\zeta,L}_{p}({\vec x,\vec y}) = \sum_{\pi \in S_p} \int \widehat \mu^{L}_{y_{\pi(1)},x_{1}}(\dd \omega_{1}) \cdots \widehat \mu^{L}_{y_{\pi(p)},x_{p}}(\dd \omega_{p}) \sum_{\Pi \in \mathfrak{P}_{p}} \prod_{\lambda \in \Pi} X^{L}((\omega_{i})_{i \in \lambda})\,.
\end{equation}
Here, $\mathfrak{P}_{p}$ denotes the set of partitions of $[p] = \{1, \ldots, p\}$.
\end{itemize}
\end{proposition}

Proposition \ref{gamma_p_cluster_expansion} allows us to obtain the following uniform bounds on the free energy and reduced density matrices.

\begin{proposition} [Bounds on the free energy and reduced density matrices in a finite volume]
\label{Proposition_5.3}
Under the assumption \eqref{smallness_v_kappa}, the following bounds hold for all $L \in \N^*$.
\begin{itemize}
\item[(i)] The specific Gibbs potential \eqref{specific_Gibbs_potential} of the Bose gas satisfies
\begin{equation}
\label{free_energy_bound_1}
g^{\nu,\zeta,L} =\mathcal {O}_{d,\zeta,v}(1)\,.
\end{equation}
\item[(ii)] For $p \in \N^*$, the $p$-particle reduced density matrix \eqref{reduced_p_particle_density_matrix} satisfies 
\begin{equation} 
\label{operator_norm_bound_1}
\left\|\Gamma^{\nu,\zeta,L}_p\right\|_{L^{\infty}_{\vec{x} } L^1_{\vec{y}}}=\mathcal {O}_{d,\zeta,p,v}(1)\,.
\end{equation}
Moreover, recalling \eqref{Ursell_function_representation_absolute_value} and  defining 
\begin{equation}
\label{gamma_p_cluster_expansion_absolute_value}
|\Gamma|^{\nu,\zeta,L}_{p}({\vec x,\vec y}) \coloneqq \sum_{\pi \in S_p} \int \widehat \mu^{L}_{y_{\pi(1)},x_{1}}(\dd \omega_{1}) \cdots \widehat \mu^{L}_{y_{\pi(p)},x_{p}}(\dd \omega_{p}) \sum_{\Pi \in \mathfrak{P}_{p}} \prod_{\lambda \in \Pi} |X|^{L}((\omega_{i})_{i \in \lambda})\,,
\end{equation}
we have the stronger bound
\begin{equation}
\label{operator_norm_bound_1_*}
\left\||\Gamma|^{\nu,\zeta,L}_p\right\|_{L^{\infty}_{\vec{x} } L^1_{\vec{y}}}=\mathcal {O}_{d,\zeta,p,v}(1)\,.
\end{equation}
In particular, the series in \eqref{gamma_p_cluster_expansion_claim} is absolutely summable in the $L^{\infty}_{\vec{x} } L^1_{\vec{y}}$ norm, uniformly in $L$ in the sense of \eqref{operator_norm_bound_1_*} above.
\end{itemize}
\end{proposition}

\subsection{Proof of Proposition \ref{gamma_p_cluster_expansion}}
Before proceeding with the proof of Proposition \ref{gamma_p_cluster_expansion}, we note the following result, by which we iteratively integrate out a vertex of a graph (as in \eqref{Ursell_function_definition}, we are considering graphs on $[n]=\{1,\ldots,n\}$, where vertices $i,j$ are connected if the factor $\xi^{L}(\omega_i,\omega_j)$ occurs). This will also be the basis of our integration algorithm described later in Lemma \ref{integration_algorithm} below.

\begin{lemma}[Integrating out a vertex]
\label{Lemma_5.16*}
Let $\omega \in \Omega^{L,T(\omega)}$ with $T(\omega) \in \nu \N^*$, $q \in \N$, and $x \in \Lambda_L$ be given. Recalling \textup{\eqref{measure_mu_1}--\eqref{Ursell_function_definition}} and \eqref{kappa_definition}, we have the following estimates.
\begin{itemize}
\item[(i)] 
$\int \mu^L(\dd \tilde \omega) \,T(\tilde \omega)^q \,|\xi^L(\omega,\tilde{\omega})| \lesssim_d  \nu^{q-1}\,\frac{q!}{\kappa^{q+1}}\,T(\omega)\, \|v\|_{L^1(\R^d)}\,.$
 \item[(ii)] $\int_{\Lambda_L} \dd y\, \int \widehat{\mu}^L_{y,x}(\dd \tilde{\omega})\, T(\tilde{\omega})^q\,|\xi^L(\omega,\tilde{\omega})| \lesssim \nu^{q-1}\,\frac{(q+1)!}{\kappa^{q+2}}\, T(\omega)\, \|v^L\|_{L^{\infty}(\Lambda_L)} \lesssim_d \nu^{q-1}\,\frac{(q+1)!}{\kappa^{q+2}}\, T(\omega)$.
\item[(iii)] $\int \mu^L(\dd \tilde{\omega}) \,T(\tilde{\omega})^q \lesssim_d \nu^q \left(\mathbf{1}_{q=0}+ \mathbf{1}_{q \geq 1} \frac{(q-1)!}{\kappa^q} \right)\,|\Lambda_L|\,.$
\item[(iv)] $\int_{\Lambda_L} \dd y\,\int \widehat{\mu}^{L}_{y,x} (\dd \tilde{\omega})\,T(\tilde{\omega})^q \lesssim_d \nu^q \frac{q!}{\kappa^{q+1}}\,.$
\end{itemize}
\end{lemma}

We now show Proposition \ref{gamma_p_cluster_expansion} assuming Lemma \ref{Lemma_5.16*}.

\begin{proof}[Proof of Proposition \ref{gamma_p_cluster_expansion}]
Statement (i) follows directly from \cite[Theorem 1]{Ueltschi_2004}. 
Let us justify the application of \cite[Theorem 1]{Ueltschi_2004} here.
By Lemma \ref{Lemma_5.16*} (i), we have that
\begin{multline}
\label{Cluster_expansion_1}
\int \mu^L(\dd \tilde{\omega})\, \ee^{\frac{1}{\nu} T(\tilde{\omega})}\,|\xi^{L}(\omega,\tilde{\omega})|  
=
\sum_{q=0}^{\infty} \int \mu^L(\dd \tilde{\omega})\, \frac{T^q(\tilde{\omega})}{\nu^q\,q!}\,|\xi^{L}(\omega,\tilde{\omega})| 
\lesssim_d  \frac{1}{\nu}\, T(\omega) \,\|v\|_{L^1(\R^d)}\,\sum_{q=0}^{\infty}\frac{1}{\kappa^{q+1}} 
\\
= \left( \frac{1}{\kappa-1}  \,\|v\|_{L^1(\R^d)}\right)\, \frac{1}{\nu}\, T(\omega)\,.
\end{multline}
In \eqref{Cluster_expansion_1}, we used \eqref{kappa_definition}.
Using \eqref{Cluster_expansion_1} and \eqref{smallness_v_kappa} (for a suitably small quantity depending on $d$ on the right hand side), it follows that 
\begin{equation}
\label{Cluster_expansion_2}
\int \mu^L(\dd \tilde{\omega})\, \ee^{\frac{1}{\nu} T(\tilde{\omega})}\,|\xi^{L}(\omega,\tilde{\omega})| \leq \frac{1}{\nu}\, T(\omega)\,.
\end{equation}
We note that \eqref{Cluster_expansion_2} corresponds to \cite[equation (3)]{Ueltschi_2004} when $a(\omega)=\frac{T(\omega)}{\nu}$.
Furthermore, by Lemma \ref{Lemma_5.16*} (iii) and \eqref{kappa_definition}, we deduce that 
\begin{equation}
\label{Cluster_expansion_3}
\int \mu^L(\dd \tilde{\omega})\, \ee^{\frac{1}{\nu} T(\tilde{\omega})} \lesssim_d \left(\sum_{q=0}^{\infty} \frac{1}{\kappa^q}\right)\, |\Lambda_L|<\infty\,.
\end{equation}
Furthermore, by \eqref{Ursell_function_definition} and the nonnegativity of $v$, it follows that 
\begin{equation}
\label{Cluster_expansion_4}
\left|1+\xi^L(\omega,\tilde{\omega})\right| \leq 1\,.
\end{equation}
From \eqref{Cluster_expansion_2}--\eqref{Cluster_expansion_4}, we can indeed justify the application of \cite[Theorem 1]{Ueltschi_2004}.

To prove (ii), we use the Ginibre representation \eqref{Ginibre_representation_4} for the reduced $p$-particle density matrix in conjunction with \eqref{measure_mu_1}--\eqref{measure_mu_2} to obtain
\begin{equation}
    \Gamma^{\nu,\zeta,L}_{p}(\vec{x},\vec{y})=\sum_{\pi \in S_p} \int \widehat \mu^{L}_{y_{\pi(1)},x_1}(\dd \omega_1) \cdots \widehat \mu^{L}_{y_{\pi(p)},x_p}(\dd \omega_p) \,\frac{\tilde X^{L}(\omega_1, \dots, \omega_p)}{\tilde X^{L}}\,,
\end{equation}
where $\tilde X^{L}(\omega_1, \ldots, \omega_p) \equiv \tilde X^{\nu,\zeta,L}(\omega_1, \ldots, \omega_p)$ and $\tilde X^{L} \equiv \tilde X^{\nu,\zeta,L}$ are defined as
\begin{multline*}
\tilde X^{L}(\omega_1, \dots, \omega_p) \coloneqq\sum_{n = p}^\infty \frac{1}{(n - p)!} \int \mu^{L}(\dd \omega_{p+1}) \cdots \mu^{L}(\dd \omega_n) \prod_{1 \leq i<j \leq n} \exp\left(-\mathcal{V}^{\nu,L}(\omega_i, \omega_j)\right),
\\
\tilde X^{L} \coloneqq \tilde X^{L}(\emptyset)\,.
\end{multline*}
Using \cite[Theorem 2]{Ueltschi_2004} we have
\begin{equation}
\label{Xtilde_X_equation}
\frac{\tilde X^{L}(\omega_1, \dots, \omega_p)}{\tilde X^{L}} = \sum_{\Pi \in \mathfrak{P}_p} \prod_{\lambda \in \Pi} X^{L}((\omega_i)_{i \in \lambda})\,.
\end{equation}
The application of \cite[Theorem 2]{Ueltschi_2004} is justified by \eqref{Cluster_expansion_2}--\eqref{Cluster_expansion_4} as in the proof of (i). We hence deduce claim (ii).
\end{proof}

We now show Lemma \ref{Lemma_5.16*}.
For the proof of Lemma \ref{Lemma_5.16*}, we record the following result which follows directly from \cite[Lemma 5.9]{FKSS_2020_2}.

\begin{lemma}
\label{polylogarithm}
Recalling \eqref{zeta_infinite_volume}, we have 
\begin{equation*}
\sum_{k \in \N^*} \zeta^k k^q \lesssim \frac{q!}{\kappa^{q+1}}\,.
\end{equation*}
\end{lemma}

\begin{proof}[Proof of Lemma \ref{Lemma_5.16*}]
Throughout the proof, we repeatedly use the nonnegativity of $v^L$, which follows from Lemma \ref{v^L_lemma} (i).
Let us first prove claim (i). We use \eqref{xi_bound}--\eqref{measure_mu_1_bound} followed by \eqref{loop_interaction} to write 
\begin{multline}
\label{Lemma_5.16*_1}
\int \mu^L(\dd \tilde \omega) \,T(\tilde \omega)^q \,|\xi^L(\omega,\tilde{\omega})|
\leq \sum_{\tilde{T} \in {\nu \N^*}} z^{\tilde{T}}\, \tilde{T}^{q-1}\,\sum_{r \in \nu \N} \vec{1}_{r<T(\omega)}\, \sum_{s \in \nu \N} \vec{1}_{s<\tilde{T}}\, 
\\
\int_0^{\nu} \dd t\, \left[\int_{\Lambda_L} \dd x \, \int \mathbb{W}^{L,\tilde{T}}_{x,x}(\dd \tilde{\omega})\, v^L\bigl(\omega(t+r)-\tilde{\omega}(t+s)\big)\right].
\end{multline}
For fixed $\tilde{T} \in \nu \N^*$, $r,s \in \nu \N$ with $r<T(\omega), s<\tilde{T}$, and $t \in [0,\nu]$, we let $y \coloneqq \omega(t+r)$ 
and we rewrite the expression in square brackets in \eqref{Lemma_5.16*_1} as
\begin{multline}
\label{Lemma_5.16*_2}
\int_{\Lambda_L} \dd x \int_{\Lambda_L} \dd x' \int \mathbb{W}^{L,t+s}_{x',x}(\dd \tilde{\omega}_1)\, \int \mathbb{W}_{x,x'}^{L,\tilde{T}-(t+s)}(\dd \tilde{\omega}_2)\, v^L(y-x')
\\
=\int_{\Lambda_L} \dd x'\, \left\{\int_{\Lambda_L} \dd x\, \int \mathbb{W}^{L,\tilde{T}-(t+s)}_{x,x'}(\dd \tilde{\omega}_2)\, 
\int \mathbb{W}_{x',x}^{L,t+s}(\dd \tilde{\omega}_1) \right\}\, v^L(y-x')
\\
=\int_{\Lambda_L} \dd x' \int \mathbb{W}^{\tilde{T}}_{x',x'}(\dd \tilde{\omega})\, v^L(y-x')
=\int_{\Lambda_L} \dd x' \,\psi^{\tilde{T}}(0)\, v^L(y-x')=\|v^L\|_{L^1(\Lambda_L)} \,\psi^{\tilde{T}}(0) 
\\
= \|v\|_{L^1(\R^d)} \, \psi^{\tilde{T}}(0)\,.
\end{multline}
In \eqref{Lemma_5.16*_2}, we used Fubini's theorem, followed by \eqref{W^T}--\eqref{heat_kernel_integral}, and Lemma \ref{v^L_lemma} (ii). Moreover, $\tilde{\omega}=\tilde{\omega}_2 \oplus \tilde{\omega}_1$ is the concatenation of Brownian paths
\begin{equation}
\label{concatenation_paths}
\tilde{\omega}(\tau) \coloneqq 
\begin{cases}
\tilde{\omega}_2(\tau) & \text{if } \tau \leq T(\tilde{\omega}_2)
\\
\tilde{\omega}_1(\tau-T(\tilde{\omega}_2)) & \text{if } T(\tilde{\omega}_2)  < \tau \leq T(\tilde{\omega}_2)+T(\tilde{\omega}_1)\,.
\end{cases}
\end{equation}
Using \eqref{Lemma_5.16*_2}, Lemma \ref{heat_kernel_estimates} (i), \eqref{fugacity_condition}, and Lemma \ref{polylogarithm}, we get that 
\begin{multline}
\label{Lemma_5.16*_3}
\eqref{Lemma_5.16*_1} \lesssim_d \frac{1}{\nu} \sum_{\tilde{T} \in \nu \N^*} z^{\tilde{T}}\,\tilde{T}^q\, \biggl(
\frac{1}{L^d}+\frac{1}{\tilde{T}^{d/2}}\biggr)\,T(\omega)\, \|v\|_{L^1(\R^d)}
\\
=\frac{1}{\nu} \sum_{k \in \N^*} z^{k\nu} (k\nu)^q \biggl(\frac{1}{L^d}+\frac{1}{(k\nu)^{d/2}} \biggr)\, T(\omega)\,  \|v\|_{L^1(\R^d)}
\\
=\nu^{q-1}\,\sum_{k \in \N^*} \zeta^k k^q\, \nu^{\frac{dk}{2}} \biggl(\frac{1}{L^d}+\frac{1}{(k \nu)^{d/2}} \biggr)\,T(\omega)\, \|v\|_{L^1(\R^d)}
\lesssim \nu^{q-1}\, \sum_{k \in \N^*} \zeta^k k^q \,T(\omega)\, \|v\|_{L^1(\R^d)} 
\\
\lesssim \nu^{q-1}\,\frac{q!}{\kappa^{q+1}}\, T(\omega)\, \|v\|_{L^1(\R^d)}\,.
\end{multline}
In \eqref{Lemma_5.16*_3}, we used 
\begin{equation}
\label{Lemma_5.16*_3_B}
\nu^{\frac{dk}{2}} \biggl(\frac{1}{L^d}+\frac{1}{(k \nu)^{d/2}} \biggr) \lesssim 1
\end{equation}
for all $k \in \N^*$.
Thus, we obtain (i).

We now prove claim (ii). By Lemma \ref{v^L_lemma} (v), it suffices to show the first inequality.
By \eqref{xi_bound}, \eqref{measure_mu_2_bound}, followed by \eqref{loop_interaction}, we have
\begin{multline}
\label{Lemma_5.16*_4}
\int_{\Lambda_L} \dd y\, \int \widehat{\mu}^L_{y,x}(\dd \tilde{\omega})\, T(\tilde{\omega})^q\,|\xi^L(\omega,\tilde{\omega})|
\leq \frac{1}{\nu} \, \sum_{\tilde{T} \in \nu \N^*} z^{\tilde{T}}\, \tilde{T}^{q}\sum_{r \in \nu \N} \vec{1}_{r<T(\omega)} \sum_{s \in \nu \N} \vec{1}_{s<\tilde{T}}
\\
\times
\int_{\Lambda_L} \dd y\,\int \mathbb{W}^{L,\tilde{T}}_{y,x}(\dd \tilde{\omega})\, \int_0^{\nu} \dd t\, v^L\left(\omega(t+r)-\tilde{\omega}(t+s)\right).
\end{multline}
We use \eqref{heat_kernel_integral} followed by \eqref{heat_kernel_integral_1} to deduce that
\begin{equation}
\label{Lemma_5.16*_6}
\int_{\Lambda_L} \dd y\, \int \mathbb{W}^{L,\tilde{T}}_{y,x}(\dd \tilde{\omega})\, v^L\left(\omega(t+r)-\tilde{\omega}(t+s)\right)
\leq \|v^L\|_{L^{\infty}(\Lambda_L)}\,.
\end{equation}
Using \eqref{Lemma_5.16*_6}, followed by \eqref{fugacity_condition}
and Lemma \ref{polylogarithm}, we get that
\begin{multline}
\label{Lemma_5.16*_7}
\eqref{Lemma_5.16*_4} \lesssim \nu^{-2} \, \sum_{\tilde{T} \in \nu \N^*} z^{\tilde{T}} \, \tilde{T}^{q+1}\, T(\omega)\, \|v^L\|_{L^{\infty}(\Lambda_L)} =\nu^{q-1} \sum_{k \in \N^*} z^{k\nu}\,k^{q+1}\, T(\omega)\, \|v\|_{L^{\infty}(\Lambda_L)}
\\
\leq \nu^{q-1}\, \sum_{k \in \N^*} \zeta^k k^{q+1}\, T(\omega)\, \|v^L\|_{L^{\infty}(\Lambda_L)} \lesssim \nu^{q-1}\, \frac{(q+1)!}{\kappa^{q+2}}\, T(\omega)\, \|v^L\|_{L^{\infty}(\Lambda_L)}\,.
\end{multline}
This completes the proof of (ii).

We now prove claim (iii). By using \eqref{measure_mu_1_bound}, followed by  \eqref{heat_kernel_integral} and Lemma \ref{heat_kernel_estimates} (i), we have 

\begin{multline}
\label{Triangle_1}
\int \mu^L(\dd \tilde{\omega})\, T(\tilde{\omega})^q \leq \nu \sum_{\tilde{T} \in \nu \N^*} z^{\tilde{T}}\,\tilde{T}^{q-1}\,\int_{\Lambda_L}
\dd x\, \int \mathbb{W}^{L,\tilde{T}}_{x,x}(\dd \tilde{\omega}) 
\\
\lesssim_d \nu  \sum_{\tilde{T} \in \nu \N^*}
z^{\tilde{T}}\,\tilde{T}^{q-1}\,\left(\frac{1}{L^d}+\frac{1}{\tilde{T}^{d/2}}\right)\,|\Lambda_L|\,,
\end{multline}
which we can rewrite and estimate as
\begin{equation}
\label{Triangle_2}
=\nu^q \sum_{k \in \N^*}\zeta^k \,
k^{q-1}\,\nu^{\frac{dk}{2}}\, \left(\frac{1}{L^d}+\frac{1}{(k \nu)^{d/2}}\right)\,|\Lambda_L|
\lesssim \nu^q \,\frac{(q-1)!}{\kappa^q}\,|\Lambda_L|\,.
\end{equation}
when $q \geq 1$.
Above, we used \eqref{fugacity_condition}, \eqref{Lemma_5.16*_3_B}, 
and Lemma \ref{polylogarithm}. The analysis for $q=0$ is analogous. Claim (iii) now follows.

Finally, we prove claim (iv). By using \eqref{measure_mu_2_bound}, followed by  \eqref{heat_kernel_integral} and \eqref{heat_kernel_integral_1}, we have
\begin{equation*}
\int_{\Lambda_L} \dd y\, \int \widehat{\mu}^{L}_{y,x}(\dd \tilde{\omega})\,T(\tilde{\omega})^q \leq \sum_{\tilde{T} \in \nu \N^*}
z^{\tilde{T}}\,\tilde{T}^q \,\int_{\Lambda_L} \dd y\, \int \mathbb{W}_{y,x}^{L,\tilde{T}}(\dd \tilde{\omega})\,,
\end{equation*}
which by \eqref{fugacity_condition} and Lemma \ref{polylogarithm} is
\begin{equation}
\label{Proof_of_iv}
=\nu^{q} \sum_{k \in \N^*} z^{k\nu}\,k^q =\nu^{q} \sum_{k \in \N^*} \zeta^k\,\nu^{dk/2}\,k^q \leq 
\nu^{q} \sum_{k \in \N^*} \zeta^k\,k^q \lesssim \nu^q\, \frac{q!}{\kappa^{q+1}}\,.
\end{equation}
Claim (iv) now follows.
\end{proof} 
\begin{remark}
\label{5.27_remark}
Calculations similar to those in \eqref{Lemma_5.16*_3} show that for any $q \in \N$, we have 
\begin{equation}
\label{5.27_remark_bound}
\nu \sum_{T \in \nu \N^*}z^T\,T^{q-1-\frac{d}{2}}\lesssim_{\zeta,q} 1\,,
\end{equation}
uniformly in $\nu \in (0,\nu_0]$. More precisely, we use \eqref{fugacity_condition} and Lemma \ref{polylogarithm} to write
\begin{multline}
\label{5.27_remark_bound_proof}
\nu \sum_{T \in \nu \N^*}z^T\,T^{q-1-\frac{d}{2}}=\nu \sum_{k \in \N^*} z^{\nu k}\, (\nu k)^{q-1-\frac{d}{2}}=\nu \sum_{k \in \N^*} \zeta^k \, \nu^{\frac{dk}{2}}\, \nu^{q-1-\frac{d}{2}}\, k^{q-1-\frac{d}{2}}
\\
\leq \nu^q \, \sum_{k \in  \N^*}  \zeta^k \, k^{q} \lesssim_{\zeta,q} 1\,,
\end{multline}
and we obtain \eqref{5.27_remark_bound} from \eqref{5.27_remark_bound_proof}.
\end{remark}

\begin{remark}
\label{Lemma_5.4_remark}
The upper bound in Lemma \ref{Lemma_5.16*} is given in terms of $\|v^L\|_{L^{\infty}(\Lambda_L)}$. The analogous bound on the lattice is given in terms of the $L^1$ norm of $v^L$, analogously to Lemma \ref{Lemma_5.16*} (i); see \cite[Lemma 5.17 (ii)]{FKSS_2023} and \cite[Lemma 5.7 (ii)]{FKSS_2023} for details. Unlike in the case of the lattice, the heat kernel is not pointwise bounded in the continuum. As a result, our methods can only give us an upper bound in terms of $\|v^L\|_{L^{\infty}(\Lambda_L)}$. By Lemma \ref{v^L_lemma} (iii) and (v), this is sufficient for the remainder of the analysis.
\end{remark}

\subsection{Proof of Proposition \ref{Proposition_5.3}}
We now proceed to the proof of Proposition \ref{Proposition_5.3}. We first note several auxiliary results.
The first observation is that, after taking absolute values in our estimates, we can reduce the sum over connected graphs in \eqref{Ursell_function_definition} to the sum over trees. More explicitly, we have the following lemma. 
\begin{lemma}[Tree bound]
\label{tree_bound_estimate}
For $n \in \N^*$ and $\omega_1, \ldots, \omega_n \in \Omega^L$ we have 
\begin{equation}
\label{tree_bound}
|\varphi^L(\omega_1, \dots, \omega_n)| \leq \frac{1}{n!} \sum_{\mathcal T \in {\mathfrak T}_n} \prod_{\{i,j\} \in \mathcal T}|\xi^L(\omega_i,\omega_j)|\,.
\end{equation}
Here, ${\mathfrak T}_n$ denotes the set of all trees on $[n]=\{1,2,\ldots,n\}$.
\end{lemma}
The proof of Lemma \ref{tree_bound_estimate} is based on a resummation argument and is analogous to the one given on the lattice in \cite[Lemma 5.6]{FKSS_2023}. The resummation is based on Kruskal's algorithm and relies only on the assumption that $v$ is nonnegative. We refer the reader to \cite[Appendix D.1]{FKSS_2023} for the full details.

Let us now introduce further notation. Given $n \in \N^*$ and $(\delta_1,\ldots,\delta_n) \in \N^n$, we define
\begin{equation} 
\label{tree_delta}
{\mathfrak{T}}_n^{\delta_1,\ldots,\delta_n} \coloneqq \bigl\{\mathcal  T \in {\mathfrak T}_n\,, \mathrm{deg}(i)=\delta_i\,, i=1\,,\ldots,n \bigr\}\,,
\end{equation}
where in \eqref{tree_delta}, $\mathrm{deg}(i)$ denotes the degree of a vertex $i$ in $\mathcal  T$, i.e., the number of edges adjacent to $i$. Note that ${\mathfrak T}_1^{0}$ is a single vertex, and that ${\mathfrak T}_1^{\delta_1}=\emptyset$ if $\delta_1 \neq 0$. For $n \geq 2$, the set
${\mathfrak T}_n^{\delta_1,\ldots,\delta_n}$ is nonempty if and only if $(\delta_1,\ldots,\delta_n) \in (\N^*)^n$ and $\sum_{i=1}^{n} \delta_i=2(n-1)$. We recall Pr\"{u}fer's theorem which states that
\begin{equation}\label{number_of_labeled_trees_Prufer}
    |{\mathfrak{T}}_n^{\delta_1,\ldots,\delta_n}|=\frac{(n-2)!}{\prod_{i=1}^n(\delta_i-1)!}\,;
\end{equation}
see e.g.\ \cite[Theorem 5.3.4]{Stanley_EC2} for a self-contained discussion.

We now set up the integration algorithm based on Lemma \ref{Lemma_5.16*} above.
\begin{lemma}[Integration algorithm] 
\label{integration_algorithm}
Let $n \! \in \! \N^*$ and $(\delta_1, \ldots, \delta_n) \! \in \! (\N^*)^n$ with $\sum_{i=1}^{n} \delta_i=2(n-1)$ be given. Recalling \eqref{tree_delta}, let $\mathcal{T} \in  {\mathfrak T}_n^{\delta_1,\ldots,\delta_n}$ be given. Furthermore, consider sets $\mathbf{O}, \mathbf{C} \subseteq \{2,3,\ldots,n\}$ such that 
\begin{equation*}
\mathbf{O} \cup \mathbf{C}=\{2,3,\ldots,n\}\,, \quad \mathbf{O} \cap \mathbf{C}=\emptyset \,,
\end{equation*}
and
$(x_i)_{i \in \mathbf{O}} \in \Lambda_L^{\mathbf{O}}$. For $i \in \{2,3,\ldots,n\}$ we let 
\begin{equation}
\label{Theta_measure} 
\hat \delta_i \coloneqq
\begin{cases}
\delta_i & \text{if } i \in \mathbf{C}
\\
\delta_i+1 & \text{if } i \in \mathbf{O}
\end{cases}
\qquad
\Theta_i (\dd \omega) \coloneqq
\begin{cases}
\mu^{L} (\dd \omega) & \text{if } i \in \mathbf{C}
\\
 \,\int_{\Lambda_L} \dd y \,\widehat{\mu}^{L}_{y,x_i} (\dd \omega) & \text{if } i \in \mathbf{O}\,.
\end{cases}
\end{equation}
With this notation, we have that for $\omega_1 \in \Omega^{L,T(\omega_1)}$,
\begin{multline}  
\label{tree_integration_bound}
\int \Theta_2(\dd \omega_2)\,\Theta_3(\dd \omega_3) \cdots \Theta_n(\dd \omega_n)\,\prod_{\{i,j\} \in \mathcal  T} |\xi^L(\omega_i,\omega_j)| 
\\
\leq C_0^{n-1} \left(\|v\|_{L^1(\R^d)}+1\right)^{n-1}\prod_{i=2}^{n} \Bigl(\kappa^{-\hat \delta_i}\nu^{ \delta_i-2}\,(\hat \delta_i-1)!\Bigr)\,T(\omega_1)^{\delta_1}\,,
\end{multline}
for some $C_0 \equiv C_0(d)>0$.
\end{lemma}

Note that in the statement of Lemma \ref{integration_algorithm}, $\mathbf{O}$ denotes the set of labels in $\{2,3,\ldots,n\}$ corresponding to open paths, and $\mathbf{C}$ denotes the set of labels in $\{2,3,\ldots,n\}$ corresponding to closed paths.
\begin{proof}[Proof of Lemma \ref{integration_algorithm}]
We prove the claim by induction on $n$.\\

\noindent\textbf{Base case}. The claim trivially holds when $n=1$.\\ 

\noindent\textbf{Induction step}. Suppose that $n \geq 2$ and that \eqref{tree_integration_bound} holds for all trees on at most $n-1$ vertices. Let $k \coloneqq \delta_1$. Then the vertex with label $1$ is connected to vertices with labels $i_1,\ldots,i_k$ where $1 < i_1 <i_2 < \cdots <i_k \leq n$. Graphically, we represent this as
$\omega_1$ being connected to $\omega_{i_1},\omega_{i_2}, \ldots, \omega_{i_k}$ by means of a factor of $\xi^L(\omega_1,\omega_{i_j})$.
By deleting the edges $\{1,i_1\}, \ldots, \{1,i_k\}$ from $\mathcal  T$, we obtain a forest. Let us denote its corresponding connected components by $\mathcal  T_1 , \ldots , \mathcal  T_k$, where $i_{\ell} \in \mathrm{V}(\mathcal  T_\ell)$ for $\ell=1,\ldots,k$. Here, $\mathrm{V}(G)$ denotes the vertex set of a graph $G$.

For $\ell=1,\ldots,k$, we define $n_{\ell} \coloneqq|\mathrm{V}(\mathcal  T_{\ell})|$. By the induction hypothesis, we have 
\begin{multline}
\label{tree_integration_induction}
\int \prod_{i \in \mathrm{V}(\mathcal  T_{\ell}) \setminus \{i_\ell\}} \Theta_i(\dd \omega_i) \, \prod_{\{i,j\} \in \mathcal  T_\ell} |\xi^L(\omega_i,\omega_j)|
\\
\leq C_0^{n_\ell-1}\,\left(\|v\|_{L^1(\R^d)}+1\right)^{n_\ell-1}\,\prod_{i \in \mathrm{V}(\mathcal  T_{\ell}) \setminus \{i_\ell\}} \Bigl(\kappa^{-\hat \delta_i}\nu^{ \delta_i-2}\,(\hat \delta_i-1)!\Bigr)\,T(\omega_{i_\ell})^{\delta_{i_\ell}-1}\,,
\end{multline}
for all $\ell=1,\ldots,k$. Here, we recall the definition of $\hat{\delta}_i$ from \eqref{Theta_measure}. We now use \eqref{tree_integration_induction} to deduce that the left-hand side of \eqref{tree_integration_bound} is bounded from above by
\begin{multline}
\label{tree_integration_induction_2}
 C_0^{n-k-1}\,\left(\|v\|_{L^1(\R^d)}+1\right)^{n-k-1}\,\prod_{i \in \{2,\ldots,n\} \setminus \{i_1,\ldots,i_k\}} \Bigl(\kappa^{-\hat \delta_i}\nu^{ \delta_i-2}\,(\hat \delta_i-1)!\Bigr)\,
\\
\times
\prod_{\ell=1}^{k} \Biggl(\int \Theta_{i_\ell}(\dd \omega_{i_\ell})\,T(\omega_{i_\ell})^{\delta_{i_\ell}-1}\,|\xi^L(\omega_1,\omega_{i_\ell})|\Biggr)\,.
\end{multline}
We now apply Lemma \ref{Lemma_5.16*} in each of the $k$ factors in \eqref{tree_integration_induction_2}. More precisely, we apply Lemma \ref{Lemma_5.16*} (i) if $i_\ell \in \mathbf{O}$ and Lemma \ref{Lemma_5.16*}  (ii) if $i_\ell \in \mathbf{C}$. We hence deduce the bound claimed in \eqref{tree_integration_bound}.
\end{proof}

\begin{remark}
\label{Lemma_5.9_remark}
If in the proof of Lemma \ref{integration_algorithm} we kept explicit track of which paths $\omega_j$ for $j=2,\ldots,n$ are closed or open (and thus whether we apply the bound from  Lemma \ref{Lemma_5.16*} (i) or Lemma \ref{Lemma_5.16*} (ii) respectively), we obtain 
\begin{multline}  
\label{tree_integration_bound_improved}
\int \Theta_2(\dd \omega_2)\,\Theta_3(\dd \omega_3) \cdots \Theta_n(\dd \omega_n)\,\prod_{(i,j) \in \mathcal  T} |\xi^L(\omega_i,\omega_j)| 
\\
\leq C_0^{n-1} \|v\|_{L^1(\R^d)}^{n-1-|\mathbf{O}|} \prod_{i=2}^{n} \Bigl(\kappa^{-\hat \delta_i}\nu^{ \delta_i-2}\,(\hat \delta_i-1)!\Bigr)\,T(\omega_1)^{\delta_1}\,,
\end{multline}
for some $C_0 \equiv C_0(d)>0$. The bound \eqref{tree_integration_bound_improved} is an improvement of \eqref{tree_integration_bound}. The bound \eqref{tree_integration_bound} is however sufficient for our analysis.
\end{remark}

\begin{proof}[Proof of Proposition \ref{Proposition_5.3}]
We will first prove (i). 
By Proposition \ref{gamma_p_cluster_expansion} (i), \eqref{Ursell_function_representation}, \eqref{Ursell_function_definition}, and the tree bound from Lemma \ref{tree_bound_estimate}, we have
\begin{equation}
\label{logXi_1}
 \Bigl|\log \Xi^{\nu,\zeta,L}\Bigr|  \leq \sum_{n=1}^{\infty} \frac{1}{n!} \sum_{\mathcal  T \in {\mathfrak T}_n}\int \mu^{L}(\dd \omega_{1}) \cdots \mu^{L}(\dd \omega_n) \prod_{\{i,j\} \in \mathcal  T} |\xi^L(\omega_i,\omega_j)|\,.
\end{equation}
By recalling \eqref{tree_delta}, we rewrite the right-hand side of \eqref{logXi_1} as 
\begin{equation}
\label{logXi_2}
\int \mu^L(\dd \omega)+ \sum_{n=2}^{\infty} \frac{1}{n!} \sum_{\substack{(\delta_1,\ldots,\delta_n) \in \N^n \\ \delta_1+\cdots+\delta_n=2(n-1)}} \sum_{\mathcal  T \in {\mathfrak T}_n^{\delta_1,\ldots,\delta_n}}
\int \mu^{L}(\dd \omega_1) \cdots \mu^{L}(\dd \omega_n) \prod_{\{i,j\} \in \mathcal  T} 
|\xi^L(\omega_i,\omega_j)|\,,
\end{equation}
The $n=1$ term in \eqref{logXi_2} is
\begin{equation}
\label{logXi_2_{n=1}}
\int \mu^L(\dd \omega) \lesssim_d |\Lambda_L|\,,
\end{equation}
by using Lemma \ref{Lemma_5.16*} (iii) with $q=0$.
By using \eqref{logXi_1}--\eqref{logXi_2_{n=1}}
 and Lemma \ref{integration_algorithm} for $n \geq 2$ with 
$\mathbf{C}=\{2,\ldots,n\}$ in each term of \eqref{logXi_2}, we deduce that 
\begin{equation} 
\label{logXi_3}
 \Bigl|\log \Xi^{\nu,\zeta,L}\Bigr| 
\lesssim_{d} \sum_{n=1}^{\infty} \frac{1}{n}\,\biggl(\frac{C}{\kappa^2}\,\|v\|_{L^1(\R^d)}\biggr)^{n-1}\,|\Lambda_L|\,.
\end{equation}
Here we used that for $n \geq 2$ and $(\delta_1,\ldots,\delta_n)$ as above
\begin{equation}
\label{logXi_3_proof}
\prod_{i=2}^{n} \Bigl(\kappa^{-\delta_i}\nu^{\delta_i-2}\,(\delta_i-1)!\Bigr)\, \int \mu^{L} (\dd \omega_1)\,T(\omega_1)^{\delta_1} \lesssim_{d} \frac{1}{\kappa^{2n-2}}\, \prod_{i=1}^{n} (\delta_i-1)!\,|\Lambda_L|\,,
\end{equation}
which follows from Lemma \ref{Lemma_5.16*} (iii). 
Note that the power of $\nu$ in the estimate on the right-hand side of \eqref{logXi_3_proof} equals

\begin{equation*}
\sum_{i=2}^{n} (\delta_i-2) + \delta_1=2(n-1)-2(n-1)=0\,.
\end{equation*}

In other words the factors of $\nu$ cancel out by construction.
We deduce \eqref{logXi_3} from \eqref{logXi_3_proof} by using the fact that for $n \geq 2$, there are $\binom{2n-3}{n-1} \leq (2 \ee)^{n-1}$ possible positive integer solutions of $\delta_1+\cdots+\delta_n=2(n-1)$, and by using 
\eqref{number_of_labeled_trees_Prufer}. The claim follows by recalling \eqref{kappa_definition} and \eqref{smallness_v_kappa}.
 
Let us now prove (ii). We show \eqref{operator_norm_bound_1_*}, which implies \eqref{operator_norm_bound_1} by recalling \eqref{X^L_absolute_value_bound}, \eqref{gamma_p_cluster_expansion_claim}, and \eqref{gamma_p_cluster_expansion_absolute_value}. By \eqref{gamma_p_cluster_expansion_absolute_value}, it suffices to show that for all $\vec x \in \Lambda_L^p$, we have
\begin{equation} 
\label{Schur_test_2}
 \,\int_{\Lambda_L^p} \dd \vec{y} \, 
\int \widehat \mu^{L}_{y_1,x_1}(\dd \omega_1) \cdots \widehat \mu^{L}_{y_p,x_p}(\dd \omega_p) \sum_{\Pi \in \mathfrak P_p} \prod_{\lambda \in \Pi} \bigl|X^{L}((\omega_i)_{i \in \lambda})\bigr|= \mathcal  O_{d,\zeta,p,v}(1)\,.
\end{equation}
Suppose that $\vec{x},\vec{y}\in \Lambda_L^p$, and that $\lambda \in \Pi$ with $\lambda=\{a_1,\ldots,a_{|\lambda|}\} \subset [p]$ and $1\leq a_1<a_2<\cdots<a_{|\lambda|}$ are given. Let us define
\begin{equation*} 
\Gamma^{\nu,\zeta,L,\lambda}_p((x_i)_{i \in \lambda},(y_i)_{i \in \lambda}) \coloneqq
\int \prod_{i\in \lambda} \widehat \mu^{L}_{y_i,x_i}(\dd \omega_i)  \,\left|{X}^{L}((\omega_{i})_{{i} \in {\lambda}})\right|\,.
\end{equation*}
Recalling \eqref{Ursell_function_representation} and arguing as in \eqref{logXi_1}--\eqref{logXi_2}, we have
\begin{multline}
\label{gamma_lambda}\Gamma^{\nu,\zeta,L,\lambda}_p((x_i)_{i \in \lambda},(y_i)_{i \in \lambda}) \leq  \sum_{n \geq |\lambda|} \frac{1}{(n-|\lambda|)!} \sum_{\substack{(\delta_1,\ldots,\delta_n) \in \N^n \\ \delta_1+\cdots+\delta_n=2(n-1)}}  \sum_{\mathcal  T \in {\mathfrak T}_n^{\delta_1,\ldots,\delta_n}}
\int \prod_{i\in \lambda } \widehat \mu^{L}_{y_i,x_i}(\dd \omega_i)
\\
\times \mu^{L}(\dd \omega_{a_{|\lambda|}+1}) \cdots \mu^{L}(\dd \omega_n) \prod_{\{i,j\} \in \mathcal  T} |\xi^L(\omega_i,\omega_j)|\,.
\end{multline}
We first integrate \eqref{gamma_lambda} with respect to $(y_i)_{i \in \lambda \setminus \{a_1\}}$.
By Lemma \ref{integration_algorithm} with 
$$
\mathbf{O}=\lambda\setminus\{a_1\}\,,\quad \mathbf{C}=\{a_{|\lambda|}+1,\ldots,n\}$$ 
in each term of \eqref{gamma_lambda}, we deduce that 
\begin{multline} 
\label{gamma_lambda_2}
\int_{\Lambda_L^{|\lambda|-1}}\prod_{i \in \lambda\setminus\{a_1\}}\dd y_i\, \Gamma^{\nu,\zeta,L,\lambda}_p((x_i)_{i \in \lambda},(y_i)_{i \in \lambda}) \leq  \vec{1}_{\{|\lambda|=1\}} \int \widehat{\mu}^{L}_{y_{a_1},x_{a_1}} (\dd \omega_1)
\\
+ \sum_{n \geq |\lambda| \vee 2} \frac{1}{(n-|\lambda|)!}  \sum_{\substack{(\delta_1,\ldots,\delta_n) \in \N^n \\ \delta_1+\cdots+\delta_n=2(n-1)}} \sum_{\mathcal  T \in {\mathfrak T}_n^{\delta_1,\ldots,\delta_n}}  C_0^{n-1} \|v\|_{L^1(\R^d)}^{n-1} 
\\
\times
\prod_{i=2}^{n} \Bigl(\nu^{ \delta_i-2}\kappa^{-\hat \delta_i}\,(\hat \delta_i-1)!\Bigr)\, \int \widehat{\mu}^{L}_{y_{a_1},x_{a_1}} (\dd \omega_1)\,T_1^{\delta_1}\,.
\end{multline}
We now use \eqref{gamma_lambda_2}, Lemma \ref{Lemma_5.16*} (iv), and $\delta_1,\ldots,\delta_n \leq n-1$, and argue as in \eqref{logXi_3} to obtain
\begin{equation} 
\label{gamma_lambda_3}
\int_{\Lambda_L^{|\lambda|}} \prod_{i \in \lambda} \dd y_i\, \Gamma^{\nu,\zeta,L,\lambda}_p((x_i)_{i \in \lambda},(y_i)_{i \in \lambda}) 
\leq \frac{C}{\kappa} + \sum_{n \geq |\lambda| \vee 2} \biggl(\frac{(n-1)^2}{\kappa}\biggr)^{|\lambda|}\,\biggl(\frac{C}{\kappa^2}\,\|v\|_{L^1(\R^d)}\biggr)^{n-1-|\lambda|} \, \biggl(\frac{C}{\kappa^2}\biggr)^{|\lambda|}\,.
\end{equation}

In \eqref{gamma_lambda_3}, we used the observation that 
\begin{equation*}
\frac{(n-2)!}{(n-|\lambda|)!}\, \frac{\prod_{i=2}^{n} (\hat{\delta}_i-1)\,\delta_1!}{\prod_{i=1}^{n} (\delta_i-1)!}
=(n-2) (n-3)  \cdots (n-|\lambda|+1) \,\delta_1  \delta_2\, \cdots  \delta_{|\lambda|} \leq (n-1)^{2|\lambda|}\,.
\end{equation*}
Using \eqref{gamma_lambda_3} we estimate the term on the left-hand side of \eqref{Schur_test_2} coming from $\Pi=\lambda$. We hence obtain \eqref{Schur_test_2} from \eqref{gamma_lambda_3}. This proves claim (ii).
\end{proof}

\subsection{The limit $L \rightarrow \infty$. Contributions of long and far-away paths}
Our goal now is to let $L \rightarrow \infty$ with the aim of proving Theorems \ref{Infinite_volume_Theorem_2} and \ref{free_energy_infinite_volume}. Let us first state two modifications of Lemma \ref{Lemma_5.16*} that we will use henceforth. They are given in Lemmas \ref{Lemma_5.20*} and \ref{Lemma_5.21*} below. These modifications give us bounds similar to that of Lemma \ref{Lemma_5.16*} when we consider long paths and far-away paths respectively.

In what follows, we fix $c>0$. Recalling \eqref{periodic_Euclidean_norm}, we define
\begin{equation}
\label{D^L_c_definition}
\mathscr{D}^L_c\coloneqq\left\{\omega \in \bigcup_{T>0} \Omega^{L,T}\,,\,|\omega(t)-\omega(s)|_L \geq cL \text{ for some }s,t \in [0,T(\omega)]\right\}.
\end{equation}
All of the estimates in the rest of the paper depend on $c$, but we do not keep explicit track of this dependence.
The first modification of Lemma \ref{Lemma_5.16*} tells us that we get small contributions when we integrate over long paths (in the sense of \eqref{D^L_c_definition} above).

\begin{definition} 
\label{epsilon_L}
Let $(A_L)_{L \geq 1}$ be a family of positive numbers indexed by $L \geq 1$. Given $B>0$ and a finite list of parameters $a_1,\ldots,a_k$, we say that 
\begin{equation*}
A_L  \lesssim_{a_1,\ldots,a_k} o_L(B) 
\end{equation*}
if the following conditions hold.
\begin{itemize}
\item[(i)] There exists $C\equiv C(a_1,\ldots,a_k)>0$, independent of $L$ such that $A_L \leq C B$ for all $L \geq 1$.
\item[(ii)] $\lim_{L \rightarrow \infty} A_L=0$.
\end{itemize}
If $k=0$, we write $A_L=o_L(B)$.
\end{definition}
We note that the notation in Definition \ref{epsilon_L} is a bit nonstandard and does not completely follow our earlier conventions from Section \ref{Notation and conventions}. As it is convenient for our later calculations, we apply throughout the sequel.
With notation as in Definition \ref{epsilon_L}, we have the following result.

\begin{lemma}[Integration over long paths]
\label{Lemma_5.20*}
Let $\omega \in \Omega^{L,T(\omega)}$ with $T(\omega) \in \nu \N^*$, $q \in \N$ be given. Recalling \textup{\eqref{measure_mu_1}--\eqref{Ursell_function_definition}} and \eqref{D^L_c_definition}, we have the following estimates.
\begin{itemize}
\item[(i)] 
$\int \mu^L(\dd \tilde{\omega})\,T(\tilde{\omega})^q\,|\xi^L(\omega,\tilde{\omega})|\,\mathbf{1}_{\mathscr{D}^L_c}(\tilde{\omega}) \lesssim_{\zeta,q,d} o_L \left(\nu^{q-1}\,
T(\omega)\,\|v\|_{L^1(\R^d)}\right)$.
\item[(ii)] Let us fix
\begin{equation}
\label{L_0}
L_0 \in (0,L/4]\,.
\end{equation}
Then, for $x \in \Lambda_{L_0}$, we have
\begin{equation*}
\int_{\Lambda_{L_0}} \dd y\,\int \widehat{\mu}^L_{y,x}(\dd \tilde{\omega}) \,T(\tilde{\omega})^q\,|\xi^L(\omega,\tilde{\omega})|\,\mathbf{1}_{\mathscr{D}^L_c}(\tilde{\omega}) \lesssim_{\zeta,q,d} o_L \left(\nu^{q-1}\,T(\omega)\,\right).
\end{equation*}
\item[(iii)] $\frac{1}{\,|\Lambda_L|}\,\int \mu^L(\dd \tilde{\omega})\,T(\tilde{\omega})^q\,\mathbf{1}_{\mathscr{D}^L_c}(\tilde{\omega}) \lesssim_{\zeta,q,d} o_L \left(\nu^q\right).$
\item[(iv)] For $L_0$ as in \eqref{L_0} and  $x \in \Lambda_{L_0}$, we have
\begin{equation*}
\int_{\Lambda_{L_0}} \dd y\, \int \widehat{\mu}^L_{y,x} (\dd \tilde{\omega}) \, T(\tilde{\omega})^q\, \mathbf{1}_{\mathscr{D}^L_c}(\tilde{\omega}) \lesssim_{\zeta,q,d} o_L(\nu^q)\,.
\end{equation*}
\end{itemize}
\end{lemma}

The second modification of Lemma \ref{Lemma_5.16*} tells us that we get small contributions from interacting paths which are far away. To state this claim precisely, let us introduce some notation. Given $c>0$ as in \eqref{D^L_c_definition}, and $\omega \in \Omega^{T(\omega)}, \tilde{\omega} \in \Omega^{T(\tilde{\omega})}$ with $T(\omega), T(\tilde{\omega}) \in \nu \N^*$, we define the interaction $\mathcal{V}^{\nu}(\omega,\tilde{\omega})$ of loops as in \eqref{loop_interaction}, where the interaction potential $v^L \colon \Lambda_L \rightarrow \R$ is now given as
\begin{equation}
\label{v^{L}_{c}_definition}
v^{L}_{c}(x) \coloneqq v^{L}(x)\,\mathbf{1}_{|x|_{L} \geq cL}\,.
\end{equation}
Recalling \eqref{v^L} and using $\|v\|_{L^1(\R^d)}<\infty$ (which follows from Lemma \ref{v^L_lemma} (i)), we obtain from \eqref{v^{L}_{c}_definition} that 
\begin{equation}
\label{v^{L}_c_limit}
\lim_{L \rightarrow \infty} \|v^L_c \|_{L^1(\Lambda_{L})} = 0\,.
\end{equation}
Similarly, using \eqref{v_assumption_iii}--\eqref{v^L}, and \eqref{v^{L}_{c}_definition}, we have that
\begin{equation}
\label{v^{L}_c_limit_*}
\lim_{L \rightarrow \infty} \|v^L_c \|_{L^{\infty}(\Lambda_{L})} = 0\,.
\end{equation}
As in \eqref{Ursell_function_definition}, we define
\begin{equation}
\label{xi^L_c}
\xi^{L}_{c}(\omega,\tilde{\omega}) \coloneqq  \exp\bigl(-\mathcal{V}^{\nu,L}_{c}(\omega,\tilde \omega)\bigr) - 1\,,
\end{equation}
where $\mathcal{V}^{\nu,L}_{c}(\omega,\tilde \omega)$ denotes \eqref{loop_interaction} with interaction potential \eqref{v^{L}_{c}_definition}.
By arguing as in the proof of Lemma \ref{Lemma_5.16*} (i)--(ii), and by using \eqref{v^{L}_c_limit}--\eqref{v^{L}_c_limit_*}, we obtain the following result.
\begin{lemma}[Contributions from far-away paths]
\label{Lemma_5.21*}
Let $\omega \in \Omega^{L,T(\omega)}$ with $T(\omega) \in \nu \N^*$, $q \in \N$, and $x \in \Lambda_L$ be given. Recalling \textup{\eqref{measure_mu_1}--\eqref{measure_mu_2}}, \eqref{v^{L}_{c}_definition}, and \eqref{xi^L_c}, we have the following estimates.
\begin{itemize}
\item[(i)] 
$\int \mu^L(\dd \tilde \omega) \,T(\tilde \omega)^q \,|\xi^L_c(\omega,\tilde{\omega})| \lesssim_{\zeta,q,d}  \nu^{q-1}\,T(\omega)\, \|v^L_c\|_{L^1(\Lambda_L)}=o_L(\nu^{q-1}\, T(\omega))\,.$
\item[(ii)]  $\int_{\Lambda_L} \dd y\, \int \widehat{\mu}^L_{y,x}(\dd \tilde{\omega})\, T(\tilde{\omega})^q\,|\xi^L_c(\omega,\tilde{\omega})| \lesssim_{\zeta,q} \nu^{q-1}\, T(\omega)\, \|v^L_c\|_{L^{\infty}(\Lambda_L)}=o_L(\nu^{q-1}\, T(\omega))$.
\end{itemize}
\end{lemma}

Let us now rewrite the Wiener measure on $\Lambda_L$ in terms of a periodisation of the Wiener measure on $\R^d$. In order to do this, we introduce some notation.
Let $\pi_L \colon \R^d \rightarrow \Lambda_L$ denote the canonical projection map given by 
\begin{equation}
\label{pi_L}
\pi_L(x+Lk) \coloneqq x \quad  \text{for } \,\, x \in \Lambda_L\,,  k \in \Z^d\,.
\end{equation}
Given $T>0$, we let $\Omega^{\infty,T}$ denote the set of continuous paths $\omega \colon [0,T] \rightarrow \R^d$. Likewise, given $x,y \in \R^d$, we let $\Omega^{\infty,T}_{y,x}$ denote the set of paths $\omega \in \Omega^{\infty,T}$ such that $\omega(0)=x$ and $\omega(T)=y$. Given $t>0$, we denote by $\psi^{\infty,t} \colon \R^d \rightarrow \R$ the heat kernel 
\begin{equation}
\label{heat_kernel_infty}
\psi^{\infty,t}(x)=\frac{1}{(2\pi t)^{d/2}}\, \ee^{-\frac{|x|^2}{2t}}\,.
\end{equation}
In particular, we can rewrite \eqref{heat_kernel} as 
\begin{equation}
\label{heat_kernel_L_rewritten}
\psi^{L,t}(x)=\sum_{n \in \Z^d} \psi^{\infty,t}(x+Ln)\,.
\end{equation}
Given $x \in \R^d$ and $T>0$, let $\mathbb{P}^{\infty,T}_{x}$ denote the law on $\Omega^{\infty,T}$ of the standard Brownian motion on $\R^d$ equal to $x$ at time $0$, which is characterised as
\begin{multline*}
\int \mathbb{P}^T_{x} (\dd \omega) f(\omega(t_1),\ldots,\omega(t_n))
\\
=\int_{(\R^d)^n} \dd x_1\,\cdots \,\dd x_n\,\psi^{\infty,t_1}(x_1-x)\, \psi^{\infty,t_2-t_1}(x_2-x_1)\,\cdots\,\psi^{\infty,t_n-t_{n-1}}(x_n-x_{n-1})\,f(x_1,\ldots,x_n)
\end{multline*}
for all $n \in \N^*$, $f \colon (\R^d)^n \rightarrow \mathbb{C}$ continuous, and $0 <t_1<\cdots < t_n <T$. 
Moreover, given $x,y \in \R^d$, we denote by $\mathbb{P}^{T,\infty}_{y,x}$ the law of the Brownian bridge on $\R^d$ equal to $x$ at time $0$ and equal to $y$ at time $T$. We define the Wiener measure $\mathbb{W}^{\infty,T}_{y,x}$ analogously to \eqref{Wiener_measure} as
\begin{equation*}
\mathbb{W}^{\infty,T}_{y,x}(\dd \omega)\coloneqq \psi^{\infty,T}(y-x)\,\mathbb{P}^{T,\infty}_{y,x}(\dd \omega)
\end{equation*}
Moreover, for all $n \in \N^*$, $f \colon (\R^d)^n \rightarrow \mathbb{C}$ continuous, and $0 <t_1<\cdots < t_n <T$,
we have
\begin{multline}
\label{W^T_infty}
\int \mathbb{W}^{\infty,T}_{y,x} (\dd \omega) f(\omega(t_1),\ldots,\omega(t_n))
=\int_{(\R^d)^n} \dd x_1\,\cdots\, \dd x_n\,\psi^{\infty,t_1}(x_1-x)\, \psi^{\infty,t_2-t_1}(x_2-x_1)\,\cdots\,
\\
\times \cdots\,
\psi^{\infty,t_n-t_{n-1}}(x_n-x_{n-1})\,\psi^{\infty,T-t_n}(y-x_n)\,
f(x_1,\ldots,x_n)\,,
\end{multline}
as in \eqref{W^T}. 
In particular, we have 
\begin{equation}
\label{heat_kernel_integral_infty}
\int \mathbb{W}_{y, x}^{\infty,T}(\dd \omega)=\psi^{\infty,T}(y-x)\,.
\end{equation}
As in the case of finite $L$, we let $T(\omega)$ denote the duration of the path $\omega$.

\begin{lemma}[Rewriting the Wiener measure on the torus]
\label{Wiener_measure_pi_L}
Recalling \eqref{pi_L}, we have that
\begin{equation*}
\mathbb{W}^{L,T}_{y,x}(\dd \omega)=\sum_{k \in \Z^d} \left(\pi_L\right)_{\sharp} \mathbb{W}^{\infty,T}_{y+Lk,x}(\dd \omega)
\end{equation*}
for all $T>0$ and $x,y \in \Lambda_L$.
\end{lemma}

\begin{proof}
Let $n \in \N$, $f \colon (\Lambda_L)^n \rightarrow \mathbb{C}$ continuous and $0<t_1<\cdots<t_n<T$ be given.
We then compute
\begin{multline}
\label{Wiener_measure_pi_L_1}
\int_{\Omega^{L,T}} \sum_{k \in \Z^d} \left(\pi_L \right)_{\sharp} \mathbb{W}^{\infty,T}_{y+Lk,x}(\dd \omega)\, f(\omega(t_1), \ldots, \omega(t_n)) 
\\
= \sum_{k \in \Z^d}  \int_{\Omega^{L,T}}\left(\pi_L \right)_{\sharp} \mathbb{W}^{\infty,T}_{y+Lk,x}(\dd \omega)\, f(\omega(t_1), \ldots, \omega(t_n)) 
\\
=
\sum_{k \in \Z^d} \int_{\Omega^{\infty,T}}  \mathbb{W}^{\infty,T}_{y+Lk,x}(\dd \tilde{\omega})\,  f \circ \pi_L (\tilde{\omega}(t_1), \ldots, \tilde{\omega}(t_n))\,.
\end{multline}
When $f \geq 0$, all of the quantities are nonnegative so we can indeed interchange the order of summation and integration by using the monotone convergence theorem in the first equality of \eqref{Wiener_measure_pi_L_1}. In general, we decompose $f$ into its real and imaginary parts, which we then decompose into its nonnegative and negative parts. This is justified since the contribution of every part is finite by \eqref{Wiener_measure_pi_L_3}. In the last expression of \eqref{Wiener_measure_pi_L_1}, we emphasise that the integration on the left-hand side is over $\omega \in \Omega^{L,T}$, whereas the integration on the right-hand side is over $\tilde{\omega} \equiv \pi_L \omega \in \Omega^{\infty,T}$.
Since $f \circ \pi_L \colon (\R^d)^n \rightarrow \mathbb{C}$ is continuous, we can rewrite the right-hand side of \eqref{Wiener_measure_pi_L_1} using \eqref{W^T_infty} as
\begin{multline}
\label{Wiener_measure_pi_L_2}
\sum_{k \in \Z^d} \int_{\R^d} \dd x_1 \cdots \int_{\R^d} \dd x_n \, \psi^{\infty,t_1}(x_1-x)\, \psi^{\infty,t_2-t_1}(x_2-x_1)\, \cdots \psi^{\infty,T-t_n}(y+Lk-x_n)\, 
\\
\times
f(\pi_L x_1,\ldots, \pi_L x_n)
\\
=\sum_{k \in \Z^d} \sum_{k_1,\ldots,k_n \in \Z^d} \int_{\Lambda_L + Lk_1} \dd x_1 \cdots \int_{\Lambda_L+Lk_n} \dd x_n \,\psi^{\infty,t_1}(x_1-x)\, \psi^{\infty,t_2-t_1}(x_2-x_1) \cdots\\
\times
\cdots  \psi^{\infty,t_n-t_{n-1}} (x_n-x_{n-1})\,\psi^{\infty,T-t_n}(y+Lk-x_n)\,f(\pi_L x_1, \ldots, \pi_L x_n)
\\
=\sum_{k \in \Z^d} \sum_{k_1,\ldots,k_n \in \Z^d} \int_{\Lambda_L} \dd x_1 \cdots \int_{\Lambda_L} \dd x_n \,\psi^{\infty,t_1}(x_1+Lk_1-x)\, \psi^{\infty,t_2-t_1}(x_2+Lk_2-x_1-Lk_1) \cdots\\
\times
\cdots  \psi^{\infty,t_n-t_{n-1}} (x_n+Lk_n-x_{n-1}-Lk_{n-1})\,\psi^{\infty,T-t_n}(y+Lk-x_n-Lk_n)\,f(x_1, \ldots,x_n)
\,.
\end{multline}
Above, we wrote 
\begin{equation*}
\Lambda_L + Lk_j \equiv \{w+Lk_j\,, w \in \Lambda_L\}\,, \quad j=1,\ldots,n\,.
\end{equation*}
In \eqref{Wiener_measure_pi_L_2}, we first sum in $k$, and then in $k_n, k_{n-1}, \ldots, k_1$. Using \eqref{heat_kernel_L_rewritten}, we obtain 
\begin{multline}
\label{Wiener_measure_pi_L_3}
\eqref{Wiener_measure_pi_L_2}=\int_{\Lambda_L} \dd x_1\, \cdots \int_{\Lambda_L} \dd x_n\, \psi^{L,t_1}(x_1-x)\, 
\psi^{L,t_2-t_1}(x_2-x_1)\, \cdots\, \psi^{L,t_n-t_{n-1}}(x_n-x_{n-1})\, 
\\
\times 
\psi^{L,T-t_n}(y-x_n)\, f(x_1, \ldots, x_n)
=\int_{\Omega^{L,T}} \mathbb{W}^{L,T}(\dd \omega) \, f(\omega(t_1),\ldots,\omega(t_n))\,. 
\end{multline}
In \eqref{Wiener_measure_pi_L_3}, we used \eqref{W^T}. The claim follows from \eqref{Wiener_measure_pi_L_3} since $\mathbb{W}^T_{y,x}$ is characterised by its finite-dimensional distributions.
\end{proof}

\begin{proof}[Proof of Lemma \ref{Lemma_5.20*}]
We first prove claim (i). By \eqref{xi_bound}, \eqref{measure_mu_1_bound}, and \eqref{loop_interaction}, we have 
\begin{multline}
\label{Lemma_5.20*_1}
\int \mu^L(\dd \tilde{\omega})\,T(\tilde{\omega})^q\,|\xi^L(\omega,\tilde{\omega})|\,\mathbf{1}_{\mathscr{D}^L_c}(\tilde{\omega})
\leq \sum_{\tilde{T} \in \nu \N^*} z^{\tilde{T}}\,\tilde{T}^{q-1}\,\sum_{r \in \nu \N} \mathbf{1}
_{r<T(\omega)}\,\sum_{s \in \nu \N} \mathbf{1}_{s<\tilde{T}}\,\int_0^{\nu}\dd t\, 
\\
\left[\int_{\Lambda_L} \dd x\, \int \mathbb{W}^{L,\tilde{T}}_{x,x}(\dd \tilde{\omega}) \mathbf{1}_{\mathscr{D}^L_c}(\tilde{\omega})\,
v^L \left(\omega(t+r)-\tilde{\omega}(t+s)\right)
\right].
\end{multline}
We now estimate the expression in the square brackets in \eqref{Lemma_5.20*_1} for fixed $\tilde{T} \in \nu \N^*$, $t \in [0, \nu)$, $r \in \nu \N$ with $r<T(\omega)$, and $s \in \nu \N$ with $s<\tilde{T}$. To this end, for a given $y \in \Lambda_L$, we recall \eqref{concatenation_paths} and \eqref{D^L_c_definition} to write
\begin{multline}
\label{Lemma_5.20*_2}
\int_{\Lambda_L} \dd x\, \int \mathbb{W}^{L,\tilde{T}}_{x,x}(\dd \tilde{\omega}) \mathbf{1}_{\mathscr{D}^L_c}(\tilde{\omega})\,
v^L \left(y-\tilde{\omega}(t+s)\right)
\\
=\int_{\Lambda_L} \dd x\,\int_{\Lambda_L} \dd x'\,\int \mathbb{W}^{L,t+s}_{x',x}(\dd \tilde{\omega}_1) \int \mathbb{W}^{L,\tilde{T}-(t+s)}_{x,x'}(\dd \tilde{\omega}_2)\,\mathbf{1}_{\mathscr{D}^L_c}(\tilde{\omega}_1 \oplus \tilde{\omega}_2)\,v^L(y-x')
\\
=\int_{\Lambda_L} \dd x'\,\int_{\Lambda_L} \dd x\,\int \mathbb{W}^{L,\tilde{T}-(t+s)}_{x,x'}(\dd \tilde{\omega}_2) \int \mathbb{W}^{L,t+s}_{x',x}(\dd \tilde{\omega}_1)\,\mathbf{1}_{\mathscr{D}^L_c}(\tilde{\omega}_2 \oplus \tilde{\omega}_1)\,v^L(y-x')
\\
=\int_{\Lambda_L}\dd x'\, v^L(y-x')\,\left\{\int \mathbb{W}^{L,\tilde{T}}_{x',x'}(\dd \tilde{\omega})\,\mathbf{1}_{\mathscr{D}^L_c}(\tilde{\omega})\right\}.
\end{multline}
Let us now estimate the quantity in curly brackets in \eqref{Lemma_5.20*_2} for fixed $x' \in \Lambda_L$.
To this end, we recall \eqref{D^L_c_definition} and note that, by the triangle inequality, we have
\begin{equation}
\label{D^L_c_inclusion_1}
\mathscr{D}^L_c \subset \widehat{\mathscr{D}}^L_c(x')\,,
\end{equation}
where for $a \in \Lambda_L$ we define
\begin{equation}
\label{hat_D^L_c(z)_definition}
\widehat{\mathscr{D}}^L_c(a)\coloneqq\left\{\omega \in \bigcup_{T>0} \Omega^{L,T}\,,\,|\omega(s)-a|_L \geq \frac{cL}{2} \,\,\,\text{for some} \,\,\,s\in [0,T(\omega)]\right\}.
\end{equation}
Furthermore, we define 
\begin{equation}
\label{D^L_c(z)_definition}
\mathscr{D}^L_c(a)\coloneqq\left\{\omega \in \bigcup_{T>0} \Omega^{\infty,T}\,,\,|\omega(s)-a| \geq \frac{cL}{2} \,\,\,\text{for some} \,\,\,s\in [0,T(\omega)]\right\}.
\end{equation}
Recalling \eqref{periodic_Euclidean_norm} and \eqref{pi_L}, it follows from \eqref{hat_D^L_c(z)_definition}--\eqref{D^L_c(z)_definition} that 
\begin{equation}
\label{D^L_c_inclusion_2}
\widehat{\mathscr{D}}^L_c(a) \subset \pi_L \left( \mathscr{D}^L_c(a) \right).
\end{equation}
Combining Lemma \ref{Wiener_measure_pi_L}, \eqref{D^L_c_inclusion_1}, \eqref{D^L_c_inclusion_2}, and recalling \eqref{heat_kernel_integral_infty}, we deduce that 
\begin{multline}
\label{Lemma_5.20*_3}
\int \mathbb{W}^{L,\tilde{T}}_{x',x'}(\dd \tilde{\omega})\,\mathbf{1}_{\mathscr{D}^L_c}(\tilde{\omega}) \leq 
\int \mathbb{W}^{\infty,\tilde{T}}_{x',x'}(\dd \tilde{\omega})\,\mathbf{1}_{\mathscr{D}^L_c(x')}(\tilde{\omega})+
\sum_{k \in \Z^d \setminus \{0\}} \int \mathbb{W}^{\infty,\tilde{T}}_{x'+Lk,x'}(\dd \tilde \omega)\,\mathbf{1}_{\mathscr{D}^L_c(x')}(\tilde{\omega})
\\
\leq 
\int \mathbb{W}^{\infty,\tilde{T}}_{x',x'}(\dd \tilde{\omega})\,\mathbf{1}_{\mathscr{D}^L_c(x')}(\tilde{\omega})+ \sum_{k \in \Z^d \setminus \{0\}} \psi^{\infty,\tilde{T}}
(Lk)\,.
\end{multline}
For \eqref{Lemma_5.20*_3}, we used $\mathbf{1}_{\mathscr{D}^L_c(x')}(\tilde{\omega}) \leq 1$ in the sum over $k \in \Z^{d} \setminus \{0\}$.

By \eqref{D^L_c(z)_definition}, \eqref{heat_kernel_integral_infty}, Definition \ref{epsilon_L}, the dominated convergence theorem, and  \eqref{heat_kernel_infty}, we estimate the first term in \eqref{Lemma_5.20*_3} as
\begin{equation}
\label{Lemma_5.20*_3_1}
\int \mathbb{W}^{\infty,\tilde{T}}_{x',x'}(\dd \tilde{\omega})\,\mathbf{1}_{\mathscr{D}^L_c(x')}(\tilde{\omega})=\int \mathbb{W}^{\infty,\tilde{T}}_{0,0}(\dd \tilde{\omega})\,\mathbf{1}_{\mathscr{D}^L_c(0)}(\tilde{\omega})=o_L 
\left(\psi^{\infty,\tilde{T}}(0)\right) \lesssim_d o_L \left( \frac{1}{\tilde{T}^{d/2}}\right).
\end{equation}
We recall \eqref{heat_kernel_infty} and estimate the second term in \eqref{Lemma_5.20*_3} using
\begin{equation}
\label{Lemma_5.20*_3_2}
\sum_{k \in \Z^d \setminus \{0\}} \psi^{\infty,\tilde{T}}(Lk) \lesssim \sum_{k \in \Z^d \setminus \{0\}} \frac{1}{\tilde{T}^{d/2}}\,\ee^{-\frac{|Lk|^2}{2 \tilde{T}}} \lesssim \sum_{k \in \Z^d \setminus \{0\}} \frac{1}{\tilde{T}^{d/2}}\,\frac{\tilde{T}^{\frac{d+2}{2}}}{|Lk|^{d+2}} \lesssim_d \frac{\tilde{T}}{L^{d+2}}\,.
\end{equation}
Using \eqref{Lemma_5.20*_3}--\eqref{Lemma_5.20*_3_2}, followed by Lemma \ref{v^L_lemma} (ii), we obtain that 
\begin{equation}
\label{Lemma_5.20*_4A}
\eqref{Lemma_5.20*_2} \lesssim_d o_L \left(\left(\frac{1}{\tilde{T}^{d/2}}+\tilde{T}\right)\,\|v\|_{L^1(\R^d)} \right).
\end{equation}
By using \eqref{Lemma_5.20*_4A}, it follows that
\begin{multline*}
\eqref{Lemma_5.20*_1} \lesssim_{d} o_L \left( 
\sum_{\tilde{T} \in \nu \N^*} z^{\tilde{T}}\, \tilde{T}^{q-1}\, \sum_{r \in \nu \N} 
\mathbf{1}_{r<T(\omega)}\,\sum_{s \in \nu \N} \mathbf{1}_{s<\tilde{T}} \, \int_0^{\nu}\, \dd t\, \left(\frac{1}{\tilde{T}^{d/2}}+\tilde{T}\right)\, \|v\|_{L^1(\R^d)}
\right)
\\
=o_L \left( 
T(\omega)\, \|v\|_{L^1(\R^d)}\,\sum_{\tilde{T} \in \nu \N^*} z^{\tilde{T}}\, \tilde{T}^{q}\,\nu^{-1}\,\left(\frac{1}{\tilde{T}^{d/2}}+\tilde{T}\right) \right)
\lesssim_{\zeta,q,d} o_L \left(
\nu^{q-1}\,
T(\omega)\,\|v\|_{L^1(\R^d)}
\right),
\end{multline*}
where, in the line above, we used \eqref{5.27_remark_bound_proof} from Remark \ref{5.27_remark} to estimate 
\begin{equation*}
\sum_{\tilde{T} \in \nu \N^*} z^{\tilde{T}}\,\tilde{T}^{q-\frac{d}{2}} \, \nu^{-1} \lesssim_{\zeta,q} \nu^{q-1}\,, \qquad
\sum_{\tilde{T} \in \nu \N^*} z^{\tilde{T}} \, \tilde{T}^{q+1} \, \nu^{-1} \lesssim_{\zeta,q} \nu^{q+\frac{d}{2}}\,.
\end{equation*}
Claim (i) now follows.

We now prove claim (ii). By \eqref{xi_bound}, \eqref{measure_mu_2_bound}, \eqref{loop_interaction}, Lemma \ref{v^L_lemma}  (iii), and \eqref{D^L_c_inclusion_1},  we have 

\begin{multline}
\label{Lemma_5.20*_4}
\int_{\Lambda_{L_0}} \dd y\,\int \widehat{\mu}^L_{y,x}(\dd \tilde{\omega}) \,T(\tilde{\omega})^q\,|\xi^L(\omega,\tilde{\omega})|\,\mathbf{1}_{\mathscr{D}^L_c}(\tilde{\omega}) 
\leq \frac{1}{\nu} \sum_{\tilde{T} \in \nu \N^*} z^{\tilde{T}}\,\tilde{T}^q\,\sum_{r \in \nu \N} \mathbf{1}_{r<T(\omega)} \sum_{s \in \nu \N} \mathbf{1}_{s<\tilde{T}} \int_0^{\nu} \dd t\, 
\\
\times \int_{\Lambda_{L_0}} \dd y\,\int \mathbb{W}^{L,\tilde{T}}_{y,x}(\dd \tilde{\omega})\, v^L \left(\omega(t+r)-\tilde{\omega}(t+s) \right)\,\mathbf{1}_{\mathscr{D}^L_c}(\tilde{\omega})
\\
\lesssim \frac{1}{\nu}\,\sum_{\tilde{T} \in \nu \N^*} z^{\tilde{T}}\,\tilde{T}^q\,\sum_{r \in \nu \N} \mathbf{1}_{r<T(\omega)} \sum_{s \in \nu \N} \mathbf{1}_{s<\tilde{T}} \int_0^{\nu} \dd t\, \left[\int_{\Lambda_{L_0}} \dd y\,\int \mathbb{W}^{L,\tilde{T}}_{y,x}(\dd \tilde{\omega})\,\mathbf{1}_{\widehat{\mathscr{D}}^L_c(x)}(\tilde{\omega})\right].
\end{multline}
For fixed $\tilde{T} \in \nu \N^*$, $r \in \nu \N$ with $r<T(\omega)$, and $s \in \nu \N$ with $s<\tilde{T}$, we now estimate the expression in square brackets in \eqref{Lemma_5.20*_4}.
Using Lemma \ref{Wiener_measure_pi_L}, \eqref{D^L_c_inclusion_2}, and arguing as in \eqref{Lemma_5.20*_3}, we have
\begin{multline}
\label{Lemma_5.20*_5}
\int_{\Lambda_{L_0}} \dd y\,\int \mathbb{W}^{L,\tilde{T}}_{y,x}(\dd \tilde{\omega})\,\mathbf{1}_{\widehat{\mathscr{D}}^L_c(x)}(\tilde{\omega}) \leq 
\int_{\Lambda_{L_0}} \dd y\,\int \mathbb{W}^{\infty,\tilde{T}}_{y,x}(\dd \tilde{\omega})\,\mathbf{1}_{\mathscr{D}^L_c(x)}(\tilde{\omega})
\\
+\sum_{k \in \Z^d \setminus \{0\}}\int_{\Lambda_{L_0}} \dd y\,\int \mathbb{W}^{\infty,\tilde{T}}_{y+Lk,x}(\dd \tilde{\omega})\,.
\end{multline}
We note that 
\begin{equation}
\label{Lemma_5.20*_5_1}
\int_{\Lambda_{L_0}} \dd y\,\int \mathbb{W}^{\infty,\tilde{T}}_{y,x}(\dd \tilde{\omega})\,\mathbf{1}_{\mathscr{D}^L_c(x)}(\tilde{\omega})=o_L(1)\,,
\end{equation}
uniformly in $x \in \Lambda_{L_0}$.
In order to see \eqref{Lemma_5.20*_5_1}, we first use translation invariance of the Wiener measure on $\R^d$ and \eqref{D^L_c(z)_definition} to write 
\begin{multline*}
\int_{\Lambda_{L_0}} \dd y\,\int \mathbb{W}^{\infty,\tilde{T}}_{y,x}(\dd \tilde{\omega})\,\mathbf{1}_{\mathscr{D}^L_c(x)}(\tilde{\omega})=\int_{\Lambda_{L_0}} \dd y\,\int \mathbb{W}^{\infty,\tilde{T}}_{y-x,0}(\dd \tilde{\omega})\,\mathbf{1}_{\mathscr{D}^L_c(0)}(\tilde{\omega}) 
\\
\leq \int_{\R^d} \dd y\,\int \mathbb{W}^{\infty,\tilde{T}}_{y,0}(\dd \tilde{\omega})\,\mathbf{1}_{\mathscr{D}^L_c(0)}(\tilde{\omega})\,.
\end{multline*}
By \eqref{heat_kernel_integral_infty}, we have that
\begin{equation*}
\int_{\R^d}\dd y\,\int \mathbb{W}^{\infty,\tilde{T}}_{y,0}(\dd \tilde{\omega})
=
\int_{\R^d} \dd y \,\psi^{\infty,\tilde{T}}(y) = 1\,.
\end{equation*}
Finally, we apply the dominated convergence theorem to deduce that \eqref{Lemma_5.20*_5_1} holds uniformly in $x \in \Lambda_{L_0}$.

We now estimate the second term in \eqref{Lemma_5.20*_5}.
Let us observe that by \eqref{L_0}, it follows that 
\begin{equation}
\label{x-y-Lk}
|y+Lk-x| \geq |Lk|-|x|-|y| \gtrsim L\,.
\end{equation}
Using \eqref{heat_kernel_integral_infty} and \eqref{x-y-Lk}, followed by \eqref{heat_kernel_infty}, the second term in \eqref{Lemma_5.20*_5} is
\begin{multline}
\label{Lemma_5.20*_5_2}
\sum_{k \in \Z^d \setminus \{0\}}\int_{\Lambda_{L_0}} \dd y\,\psi^{\infty,\tilde{T}}(y+Lk-x) 
\leq  \int_{\R^d} \dd y \,\mathbf{1}_{|y| \gtrsim L}\, \psi^{\infty,\tilde{T}}(y)
\lesssim \int_{\R^d} \dd y \,\mathbf{1}_{|y| \gtrsim L}\, \frac{1}{\tilde{T}^{d/2}}\,\ee^{-\frac{|y|^2}{2 \tilde{T}}}
\\ \lesssim \int_{\R^d} \dd y\,\mathbf{1}_{|y| \gtrsim \frac{L}{\sqrt{\tilde{T}}}} \,\ee^{-|y|^2}
\lesssim_d \int_0^{\infty} \dd r\,\mathbf{1}_{r \gtrsim \frac{L}{\sqrt{\tilde{T}}}}\,r^{d-1}\,\ee^{-r^2}
\leq \frac{\tilde{T}}{L^2}\,\int_0^{\infty} \dd r\,\mathbf{1}_{r \gtrsim \frac{L}{\sqrt{\tilde{T}}}}\,r^{d+1}\,\ee^{-r^2}
\lesssim_{d} \frac{\tilde{T}}{L^2}\,,
\end{multline}
uniformly in $x \in \Lambda_{L_0}$.
From \eqref{Lemma_5.20*_5_1} and \eqref{Lemma_5.20*_5_2}, it follows that
\begin{equation}
\label{Lemma_5.20*_5_3}
\eqref{Lemma_5.20*_5} \lesssim_d o_L(1+\tilde{T})\,.
\end{equation}
Using \eqref{Lemma_5.20*_5_3} in \eqref{Lemma_5.20*_4}, and arguing as in \eqref{Lemma_5.16*_7}, we deduce claim (ii).

Let us now prove claim (iii). As in \eqref{Triangle_1}, we have
\begin{equation}
\label{Triangle_1_*}
\int \mu^L(\dd \tilde{\omega})\,T(\tilde{\omega})^q\,\mathbf{1}_{\mathscr{D}^L_c}(\tilde{\omega})  \leq \nu \sum_{\tilde{T} \in \nu \N^*} z^{\tilde{T}}\,\tilde{T}^{q-1}\,\int_{\Lambda_L}
\dd x\, \int \mathbb{W}^{L,\tilde{T}}_{x,x}(\dd \tilde{\omega}) \,\mathbf{1}_{\mathscr{D}^L_c}(\tilde{\omega})\,.
\end{equation}
We now apply \eqref{Lemma_5.20*_3}--\eqref{Lemma_5.20*_3_2} in \eqref{Triangle_1_*} and deduce the claim by arguing as in \eqref{Triangle_1}--\eqref{Triangle_2}.

To conclude, we prove claim (iv). To this end, we first note that by \eqref{measure_mu_2_bound} and \eqref{D^L_c_inclusion_1}, we have
\begin{equation}
\label{Lemma_5.20*_7}
\int_{\Lambda_{L_0}} \dd y \int \widehat{\mu}^L_{y,x} (\dd \tilde{\omega})\, T(\tilde{\omega})^q\, \mathbf{1}_{\mathscr{D}^L_c}(\tilde{\omega})
\leq \sum_{\tilde{T} \in \nu \N^*} z^{\tilde{T}}\,\tilde{T}^q\, \left[\int_{\Lambda_{L_0}} \dd y\, \int \mathbb{W}^{L,\tilde{T}}_{y,x}(\dd \tilde{\omega})\,\mathbf{1}_{\widehat{\mathscr{D}}^L_c(x)}(\tilde{\omega})\right].
\end{equation}
We estimate the quantity in square brackets in \eqref{Lemma_5.20*_7} by using \eqref{Lemma_5.20*_5} and \eqref{Lemma_5.20*_5_3}. Hence
\begin{equation}
\label{Lemma_5.20*_8}
\eqref{Lemma_5.20*_7} \lesssim_d o_L \left(\sum_{\tilde{T} \in \nu \N^*} z^{\tilde{T}}\,\tilde{T}^q (1+\tilde{T})\right) \lesssim_{\zeta,q} o_L(\nu^q)\,.
\end{equation}
For the last estimate in \eqref{Lemma_5.20*_8}, we argued as in \eqref{Proof_of_iv}. We now deduce claim (iv).

\end{proof}

\begin{remark}
\label{quantitative_estimate_remark}
One can obtain a more precise quantitative estimate on \eqref{Lemma_5.20*_3_1} and \eqref{Lemma_5.20*_5_1} in terms of $L$. More precisely, to estimate \eqref{Lemma_5.20*_3_1}, we first note that $\mathscr{D}^L_c(0)$ in \eqref{D^L_c(z)_definition} is a tail event of the random variable $\sup_{t \in [0,\widetilde{T}]}|B_t|$, where $(B_t)_{t \in [0,\tilde{T}]}$ is a $d$-dimensional Brownian bridge with initial and terminal site $0$ with time horizon $[0,\tilde{T}]$. If $d=1$, then we may write the Brownian bridge measure as
\begin{equation}\label{Brownian_bridge_bound}
    \mathbb{W}_{0,0}^{\infty,\tilde{T}}(\mathscr{D}^L_c(0))=\frac{\mathbb{P}_0(\sup_{t \in [0,\widetilde{T}]}|W_t|\ge cL/2 , W_{\tilde{T}}\in \dd y)}{\dd y}\Bigg |_{y=0},
\end{equation}
where $(W_t)_{t \ge 0}$ is a standard Brownian motion. We can now obtain an upper bound for \eqref{Brownian_bridge_bound} using \cite[(1.15.8), Page 180]{BS02}. If $d>1$, we may use dimension reduction arguments, since the components of the $d$-dimensional vector $B_t$ are independent, one-dimensional Brownian bridges. Hence, one may obtain analogous tail bounds, up to dimension-dependent constants. We use analogous arguments to estimate \eqref{Lemma_5.20*_5_1}. However, our analysis does not directly rely on these improved estimates, and the above qualitative estimates suffice for our purposes.
\end{remark}

\begin{proof}[Proof of Lemma \ref{Lemma_5.21*}]
The proof of Lemma \ref{Lemma_5.21*} is analogous to that of Lemma \ref{Lemma_5.16*} (i)--(ii). The only difference is that we do not use Lemma \ref{v^L_lemma} (ii). Instead, we keep the factor $\|v^L_c\|_{L^1(\Lambda_L)}$ and obtain the estimates stated above.
\end{proof}

\begin{remark}
\label{long_far_away_path_comment}
By the same proofs, we can modify the results of Lemmas \ref{Lemma_5.20*} and  \ref{Lemma_5.21*} as follows. Let $1 \leq \tilde{L} < L = \infty$ be given. Here, we take $|\cdot|_{\infty} \equiv |\cdot|$ to be the Euclidean norm on $\R^d$, and for the interaction potential, we take $v^{\infty}\equiv v$. In Lemma \ref{Lemma_5.20*}, we can replace the indicator $\mathbf{1}_{\mathscr{D}^L_c}(\tilde{\omega})$ by 
$\mathbf{1}_{\mathscr{D}^{\infty,\tilde{L}}_c}(\tilde{\omega})$, where 
\begin{equation}
\label{D^{L,tilde_L}_c}
\mathscr{D}^{\infty,\tilde{L}}_c\coloneqq\left\{\omega \in \bigcup_{T>0} \Omega^{\infty,T}\,,\,|\omega(t)-\omega(s)| \geq c\tilde{L} \,\,\,\text{for some} \,\,\,s,t \in [0,T(\omega)]\right\}.
\end{equation}
Finally, we define $\mu^{\infty}(\dd \omega)$ and $\widehat{\mu}_{y,x}^{\infty}(\dd \omega)$ by setting $L=\infty$ in \eqref{measure_mu_1} and \eqref{measure_mu_2}, respectively.
Then, all of the upper bounds in Lemma \ref{Lemma_5.20*} change from $o_L(\cdots)$ to $o_{\tilde{L}}(\cdots)$. In other words, we have the following estimates.

\begin{itemize}
\item[(i)] 
$\int \mu^{\infty}(\dd \tilde{\omega})\,T(\tilde{\omega})^q\,|\xi^{\infty}(\omega,\tilde{\omega})|\,\mathbf{1}_{\mathscr{D}^{\infty,\tilde{L}}_c}(\tilde{\omega}) \lesssim_{\zeta,q,d} o_{\tilde{L}} \left(\nu^{q-1}\,
T(\omega)\,\|v\|_{L^1(\R^d)}\right)$.
\item[(ii)] 
$\int_{\R^d} \dd y\,\int \widehat{\mu}^{\infty}_{y,x}(\dd \tilde{\omega}) \,T(\tilde{\omega})^q\,|\xi^{\infty}(\omega,\tilde{\omega})|\,\mathbf{1}_{\mathscr{D}^{\infty,\tilde{L}}_c}(\tilde{\omega}) \lesssim_{\zeta,q,d} o_{\tilde{L}} \left(\nu^{q-1}\,T(\omega)\,\right)$.
\item[(iii)] 
Instead of Lemma \ref{Lemma_5.20*} (iii), we have the estimate
\begin{equation}
\label{Remark_5.14_iii}
\frac{1}{|\Lambda_{\tilde{L}}|}\,\int_{\Lambda_{\tilde{L}}} \dd \tilde{x}\,\nu \sum_{\tilde{T} \in \nu \N^*} \frac{z^{\tilde{T}}}{\tilde{T}}\, \int \mathbb{W}^{\infty,\tilde{T}}_{\tilde{x},\tilde{x}}(\dd \tilde{\omega})\, 
\mathrm{e}^{-\mathcal{V}^{\nu,\infty}(\tilde{\omega})}\,\mathbf{1}_{\mathscr{D}^{\infty,\tilde{L}}_c}(\tilde{\omega})\,\tilde{T}^q  \lesssim_{\zeta,q,d} o_{\tilde{L}}\left(\nu^q\right).
\end{equation}
\item[(iv)] 
$\int_{\R^d} \dd y\, \int \widehat{\mu}^{\infty}_{y,x} (\dd \tilde{\omega}) \, T(\tilde{\omega})^q\, \mathbf{1}_{\mathscr{D}^{\infty,\tilde{L}}_c}(\tilde{\omega}) \lesssim_{\zeta,q,d} o_{\tilde{L}}(\nu^q)$. $\hfill\refstepcounter{equation}(\theequation)\label{Remark_5.14_iv}$
\end{itemize}
Note that for \eqref{Remark_5.14_iii} in part (iii) above, we are integrating over the finite volume $\Lambda_{\tilde{L}}$ in $\tilde{x}$.

Suppose that instead of \eqref{v^{L}_{c}_definition}, we consider
\begin{equation}
\label{v^{L,tilde_L}_{c}_definition}
v^{\infty,\tilde{L}}_{c}(x)\coloneqq v(x)\,\mathbf{1}_{|x| \geq c\tilde{L}}\,, 
\end{equation}
and instead of \eqref{xi^L_c} we consider 
\begin{equation*}
\xi^{\infty,\tilde{L}}_{c}(\omega,\tilde{\omega}) \coloneqq  \exp\bigl(-\mathcal{V}^{\nu,\infty,\tilde{L}}_{c}(\omega,\tilde \omega)\bigr) - 1\,,
\end{equation*}
where $\mathcal{V}^{\nu,\infty,\tilde{L}}_{c}(\omega,\tilde \omega)$ denotes \eqref{loop_interaction} with interaction potential \eqref{v^{L,tilde_L}_{c}_definition}.
Then, the factors of $\|v^L_c\|_{L^1(\Lambda_L)}$ in the upper bounds in Lemma \ref{Lemma_5.21*} change to $\|v^{\infty,\tilde{L}}_c\|_{L^1(\Lambda_L)}$. In other words, instead of Lemma \ref{Lemma_5.21*}, we have the following estimates.
\begin{itemize}
\item[(i')] $\int \mu^{\infty}(\dd \tilde \omega) \,T(\tilde \omega)^q \,|\xi^{\infty,\tilde{L}}_c(\omega,\tilde{\omega})| \lesssim_{\zeta,q,d}  \nu^{q-1}\,T(\omega)\, \|v^{\infty,\tilde{L}}_c\|_{L^1(\R^d)} =o_{\tilde{L}}(\nu^{q-1}\, T(\omega))$.
\item[(ii')] $\int_{\R^d} \dd y\, \int \widehat{\mu}^{\infty}_{y,x}(\dd \tilde{\omega})\, T(\tilde{\omega})^q\,|\xi^{\infty,\tilde{L}}_c(\omega,\tilde{\omega})| \lesssim_{\zeta,q} \nu^{q-1}\, T(\omega)\, \|v^{\infty,\tilde{L}}_c\|_{L^{\infty}(\R^d)}= o_{\tilde{L}}(\nu^{q-1}\, T(\omega))$.
\end{itemize}
Here, we used \eqref{v^{L,tilde_L}_{c}_definition} to deduce that the above quantities are indeed of the magnitude $o_{\tilde{L}}(\nu^{q-1}\, T(\omega))$.
\end{remark}

\begin{remark}
Let us note that, in the proof of Lemmas \ref{Lemma_5.16*}, \ref{Lemma_5.20*}, and \ref{Lemma_5.21*}, we only used $\kappa>0$ and not the full assumption $\kappa>1$ from \eqref{zeta_infinite_volume}--\eqref{kappa_definition}. 
\end{remark}

In what follows, we prove Theorems \ref{free_energy_infinite_volume} and \ref{Infinite_volume_Theorem_2}. The proofs follow similar arguments, based on the ideas we have developed in this section so far. We first prove Theorem  \ref{Infinite_volume_Theorem_2} in Section \ref{Proof_of_Theorem_{Infinite_volume_Theorem_2}}, as its proof is slightly simpler. We then prove Theorem \ref{free_energy_infinite_volume} in Section \ref{Proof_of_Theorem_free_energy_infinite_volume}.

\subsubsection{Analysis of reduced density matrices in the infinite volume. Proof of Theorem \ref{Infinite_volume_Theorem_2}}
\label{Proof_of_Theorem_{Infinite_volume_Theorem_2}}

Let us now prove Theorem \ref{Infinite_volume_Theorem_2}. 
Before proceeding with the proof, we note an explicit representation of reduced density matrices in the infinite volume.
To this end, we use Lemma \ref{Lemma_5.16*} (i)--(ii) and (iv), as well as 
the integration algorithm from Lemma \ref{integration_algorithm}, and argue as in the proof of Proposition \ref{Proposition_5.3} (ii) adapted for $L=\infty$, to deduce the following.
\begin{corollary}[Reduced density matrices in the infinite volume]
\label{operator_norm_bound_corollary}
 Let $p \in \N^*$ and be given. Under the assumption \eqref{smallness_v_kappa}, we can take $L=\infty$ in \eqref{gamma_p_cluster_expansion_claim} and \eqref{gamma_p_cluster_expansion_absolute_value} to obtain operators $\Gamma^{\nu,\zeta,\infty}_p$ and $|\Gamma|^{\nu,\zeta,\infty}_p$ acting in the variables $(\R^{d})^{\otimes p}$  with kernels 
\begin{equation}
\label{kernel_gamma_infty}
\Gamma^{\nu,\zeta,\infty}_{p}({\vec x,\vec y}) \coloneqq \sum_{\pi \in S_p} \int \widehat \mu^{\infty}_{y_{\pi(1)},x_{1}}(\dd \omega_{1}) \cdots \widehat \mu^{\infty}_{y_{\pi(p)},x_{p}}(\dd \omega_{p}) \sum_{\Pi \in \mathfrak{P}_{p}} \prod_{\lambda \in \Pi} X^{\infty}((\omega_{i})_{i \in \lambda})
\end{equation}
and
\begin{equation}
\label{gamma_p_cluster_expansion_absolute_value_infty}
|\Gamma|^{\nu,\zeta,\infty}_{p}({\vec x,\vec y}) \coloneqq \sum_{\pi \in S_p} \int \widehat \mu^{L}_{y_{\pi(1)},x_{1}}(\dd \omega_{1}) \cdots \widehat \mu^{\infty}_{y_{\pi(p)},x_{p}}(\dd \omega_{p}) \sum_{\Pi \in \mathfrak{P}_{p}} \prod_{\lambda \in \Pi} |X|^{\infty}((\omega_{i})_{i \in \lambda})\,,
\end{equation}
In \textup{\eqref{kernel_gamma_infty}--\eqref{gamma_p_cluster_expansion_absolute_value_infty}}, we take $X^{\infty}, |X|^{\infty}$ as in \textup{\eqref{Ursell_function_representation}--\eqref{Ursell_function_representation_absolute_value}} with $L=\infty$ respectively.
Then, we have 
\begin{equation}
\label{gamma_p_cluster_expansion_absolute_value_infty_bound}
\left\||\Gamma|^{\nu,\zeta,\infty}_p\right\|_{L^{\infty}_{\vec{x} } L^1_{\vec{y}}}=\mathcal {O}_{d,\zeta,p,v}(1)
\end{equation}
and hence 
\begin{equation*}
\left\|\Gamma^{\nu,\zeta,\infty}_p\right\|_{L^{\infty}_{\vec{x} } L^1_{\vec{y}}}=\mathcal {O}_{d,\zeta,p,v}(1)
\end{equation*}
In particular, the series in \eqref{kernel_gamma_infty} is absolutely summable in the $L^{\infty}_{\vec{x}} L^1_{\vec{y}}$ norm in the sense of \eqref{gamma_p_cluster_expansion_absolute_value_infty_bound} above.
\end{corollary}

We now have all of the tools to prove Theorems \ref{free_energy_infinite_volume} and \ref{Infinite_volume_Theorem_2}. Before proceeding with the proof, let us first outline the main challenges in the analysis.
Our starting point is the Ginibre representation \eqref{gamma_p_cluster_expansion_claim}. When studying the infinite volume limit, it is important to use that $v^L$ converges to $v$ in $L^1(\Lambda)$ and $L^\infty(\Lambda)$; see Lemma \ref{v^L_lemma} (iv)--(v). Furthermore, Brownian motion on $\Lambda_L$ is subject to periodic boundary conditions, while Brownian motion on $\R^{d}$ is not. To reconcile this difference, we limit our analysis to a smaller box in the interior of $\Lambda_{L}$, which allows us to ignore the effects of boundary conditions. This reduction is justified by the negligible contribution of long and far-away paths in the Ginibre representation, and is done in \eqref{A^{L}_B^{L}}--\eqref{(ii)_star_proof1}, based on Lemmas \ref{Lemma_5.20*} and \ref{Lemma_5.21*} above.

\begin{proof}[Proof of Theorem \ref{Infinite_volume_Theorem_2}]
Let us first show that $\Gamma_{p}^{\nu,\zeta,L}$ converges to $\Gamma_{p}^{\nu,\zeta,\infty}$
with respect to \eqref{L_0_seminorm} uniformly in $\nu \in (0,\nu_0]$. Here we recall \eqref{nu_0} from Assumption \ref{Choice_of_z}, and \eqref{L_0}.
We work with the explicit representations of $\Gamma_{p}^{\nu,\zeta,L}$ and $\Gamma_{p}^{\nu,\zeta,\infty}$ in \eqref{gamma_p_cluster_expansion_claim}.

Recalling \eqref{L_0_seminorm}, let us now estimate the quantity 
\begin{equation}
\label{Gamma_difference}
\|\Gamma_{p}^{\nu,\zeta,L} - \Gamma_{p}^{\nu,\zeta,\infty}\|_{L_0,p}\,. 
\end{equation}
When estimating \eqref{Gamma_difference} it suffices to use the triangle inequality and consider the term coming from \eqref{gamma_p_cluster_expansion_claim} when $\pi=\vec{1}$ is the identity permutation.
By using \eqref{operator_norm_bound_1_*}, \eqref{gamma_p_cluster_expansion_absolute_value_infty_bound}, and the dominated convergence theorem, it suffices to show that for all $\vec{x}, \vec{y} \in (\Lambda_{L_{0}})^{p}$ and for all $n \in \N$ we have that
\begin{multline}
\label{(ii)_star}
\lim_{L \rightarrow \infty}  \,\biggl\|\int \widehat \mu^{L}_{y_1,x_1}(\dd \omega_1) \cdots \widehat \mu^{L}_{y_p,x_p}(\dd \omega_p)\,\mu^{L}(\dd \omega_{p+1}) \cdots \mu^{L}(\dd \omega_n)\,\varphi^{L}(\omega_1,\ldots,\omega_n)
\\
-\int \widehat \mu^{\infty}_{y_1,x_1}(\dd \omega_1) \cdots \widehat \mu^{\infty}_{y_p,x_p}(\dd \omega_p)\,\mu^{\infty}(\dd \omega_{p+1}) \cdots \mu^{\infty}(\dd \omega_n)\,\varphi^{\infty}(\omega_1,\ldots,\omega_n)
 \biggr\|_{L_0,p} = 0\,.
\end{multline}
When $L = \infty$, we consider \eqref{measure_mu_1}--\eqref{Ursell_function_definition} with appropriate modifications; see \eqref{Ursell_function_definition_modification_1} as well.
In \eqref{(ii)_star}, we also recall
\eqref{Ursell_function_representation_absolute_value}, \eqref{gamma_p_cluster_expansion_claim}, \eqref{gamma_p_cluster_expansion_absolute_value}, and \eqref{kernel_gamma_infty}--\eqref{gamma_p_cluster_expansion_absolute_value_infty}. In \eqref{(ii)_star} and in the remainder of the proof, all of the convergence claims are interpreted as being uniform in $\nu \in (0,\nu_0]$.

Let us now prove \eqref{(ii)_star}. We henceforth assume that $L \geq 1$, fix $L_{0}$ satisfying \eqref{L_0}, and set
\begin{equation}
\label{L_1_choice}
L_1 \coloneqq L/3 \,.
\end{equation}
It is crucial that $L_{1} < \tfrac{L}{2}$ since this guarantees that $x-y \in \Lambda_{L}$ for any $x, y \in \Lambda_{L_{1}}$, where the difference is interpreted in the non-periodic sense. More precisely, $|x-y|_{\infty} < \tfrac{L}{2}$ for any $x, y \in \R^{d}$ such that $|x|_{\infty} < \tfrac{L_{1}}{2}$ and $|y|_{\infty} < \tfrac{L_{1}}{2}$.
Finally, define
\begin{equation}
\label{A^{L}_B^{L}}
\mathscr{A}^{L} \coloneqq \biggl\{\omega \in \bigcup_{T > 0} \Omega^{\infty,T},\, \omega(t) \in \Lambda_{L_1}\, \forall t \in [0,T(\omega)]\biggr\}\,,\quad \mathscr{B}^{L} \coloneqq  \bigcup_{T > 0} \Omega^{\infty,T} \setminus \mathscr{A}^{L}\,.
\end{equation}
Let us first show that for all $k \in \{1,\ldots,n\}$,
\begin{equation}
\label{(ii)_star_proof1}
\lim_{L \rightarrow \infty} \,\biggl\|\int \widehat \mu^{L}_{y_1,x_1}(\dd \omega_1) \cdots \widehat \mu^{L}_{y_p,x_p}(\mathrm{d} \omega_p) \,\mu^{L}(\mathrm{d} \omega_{p+1}) \cdots \mu^{L}(\mathrm{d} \omega_n)\,\vec1_{\mathscr{B}^{L}}(\omega_k)\,\varphi^{L}(\omega_1,\ldots,\omega_n)\biggr\|_{L_0,p}=0\,.
\end{equation} 
Lemma \ref{tree_bound_estimate} allows us to deduce \eqref{(ii)_star_proof1} provided we show that for all $k \in \{1,\ldots,n\}$,
\begin{equation}
\label{(ii)_star_proof2}
\lim_{L \rightarrow \infty} \,\biggl\|\int \widehat \mu^{L}_{y_1,x_1}(\mathrm{d} \omega_1) \cdots \widehat \mu^{L}_{y_p,x_p}(\mathrm{d} \omega_p)\,\mu^{L}(\mathrm{d} \omega_{p+1}) \cdots  \mu^{L}(\mathrm{d} \omega_n)\,\vec{1}_{\mathscr{B}^{L}}(\omega_k)\,\prod_{\{i,j\} \in \mathcal{T}} |\xi^{L}(\omega_i,\omega_j)|\biggr\|_{L_0,p}=0
\end{equation}
for a fixed tree $\mathcal {T} \in {\mathfrak{T}}_n$.
Let us now show that \eqref{(ii)_star_proof2} follows by applying the triangle inequality together with Lemmas \ref{Lemma_5.20*} and \ref{Lemma_5.21*}.
Observe first that by \eqref{1.37*}--\eqref{L_0_seminorm}, the nonzero contribution on the left-hand side of \eqref{(ii)_star_proof2} comes from $x_1 \in \Lambda_{L_0}$. For such $x_1$, we recall that by \eqref{L_0} and \eqref{L_1_choice}--\eqref{A^{L}_B^{L}}, there exists a point on $\omega_k$, denoted by $a=\omega(t_a)$ for some $t_a \in [0,T(\omega_k)]$, such that 
\begin{equation}
\label{a-x_1}
|a-x_1|_{L} \geq c_1L\,, 
\end{equation}
where $c_1>0$ is a constant. In particular,  recalling \eqref{L_0} and \eqref{L_1_choice} one can take
$c_{1} = \tfrac{1}{6} - \tfrac{1}{8} = \tfrac{1}{24}$. It is crucial here that $L_{0} < L_{1}$.
By the triangle inequality, at least one of the following cases occurs with suitable $c\gtrsim \frac{1}{n}$ for the paths over which we are integrating in \eqref{(ii)_star_proof2}. 
\begin{itemize}
\item[(1)] There exists $i \in \{1,\ldots,n\}$ such that $\omega_i \in \mathscr{D}^{L}_{c}$. Here, we recall \eqref{D^L_c_definition}.
\item[(2)] There exists $\{i,j\} \in \mathcal {T}$ such that $\xi^{L}(\omega_i,\omega_j)=\xi^{L}_{c}(\omega_i,\omega_j)$. Here, we recall \eqref{xi^L_c}.
\end{itemize}
To see that (1) or (2) must hold, let $P_{1,k}\coloneqq(1=u_1\to u_2\to \cdots \to u_m=k)$ be the unique path in $\mathcal{T}$ that connects vertex $1$ with vertex $k$. Then, for arbitrary $t_{i} \in [0, T(\omega_{u_i})]$, $i=1,\ldots,m$, we have 
\begin{equation}
\label{tree_triangle_inequality_1}
    c_1L\leq |\omega_1(0)-\omega_k(t_a)|_{L}\leq \sum_{i=1}^{m-1}|\omega_{u_{i+1}}(t_{i+1})-\omega_{u_i}(t_{i})|_{L}+|\omega_1(0)-\omega_1(t_1)|_{L}+|\omega_{k}(t_{m})-\omega_k(t_a)|_{L}\,.
\end{equation}
Since \eqref{tree_triangle_inequality_1} holds for all $t_{i} \in [0, T(\omega_{u_i})]$, $i=1,\ldots,m$, we may choose $t_{i}$ such that
\begin{equation}
\label{distance_of_paths}
    |\omega_{u_{i+1}}(t_{i+1})-\omega_{u_i}(t_{i})|_{L}=\mathrm{dist} (\omega_{u_{i+1}},\omega_{u_{i}}), \qquad i \in \{1,\ldots,m-1\}\,.
\end{equation}
In \eqref{distance_of_paths}, we recalled that the paths we are considering are continuous. We deduce from \eqref{tree_triangle_inequality_1}--\eqref{distance_of_paths} that at least one of (1) and (2) must hold. See Remark \ref{open_path_proof_remark} for a further comment.

We can now show \eqref{(ii)_star_proof2} by arguing as in the proof of Proposition \ref{Proposition_5.3} (ii). In particular we modify the proof to address cases (1) and (2) given above. If (1) occurs, we apply Lemma \ref{Lemma_5.20*} when integrating over $\omega_i$. Namely, if $\omega_i$ is a closed loop we apply Lemma \ref{Lemma_5.20*} (i) or (iii) (we apply (iii) if $i$ corresponds to a root), otherwise, we apply Lemma \ref{Lemma_5.20*} (ii) or (iv) (we apply (iv) if $i$ corresponds to a root), as determined by the integration algorithm from Lemma \ref{integration_algorithm}. 

If (2) occurs, we apply Lemma \ref{Lemma_5.21*}  when integrating over $\omega_i$ or $\omega_j$, whichever appears earlier as determined by the algorithm from Lemma \ref{integration_algorithm}. Let us assume without loss of generality that this is label $i$. Note that by construction $i$ does not correspond to a root. In particular, we apply Lemma \ref{Lemma_5.21*} (i) if $\omega_i$ is a closed path and Lemma \ref{Lemma_5.21*} (ii) if $\omega_i$ is an open path. In this case, we also recall \eqref{v^{L}_c_limit}--\eqref{v^{L}_c_limit_*} to note that the estimates we get from Lemma \ref{Lemma_5.21*} tend to zero as $L \rightarrow \infty$. We hence deduce \eqref{(ii)_star_proof1}. All of the estimates are uniform in $\nu \in (0,\nu_0]$ by construction, that is by arguing as in the proof of Proposition \ref{Proposition_5.3} (ii).

By arguing as for \eqref{(ii)_star_proof1}, we get that for all $k \in \{1,\ldots,n\}$,
\begin{equation}
\label{(ii)_star_proof3}
\lim_{L \rightarrow \infty} \,\biggl\|\int \widehat \mu^{\infty}_{y_1,x_1}(\mathrm{d} \omega_1) \cdots \widehat \mu^{\infty}_{y_p,x_p}(\mathrm{d} \omega_p)\,\mu^{\infty}(\mathrm{d} \omega_{p+1}) \cdots \mu^{\infty}(\mathrm{d} \omega_n)\,\vec{1}_{\mathscr{B}^{L}}(\omega_k)\,\varphi^{\infty}(\omega_1,\ldots,\omega_n)\biggr\|_{L_0,p}=0\,.
\end{equation}
By \eqref{A^{L}_B^{L}}, \eqref{(ii)_star_proof1}, and \eqref{(ii)_star_proof3}, we get \eqref{(ii)_star} provided that we show
\begin{multline}
\label{(ii)_star_proof4}
\lim_{L \rightarrow \infty} \,\biggl\|\int \widehat \mu^{L}_{y_1,x_1}(\mathrm{d} \omega_1) \cdots \widehat \mu^{L}_{y_p,x_p}(\mathrm{d} \omega_p)\,\mu^{L}(\mathrm{d} \omega_{p+1}) \cdots \mu^{L}(\mathrm{d} \omega_n)\,\prod_{k=1}^{n}\vec{1}_{\mathscr{A}^{L}}(\omega_k)\,\varphi^{L}(\omega_1,\ldots,\omega_n)
\\
-\int \widehat \mu^{\infty}_{y_1,x_1}(\mathrm{d} \omega_1) \cdots \widehat \mu^{\infty}_{y_p,x_p}(\mathrm{d} \omega_p)\,\mu^{\infty}(\mathrm{d} \omega_{p+1}) \cdots \mu^{\infty}(\mathrm{d} \omega_n)\,\prod_{k=1}^{n}\vec{1}_{\mathscr{A}^{L}}(\omega_k)\,\varphi^{\infty}(\omega_1,\ldots,\omega_n)
\biggr\|_{L^{\infty}_{\vec{x}}L^1_{\vec{y}}} \! = 0\,.
\end{multline}
Let us now give the proof of \eqref{(ii)_star_proof4}. Recalling \eqref{self_interaction}, we first introduce the following modifications of the measures \eqref{measure_mu_1}--\eqref{measure_mu_2}: 
\begin{align}
\label{measure_mu_1_star}
\mu^{L,*}(\mathrm{d} \omega) &\coloneqq  \nu \sum_{T\in \nu {\N}^*} \frac{z^T}{T} \,  \mathbb{W}^{L,T}(\mathrm{d} \omega)
\,\ee^{-\mathcal{V}^{\nu,\infty}(\omega)}\,,
\\
\label{measure_mu_2_star}
\widehat{\mu}^{L,*}_{y, x}(\mathrm{d} \omega) &\coloneqq  \sum_{T \in \nu \N^*} z^T \, \mathbb{W}^{L,T}_{y,x}(\mathrm{d} \omega)\,\ee^{-\mathcal{V}^{\nu,\infty}(\omega)}\,.
\end{align}
In \eqref{measure_mu_1_star}--\eqref{measure_mu_2_star}, $\mathcal{V}^{\nu,\infty}(\omega)$ denotes \eqref{self_interaction} with interaction potential $v^{\infty} \equiv v$.
Moreover, we recall the definition of the measures $\mu^{L,0}$ and $\widehat{\mu}^{L,0}_{y, x}$ from \eqref{measure_mu_1_bound}--\eqref{measure_mu_2_bound} above. These are modifications of \eqref{measure_mu_1}--\eqref{measure_mu_2}, where we omit the self-interaction term.

Note that from \eqref{measure_mu_1}--\eqref{measure_mu_2},  \eqref{measure_mu_1_bound}--\eqref{measure_mu_2_bound}, and \eqref{measure_mu_1_star}--\eqref{measure_mu_2_star}, we that have by the nonnegativity of $\mathcal{V}^{\nu,L}(\omega)$ and $\mathcal{V}^{\nu,\infty}(\omega)$,
\begin{equation}
\label{mu_bounds}
    \mu^L(\mathrm{d} \omega),\, \mu^{L,*}(\mathrm{d} \omega)\leq \mu^{L,0}(\mathrm{d} \omega),\qquad \widehat{\mu}_{y,x}^{L}(\mathrm{d} \omega),\, \widehat{\mu}_{y,x}^{L,*}(\mathrm{d} \omega)\leq \widehat{\mu}_{y,x}^{L,0}(\mathrm{d} \omega).
\end{equation}
In addition, we note that the indicator $\vec{1}_{\mathcal{A}^L}$ erases the boundary effects on $\Lambda_L$. That is, for all integrable functions $F \colon \bigcup_{T>0}\Omega^{L,T}\to \R$ , we have
\begin{equation}\label{mu_1_boundary_effects}
    \int \mu^{L,*}(\mathrm{d} \omega)F(\omega)\vec{1}_{\mathcal{A}^L}(\omega)=\int \mu^{\infty}(\mathrm{d} \omega)F(\omega)\vec{1}_{\mathcal{A}^L}(\omega)\, ,
\end{equation}
and 
\begin{equation}\label{mu_2_boundary_effects}
    \int \widehat{\mu}_{y,x}^{L,*}(\mathrm{d} \omega)F(\omega)\vec{1}_{\mathcal{A}^L}(\omega)=\int \widehat{\mu}_{y,x}^{\infty}(\mathrm{d} \omega)F(\omega)\vec{1}_{\mathcal{A}^L}(\omega)\, .
\end{equation}
We can therefore deduce \eqref{(ii)_star_proof4} provided that we show 
\begin{multline}
\label{(ii)_star_proof5}
\lim_{L \rightarrow \infty} \,\biggl\|\int \widehat \mu^{L}_{y_1,x_1}(\mathrm{d} \omega_1) \cdots \widehat \mu^{L}_{y_p,x_p}(\mathrm{d} \omega_p)\,\mu^{L}(\mathrm{d} \omega_{p+1}) \cdots \mu^{L}(\mathrm{d} \omega_n)\,\prod_{k=1}^{n}\vec{1}_{\mathscr{A}^{L}}(\omega_k)\,\varphi^{L}(\omega_1,\ldots,\omega_n)
\\
-\!\int \!\widehat \mu^{L,*}_{y_1,x_1}(\mathrm{d} \omega_1) \cdots \widehat \mu^{L,*}_{y_p,x_p}(\mathrm{d} \omega_p)\,\mu^{L,*}(\mathrm{d} \omega_{p+1}) \cdots \mu^{L,*}(\mathrm{d} \omega_n)
\prod_{k=1}^{n}\vec{1}_{\mathscr{A}^{L}}(\omega_k)\,\varphi^{L}(\omega_1,\ldots,\omega_n)
\biggr\|_{L^{\infty}_{\vec{x}}L^1_{\vec{y}}} \! = 0\,,
\end{multline}
and 
\begin{multline}
\label{(ii)_star_proof6}
\lim_{L \rightarrow \infty} \,\biggl\|\int \widehat \mu^{L,*}_{y_1,x_1}(\mathrm{d} \omega_1) \cdots \widehat \mu^{L,*}_{y_p,x_p}(\mathrm{d} \omega_p)\,\mu^{L,*}(\mathrm{d} \omega_{p+1}) \cdots \mu^{L,*}(\mathrm{d} \omega_n)\,\prod_{k=1}^{n}\vec{1}_{\mathscr{A}^{L}}(\omega_k)\,\varphi^{L}(\omega_1,\ldots,\omega_n)
\\
-\int \widehat \mu^{\infty}_{y_1,x_1}(\mathrm{d} \omega_1) \cdots \widehat \mu^{\infty}_{y_p,x_p}(\mathrm{d} \omega_p)\,\mu^{\infty}(\mathrm{d} \omega_{p+1}) \cdots \mu^{\infty}(\mathrm{d} \omega_n)\,\prod_{k=1}^{n}\vec{1}_{\mathscr{A}^{L}}(\omega_k)\,\varphi^{\infty}(\omega_1,\ldots,\omega_n)
\biggr\|_{L^{\infty}_{\vec{x}}L^1_{\vec{y}}} \! = 0\,.
\end{multline}
Let us first prove \eqref{(ii)_star_proof5}. By recalling \eqref{measure_mu_1}--\eqref{measure_mu_2}, \eqref{measure_mu_1_star}--\eqref{measure_mu_2_star}, as well as the bounds \eqref{mu_bounds}, and using a telescoping argument combined with the mean-value theorem, it suffices to show that for all $\ell \in \{2,\ldots,n\}$,
\begin{multline}
\label{(ii)_star_proof7}
    \lim_{L \to \infty}\Biggr\|\int \widehat \mu^{L,0}_{y_1,x_1}(\mathrm{d} \omega_1) \cdots \widehat \mu^{L,0}_{y_p,x_p}(\mathrm{d} \omega_p)\,\mu^{L,0}(\mathrm{d} \omega_{p+1}) \cdots \mu^{L,0}(\mathrm{d} \omega_n)\,\prod_{k=1}^{n}\vec{1}_{\mathscr{A}^{L}}(\omega_k)\,\varphi^{L}(\omega_1,\ldots,\omega_n)\\ \times
    |\mathcal{V}^{\nu,L}(\omega_\ell)-\mathcal{V}^{\nu,\infty}(\omega_\ell)|\Biggl\|_{L^{\infty}_{\vec{x}}L^1_{\vec{y}}}=0\,.
\end{multline}
We note that telescoping was applied above in order to replace the measure 
\begin{equation*}
\widehat \mu^{L}_{y_1,x_1}(\mathrm{d} \omega_1) \cdots \widehat \mu^{L}_{y_p,x_p}(\mathrm{d} \omega_p)\,\mu^{L}(\mathrm{d} \omega_{p+1}) \cdots \mu^{L}(\mathrm{d} \omega_n)
\end{equation*} 
with 
\begin{equation*}
\widehat \mu^{L,*}_{y_1,x_1}(\mathrm{d} \omega_1) \cdots \widehat \mu^{L,*}_{y_p,x_p}(\mathrm{d} \omega_p)\,\mu^{L,*}(\mathrm{d} \omega_{p+1}) \cdots \mu^{L,*}(\mathrm{d} \omega_n)
\end{equation*}
one factor at a time.

Using Lemma \ref{tree_bound_estimate}, we can prove \eqref{(ii)_star_proof7} if we show that for every $\mathcal{T}\in \mathfrak{T}_n$,
\begin{multline}
\label{(ii)_star_proof8}
    \lim_{L \to \infty}\Biggr\|\int \widehat \mu^{L,0}_{y_1,x_1}(\dd \omega_1) \cdots \widehat \mu^{L,0}_{y_p,x_p}(\dd \omega_p)\,\mu^{L,0}(\dd \omega_{p+1}) \cdots \mu^{L,0}(\dd \omega_n)\,\prod_{k=1}^{n}\vec{1}_{\mathscr{A}^{L}}(\omega_k)\,\prod_{\{i,j\} \in \mathcal{T}}|\xi^L(\omega_i,\omega_j)|\\ \times
    |\mathcal{V}^{\nu,L}(\omega_\ell)-\mathcal{V}^{\nu,\infty}(\omega_\ell)|\Biggl\|_{L^{\infty}_{\vec{x}}L^1_{\vec{y}}}=0\,.
\end{multline}
Furthermore, since $|\xi^L(\omega_i,\omega_j)|\leq 1$, \eqref{(ii)_star_proof8} follows if we show
\begin{multline}
\label{(ii)_star_proof9}
    \lim_{L \to \infty}\Biggr\|\int \widehat \mu^{L,0}_{y_1,x_1}(\dd \omega_1) \cdots \widehat \mu^{L,0}_{y_p,x_p}(\dd \omega_p)\,\mu^{L,0}(\dd \omega_{p+1}) \cdots \mu^{L,0}(\dd \omega_n)\,\prod_{k=1}^{n}\vec{1}_{\mathscr{A}^{L}}(\omega_k)\,\prod_{\substack{\{i,j\}\in \mathcal{T}\\i,j \neq\ell}}|\xi^L(\omega_i,\omega_j)|\\ \times
    |\mathcal{V}^{\nu,L}(\omega_\ell)-\mathcal{V}^{\nu,\infty}(\omega_\ell)|\Biggl\|_{L^{\infty}_{\vec{x}}L^1_{\vec{y}}}=0\,.
\end{multline}
To prove \eqref{(ii)_star_proof9}, we first integrate in $\omega_\ell$ by means of the following estimates.
\begin{equation}
\begin{cases}
\label{(ii)_star_proof10}
       \int_{\Lambda_L}\dd y\int \widehat \mu^{L,0}_{y,x}(\dd \omega_\ell)|\mathcal{V}^{\nu,L}(\omega_\ell)-\mathcal{V}^{\nu,\infty}(\omega_\ell)|\,\vec{1}_{\mathscr{A}^{L}}(\omega_{\ell})\lesssim \frac{1}{\kappa^2}\|v-v^L\|_{L^{\infty}(\Lambda_{L})}, \text{ if } 1\leq \ell \leq p\,.\\
        \\\int  \mu^{L,0}(\dd \omega_\ell)|\mathcal{V}^{\nu,L}(\omega_\ell)-\mathcal{V}^{\nu,\infty}(\omega_\ell)|\,\vec{1}_{\mathscr{A}^{L}}(\omega_{\ell}) \lesssim_d \frac{1}{\kappa}\|v-v^L\|_{L^1(\Lambda_{L})}, \text{ if } p+1\leq \ell \leq n\,.
    \end{cases}
\end{equation}
In \eqref{(ii)_star_proof10}, we obtained the bounds by arguing as in Lemma \ref{Lemma_5.16*} (i)--(ii). Here, we recall \eqref{self_interaction} and \eqref{L_1_choice} to see that the argument of $v-v^L$ is always taken in $\Lambda_L$. Let now 
\begin{equation*}
\mathcal{T}_1\sqcup\cdots\sqcup\mathcal{T}_m 
\end{equation*}
be the forest obtained by removing from $\mathcal{T}$ the vertex $\ell$ and its adjacent edges. Then, \eqref{(ii)_star_proof9} follows by using the integration algorithm from Lemma \ref{integration_algorithm} on each connected component, $\mathcal{T}_i$, followed by an application of Lemma \ref{v^L_lemma} (iv) and (v).

Let us now prove \eqref{(ii)_star_proof6}. By \eqref{mu_1_boundary_effects}--\eqref{mu_2_boundary_effects}, it is enough to show 
\begin{multline}
\label{(ii)_star_proof11}
    \lim_{L \to \infty}\Biggl \|\int \widehat \mu^{\infty}_{y_1,x_1}(\dd \omega_1) \cdots \widehat \mu^{\infty}_{y_p,x_p}(\dd \omega_p)\,\mu^{\infty}(\dd \omega_{p+1}) \cdots \mu^{\infty}(\dd \omega_n)\,\prod_{k=1}^{n}\vec{1}_{\mathscr{A}^{L}}(\omega_k)\, \\\times\left(\varphi^{L}(\omega_1,\ldots,\omega_n)-\varphi^{\infty}(\omega_1,\ldots,\omega_n)\right)\Biggr\|_{L^{\infty}_{\vec{x}}L^1_{\vec{y}}}=0\,. 
\end{multline}
Recalling \eqref{Ursell_function_definition}, \eqref{(ii)_star_proof11} follows if we show
\begin{multline}
\label{(ii)_star_proof12}
    \lim_{L \to \infty}\Biggl \|\int \widehat \mu^{\infty}_{y_1,x_1}(\dd \omega_1) \cdots \widehat \mu^{\infty}_{y_p,x_p}(\dd \omega_p)\,\mu^{\infty}(\dd \omega_{p+1}) \cdots \mu^{\infty}(\dd \omega_n)\,\prod_{k=1}^{n}\vec{1}_{\mathscr{A}^{L}}(\omega_k)\, \\\times\left(\prod_{\{i,j\}\in \mathcal{G}}\xi^{L}(\omega_i,\omega_j)-\prod_{\{i,j\}\in \mathcal{G}}\xi^{\infty}(\omega_i,\omega_j)\right)\Biggr\|_{L^{\infty}_{\vec{x}}L^1_{\vec{y}}}=0\,,
\end{multline}
for all $\mathcal{G} \in \mathfrak{G}_n^c$. Let $\prec$ be a total order on the edges in $\mathcal{G}$. We write an edge $e$ as $\{i_e,j_e\}$. Using a telescoping argument, we have 
\begin{multline}
\label{xi_telescoping_1}
    \left|\prod_{\{i,j\}\in \mathcal{G}}\xi^{L}(\omega_i,\omega_j)-\prod_{\{i,j\}\in \mathcal{G}}\xi^{\infty}(\omega_i,\omega_j)\right|\\
    \leq \sum_{e \in \mathcal{G}}\prod_{\substack{\{i,j\} \in \mathcal{G}\\\{i,j\} \prec e}}|\xi^{L}(\omega_i,\omega_j)|\,|\xi^{L}(\omega_{i_e},\omega_{j_e})-\xi^{\infty}(\omega_{i_e},\omega_{j_e})|\prod_{\substack{\{i,j\} \in \mathcal{G}\\ e \prec \{i,j\}}}|\xi^{\infty}(\omega_i,\omega_j)|\,.
\end{multline}
For a given $e \in \mathcal{G}$, let $\mathcal{T}_e$ be an arbitrary spanning tree of $\mathcal{G}$ that contains $e$. Since $|\xi^{L}(\omega,\tilde{\omega})|\leq 1$ for all $L\in (0,\infty]$, we can bound the right-hand side of \eqref{xi_telescoping_1} from above by
\begin{equation}
\label{xi_telescoping_1_B}
    \sum_{e \in \mathcal{G}}\prod_{\substack{\{i,j\} \in \mathcal{T}_e\\\{i,j\} \prec e}}|\xi^{L}(\omega_i,\omega_j)|\,|\xi^{L}(\omega_{i_e},\omega_{j_e})-\xi^{\infty}(\omega_{i_e},\omega_{j_e})|\prod_{\substack{\{i,j\} \in \mathcal{T}_e\\ e \prec \{i,j\}}}|\xi^{\infty}(\omega_i,\omega_j)|\,.
\end{equation}
Furthermore, since $v\ge 0$, we have by the mean-value theorem that 
\begin{multline}
\label{xi_telescoping_2}
    |\xi^{L}(\omega,\tilde{\omega})-\xi^{\infty}(\omega,\tilde{\omega})|\leq |\mathcal{V}^{\nu,L}(\omega,\tilde{\omega})-\mathcal{V}^{\nu,\infty}(\omega,\tilde{\omega})|\\
    \leq \frac{1}{\nu} \sum_{s \in \nu \N} \vec{1}_{s<T(\omega)} \sum_{\tilde{s} \in \nu \N} \vec{1}_{\tilde{s}<T(\tilde{\omega})} \int_{0}^{\nu}\dd t\,| v^L-v|\bigl(\omega(s+t)-\tilde{\omega}(\tilde{s}+t)\bigr) \,.
\end{multline}
Recall $L_1$ in \eqref{L_1_choice}. As noted earlier, we obtain by this choice of $L_1$ and the triangle inequality that $\omega(\cdot)-\tilde{\omega}(\cdot)$ lies within $\Lambda_L$ whenever $\omega(\cdot),\, \tilde{\omega}(\cdot) \in \mathcal{A}^{L}$. Hence, using \eqref{xi_telescoping_2} and the definition \eqref{A^{L}_B^{L}} we have
\begin{multline}
\label{xi_telescoping_2_restricted}
   |\xi^{L}(\omega,\tilde{\omega})-\xi^{\infty}(\omega,\tilde{\omega})|\,\vec{1}_{\mathcal{A}^L}(\omega)\,\vec{1}_{\mathcal{A}^L}(\tilde{\omega})\leq |\mathcal{V}^{\nu,L}(\omega,\tilde{\omega})-\mathcal{V}^{\nu,\infty}(\omega,\tilde{\omega})|\\
    \leq \frac{1}{\nu} \sum_{s \in \nu \N} \vec{1}_{s<T(\omega)} \sum_{\tilde{s} \in \nu \N} \vec{1}_{\tilde{s}<T(\tilde{\omega})} \int_{0}^{\nu}\dd t\, |v^L-v|\,\vec{1}_{\Lambda_L} \left(\omega(s+t)-\tilde{\omega}(\tilde{s}+t)\right).
\end{multline}
This justifies replacing any factors of $|\xi^{L}(\omega_{i}, \omega_{j})|$ in Lemma \ref{integration_algorithm} with the left-hand side of \eqref{xi_telescoping_2_restricted}. Indeed, this is akin to taking $(v^L-v)\,\vec{1}_{\Lambda_L}$ instead of $v$ in \eqref{loop_interaction} and \eqref{Ursell_function_definition}. Similarly, for $\dag = L$ or $\dag = \infty$ we have
\begin{multline}
\label{xi_telescoping_2_B}
   |\xi^{\dag}(\omega,\tilde{\omega})| \vec{1}_{\mathcal{A}^L}(\omega)\vec{1}_{\mathcal{A}^L}(\tilde{\omega}) \leq
   \left(|\xi^{L}(\omega,\tilde{\omega})|+|\xi^{\infty}(\omega,\tilde{\omega})|\right)\vec{1}_{\mathcal{A}^L}(\omega)\vec{1}_{\mathcal{A}^L}(\tilde{\omega})\leq \\
    \leq \frac{1}{\nu} \sum_{s \in \nu \N} \vec{1}_{s<T(\omega)} \sum_{\tilde{s} \in \nu \N} \vec{1}_{\tilde{s}<T(\tilde{\omega})} \int_{0}^{\nu}\mathrm{d} t\,\left[\left|v^L\,\vec{1}_{\Lambda_L}\right|+\bigl|v\,\vec{1}_{\Lambda_L}\bigr|\right] \bigl(\omega(s+t)-\tilde{\omega}(\tilde{s}+t)\bigr)\,,  
\end{multline}
which justifies replacing any factors of $|\xi^{L}(\omega_{i}, \omega_{j})|$ in Lemma \ref{integration_algorithm} with $|\xi^{\dag}(\omega_{i},\omega_{j})|\vec{1}_{\mathcal{A}^L}(\omega_{i})\vec{1}_{\mathcal{A}^L}(\omega_{j})$.
By \eqref{xi_telescoping_1}--\eqref{xi_telescoping_2_B}, combined with the integration algorithm
from Lemma \ref{integration_algorithm} and Lemma \ref{Lemma_5.16*} (iv) where in both we set $L=\infty$ (which one obtains by the same arguments as for finite $L$), it follows that
\begin{multline}
\label{5.107}
    \Biggl \|\int \widehat \mu^{\infty}_{y_{1},x_{1}}(\mathrm{d} {\omega}_{1}) \cdots \widehat \mu^{\infty}_{y_p,x_p}(\mathrm{d} \omega_p)\,\mu^{\infty}(\mathrm{d} \omega_{p+1}) \cdots \mu^{\infty}(\mathrm{d} \omega_n)\,\prod_{k=1}^{n}\vec{1}_{\mathscr{A}^{L}}(\omega_k)\, \\\times\left(\prod_{\{i,j\}\in \mathcal{G}}\xi^{L}(\omega_i,\omega_j)-\prod_{\{i,j\}\in \mathcal{G}}\xi^{\infty}(\omega_i,\omega_j)\right)\Biggr\|_{L^{\infty}_{\vec{x}}L^1_{\vec{y}}}
\\
\lesssim_{d,\zeta,n}(\|v^L\|_{L^1(\Lambda_L)}+\|v\|_{L^1(\Lambda_L)}+\|v^L\|_{L^{\infty}(\Lambda_L)}+\|v\|_{L^{\infty}(\Lambda_L)})^{n-1}\, 
\\
\times
(\|v^L-v\|_{L^1(\Lambda_L)}+\|v^L-v\|_{L^{\infty}(\Lambda_L)})\,.
\end{multline}
The claim \eqref{(ii)_star} now follows from \eqref{5.107} by using Lemma \ref{v^L_lemma} (ii)--(v).

In order to conclude the proof of Theorem \ref{Infinite_volume_Theorem_2}, we note the following.
\begin{itemize}
\item[(I)] We have 
\begin{equation}
\label{convergence_item_1}
\lim_{L \rightarrow \infty} \Gamma_{p}^{\nu,\zeta,L} = \Gamma_{p}^{\nu,\zeta,\infty}
\end{equation}
exists with respect to \eqref{L_0_seminorm} uniformly in $\nu \in (0,\nu_0]$. Here we recall \eqref{nu_0} from Assumption \ref{Choice_of_z}. Let us note that \eqref{convergence_item_1} follows from \eqref{(ii)_star} shown above. 
\item[(II)] For all $L \geq 1$, we have
\begin{equation}
\label{convergence_item_2}
\lim_{\nu \rightarrow 0} \Gamma_{p}^{\nu,\zeta,L} = \Gamma_{p}^{0,\zeta,L}\,,
\end{equation}
with respect to \eqref{L_0_seminorm}. Here, we recall \eqref{density_matrix_result_2}--\eqref{density_matrix_result_3}.
We note that \eqref{convergence_item_2} follows from Theorem \ref{density_matrix_result}. Here, we use Lemma \ref{v^L_lemma} (i) to note that $v^L$ satisfies Assumption \ref{Assumption_on_v} with $B=0$ and $\alpha=1$. In particular, we deduce that the convergence in \eqref{convergence_item_2} holds with respect to \eqref{L_0_seminorm} by using the rate of convergence given by \eqref{Theorem_1.5(ii)_rate_of_convergence} above.
By \eqref{density_matrix_result_2}--\eqref{density_matrix_result_3} and \eqref{L_0_seminorm}, we have that 
\begin{equation}
\label{convergence_item_2*}
\|\Gamma_{p}^{0,\zeta,L}\|_{L_0,p}=0\,.
\end{equation}
In particular, from \eqref{convergence_item_2}--\eqref{convergence_item_2*}, we deduce that 
\begin{equation}
\label{convergence_item_2**}
\lim_{\nu \rightarrow 0} \Gamma_{p}^{\nu,\zeta,L} =0\,,
\end{equation}
with respect to \eqref{L_0_seminorm}.

\end{itemize}

Under conditions (I)--(II) above, we are able to exchange limits and deduce the following claims; see also \cite[Lemma 5.13]{FKSS_2023}.
\begin{itemize}
\item[(i)] The limit 
\begin{equation}
\label{convergence_item_3}
\lim_{L \rightarrow \infty} \Gamma_{p}^{0,\zeta,L} \eqqcolon \Gamma_{p}^{0,\zeta,\infty}
\end{equation}
exists with respect to \eqref{L_0_seminorm}.
\item[(ii)] For $\Gamma_{p}^{0,\zeta,\infty}$ given by \eqref{convergence_item_3}, we have
\begin{equation}
\label{convergence_item_4}
\lim_{\nu \rightarrow 0} \Gamma_{p}^{\nu,\zeta,\infty}=0\,, 
\end{equation}
with respect to \eqref{L_0_seminorm}.
\end{itemize}
We note that \eqref{convergence_item_3}--\eqref{convergence_item_4} imply \eqref{Infinite_volume_Theorem_2_claim}.

\end{proof}

\begin{remark}
\label{non_uniformity_in_L_0}
We note that \eqref{convergence_item_2} is obtained by \eqref{Theorem_1.5(ii)_rate_of_convergence}, which does not give us uniform convergence in $L_0$.
\end{remark}

\begin{remark}
\label{open_path_proof_remark}
If $k \leq p$ in \eqref{(ii)_star_proof2}, we can argue as for \eqref{a-x_1} to obtain that $|a - x_{k}|_{L} \geq \tfrac{L}{24}$ for some point $a$ on $\omega_k$. In particular, we can automatically deduce that case (1) in the proof of \eqref{(ii)_star_proof2} occurs with $i=k$ and $\omega_i$ being an open path. Hence, in this case, the argument above simplifies. For $k \in \{p+1,\ldots,n\}$ we have to consider cases (1) and (2) in the proof of \eqref{(ii)_star_proof2} as explained above.
\end{remark}

\subsubsection{Analysis of the free energy in the infinite volume. Proof of Theorem \ref{free_energy_infinite_volume}}
\label{Proof_of_Theorem_free_energy_infinite_volume}

In this section, we prove our main result on the free energy in the infinite-volume limit, given by Theorem \ref{free_energy_infinite_volume} above. The idea of the proof is similar to that of Theorem \ref{Infinite_volume_Theorem_2}. The challenge now is that all Brownian paths in the Ginibre representation \eqref{Ginibre_representation_1} are closed, meaning we cannot fix the endpoints of the open paths as before. This requires some new ingredients in the analysis. We will outline the main differences in the proof given below.

\begin{proof}[Proof of Theorem \ref{free_energy_infinite_volume}]
Let us note that claim (ii) follows from claim (i). Namely, we argue as in \eqref{convergence_item_1}--\eqref{convergence_item_2} in the proof of Theorem \ref{Infinite_volume_Theorem_2} above. Instead of  \eqref{convergence_item_1}, we use \eqref{free_energy_infinite_volume_i}. Instead of \eqref{convergence_item_2}, we recall \eqref{specific_Gibbs_potential}, \eqref{classical_specific_Gibbs_potential}, and use Theorem \ref{partition_function_result} to obtain that
\begin{equation}
\label{convergence_item_2_*}
\lim_{\nu \rightarrow 0} g^{\nu,\zeta,L} = g^{\mathrm{cl},\zeta,L}\,.
\end{equation}
For \eqref{convergence_item_2_*}, we again use Lemma \ref{v^L_lemma} (i) to note that $v^L$ satisfies Assumption \ref{Assumption_on_v} with $B=0$ and $\alpha=1$. See also \cite[Lemma 5.13]{FKSS_2023} and set $\rho(L,\nu) \equiv g^{\nu,\zeta,L}$ with notation as in \cite{FKSS_2023}.

We now prove claim (i).
Let us first use Proposition \ref{gamma_p_cluster_expansion} (i) and recall \eqref{measure_mu_1} and \eqref{Ursell_function_representation}, as well as \eqref{W^T_definition} to write 
\begin{multline}
\label{free_energy_infinite_volume_1}
g^{\nu,\zeta,L}=\frac{1}{|\Lambda_L|} \log \Xi^{\nu,\zeta,L}=
\frac{1}{|\Lambda_L|} \int_{\Lambda_L} \dd x_1\,\sum_{n=1}^{\infty} \nu \sum_{T_1 \in \nu \N^*} \frac{z^{T_1}}{T_1}\, \int \mathbb{W}^{L,T_1}_{x_1,x_1}(\dd \omega_1)\, 
\\
\times \mathrm{e}^{-\mathcal{V}^{\nu,L}(\omega_1)}\, \mu^L(\dd \omega_2) \cdots \mu^L(\dd \omega_n)\, \varphi^L(\omega_1,\ldots,\omega_n)\,.
\end{multline}
We analyse \eqref{free_energy_infinite_volume_1} using several steps.
\paragraph{\textbf{\emph{Step 1: Reduce to integrating in $x_1$ away from the boundary of $\Lambda_L$}}}
We first show that removing a relatively small strip at the boundary of $\Lambda_L$ from the integral over $x_1$ has a negligible effect. Combined with Step 2, this will allow us to manipulate the boundary conditions of Brownian motion. 

Given $\delta>0$ small, we write the average in $x_1$ in \eqref{free_energy_infinite_volume_1} as
\begin{equation}
\label{x_1_average}
\frac{1}{|\Lambda_L|} \int_{\Lambda_L} \dd x_1=
\frac{1}{|\Lambda_L|} \int_{(1-\delta)\Lambda_L} \dd x_1+\frac{1}{|\Lambda_L|} \int_{\Lambda_L \setminus (1-\delta)\Lambda_L} \dd x_1\,.
\end{equation}
In \eqref{x_1_average} and onwards, we write $(1-\delta)\Lambda_L \equiv \Lambda_{(1-\delta)L,d}$.
By arguing as in the proof of Lemma \ref{Lemma_5.16*} (iii), we obtain that 
\begin{multline}
\label{free_energy_infinite_volume_2}
\int_{\Lambda_L \setminus (1-\delta) \Lambda_L} \nu \sum_{T \in \nu \N^*} \frac{z^T}{T} \int \mathbb{W}^{L,T}_{x,x}(\dd 
\omega) \, T^q \lesssim_{d} \nu^q\, \frac{(q-1)!}{\kappa^q}\ \left| \Lambda_L \setminus (1-\delta) \Lambda_L \right|
\\
= \nu^q\, \frac{(q-1)!}{\kappa^q}\ |\Lambda_L|\,\left(1-(1-\delta)^d\right),
\end{multline}
for all $q \in \N$. Here, we also recall \eqref{kappa_definition}. More precisely, we get \eqref{free_energy_infinite_volume_2} from \eqref{Triangle_1}--\eqref{Triangle_2} above by replacing the spatial domain $\Lambda_L$ with $\Lambda_L \setminus (1-\delta) \Lambda_L$.
We now apply the integration algorithm from Lemma \ref{integration_algorithm} to the remaining integral over $\Lambda_L \setminus (1-\delta) \Lambda_L$ in \eqref{free_energy_infinite_volume_1}. We do this as in the proof of Proposition \ref{Proposition_5.3} (i) to integrate out the root of the tree, but using \eqref{free_energy_infinite_volume_2} instead of Lemma \ref{Lemma_5.16*}  (iii). 
By \eqref{free_energy_infinite_volume_1}--\eqref{free_energy_infinite_volume_2}, we can then write
\begin{equation}
\label{free_energy_infinite_volume_3}
g^{\nu,\zeta,L}=\mathrm{I}^{\nu,\zeta,L}+ \mathcal{O}_{d,\zeta,v} (\delta)\,,
\end{equation}
where
\begin{multline}
\label{free_energy_infinite_volume_4}
\mathrm{I}^{\nu,\zeta,L} \coloneqq \
\frac{1}{|\Lambda_L|} \int_{(1-\delta)\Lambda_L} \dd x_1\,\sum_{n=1}^{\infty} \nu \sum_{T_1 \in \nu \N^*} \frac{z^{T_1}}{T_1}\, \int \mathbb{W}^{L,T_1}_{x_1,x_1}(\dd \omega_1)\, 
\\
\times \mathrm{e}^{-\mathcal{V}^{\nu,L}(\omega_1)}\, \mu^L(\dd \omega_2) \cdots \mu^L(\dd \omega_n)\, \varphi^L(\omega_1,\ldots,\omega_n)\,. 
\end{multline}

\paragraph{\textbf{\emph{Step 2: Reduce to integrating over loops $\omega_j$ which are close to $x_1$}}}
To complement Step 1, we now show that ignoring Brownian bridges that are relatively close to the boundary of $\Lambda_L$ has a negligible effect. This will allow us to manipulate the boundary conditions of Brownian motion. 

Given $n \in \N^*$, let us define the set 
\begin{multline}
\label{free_energy_infinite_volume_5}
\mathcal{R}_{n}^{L,\delta} \coloneqq \Bigl\{(\omega_1, \ldots,\omega_n) \in (\Omega^L)^n,\, \,\, |\omega_i(t_i)-\omega_j(t_j)|_L =|\omega_i(t_i)-\omega_j(t_j)| \leq L \delta/2\,
\\
\forall 1 \leq i,j \leq n\,,\,\, \forall t_i \in [0,T(\omega_i)]\,,\,\, \forall t _j \in [0,T(\omega_j)]\Bigr\}\,.
\end{multline}
In \eqref{free_energy_infinite_volume_5} we used the observation that for $\delta>0$ small, we have $|w|_L \leq L\delta/2$ if and only if $|w| \leq L\delta/2$. See Figure \ref{Loops_example}.

\begin{figure}[!ht]
\begin{center}
\includegraphics[scale=0.5]{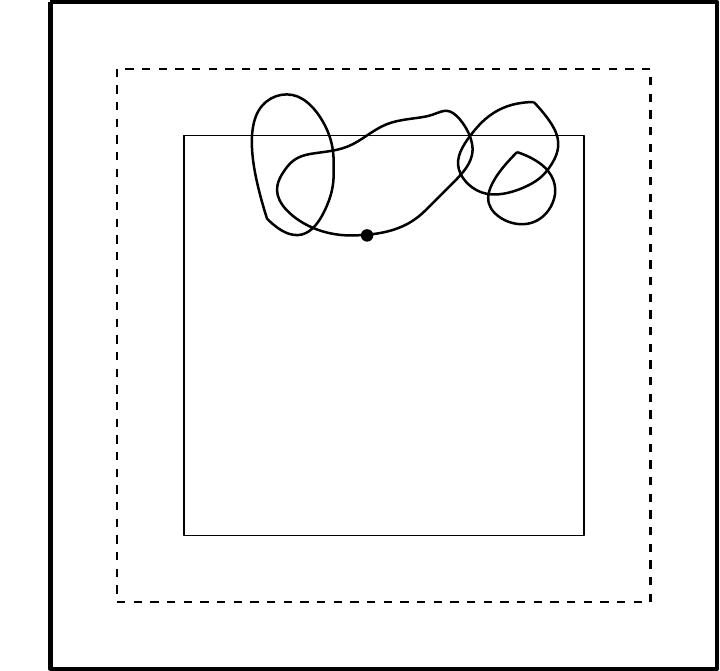}
\end{center}
\caption{Here, we take $d=2$ and $n=4$. The concentric rectangles are $(1-\delta) \Lambda_L$, $(1-\delta/2) \Lambda_L$, and $\Lambda_L$. The marked point is $x_1$, which lies on the loop $\omega_1$. The points on all the loops $\omega_1,\omega_2,\omega_3,\omega_4$ are at a distance of at most $L\delta/2$ from $x_1$.
\label{Loops_example}
}
\end{figure}

Using \eqref{free_energy_infinite_volume_4}--\eqref{free_energy_infinite_volume_5}, and arguing as in \eqref{(ii)_star_proof1}--\eqref{(ii)_star_proof2}, we have that
\begin{equation}
\label{free_energy_infinite_volume_6}
\mathrm{I}^{\nu,\zeta,L}=\mathrm{II}^{\nu,\zeta,L}+ C_1(\nu,L,\delta)\,, \qquad |C_1(\nu,L,\delta)| \lesssim_{d,\zeta,v} o_{L}^{(\delta)}(1)\,,
\end{equation}
where
\begin{multline}
\label{free_energy_infinite_volume_7}
\mathrm{II}^{\nu,\zeta,L} \coloneqq \
\frac{1}{|\Lambda_L|} \int_{(1-\delta)\Lambda_L} \dd x_1\,\sum_{n=1}^{\infty}  \nu \sum_{T_1 \in \nu \N^*} \frac{z^{T_1}}{T_1}\, \int \mathbb{W}^{L,T_1}_{x_1,x_1}(\dd \omega_1)\, 
\\
\times \mathrm{e}^{-\mathcal{V}^{\nu,L}(\omega_1)}\, \mu^L(\dd \omega_2) \cdots \mu^L(\dd \omega_n)\, \varphi^L(\omega_1,\ldots,\omega_n)\,\mathbf{1}_{\mathcal{R}_{n}^{L,\delta}}(\omega_1,\ldots,\omega_n)\,.
\end{multline}
In \eqref{free_energy_infinite_volume_6} and onwards, $o_{L}^{(\delta)}(1)$ denotes a quantity as in Definition \ref{epsilon_L} where the convergence to zero as $L \rightarrow \infty$ depends on $\delta$ (which is now fixed).
More precisely, for fixed $n \geq 1$, we bound $|\varphi^L(\omega_1,\ldots,\omega_n)|$ using \eqref{tree_bound} from Lemma \ref{tree_bound_estimate}. If 
\begin{equation*}
(\omega_1, \ldots,\omega_n) \in (\Omega^L)^n \setminus \mathcal{R}_{n}^{L,\delta}\,
\end{equation*} 
then one of the following cases must occur for a fixed tree $\mathcal{T}$ on the right-hand side of \eqref{tree_bound}.
\begin{itemize}
\item[(1)] There exists $i \in \{1,\ldots,n\}$ such that $\omega_i \in \mathscr{D}^{L}_{c}$.
\item[(2)] There exists $\{i,j\} \in \mathcal {T}$ such that $\xi^{L}(\omega_i,\omega_j)=\xi^{L}_{c}(\omega_i,\omega_j)$.
\end{itemize}
Here, we recall \eqref{D^L_c_definition} and \eqref{xi^L_c} above.
We can hence use Lemma \ref{Lemma_5.20*} (i) and (iii) as well as Lemma \ref{Lemma_5.21*} (i) in the appropriate steps of the integration algorithm from Lemma \ref{integration_algorithm}, and argue as for \eqref{(ii)_star_proof1}--\eqref{(ii)_star_proof2} in the proof of Theorem \ref{Infinite_volume_Theorem_2} (where now $p=0$ and there are no open paths). 
We therefore deduce \eqref{free_energy_infinite_volume_6}--\eqref{free_energy_infinite_volume_7} as claimed.

\paragraph{\textbf{\emph{Step 3: Replace the interaction potential $v^L$ by $v$}}}
By arguing as in \eqref{(ii)_star_proof5}--\eqref{(ii)_star_proof6}, and by recalling \eqref{(ii)_star_proof10} and \eqref{5.107}, we have that
\begin{equation}
\label{free_energy_infinite_volume_8}
\mathrm{II}^{\nu,\zeta,L}=\mathrm{III}^{\nu,\zeta,L}+\mathcal{O}_{d,\zeta,v}\left(\|v^L-v\|_{L^1(\Lambda_L)}\right),
\end{equation}
where
\begin{multline}
\label{free_energy_infinite_volume_9}
\mathrm{III}^{\nu,\zeta,L} \coloneqq \
\frac{1}{|\Lambda_L|} \int_{(1-\delta)\Lambda_L} \dd x_1\,\sum_{n=1}^{\infty}  \nu \sum_{T_1 \in \nu \N^*} \frac{z^{T_1}}{T_1}\, \int \mathbb{W}^{L,T_1}_{x_1,x_1}(\dd \omega_1)\, 
\\
\times \mathrm{e}^{-\mathcal{V}^{\nu,\infty}(\omega_1)}\, \mu^{L,*}(\dd \omega_2) \cdots \mu^{L,*}(\dd \omega_n)\, \varphi^{\infty}(\omega_1,\ldots,\omega_n)\, \mathbf{1}_{\mathcal{R}_{n}^{L,\delta}}(\omega_1,\ldots,\omega_n)\,.
\end{multline}
Here, we also recall \eqref{measure_mu_1_star}.
For \eqref{free_energy_infinite_volume_8}--\eqref{free_energy_infinite_volume_9}, we used (as in \eqref{(ii)_star_proof5}--\eqref{(ii)_star_proof6} in the proof of Theorem \ref{Infinite_volume_Theorem_2}) the observation that 
\begin{equation*}
\omega_i(\cdot)-\omega_j(\cdot) \in \Lambda_L \qquad \forall 1 \leq i, j \leq n\,,
\end{equation*} 
whenever $(\omega_1,\ldots,\omega_n) \in \mathcal{R}_{n}^{L,\delta}$, as defined in \eqref{free_energy_infinite_volume_5}. 
Since all the paths $\omega_j$ are closed, we are only using the second estimate in \eqref{(ii)_star_proof10} above when repeating the argument from \eqref{(ii)_star_proof5}--\eqref{(ii)_star_proof6}. Therefore, the error term in \eqref{free_energy_infinite_volume_8} is given only in terms of $\|v^L-v\|_{L^1(\Lambda_L)}$.

Using \eqref{free_energy_infinite_volume_5} and the assumption that $\omega_1(0)=x_1 \in (1-\delta) \Lambda_L$, the continuity of the $\omega_i$, and the assumption that $\delta$ is chosen to be small, we obtain that
\begin{equation}
\label{free_energy_infinite_volume_10}
\omega_i \colon [0,T(\omega_i)] \rightarrow (1-\delta/2) \Lambda_L \quad \forall i \in \N^*
\end{equation}
in the integral \eqref{free_energy_infinite_volume_9}. In other words, all loops $\omega_i$ lie within $(1-\delta/2) \Lambda_L$. Recalling \eqref{measure_mu_1_star}, we can use \eqref{free_energy_infinite_volume_10} to erase the boundary effects in \eqref{free_energy_infinite_volume_9} and write
\begin{multline}
\label{free_energy_infinite_volume_11}
\mathrm{III}^{\nu,\zeta,L} 
=
\frac{1}{|\Lambda_L|} \sum_{n=1}^{\infty}  \int_{(1-\delta)\Lambda_L} \dd x_1\, \int_{\Lambda_L} \dd x_2\, \cdots \int_{\Lambda_L} \dd x_n
\\
\prod_{j=1}^{n} \left[\nu \sum_{T_j \in \nu \N^*} \frac{z^{T_j}}{T_j}\, \int \mathbb{W}^{\infty,T_j}_{x_j,x_j}(\dd \omega_j)\, 
\mathrm{e}^{-\mathcal{V}^{\nu,\infty}(\omega_j)} \right]
\,\varphi^{\infty}(\omega_1,\ldots,\omega_n)\,
\mathbf{1}_{\mathcal{R}_{n}^{L,\delta}}(\omega_1,\ldots,\omega_n) \,.
\end{multline}
\paragraph{\textbf{\emph{Step 4: Relax the assumption that the loops $\omega_j$ are close to $x_1$}}}
We now reintroduce the paths in $(\Omega^L)^{n} \setminus \mathcal{R}_{n}^{L,\delta}$ that we removed in Step 2, to again obtain integrals over all Brownian bridges. 
Let us note that
\begin{multline}
\label{free_energy_infinite_volume_12_C}
\frac{1}{|\Lambda_L|} \sum_{n=1}^{\infty} \int_{(1-\delta)\Lambda_L} \dd x_1\, \int_{\Lambda_L} \dd x_2\, \cdots \int_{\Lambda_L} \dd x_n
\\
\prod_{j=1}^{n} \left[\nu \sum_{T_j \in \nu \N^*} \frac{z^{T_j}}{T_j}\, \int \mathbb{W}^{\infty,T_j}_{x_j,x_j}(\dd \omega_j)\, 
\mathrm{e}^{-\mathcal{V}^{\nu,\infty}(\omega_j)} \right]
\,\left|\varphi^{\infty}(\omega_1,\ldots,\omega_n)\right|\,
\left(1-\mathbf{1}_{\mathcal{R}_{n}^{L,\delta}}(\omega_1,\ldots,\omega_n)\right)
\\
\lesssim_{d,\zeta,v} o_{L}^{(\delta)}(1)\,.
\end{multline}
In \eqref{free_energy_infinite_volume_12_C}, we recall the notation from \eqref{free_energy_infinite_volume_6} above.
To show \eqref{free_energy_infinite_volume_12_C}, we first fix $n \geq 1$ and apply the tree bound \eqref{tree_bound} to $\left|\varphi^{\infty}(\omega_1,\ldots,\omega_n)\right|$. When estimating the contribution from the trees, we set $\omega_1$ as the root.
In particular, when integrating with respect to $\omega_1$ in the integration algorithm from Lemma \ref{integration_algorithm}, we use one of the following estimates.
\begin{equation}
\label{omega_1_root_1}
\int_{(1-\delta)\Lambda_L} \dd x_1\,\nu \sum_{T_1 \in \nu \N^*} \frac{z^{T_1}}{T_1}\, \int \mathbb{W}^{\infty,T_1}_{x_1,x_1}(\dd \omega_1)\, 
\mathrm{e}^{-\mathcal{V}^{\nu,\infty}(\omega_1)}\,T_1^q \lesssim_{\zeta,q,d} |(1-\delta)\Lambda_L|\,.
\end{equation}
\begin{equation}
\label{omega_1_root_2}
\frac{1}{|(1-\delta)\Lambda_L|}\,\int_{(1-\delta)\Lambda_L} \dd x_1\,\nu \sum_{T_1 \in \nu \N^*} \frac{z^{T_1}}{T_1}\, \int \mathbb{W}^{\infty,T_1}_{x_1,x_1}(\dd \omega_1)\, 
\mathrm{e}^{-\mathcal{V}^{\nu,\infty}(\omega_1)}\,\mathbf{1}_{\mathscr{D}^{\infty,L}_{c\delta}}(\omega_1)\,T_1^q  \lesssim_{\zeta,q,d} o^{(\delta)}_{L}\left(1\right).
\end{equation}
Here, we recall \eqref{D^{L,tilde_L}_c}.
The estimate \eqref{omega_1_root_1} follows by arguing as for Lemma \ref{Lemma_5.16*} (iii). The estimate \eqref{omega_1_root_2} follows by arguing as for Lemma \ref{Lemma_5.20*} (iii), with the modifications given in 
Remark \ref{long_far_away_path_comment} with $\tilde{L}=L$; see \eqref{Remark_5.14_iii} above\footnote{As in \eqref{Remark_5.14_iii}, it is important to note that we are integrating over the finite volume $(1-\delta)\Lambda_L$ in $x_1$.}. 

Now, by \eqref{free_energy_infinite_volume_5}, if $(\omega_1,\ldots,\omega_n) \notin \mathcal{R}_{n}^{L,\delta}$ there exist $1 \leq i,j \leq n$ and $t_i \in [0,T(\omega_i)], \, t _j \in [0,T(\omega_j)]$ such that 
\begin{equation}
\label{L_delta_2_distance_paths}
|\omega_i(t_i)-\omega_j(t_j)| > L \delta/2\,.
\end{equation}
Using \eqref{L_delta_2_distance_paths}, followed by \eqref{omega_1_root_1}--\eqref{omega_1_root_2} and the modifications of Lemmas \ref{Lemma_5.16*}, \ref{Lemma_5.20*}, and \ref{Lemma_5.21*} from Remark \ref{long_far_away_path_comment} with $\tilde{L}=L$, the arguments for \eqref{free_energy_infinite_volume_6}--\eqref{free_energy_infinite_volume_7} imply \eqref{free_energy_infinite_volume_12_C}.

Using \eqref{free_energy_infinite_volume_12_C}, it follows that
\begin{equation}
\label{free_energy_infinite_volume_12}
\mathrm{III}^{\nu,\zeta,L}=\mathrm{IV}^{\nu,\zeta,L}+ C_2(\nu,L,\delta)\,, \qquad |C_2(\nu,L,\delta)| \lesssim_{d,\zeta,v} o_{L}^{(\delta)}(1)\,,
\end{equation}
where
\begin{multline}
\label{free_energy_infinite_volume_12_B}
\mathrm{IV}^{\nu,\zeta,L} \coloneqq \frac{1}{|\Lambda_L|} \sum_{n=1}^{\infty}  \int_{(1-\delta)\Lambda_L} \dd x_1\, \int_{\Lambda_L} \dd x_2\, \cdots \int_{\Lambda_L} \dd x_n
\\
\prod_{j=1}^{n} \left[\nu \sum_{T_j \in \nu \N^*} \frac{z^{T_j}}{T_j}\, \int \mathbb{W}^{\infty,T_j}_{x_j,x_j}(\dd \omega_j)\, 
\mathrm{e}^{-\mathcal{V}^{\nu,\infty}(\omega_j)} \right]
\,\varphi^{\infty}(\omega_1,\ldots,\omega_n)\,.
\end{multline}

\paragraph{\textbf{\emph{Step 5: Integrate in space over all of $\R^d$}}}
Let us now define 
\begin{multline}
\label{free_energy_infinite_volume_12_D}
\mathrm{V}^{\nu,\zeta,L} \coloneqq \frac{1}{|\Lambda_L|} \sum_{n=1}^{\infty}  \int_{(1-\delta)\Lambda_L} \dd x_1\, \int_{\R^d} \dd x_2\, \cdots \int_{\R^d} \dd x_n
\\
\prod_{j=1}^{n} \left[\nu \sum_{T_j \in \nu \N^*} \frac{z^{T_j}}{T_j}\, \int \mathbb{W}^{\infty,T_j}_{x_j,x_j}(\dd \omega_j)\, 
\mathrm{e}^{-\mathcal{V}^{\nu,\infty}(\omega_j)} \right]
\,\varphi^{\infty}(\omega_1,\ldots,\omega_n)\,.
\end{multline}
By \eqref{free_energy_infinite_volume_12_B}--\eqref{free_energy_infinite_volume_12_D}, a telescoping argument, and enlarging the integration domain, we deduce that
\begin{multline}
\label{free_energy_infinite_volume_12_E}
\left| \mathrm{IV}^{\nu,\zeta,L} -\mathrm{V}^{\nu,\zeta,L} \right| \leq  \frac{1}{|\Lambda_L|} \sum_{n=1}^{\infty} (n-1)\, \int_{(1-\delta)\Lambda_L} \dd x_1\, \int_{\R^d \setminus \Lambda_L} \dd x_2\, \int_{\R^d} \dd x_3\, \cdots \int_{\R^d} \dd x_n
\\
\prod_{j=1}^{n} \left[\nu \sum_{T_j \in \nu \N^*} \frac{z^{T_j}}{T_j}\, \int \mathbb{W}^{\infty,T_j}_{x_j,x_j}(\dd \omega_j)\, 
\mathrm{e}^{-\mathcal{V}^{\nu,\infty}(\omega_j)} \right]
\,\left|\varphi^{\infty}(\omega_1,\ldots,\omega_n)\right|\,.
\end{multline}
Note that in the integral on the right-hand side of \eqref{free_energy_infinite_volume_12_E}, we have 
\begin{equation}
\label{free_energy_infinite_volume_12_F}
|\omega_1(0)-\omega_2(0)|=|x_1-x_2| \geq L \delta\,.
\end{equation}
Using \eqref{free_energy_infinite_volume_12_F} in \eqref{free_energy_infinite_volume_12_E}, and arguing as for \eqref{free_energy_infinite_volume_12_C}, we obtain that
\begin{equation}
\label{free_energy_infinite_volume_12_G}
\left| \mathrm{IV}^{\nu,\zeta,L} -\mathrm{V}^{\nu,\zeta,L} \right| \lesssim_{d,\zeta,v} o_{L}^{(\delta)}(1)\,.
\end{equation}
\paragraph{\textbf{\emph{Step 6: Use translation invariance to conclude the proof}}}
We now use translation invariance and $\frac{|(1-\delta)\Lambda_L|}{|\Lambda_L|}=(1-\delta)^d$ to rewrite \eqref{free_energy_infinite_volume_12_D} as
\begin{multline}
\label{free_energy_infinite_volume_12_H}
\mathrm{V}^{\nu,\zeta,L} = (1-\delta)^d\, \sum_{n=1}^{\infty} \int_{\R^d} \dd x_2\, \cdots \int_{\R^d} \dd x_n\,\left[\nu \sum_{T_1 \in \nu \N^*} \frac{z^{T_1}}{T_1}\, \int \mathbb{W}^{\infty,T_1}_{0,0}(\dd \omega_1)\,\mathrm{e}^{-\mathcal{V}^{\nu,\infty}(\omega_1)}\right] 
\\
\prod_{j=2}^{n} \left[\nu \sum_{T_j \in \nu \N^*} \frac{z^{T_j}}{T_j}\, \int \mathbb{W}^{\infty,T_j}_{x_j,x_j}(\dd \omega_j)\, 
\mathrm{e}^{-\mathcal{V}^{\nu,\infty}(\omega_j)} \right]
\,\varphi^{\infty}(\omega_1,\ldots,\omega_n)\,.
\end{multline}
By Proposition \ref{Proposition_5.3} (i), we know that $g^{\nu,\zeta,L}$ is bounded uniformly in $\nu \in (0,\nu_0]$. By \eqref{free_energy_infinite_volume_3}, \eqref{free_energy_infinite_volume_6}, \eqref{free_energy_infinite_volume_8}, \eqref{free_energy_infinite_volume_12}, and \eqref{free_energy_infinite_volume_12_G}, it follows that the same is true for \eqref{free_energy_infinite_volume_12_H}.
We now use \eqref{free_energy_infinite_volume_3}, \eqref{free_energy_infinite_volume_6}, \eqref{free_energy_infinite_volume_8}, \eqref{free_energy_infinite_volume_12} (combined with Lemma \ref{v^L_lemma} (iv)), \eqref{free_energy_infinite_volume_12_G}, and \eqref{free_energy_infinite_volume_12_H} 
to deduce that
\begin{multline}
\label{free_energy_infinite_volume_conclusion}
\lim_{L \rightarrow \infty} g^{\nu,\zeta,L}=\sum_{n=1}^{\infty} \int_{\R^d} \dd x_2\, \cdots \int_{\R^d} \dd x_n\,\left[\nu \sum_{T_j \in \nu \N^*} \frac{z^{T_1}}{T_1}\, \int \mathbb{W}^{\infty,T_1}_{0,0}(\dd \omega_1)\mathrm{e}^{-\mathcal{V}^{\nu,\infty}(\omega_1)}\right] 
\\
\prod_{j=2}^{n} \left[\nu \sum_{T_j \in \nu \N^*} \frac{z^{T_j}}{T_j}\, \int \mathbb{W}^{\infty,T_j}_{x_j,x_j}(\dd \omega_j)\, 
\mathrm{e}^{-\mathcal{V}^{\nu,\infty}(\omega_j)} \right]
\,\varphi^{\infty}(\omega_1,\ldots,\omega_n)\,,
\end{multline}
uniformly in $\nu \in (0,\nu_0]$. This concludes the proof of claim (i) and the result follows.
\end{proof}

\begin{remark}[Theorem \ref{Infinite_volume_Theorem_2} for the Fermi and Boltzmann statistics]
\label{fermions_infinite_volume}
The methods used to analyse the infinite-volume limits in the case of bosons extend to fermions and the Boltzmann statistics as well. We do not pursue this in detail, but we give a very brief justification.

Recall Remark \ref{Remark_Boltzmann_gas_partition_function}. To get the Ginibre representation of the partition function for the Boltzmann statistics, one redefines the measures \eqref{measure_mu_1}--\eqref{measure_mu_2}, \eqref{measure_mu_1_bound}--\eqref{measure_mu_2_bound}, and \eqref{measure_mu_1_star}--\eqref{measure_mu_2_star} by only keeping the $k=1$ term in their defining sums. One then uses the same arguments as for the Bose gas.

To derive the Ginibre representation of the Fermi gas, one adds appropriate permutation signs to \eqref{measure_mu_1}--\eqref{measure_mu_2}, \eqref{measure_mu_1_bound}--\eqref{measure_mu_2_bound}, and \eqref{measure_mu_1_star}--\eqref{measure_mu_2_star} as in \eqref{single_loop_measure_Fermi}, and to \eqref{gamma_p_cluster_expansion_claim} as in \eqref{Ginibre_representation_4_Fermi}. Note however that in \eqref{gamma_p_cluster_expansion_claim}, the permutation signs that correspond to $\prod_{i=1}^p (-1)^{k_i-1}$ in \eqref{gamma_p_cluster_expansion_claim} are already contained in the measures \eqref{measure_mu_2}.

The cluster expansion of \cite{Ueltschi_2004} applies to complex measures with finite total variation, so Proposition \ref{gamma_p_cluster_expansion} still holds. Most of the estimates in Section \ref{section_infinite_volume} are done on nonnegative integrands, for example those containing $|\varphi^{L}(\omega_{1}, \ldots, \omega_{n})|$. In these cases, we get the desired estimates for fermions by applying the triangle inequality, and then using the existing estimates for bosons. One exception are equations where a telescoping argument is applied to signed sums, for example \eqref{(ii)_star_proof5}. When dealing with these cases for fermions, we apply the telescoping argument before the triangle inequality.

In short, at no point do we use cancellations between individual summands in \eqref{measure_mu_1}--\eqref{measure_mu_2}, \eqref{measure_mu_1_bound}--\eqref{measure_mu_2_bound}, and \eqref{measure_mu_1_star}--\eqref{measure_mu_2_star}. Therefore, the arguments for bosons extend to fermions.
\end{remark}

\appendix

\section{Heat kernel estimates. Proofs of Lemmas \ref{heat_kernel_estimates} and \ref{Riemann_sums_estimate}.}
\label{appendix_heat_kernel_estimates}
In this appendix, we prove Lemmas \ref{heat_kernel_estimates} and \ref{Riemann_sums_estimate}, which give us quantitative estimates on the heat kernel on $\Lambda_L$ given by \eqref{heat_kernel} above.

\begin{proof}[Proof of Lemma \ref{heat_kernel_estimates}]
We give a self-contained proof of (i) and (ii) for completeness.
The proof of (iii) relies on the estimate (i) and is given in the proof of \cite[Lemma 5.4 (ii)]{FKSS_2022}. We note that in \cite{FKSS_2022}, one considers $L=1$, but the estimate for general $L$ given in \eqref{FKSS_2022_estimate} follows analogously, once we know the estimate \eqref{heat_kernel_estimates_1} from part (i).

We prove (i) inductively on the dimension $d$. The nonnegativity of $\psi^t(x)$ follows immediately from \eqref{heat_kernel}, so we just need to show the upper bound in \eqref{heat_kernel_estimates_1}. For $d = 1$, we expand \eqref{heat_kernel} to obtain
\begin{equation} \label{heat-psi-d=1}
		\psi^{t}(x) = \frac{1}{(2\pi t)^{1/2}}\,\ee^{-\frac{x^{2}}{2t}}+ \sum_{n \in \Z \setminus \{0\}} \frac{1}{(2 \pi t)^{1/2}} \,\ee^{-\frac{(x-Ln)^{2}}{2t}}\,.
	\end{equation}
	Since $\ee^{-\frac{x^{2}}{2t}} \leq 1$, we only need to analyse the second term in  \eqref{heat-psi-d=1}. By interpreting \eqref{heat-psi-d=1} as a lower Riemann sum of a translated Gaussian function with mesh size $L t^{-1/2}$, we get
	\begin{equation}
	 \label{heat-psi-d=1_B}
		\sum_{n \in \Z \setminus \{0\}} \frac{1}{(2\pi t)^{1/2}} \,\ee^{-\frac{(x-Ln)^{2}}{2t}} \lesssim \frac{1}{L} \int_{-\infty}^{\infty} \ee^{-\frac{(x-y)^{2}}{2}}\, \dd y \lesssim \frac{1}{L}\,.
	\end{equation}
	Substituting \eqref{heat-psi-d=1_B}
 into \eqref{heat-psi-d=1} yields \eqref{heat_kernel_estimates_1} when $d=1$.

For the induction step, let us suppose that \eqref{heat_kernel_estimates_1} holds up to dimension $d - 1$ with $d \geq 2$. By \eqref{heat_kernel}, we can factorise $\psi^{t}(x)$ as
	\begin{equation}
	 \label{heat-psi_1}
		\psi^{t}(x) = \left( \sum_{n \in \Z} \frac{1}{(2 \pi t)^{1/2}} \,\ee^{-\frac{(x_1-n)^2}{2t}} \right) \left( \sum_{n' \in \Z^{d-1}} \frac{1}{(2 \pi t)^{(d-1)/2}}\, \ee^{-\frac{|x'-n'|^2}{2t}} \right).
\end{equation}
where $x'=(x_2,\ldots,x_n) \in \R^{n-1}$.
Using the induction hypothesis in \eqref{heat-psi_1}, we obtain \eqref{heat_kernel_estimates_1}.

We now prove (ii). By~\eqref{heat_kernel_estimates_tilde_A} and~\eqref{heat_kernel} it suffices to show the upper bound in \eqref{heat_kernel_estimates_tilde}. Using \eqref{heat_kernel_estimates_tilde_A} again and noting that the number of $n \in \Z^d$ with $|n| \in [k, k+1)$ is $\mathcal{O}_d(k^{d-1})$, it follows that 
\begin{equation}
\label{heat_kernel_estimates_tilde_B}
\sum_{n \in \Z^d \setminus \{0\}} \frac{1}{(2\pi t)^{d/2}} \,\ee^{-\frac{|Ln|^2}{2t}} \lesssim_d \sum_{k=1}^{\infty}  \frac{1}{(2\pi t)^{d/2}} \,\ee^{-\frac{L^2 k^2}{2t}} \,k^{d-1}\sim_d \sum_{k=1}^{\infty}   \ee^{-\frac{L^2 k^2}{2t}} \,\biggl(\frac{k}{\sqrt{t}}\biggr)^{d-1}\,\frac{1}{\sqrt{t}}\,.
\end{equation}
By considering Riemann sums of mesh size $t^{-1/2}$, the expression in \eqref{heat_kernel_estimates_tilde_B} is
\begin{equation}
\label{heat_kernel_estimates_tilde_B_integral}
\int_0^{+\infty} \dd y\, \ee^{-\frac{L^2 y^2}{2}}\, y^{d-1} \lesssim_d \frac{1}{L^d}\,.
\end{equation}
In \eqref{heat_kernel_estimates_tilde_B_integral}, we used the change of variables $y'=Ly$ to deduce the upper bound.
We hence obtain (ii).
\end{proof}

\begin{proof}[Proof of Lemma \ref{Riemann_sums_estimate}]
Recalling \eqref{heat_kernel} and the definition of $n_0$ given by \eqref{n_0_choice}, we write
\begin{equation}
\label{Riemann_sum_estimate_2}
\psi^t(x)= \frac{1}{(2\pi t)^{d/2}}\,\ee^{-\frac{|x-Ln_0|^2}{2t}}+\sum_{n \in \Z^d\backslash\{n_0\}}\frac{1}{(2\pi t)^{d/2}}\,\,\ee^{-\frac{|x-Ln|^2}{2t}}\,.
\end{equation}
We obtain \eqref{Riemann_sums_estimate_1B} and hence \eqref{Riemann_sums_estimate_1} from \eqref{Riemann_sum_estimate_2} provided we show 
\begin{equation}
\label{Riemann_sum_estimate_3}
\sum_{n \in \Z^d\backslash\{n_0\}}\frac{1}{(2\pi t)^{d/2}}\,\,\ee^{-\frac{|x-Ln|^2}{2t}} =
\sum_{n \in \Z^d\backslash\{0\}}\frac{1}{(2\pi t)^{d/2}}\,\,\ee^{-\frac{|x-Ln_{0}-Ln|^2}{2t}}
\lesssim_d \frac{1}{L^d}\,.
\end{equation}
We get the second sum from the first by introducing a new variable $n - n_0$. To show \eqref{Riemann_sum_estimate_3}, we use that for all $n \in \Z^d\backslash\{0\}$,
\begin{equation}
\label{Riemann_sum_estimate_4}
|x-Ln_0-Ln| \geq \frac{L|n|}{2\sqrt{d}}\,.
\end{equation}
We first prove \eqref{Riemann_sum_estimate_4}. By \eqref{n_0_choice}, we have
\begin{equation}
\label{Riemann_sum_estimate_4A}
|x-Ln_0| = |x|_L \leq \max_{y\in\Lambda_L} |y|_L = \frac{L\sqrt{d}}{2} \,.
\end{equation}
We also have that for all $n \in \Z^{d} \setminus \{0\}$,
\begin{equation}
\label{Riemann_sum_estimate_4B}
|x-Ln_0-Ln| \geq \frac{L}{2} \,.
\end{equation}
Indeed, suppose that $|x-Ln_0-Ln| < \frac{L}{2}$. Since $n \in \Z^d \setminus \{0\}$, we then have that
\begin{equation*}
|x-Ln_{0}| \geq L|n| - |x-Ln_{0}-Ln| > L-\frac{L}{2} = \frac{L}{2} \,.
\end{equation*}
This contradicts the definition of $n_{0}$, namely that $|x|_L = |x-Ln_0| = \min_{k \in \Z^{d}} |x-Lk|$ by \eqref{n_0_choice} and \eqref{periodic_Euclidean_norm}. Hence, \eqref{Riemann_sum_estimate_4B} holds as claimed.

Let now $n \in \Z^{d} \setminus \{0\}$ be arbitrary. If $|n| \leq \sqrt{d}$, then by \eqref{Riemann_sum_estimate_4B},
\begin{equation}
\label{Riemann_sum_estimate_4C}
|x-Ln_{0}-Ln| \geq \frac{L}{2} \geq \frac{L|n|}{2\sqrt{d}} \,.
\end{equation}
If on the other hand $|n| > \sqrt{d}$ we have by \eqref{Riemann_sum_estimate_4A} that
\begin{equation}
\label{Riemann_sum_estimate_4D}
|x-Ln_{0}-Ln| \geq L|n| - |x-Ln_0| \geq L|n| - \frac{L\sqrt{d}}{2} > \frac{L|n|}{2} \geq \frac{L|n|}{2\sqrt{d}} \,.
\end{equation}
We deduce \eqref{Riemann_sum_estimate_4} from \eqref{Riemann_sum_estimate_4C}--\eqref{Riemann_sum_estimate_4D}.

We now show \eqref{Riemann_sum_estimate_3}. By recalling  \eqref{heat_kernel_estimates_tilde_A}, and using \eqref{Riemann_sum_estimate_4} followed by Lemma~\ref{heat_kernel_estimates}~(ii), we have
\begin{equation}
\sum_{n \in \Z^d\backslash\{0\}}\frac{1}{(2\pi t)^{d/2}}\,\,\ee^{-\frac{|x-Ln_{0}-Ln|^2}{2t}} \leq 
\sum_{n \in \Z^d\backslash\{0\}}\frac{1}{(2\pi t)^{d/2}}\,\,\ee^{-\frac{|Ln|^2}{2t \cdot 4d}} = (4d)^{d/2} \eta^{4d t} \lesssim_{d} \frac{1}{L^d} \,.
\end{equation}
This concludes the proof Lemma \ref{Riemann_sums_estimate}.
\end{proof}

\section{Further remarks on the Ginibre representation}
\label{Ginibre_representation_appendix} 
In this appendix, we explain how to obtain Lemma \ref{Ginibre_representation} from the arguments used to prove \cite[Proposition 2.3]{FKSS_2023}. The proof of the latter is based on \cite{Ginibre} and is given in \cite[Appendix A.2]{FKSS_2023}. The conventions used in \cite{FKSS_2023} for the interaction potential, the $N$-body Hamiltonian, and the self-interactions are slightly different from the ones we use in the current work, so they need to be reconciled. Through this section, the superscript `$\mathrm{old}$' denotes objects defined in \cite{FKSS_2023}.

\begin{proof}[Summary of the argument]
We briefly summarise how to obtain the claim in Lemma \ref{Ginibre_representation} (i) from the arguments in \cite[Appendix A.2]{FKSS_2023}. The claim in (ii) follows by analogous arguments and we omit the details.
Let $N \in \N^*$ be given. On the $N$-particle space defined as in \eqref{n_particle_space}, the convention for the 
$N$-body Hamiltonian in \cite{FKSS_2023} is given by
\begin{equation}
\label{n_body_Hamiltonian_old}
H_N^{\nu,\mathrm{old}}\coloneqq -\frac{\nu}{2} \sum_{i=1}^{N} \Delta_i +\frac{1}{2} \sum_{i,j=1}^{N} v(x_i-x_j)\,;
\end{equation}
see \cite[(1.18)]{FKSS_2023}. In particular, from \eqref{n_body_Hamiltonian} and \eqref{n_body_Hamiltonian_old}, we obtain that 
\begin{equation}
\label{n_body_Hamiltonian_difference}
H_N^{\nu}=H_N^{\nu,\mathrm{old}}-\frac{N}{2} v(0)\,.
\end{equation}
In \eqref{n_body_Hamiltonian_difference}, we also used that $v$ is even by Assumption \ref{Assumption_on_v}.
The grand-canonical partition function considered in \cite{FKSS_2023} is different from the one in \eqref{grand_canonical_ensemble} and is given by
\begin{equation}
\label{grand_canonical_partition_old}
\Xi^{\nu,\zeta,\mathrm{old}}\coloneqq 1 + \sum_{N=1}^{\infty} \mathrm{Tr}_{\mathcal{H}_N} \left(\ee^{-H_{N}^{\nu,\mathrm{old}}}\,z^{N \nu} \right).
\end{equation}
In \cite{FKSS_2023}, the interaction $\mathcal{V}^{\nu,\infty}(\omega,\tilde{\omega})$ between two loops $\omega \in \Omega^{T(\omega)}, \tilde{\omega} \in \Omega^{T(\tilde{\omega})}$ with $T(\omega), T(\tilde{\omega}) \in \nu \N^*$ is given as in \eqref{loop_interaction} above. The convention for the self-interactions is different from that given by \eqref{self_interaction}. In particular, 
given $\omega \in \Omega^{T(\omega)}$ with $T(\omega) \in \nu \N$, its self-interaction in \cite{FKSS_2023} is given by $\mathcal{V}^{\nu,\infty}(\omega,\omega)$. From \eqref{loop_interaction} and \eqref{self_interaction}, it follows that for $\omega \in \Omega^{T(\omega)}$ with $T(\omega) \in \nu \N^*$, we have
\begin{equation}
\label{self_interaction_difference}
\frac{1}{2} \mathcal{V}^{\nu,\infty}(\omega,\omega)=\mathcal{V}^{\nu,\infty}(\omega)+\frac{T(\omega)}{2\nu} v(0)\,.
\end{equation}
Given $n \in \N^*$ and $\vec{\omega}=(\omega_1,\ldots,\omega_n) \in \Omega^{\vec{T}}$ such that $T_i=T(\omega_i) \in \nu \N^*$ for $i = 1, \ldots, n$, instead of \eqref{loop_interaction_2}, in \cite[(1.4)]{FKSS_2023} one considers
\begin{equation}
\label{loop_interaction_2_old}
\mathcal{V}^{\nu,\mathrm{old}}(\vec{\omega})\coloneqq \frac{1}{2} \sum_{i,j=1}^{n} \mathcal{V}^{\nu,\infty}(\omega_i,\omega_j)\,.
\end{equation}
The result of \cite[Proposition 2.3]{FKSS_2023} lets us rewrite \eqref{grand_canonical_partition_old} as 
\begin{equation}
\label{Ginibre_representation_1_old}
\Xi^{\nu,\zeta,\mathrm{old}}= 1 + \sum_{n=1}^{\infty} \frac{1}{n!} \int \mathbb{L}^{\nu,\zeta}(\dd \omega_1) \cdots\, \mathbb{L}^{\nu,\zeta}(\dd \omega_n) \,\exp\bigl(-\mathcal{V}^{\nu,\mathrm{old}}(\vec{\omega})\bigr)\,,
\end{equation}
where $\mathbb{L}^{\nu,\zeta}(\dd \omega)$ is given by \eqref{single_loop_measure}.
From \eqref{loop_interaction_2}, \eqref{self_interaction_difference}, and \eqref{loop_interaction_2_old}, we deduce that
\begin{equation}
\label{loop_interaction_difference}
\mathcal{V}^{\nu,\mathrm{old}}(\vec{\omega})=\mathcal{V}^{\nu,\infty}(\vec{\omega})+\frac{1}{2} \sum_{i=1}^{n} \frac{T(\omega_i)}{\nu} v(0)\,.
\end{equation}
Recalling \eqref{grand_canonical_ensemble} and \eqref{n_body_Hamiltonian_difference}, we can write 
\begin{equation}
\label{grand_canonical_partition_neq}
\Xi^{\nu,\zeta}=1 + \sum_{N=1}^{\infty} \mathrm{Tr}_{\mathcal{H}_N} \left(\ee^{-H_{N}^{\nu,\mathrm{old}}}\,\ee^{\frac{N}{2} v(0)}\,z^{N \nu} \right).
\end{equation}
Using \eqref{grand_canonical_partition_old}, \eqref{Ginibre_representation_1_old}, and \eqref{grand_canonical_partition_neq} with
\begin{equation}
\label{N_choice}
N=\frac{T(\omega_1)}{\nu}+\cdots+\frac{T(\omega_n)}{\nu}\,,
\end{equation}
and applying the arguments in \cite[Appendix A.2]{FKSS_2023}, one can rewrite \eqref{grand_canonical_partition_neq} as
\begin{equation}
\label{Ginibre_representation_1_new}
\Xi^{\nu,\zeta}= 1 + \sum_{n=1}^{\infty} \frac{1}{n!} \int \mathbb{L}^{\nu,\zeta}(\dd \omega_1) \cdots\, \mathbb{L}^{\nu,\zeta}(\dd \omega_n) \,\exp\left(-\mathcal{V}^{\nu,\mathrm{old}}(\vec{\omega})\right)\,\exp \left[\left(\frac{T(\omega_1)}{2\nu} + \cdots + \frac{T(\omega_n)}{2\nu}\right) v(0)\right].
\end{equation}
Note that in \eqref{Ginibre_representation_1_new}, in the term coming from the trace over the $N$-particle space, we are adding a factor of $\ee^{\frac{N}{2} v(0)}$, and we are using \eqref{N_choice}.
Substituting \eqref{loop_interaction_difference} into \eqref{Ginibre_representation_1_new}, we conclude that
\begin{equation}
\label{Ginibre_representation_1_new_B}
\Xi^{\nu,\zeta}= 1 + \sum_{n=1}^{\infty} \frac{1}{n!} \int \mathbb{L}^{\nu,\zeta}(\dd \omega_1) \cdots\, \mathbb{L}^{\nu,\zeta}(\dd \omega_n) \,\exp\bigl(-\mathcal{V}^{\nu}(\vec{\omega})\bigr)\,,
\end{equation}
as was claimed. We note that \eqref{Ginibre_representation_1_new_B} is indeed a convergent sum. To see this, we use Lemma \ref{lemma_loop_interactions}, consider $\tilde{z}$ as in \eqref{definition_of_z_tilde}, and argue as in \eqref{Proposition_2_1_2} and \eqref{fin-vol-telescoping-bound} to write
\begin{multline}
\label{Ginibre_representation_1_new_C}
1 + \sum_{n=1}^{\infty} \frac{1}{n!} \int \mathbb{L}^{\nu,\zeta}(\dd \omega_1) \cdots\, \mathbb{L}^{\nu,\zeta}(\dd \omega_n) \,\exp\bigl(-\mathcal{V}^{\nu}(\vec{\omega})\bigr)
\\
\leq 
1+\sum_{n = 1}^{\infty} \frac{1}{n!} \sum_{\vec{k} \in (\N^*)^n} \prod_{i = 1}^{n} \frac{z^{k_{i}\nu}}{k_{i}} \int \mathbb{W}^{\nu \vec{k}} (\dd \vec{\omega})\exp\!\left(B\sum_{i=1}^nk_i\right)
\\ 
\leq 
1+ \sum_{n=1}^{\infty} \frac{1}{n!} \left(\sum_{\ell=1}^{\infty} \frac{\tilde{z}^{\ell \nu}}{\ell} \int_{\Lambda} \dd x \int \mathbb{W}_{x, x}^{\ell \nu}(\dd \omega)\right)^n \leq C(d,\zeta,B,L)\,.
\end{multline}
For the final bound in \eqref{Ginibre_representation_1_new_C}, we recalled \eqref{fin-vol-bound1}--\eqref{fin-vol-bound2}, which in turn rely on \eqref{fugacity_condition}.
\end{proof}

\section{Tuning the particle density}
\label{tuning_particle_density}
We briefly explain how one can tune the particle density of the interacting Bose gas in the context of the large-mass limit. The latter is necessary in the study of the mean-field limit in dimensions $d>1$; see \cite{FKSS_2020_1} for a further discussion. As is seen in Theorem \ref{partition_function_result}, Corollary \ref{relative_partition_function_convergence},  and Theorem \ref{density_matrix_result} above, there is no need to tune the particle density in the large-mass limit. We study this only in the context of the relative partition function in a fixed finite volume, as in Corollary \ref{relative_partition_function_convergence}. 
Our other results can also be rephrased in this context; we omit the details. 

In order to state a precise result, we follow the setup given in \cite[Section 2]{FKSS_2020_2}. 
In particular, we work in the \emph{bosonic Fock space} defined as 
\begin{equation}
\label{Fock_space}
\mathcal{F}\coloneqq \bigoplus_{n \in \N} \mathcal{H}_n\,,
\end{equation}
where we recall \eqref{n_particle_space}. We denote elements of $\mathcal{F}$ as vectors $\Psi=(\Psi^{(n)})_{n \in \mathbb{N}}$. One obtains that \eqref{Fock_space} is a Hilbert space if given $\Psi_1=(\Psi_1^{(n)})_{n \in \mathbb{N}}, \Psi_2=(\Psi_2^{(n)})_{n \in \mathbb{N}} \in \mathcal{F}$, we define their inner product as
\begin{equation*}
 \langle \Psi_1, \Psi_2 \rangle\coloneqq \sum_{n \in \N} \big\langle \Psi_1^{(n)}, \Psi_2^{(n)} \big \rangle_{\mathcal{H}_n}\,. 
\end{equation*}
Recalling \eqref{n_body_Hamiltonian}, we consider the \emph{many-body Hamiltonian} $\mathbb{H}^{\nu}$ acting on $\mathcal{F}$ given by 
\begin{equation}
\label{H_sgv_1}
\mathbb{H}^{\nu}=\bigoplus_{n \in \N} H^{\nu}_n\,.
\end{equation}

Let us now rewrite \eqref{H_sgv_1} in second quantised form. 
To this end, given  $f \in L^2(\Lambda)$, we define the bosonic annihilation and creation operators $a(f)$ and $a^{*}(f)$ as follows.
\begin{align}
\label{a(f)}
(a(f) \Psi)^{(n)}(x_1,\ldots,x_n)&\coloneqq \sqrt{n+1} \int_{\Lambda} \dd x\, \overline{f(x)}\,\Psi^{(n+1)}(x,x_1, \ldots, x_n)\,,
\\
\label{a^*(f)}
(a^*(f) \Psi)^{(n)}(x_1,\ldots,x_n)&\coloneqq \frac{1}{\sqrt{n}} \sum_{j=1}^{n} f(x_j)\,\Psi^{(n-1)}(x_1, \ldots, x_{j-1},x_{j+1}, \ldots,x_n)\,.
\end{align}
The operators \eqref{a(f)}--\eqref{a^*(f)} are unbounded and closed on $\mathcal{F}$. Under suitable choice of domain, they are each other's adjoints. Furthermore, they satisfy the \emph{canonical commutation relations}, i.e.\ for all $f_1, f_2 \in L^2(\Lambda)$, we have
\begin{equation}
\label{Canonical_commutation_relations_1}
[a(f_1),a^{*}(f_2)]=\langle f_1, f_2 \rangle_{L^2}\,, \qquad [a(f_1),a(f_2)]=[a^{*}(f_1),a^{*}(f_2)]=0\,,
\end{equation}
where $[X,Y] \equiv XY -YX$ denotes the commutator.

With $\nu>0$ fixed and $f \in L^2(\Lambda)$, we use \eqref{a(f)}--\eqref{a^*(f)} and henceforth consider the \emph{rescaled} annihilation and creation operators, defined as 
\begin{equation}
\label{rescaled_annihilation_creation}
a_{\nu} (f)\coloneqq \sqrt{\nu} a(f)\,,\qquad a_{\nu}^{*}(f)\coloneqq \sqrt{\nu} a^*(f)\,.
\end{equation}
We can view $a_{\nu}, a_{\nu}^{*}$ defined in \eqref{rescaled_annihilation_creation} as being operator-valued distributions.
In this framework, we can rewrite \eqref{rescaled_annihilation_creation} in terms of their distribution kernels $a_{\nu}(x), a_{\nu}^{*}(x)$ as
\begin{equation}
a_{\nu} (f)=\int_{\Lambda} \dd x\,\overline{f(x)}\, a_{\nu}(x)\,,\qquad a_{\nu}^{*} (f)=\int_{\Lambda} \dd x\,f(x)\, a_{\nu}^{*}(x)\,.
\end{equation}
Furthermore, from \eqref{Canonical_commutation_relations_1} one has that for all $x,y \in \Lambda$,
\begin{equation}
\label{Canonical_commutation_relations}
[a_{\nu}(x),a_{\nu}^*(y)]=\nu \delta(x-y)\,,\qquad [a_{\nu}(x),a_{\nu}(y)]=[a_{\nu}^{*}(x),a_{\nu}^{*}(y)]=0\,.
\end{equation}
An explicit calculation using \eqref{a(f)}--\eqref{a^*(f)} allows us to rewrite \eqref{H_sgv_1} as 
\begin{equation}
\label{H_sgv}
\mathbb{H}^{\nu}=\frac{1}{2} \int_{\Lambda} \dd x\, \nabla a_{\nu}^{*}(x)\, \nabla a_{\nu}(x)+\frac{1}{2 \nu^2} \int_{\Lambda} \dd x\, \int_{\Lambda} \dd y\, a_{\nu}^{*}(x)\, a_{\nu}^{*}(y)\, v(x-y)\, a_{\nu}(x)\,a_{\nu}(y)\,.
\end{equation}
Given parameters $\kappa_0, \rho>0$, we define
\begin{multline}
\label{H^nu_kappa_rho}
\mathbb{H}^{\nu}(\kappa_0,\rho)\coloneqq \frac{1}{2}\int_{\Lambda} \dd x\, \nabla a_{\nu}^{*}(x)\, \nabla a_{\nu}(x)+\kappa_0 \int_{\Lambda} \dd x\, a_{\nu}^{*}(x)\,a_{\nu}(x)
\\
+\frac{1}{2\nu^2} \int_{\Lambda} \dd x\, \int_{\Lambda} \dd y\,\bigl(a_{\nu}^{*}(x)a_{\nu}(x)-\rho\bigr)\,v(x-y)\, \bigl(a_{\nu}^{*}(y)a_{\nu}(y)-\rho\bigr)\,.
\end{multline}
Note that in \eqref{H^nu_kappa_rho}, one is tuning the (pointwise) quantum density $a_{\nu}^{*}(x)a_{\nu}(x)$ by subtracting the parameter $\rho$.
Using \eqref{Canonical_commutation_relations} and the fact that $v$ is even (by Assumption \ref{Assumption_on_v}), we can write 
\begin{multline}
\label{H^nu_kappa_rho_B}
\frac{1}{2\nu^2} \int_{\Lambda} \dd x\, \int_{\Lambda} \dd y\,\bigl(a_{\nu}^{*}(x)a_{\nu}(x)-\rho\bigr)\,v(x-y)\, \bigl(a_{\nu}^{*}(y)a_{\nu}(y)-\rho\bigr)
\\
=\frac{1}{2\nu^2} \int_{\Lambda} \dd x\, \int_{\Lambda} \dd y\,\bigl(a_{\nu}^{*}(x)a_{\nu}(x)-\rho\bigr)\,v(x-y)\, \bigl(a_{\nu}^{*}(y)a_{\nu}(y)-\rho\bigr)
\\
=\frac{1}{2 \nu^2} \int_{\Lambda} \dd x\, \int_{\Lambda} \dd y\, a_{\nu}^{*}(x)\, a_{\nu}^{*}(y)\, v(x-y)\, a_{\nu}(x)\,a_{\nu}(y)
\\
+\biggl[\frac{1}{2\nu}\,v(0)-\frac{\rho}{\nu^2}\,\widehat{v}(0)\biggr]\,\int_{\Lambda}\dd x\, a_{\nu}^{*}(x)\,a_{\nu}(x)
+\frac{\rho^2}{2\nu^2}\,\widehat{v}(0)\,L^d\,.
\end{multline}
In \eqref{H^nu_kappa_rho_B} and onwards, we are using the convention of the Fourier transform given by 
\begin{equation*}
\widehat{u}(k) \equiv \int_{\Lambda} \dd x\, u(x)\, \ee^{-2 \pi \ii x \cdot k/L}\,.
\end{equation*}
Combining \eqref{H_sgv}--\eqref{H^nu_kappa_rho_B}, we can write
\begin{equation}
\label{Hamiltonian_rewritten}
\mathbb{H}^{\nu}(\kappa_0,\rho)=\mathbb{H}^{\nu}+\biggl[\kappa_0
+\frac{1}{2\nu}\,v(0)-\frac{\rho}{\nu^2}\,\widehat{v}(0)\biggr]\,\int_{\Lambda}\dd x\, a_{\nu}^{*}(x)\,a_{\nu}(x)
+\frac{\rho^2}{2\nu^2}\,\widehat{v}(0)\,L^d\,.
\end{equation}
From \eqref{n_body_Hamiltonian}, \eqref{H_sgv_1}, and \eqref{Hamiltonian_rewritten}, we can write
\begin{equation*}
\mathbb{H}^{\nu}(\kappa_0,\rho)=\bigoplus_{n \in \N} \mathbb{H}^{\nu}_n(\kappa_0,\rho)\,,
\end{equation*}
where for $n \in \N$,
\begin{equation}
\label{Hamiltonian_rewritten_2}
\mathbb{H}^{\nu}_n(\kappa_0,\rho)\coloneqq -\frac{\nu}{2} \sum_{i=1}^{n} \Delta_i +\sum_{1 \leq i < j \leq n} v(x_i-x_j)+\biggl[\kappa_0
+\frac{1}{2\nu}\,v(0)-\frac{\rho}{\nu^2}\,\widehat{v}(0)\biggr]\,n \nu+\frac{\rho^2}{2\nu^2}\,\widehat{v}(0)\,L^d\,.
\end{equation}
With notation as in \eqref{Hamiltonian_rewritten_2}, we define 
\begin{equation}
\label{Xi_shifted_particle_density}
\Xi^{\nu,\zeta}(\kappa_0,\rho)\coloneqq 1+\sum_{n=1}^{\infty} \mathrm{Tr}_{\mathcal{H}_n} \ee^{-\mathbb{H}^{\nu}_{n}(\kappa_0,\rho)}\,.
\end{equation}
We can now state the analogue of Corollary \ref{relative_partition_function_convergence} for convergence of relative partition functions when the many-body Hamiltonian is given by \eqref{H^nu_kappa_rho}.
\begin{proposition}
\label{partition_function_tuned_particle_density}
Given $\zeta>0$, suppose that $\kappa_0, \rho >0$ (possibly depending on $\nu$) are chosen such that

\begin{equation}
\label{fugacity_condition*}
\exp\biggl\{-\nu\biggl(\kappa_0+\frac{1}{2\nu} v(0)-\frac{\rho}{\nu^2} \widehat{v}(0)\biggr)\biggr\}\,\nu^{-d/2}=\zeta\,.
\end{equation}
With notation as in \eqref{Xi_shifted_particle_density}, we let
\begin{equation}
\label{relative_partition_function_tuned_particle_density}
\mathcal{Z}^{\nu,\zeta}(\kappa_0,\rho)\coloneqq \frac{\Xi^{\nu,\zeta}(\kappa_0,\rho)}{\Xi^{\nu,\zeta}_{(0)}(\kappa_0,\rho)}\,,\qquad 
\Xi^{\nu,\zeta}_{(0)}(\kappa_0,\rho)\coloneqq \Xi^{\nu,\zeta}(\kappa_0,\rho)\bigl|_{v=0}\,.
\end{equation}
We then have 
\begin{equation}
\label{relative_partition_function_tuned_particle_density_convergence}
\lim_{\nu \rightarrow 0}\mathcal{Z}^{\nu,\zeta}(\kappa_0,\rho)=\mathcal{Z}^{0,\zeta}\,,
\end{equation}
for $\mathcal{Z}^{0,\zeta}$ as in \eqref{relative_partition_function_convergence_1}.
\end{proposition}
\begin{proof}
The proof is analogous to that of Corollary \ref{relative_partition_function_convergence}. We note that the constant factors $\frac{\rho^2}{2\nu^2}\,\widehat{v}(0)\,L^d$ cancel out in the ratio given in \eqref{relative_partition_function_tuned_particle_density}. Furthermore, the condition \eqref{fugacity_condition} is replaced by \eqref{fugacity_condition*} above.
\end{proof}

\section{Justification of Remark \ref{Dirichlet_boundary_condition_remark} and proof of Proposition \ref{Infinite_volume_Theorem_1}}
\label{Dirichlet_boundary_conditions}

In this section, we explain Remark \ref{Dirichlet_boundary_condition_remark} and prove Proposition \ref{Infinite_volume_Theorem_1}.
We first justify Remark \ref{Dirichlet_boundary_condition_remark}.
\begin{proof}[Justification of Remark \ref{Dirichlet_boundary_condition_remark}]
In \cite[Chapter 3]{Ginibre}, Ginibre shows that in dimensions $d \geq 2$, \eqref{Ginibre_representation_1_bar} corresponds to the grand-canonical partition function \eqref{grand_canonical_ensemble} when Dirichlet boundary conditions are considered in \eqref{n_body_Hamiltonian}, barring one change: the potential Ginibre considers that is closest to our potential is not continuous at the origin, but is allowed to diverge to infinity there. We give a brief justification that this singularity is not essential to the analysis, which will make our potential admissible. We adopt Ginibre's notation in the remainder of this remark.

Assume the potential $\Phi$ in case A) in \cite[p.\ 346]{Ginibre} is continuous on an open set $\Lambda \subset \R^{\nu}$. In addition, recall that $\Phi$ is a real, even function, bounded from below by some constant $b \in \R$. Recall that $\Lambda \subset \R^{\nu}$, and recall the following notation from \cite{Ginibre}.
\begin{gather*}
	x = (x_{1}, \ldots, x_{m}) \in \Lambda^{m}
	\intertext{is a configuration of $m$ particles, where $x_{j} \in \Lambda$ is the position of the $j$-th particle,}
	U(x) \coloneq \sum_{1 \leq i < j \leq m} \Phi(x_{i} - x_{j})
	\intertext{is the total interaction potential of the particles, and}
	S_{a} \coloneq \{ x \in \R^{m \nu} : |x_{i} - x_{j}| > a \ \text{for all pairs $(i,j)$} \},
\end{gather*}
where $a \geq 0$ is fixed, is the set of admissible configurations; when $a>0$ the configurations in $\Lambda^{m} \setminus S_{a}$ are forbidden by the hard-core interaction. In the case of continuous $\Phi$, there is no need to introduce $S_{a}$, and the set of admissible configurations is $\Lambda^{m}$. In particular, we can take $a=0$.

We define the integral kernel $W^{t}(x, y)$ as in \cite[(1.37)]{Ginibre}. The interaction $U$ is still bounded from below, so the integrand is well-defined. Let $\mathcal{H} \coloneq L^{2}(\Lambda^{m})$ for a fixed $m \in \N^{*}$. Then $W^{t}(x, y)$ defines a bounded self-adjoint operator on $\mathcal{H}$, denoted by $W^{t}$. The properties of $W^{t}(x, y)$ and $W^{t}$ stated in \cite[Lemma 1.4]{Ginibre} and in the paragraph preceeding it only require that $U$ is bounded from below, real, and continuous, and therefore still hold, but with $\Lambda^{m} \cap S_{a}$ replaced by $\Lambda^{m}$.

Finally, we observe that the analysis in \cite[Lemma 1.5]{Ginibre} is done on compact subsets of $\Lambda^{m} \cap S_{a}$ where the interaction $U$ is bounded and continuous. Namely, these sets are $Q_{\varepsilon}$ and $Q_{c/2}$, and by definition lie at a positive distance away from the boundary of $\Lambda^{m} \cap S_{a}$. Note that the potentials Ginibre considers only have singularities on $\Lambda^{m} \setminus S_{a}$. On the other hand, the potentials we consider are continuous on the entire $\Lambda^{m}$. Therefore, the arguments in the proof of \cite[Lemma 1.5]{Ginibre} immediately extend to our case, where we replace $\Lambda^{m} \cap S_{a}$ by $\Lambda^{m}$.
\end{proof}

Let us now prove Proposition \ref{Infinite_volume_Theorem_1}.
\begin{proof}[Proof of Proposition \ref{Infinite_volume_Theorem_1}]
Let us first prove claim (i). 
Arguing as in the proof of Proposition \ref{Proposition_5.3} (i), we get that 
\begin{equation*}
\sup_{L \geq 1} \frac{\log \overline{\Xi}^{\,\nu,\zeta,L}}{|\Lambda_L|} <\infty\,.
\end{equation*}
Hence, we just need to show that the limit \eqref{Theorem_1.16_ii} exists.
Given $R>0$, we define the truncated interaction $v_{(R)} \colon \R^d \rightarrow \R$ as
\begin{equation}
\label{v_{(R)}}
v_{(R)} \coloneqq v\,\mathbf{1}_{|x|<R}\,.
\end{equation}
By \eqref{Lemma_4.1_iii} and \eqref{v_{(R)}}, it follows that 
\begin{equation}
\label{v_{(R)}_convergence_1}
\lim_{R \rightarrow \infty} \|v_{(R)}-v\|_{L^1(\R^d)}=0\,.
\end{equation}
Moreover, from \eqref{v_assumption_iii}, it follows that 
\begin{equation}
\label{v_{(R)}_convergence_2}
\lim_{R \rightarrow \infty} \|v_{(R)}-v\|_{L^{\infty}(\R^d)}=0\,.
\end{equation}
In the rest of the section, let $\mathcal{V}^{\nu,\infty}_{(R)}(\vec{\omega})$ denote \eqref{loop_interaction_2} with interaction potential \eqref{v_{(R)}}  in \eqref{loop_interaction}--\eqref{self_interaction}.
Recalling \eqref{single_loop_measure_bar}, we define 
\begin{equation}
\label{Ginibre_representation_1_bar_R}
\overline{\Xi}^{\,\nu,\zeta,L}_{(R)} \coloneqq  1 + \sum_{n=1}^{\infty} \frac{1}{n!} \int \overline{\mathbb{L}}^{\,\nu,\zeta,L}(\dd \omega_1) \cdots\, \overline{\mathbb{L}}^{\,\nu,\zeta,L}(\dd \omega_n) \,\exp\left(-\mathcal{V}^{\nu,\infty}_{(R)}(\vec{\omega})\right).
\end{equation}
Let us first note that 
\begin{equation}
\label{Ginibre_representation_1_bar_R_convergence}
\frac{\log \overline{\Xi}^{\,\nu,\zeta,L}}{|\Lambda_L|}-\frac{\log \overline{\Xi}^{\,\nu,\zeta,L}_{(R)}}{|\Lambda_L|}=\mathcal{O}_{d,\zeta,v}\left(\|v-v_{(R)}\|_{L^{\infty}(\R^d)}\right)+\mathcal{O}_{d,\zeta,v}\left(\|v-v_{(R)}\|_{L^1(\R^d)}\right),
\end{equation}
uniformly in $L$. 
The significance of \eqref{Ginibre_representation_1_bar_R_convergence} is that it lets us work with a finite-range interaction potential \eqref{v_{(R)}}, instead of with $v$. This is crucially used to obtain \eqref{Infinite_volume_Thm_1_Star_3_B} below.

To prove \eqref{Ginibre_representation_1_bar_R_convergence}, we first recall \eqref{I^{L}} and define the following modifications of the measure \eqref{measure_mu_1}:
\begin{align}
\label{measure_mu_1_A}
\overline{\mu}^{L}(\dd \omega) &\coloneqq  \nu \sum_{T\in \nu {\N}^*} \frac{z^T}{T} \, \int_{\Lambda_L}\dd x\, \int \mathbb{W}^{\infty,T}_{x,x}(\dd \omega)\,\ee^{-\mathcal{V}^{\nu,\infty}(\omega)}\,\mathbf{1}_{ \mathscr{I}_L}(\omega)\,,
\\
\label{measure_mu_1_B}
\overline{\mu}^{L}_{(R)}(\dd \omega) &\coloneqq \nu \sum_{T\in \nu {\N}^*} \frac{z^T}{T} \, \int_{\Lambda_L}\dd x\, \int \mathbb{W}^{\infty,T}_{x,x}(\dd \omega)\,\ee^{-\mathcal{V}^{\nu,\infty}_{(R)}(\omega)}\,\mathbf{1}_{ \mathscr{I}_L}(\omega)\,,
\\
\label{measure_mu_1_C}
\overline{\mu}^{L,0}(\dd \omega) &\coloneqq  \nu \sum_{T\in \nu {\N}^*} \frac{z^T}{T} \, \int_{\Lambda_L}\dd x\, \int \mathbb{W}^{\infty,T}_{x,x}(\dd \omega)\,\mathbf{1}_{ \mathscr{I}_L}(\omega)\,.
\end{align}
By \eqref{measure_mu_1_A}--\eqref{measure_mu_1_C} and the nonnegativity of $v$, we obtain 
\begin{equation}
\label{measure_mu_1_D}
\overline{\mu}^{L}(\dd \omega) \,,\overline{\mu}^{L}_{(R)}(\dd \omega) \leq \overline{\mu}^{L,0}(\dd \omega)\,.
\end{equation}
The arguments from the proof of Proposition \ref{gamma_p_cluster_expansion} (i) allow us to deduce that 
\begin{equation}
\label{R_cluster_expansion_1}
\frac{\log \overline{\Xi}^{\,\nu,\zeta,L}}{|\Lambda_L|}=\sum_{n \geq 1} \frac{1}{|\Lambda_L|} \,\int \overline{\mu}^L(\dd \omega_1)\,\cdots\,\overline{\mu}^L(\dd \omega_n)\,\varphi^{\infty}(\omega_1,\ldots,\omega_n)
\end{equation}
and 
\begin{equation}
\label{R_cluster_expansion_2}
\frac{\log \overline{\Xi}^{\,\nu,\zeta,L}_{(R)}}{|\Lambda_L|}=\sum_{n \geq 1} \frac{1}{|\Lambda_L|} \,\int \overline{\mu}^L_{(R)}(\dd \omega_1)\,\cdots\,\overline{\mu}^L_{(R)}(\dd \omega_n)\,\varphi^{\infty}_{(R)}(\omega_1,\ldots,\omega_n)\,.
\end{equation}
In \eqref{R_cluster_expansion_1}--\eqref{R_cluster_expansion_2}, we defined the following modifications of \eqref{Ursell_function_definition}:
\begin{equation}
\label{Ursell_function_definition_modification_1}
\varphi^{\infty}(\omega_{1},\dots, \omega_{n})
\coloneqq 
\frac{1}{n!} \sum_{\mathcal{G} \in {\mathfrak{G}}^c_n} \prod_{\{i,j\} \in \mathcal{G}} \xi^{\infty}(\omega_i, \omega_j)\,, \quad \xi^{\infty}(\omega, \tilde \omega)\coloneqq  \exp\bigl(-\mathcal{V}^{\nu,\infty}(\omega,\tilde \omega)\bigr) - 1\,,
\end{equation}
\begin{equation}
\label{Ursell_function_definition_modification_2}
\varphi^{\infty}_{(R)}(\omega_{1},\dots, \omega_{n})
\coloneqq 
\frac{1}{n!} \sum_{\mathcal{G} \in {\mathfrak{G}}^c_n} \prod_{\{i,j\} \in \mathcal{G}} \xi^{\infty}_{(R)}(\omega_i, \omega_j)\,, \quad \xi^{\infty}_{(R)}(\omega, \tilde \omega)\coloneqq  \exp\left(-\mathcal{V}^{\nu,\infty}_{(R)}(\omega,\tilde \omega)\right) - 1\,.
\end{equation}
The estimate \eqref{Ginibre_representation_1_bar_R_convergence} follows if we prove that 
\begin{equation}
\label{R_cluster_expansion_3}
\left| \frac{\log \overline{\Xi}^{\,\nu,\zeta,L}}{|\Lambda_L|}-\sum_{n \geq 1} \frac{1}{|\Lambda_L|} \,\int \overline{\mu}^L_{(R)}(\dd \omega_1)\,\cdots\,\overline{\mu}^L_{(R)}(\dd \omega_n)\,\varphi^{\infty}(\omega_1,\ldots,\omega_n)\right|= \mathcal{O}_{d,\zeta,v}\left(\|v-v_{(R)}\|_{L^{\infty}(\R^d)}\right)
\end{equation}
and 
\begin{equation}
\label{R_cluster_expansion_4}
\left|\sum_{n \geq 1} \frac{1}{|\Lambda_L|} \,\int \overline{\mu}^L_{(R)}(\dd \omega_1)\,\cdots\,\overline{\mu}^L_{(R)}(\dd \omega_n)\,\varphi^{\infty}(\omega_1,\ldots,\omega_n)-\frac{\log \overline{\Xi}^{\,\nu,\zeta,L}_{(R)}}{|\Lambda_L|}\right|= \mathcal{O}_{d,\zeta,v}\left(\|v-v_{(R)}\|_{L^1(\R^d)}\right),
\end{equation}
uniformly in $L$.

By the dominated convergence theorem, \eqref{R_cluster_expansion_1}, a telescoping argument, \eqref{Ursell_function_definition_modification_1}, and Lemma \ref{tree_bound_estimate}, the estimate \eqref{R_cluster_expansion_3} follows provided we show that for all $n \in \N^*$ and for all trees $\mathcal{T}$ on $[n]=\{1,\ldots,n\}$, we have
\begin{multline}
\label{R_cluster_expansion_5}
\frac{1}{|\Lambda_L|}\,\int \left(\overline{\mu}^L(\dd \omega_1)-\overline{\mu}^{L}_{(R)}(\dd \omega_1)\right)\,\overline{\mu}^{L}_{\sharp}(\dd \omega_2)\,\cdots\,\overline{\mu}^{L}_{\sharp}(\dd \omega_n)\,
\prod_{\{i,j\} \in \mathcal{T}} \xi^{\infty}(\omega_i,\omega_j)
\\
=\mathcal{O}_{d,\zeta,v,n,\mathcal{T}}\left(\|v-v_{(R)}\|_{L^{\infty}(\R^d)}\right),
\end{multline}
uniformly in $L$. In \eqref{R_cluster_expansion_5}, $\overline{\mu}^{L}_{\sharp}(\dd \omega_j)$ denotes either $\overline{\mu}^{L}(\dd \omega_j)$ or $\overline{\mu}^{L}_{(R)}(\dd \omega_j)$ given by  \eqref{measure_mu_1_A} or \eqref{measure_mu_1_B} respectively.
By using \eqref{measure_mu_1_D} and the integration algorithm from Lemma \ref{integration_algorithm} (as in the proof of Proposition \ref{Proposition_5.3} (i) above) with measure $\mu^L(\dd \omega)$ in \eqref{measure_mu_1} replaced by $\overline{\mu}^{L,0}(\dd \omega)$ in \eqref{measure_mu_1_C}, \eqref{R_cluster_expansion_5} follows if we show that for all $q \in \N$, we have 
\begin{equation}
\label{R_cluster_expansion_6}
\frac{1}{|\Lambda_L|}\,\int \left|\overline{\mu}^L(\dd \omega_1)-\overline{\mu}^{L}_{(R)}(\dd \omega_1)\right|\,T(\omega_1)^q=\mathcal{O}_{\zeta,q,d} \left(\|v-v_{(R)}\|_{L^{\infty}(\R^d)}\right),
\end{equation}
uniformly in $L$. 
In order to show \eqref{R_cluster_expansion_6}, we recall \eqref{measure_mu_1_A}--\eqref{measure_mu_1_B}, and use  \eqref{measure_mu_1_D}, the mean-value theorem and the nonnegativity of $\mathcal{V}^{\nu,\infty}(\omega_1)$ and $\mathcal{V}^{\nu,\infty}_{(R)}(\omega_1)$ to see that \eqref{R_cluster_expansion_6} follows if we show
\begin{equation}
\label{R_cluster_expansion_7}
\frac{1}{|\Lambda_L|}\,\int \overline{\mu}^{L,0}(\dd \omega_1) \left|\mathcal{V}^{\nu,\infty}(\omega_1)-\mathcal{V}^{\nu,\infty}_{(R)}(\omega_1)\right|\,T(\omega_1)^q=\mathcal{O}_{d,q} \left(\|v-v_{(R)}\|_{L^{\infty}(\R^d)}\right),
\end{equation}
uniformly in $L$. 
By \eqref{self_interaction}, we have that 
\begin{equation}
\label{R_cluster_expansion_8}
\left|\mathcal{V}^{\nu,\infty}(\omega_1)-\mathcal{V}^{\nu,\infty}_{(R)}(\omega_1)\right| \lesssim \frac{T(\omega_1)^2}{\nu^2}\,\|v-v_{(R)}\|_{L^{\infty}(\R^d)}\,.
\end{equation}
We deduce \eqref{R_cluster_expansion_7} by using \eqref{measure_mu_1_C}, \eqref{R_cluster_expansion_8}, and observing that by 
 \eqref{fugacity_condition}, \eqref{heat_kernel_infty}, \eqref{heat_kernel_integral_infty}, and Lemma \ref{polylogarithm}
we have 
\begin{multline}
\label{R_cluster_expansion_9}
\nu \sum_{T \in \nu \N^*} z^T\,\frac{T^{q+1}}{\nu^2}\,\psi^{\infty,T}(0)
\lesssim \sum_{k \in \N^*} z^{k \nu} \,\frac{(k\nu)^{q+1}}{\nu}\, (k\nu)^{-d/2}
=\nu^q \sum_{k \in \N^*} \zeta^k\,\nu^{dk/2}\, \, \nu^{-d/2} \,k^{q+1-d/2}
\\
\leq  \nu^q \sum_{k \in \N^*} \zeta^k\, k^{q+1-d/2}
= \mathcal{O}_{\zeta,q,d}(1)\,.
\end{multline}
Note that \eqref{R_cluster_expansion_9} is slightly different from the bound in \eqref{5.27_remark_bound}.
We hence conclude \eqref{R_cluster_expansion_3}.

We show \eqref{R_cluster_expansion_4} using similar methods. More precisely, we use \eqref{R_cluster_expansion_2} to write 
\begin{multline}
\label{R_cluster_expansion_10}
\sum_{n \geq 1} \frac{1}{|\Lambda_L|} \,\int \overline{\mu}^L_{(R)}(\dd \omega_1)\,\cdots\,\overline{\mu}^L_{(R)}(\dd \omega_n)\,\varphi^{\infty}(\omega_1,\ldots,\omega_n)-\frac{\log \overline{\Xi}^{\,\nu,\zeta,L}_{(R)}}{|\Lambda_L|}
\\
=
\sum_{n \geq 1} \frac{1}{|\Lambda_L|} \,\int \overline{\mu}^L_{(R)}(\dd \omega_1)\,\cdots\,\overline{\mu}^L_{(R)}(\dd \omega_n)\,\left(\varphi^{\infty}(\omega_1,\ldots,\omega_n)-\varphi^{\infty}_{(R)}(\omega_1,\ldots,\omega_n)\right).
\end{multline}
Using \eqref{measure_mu_1_D} and \eqref{R_cluster_expansion_10}, we note that
\eqref{R_cluster_expansion_4} follows if we show that
\begin{multline}
\label{R_cluster_expansion_11}
\sum_{n \geq 1} \frac{1}{|\Lambda_L|} \,\int \overline{\mu}^{L,0}(\dd \omega_1)\,\cdots\,\overline{\mu}^{L,0}(\dd \omega_n)\,\left|\varphi^{\infty}(\omega_1,\ldots,\omega_n)-\varphi^{\infty}_{(R)}(\omega_1,\ldots,\omega_n)\right|
\\
=\mathcal{O}_{d,\zeta}\left(\|v-v_{(R)}\|_{L^1(\R^d)}\right),
\end{multline}
uniformly in $L$.

Let us now show \eqref{R_cluster_expansion_11}. Recalling \eqref{Ursell_function_definition_modification_1}--\eqref{Ursell_function_definition_modification_2},  the integration algorithm from Lemma \ref{integration_algorithm} with measure $\mu^L(\dd \omega)$ in \eqref{measure_mu_1} replaced by $\overline{\mu}^{L,0}(\dd \omega)$ in \eqref{measure_mu_1_C}, and using the dominated convergence theorem, it suffices to show that for all $n \geq 2$ and $\mathcal{G} \in {\mathfrak{G}}^c_n$, we have
\begin{multline}
\label{R_cluster_expansion_12}
\frac{1}{|\Lambda_L|} \,\int \overline{\mu}^{L,0}(\dd \omega_1)\,\cdots\,\overline{\mu}^{L,0}(\dd \omega_n)\,\left|\prod_{\{i,j\} \in \mathcal{G}}\xi^{\infty}(\omega_i,\omega_j)-\prod_{\{i,j\} \in \mathcal{G}}\xi^{\infty}_{(R)}(\omega_i,\omega_j)\right|
\\
=\mathcal{O}_{d,\zeta,v,n}\left(\|v-v_{(R)}\|_{L^1(\R^d)}\right),
\end{multline}
uniformly in $L$.

As in \eqref{xi_telescoping_1} above, let $\prec$ be a total order of the edges in $\mathcal{G}$, and we write an edge $e$ as $\{i_e,j_e\}$. Using a telescoping argument and the same notational conventions as in \eqref{xi_telescoping_1}--\eqref{xi_telescoping_1_B}, we have
\begin{multline}
\label{R_cluster_expansion_12_B}    
\left|\prod_{\{i,j\}\in \mathcal{G}}\xi^{\infty}(\omega_i,\omega_j)-\prod_{\{i,j\}\in \mathcal{G}}\xi^{\infty}_{(R)}(\omega_i,\omega_j)\right|\\
    \leq  \sum_{e \in \mathcal{G}}\prod_{\substack{\{i,j\} \in \mathcal{T}_e\\\{i,j\} \prec e}}\Bigl|\xi^{\infty}(\omega_i,\omega_j)\Bigr|\,\left|\xi^{\infty}(\omega_{i_e},\omega_{j_e})-\xi^{\infty}_{(R)}(\omega_{i_e},\omega_{j_e})\right|\prod_{\substack{\{i,j\} \in \mathcal{T}_e\\ e \prec \{i,j\}}}\left|\xi^{\infty}_{(R)}(\omega_i,\omega_j)\right|\,.
\end{multline}
As in \eqref{xi_telescoping_1_B}, for each $e \in \mathcal{G}$ in \eqref{R_cluster_expansion_12_B}, $\mathcal{T}_e$ is an arbitrary spanning tree of $\mathcal{G}$ containing $e$.
Similarly to \eqref{xi_telescoping_2}--\eqref{xi_telescoping_2_B}, we have
\begin{multline}
\label{R_cluster_expansion_14}       
\left|\xi^{\infty}(\omega,\tilde{\omega})-\xi^{\infty}_{(R)}(\omega,\tilde{\omega})\right |\leq \left|\mathcal{V}^{\nu,\infty}(\omega,\tilde{\omega})-\mathcal{V}^{\nu,\infty}_{(R)}(\omega,\tilde{\omega})\right|\\
    \leq \frac{1}{\nu} \sum_{s \in \nu \N} \vec{1}_{s<T(\omega)} \sum_{\tilde{s} \in \nu \N} \vec{1}_{\tilde{s}<T(\tilde{\omega})} \int_{0}^{\nu}\dd t\,\left| v-v_{(R)}\right|\bigl(\omega(s+t)-\tilde{\omega}(\tilde{s}+t)\bigr) 
\end{multline}
and 
\begin{multline}
\label{R_cluster_expansion_15}      
\Bigl|\xi^{\infty}(\omega,\tilde{\omega})\Bigr|+\left|\xi^{\infty}_{(R)}(\omega,\tilde{\omega})\right|
\\
\leq \frac{1}{\nu} \sum_{s \in \nu \N} \vec{1}_{s<T(\omega)} \sum_{\tilde{s} \in \nu \N} \vec{1}_{\tilde{s}<T(\tilde{\omega})} \int_{0}^{\nu}\mathrm{d} t\,\left[\bigl|v\bigr|+\bigl|v_{(R)}\bigr|\right] \bigl(\omega(s+t)-\tilde{\omega}(\tilde{s}+t)\bigr)
\end{multline}
respectively. Using \eqref{R_cluster_expansion_12_B}--\eqref{R_cluster_expansion_15} and 
 the integration algorithm from Lemma \ref{integration_algorithm} with measure $\overline{\mu}^{L,0}(\dd \omega)$ in \eqref{measure_mu_1} replaced by $\overline{\mu}^{L,0}(\dd \omega)$ in \eqref{measure_mu_1_C}, we deduce \eqref{R_cluster_expansion_11}, and hence \eqref{R_cluster_expansion_4}. Finally, we obtain \eqref{Ginibre_representation_1_bar_R_convergence} from \eqref{R_cluster_expansion_3}--\eqref{R_cluster_expansion_4}.

By \eqref{v_{(R)}} and Assumption \ref{Assumption_on_v_2}, we have that the right-hand side of \eqref{Ginibre_representation_1_bar_R_convergence} converges to zero as $R \rightarrow \infty$. In particular, using \eqref{Ginibre_representation_1_bar_R_convergence} we obtain \eqref{Theorem_1.16_ii} and conclude the proof of claim (i) if we prove that for $R>0$ fixed, the limit 
\begin{equation}
\label{Infinite_volume_Thm_1_Star_3_B}
\lim_{L \rightarrow \infty} \frac{\log \overline{\Xi}_{(R)}^{\,\nu,\zeta,L}}{|\Lambda_L|}
\end{equation}
exists. In order to show \eqref{Infinite_volume_Thm_1_Star_3_B}, we introduce a generalisation of \eqref{I^{L}}--\eqref{Ginibre_representation_1_bar} corresponding to a Lebesgue measurable set $\Lambda \subset \R^d$ which is not necessary equal to $\Lambda_L$. Namely, we define the set
\begin{equation}
\label{I^{L}_*}
\mathscr{I}_{\Lambda} \coloneqq \biggl\{\omega \in \bigcup_{T > 0} \Omega^{\infty,T}\,,\,\omega(t) \in \Lambda\,\, \forall \, t \in [0,T(\omega)]\biggr\}\,,
\end{equation}
and the loop measure
\begin{equation}
\label{single_loop_measure_bar_*}
\overline{\mathbb{L}}^{\,\nu,\zeta,\Lambda}(\dd \omega)\coloneqq \nu \sum_{T \in \nu \N^*} \frac{z^T}{T} \, \int_{\Lambda} \dd x\,\int \mathbb{W}^{\infty,T}_{x,x}(\dd \omega)\,\mathbf{1}_{\mathscr{I}_{\Lambda}}(\omega)\,.
\end{equation}
From \eqref{I^{L}_*}--\eqref{single_loop_measure_bar_*}, we define
\begin{equation}
\label{Ginibre_representation_1_bar_*}
\overline{\Xi}^{\,\nu,\zeta,\Lambda}_{(R)} \coloneqq  1 + \sum_{n=1}^{\infty} \frac{1}{n!} \int \overline{\mathbb{L}}^{\,\nu,\zeta,\Lambda}(\dd \omega_1) \cdots\, \overline{\mathbb{L}}^{\,\nu,\zeta,\Lambda}(\dd \omega_n) \,\exp\bigl(-\mathcal{V}^{\nu,\infty}_{(R)}(\vec{\omega})\bigr)\,.
\end{equation}
From \eqref{Ginibre_representation_1_bar_*}, we have that 
\begin{equation}
\label{Ginibre_representation_1_bar_*_1}
\overline{\Xi}^{\,\nu,\zeta,\Lambda_L}_{(R)} =\overline{\Xi}^{\,\nu,\zeta,L}_{(R)}\,,
\end{equation}
where we recall \eqref{Ginibre_representation_1_bar_R}. Moreover, by translation invariance, we have that 
\begin{equation}
\label{Ginibre_representation_1_bar_*_2}
\overline{\Xi}^{\,\nu,\zeta,\Lambda+y}_{(R)}=\overline{\Xi}^{\,\nu,\zeta,\Lambda}_{(R)} \qquad \forall \,y \in \R^d\,.
\end{equation}
Let us consider Lebesgue measurable sets 
\begin{equation}
\label{Lambda^{(1)}_Lambda^{(2)}}
\Lambda^{(1)}, \Lambda^{(2)} \subset \R^d\,,\qquad \mathrm{dist} (\Lambda^{(1)}, \Lambda^{(2)}) \geq R\,.
\end{equation}
For such $\Lambda^{(1)}, \Lambda^{(2)}$, let us note that
\begin{equation}
\label{Lambda^{(1)}_Lambda^{(2)}_A}
\overline{\Xi}^{\,\nu,\zeta,\Lambda^{(1)} \cup \Lambda^{(2)}}_{(R)} \geq \overline{\Xi}^{\,\nu,\zeta,\Lambda^{(1)}}_{(R)}\,\overline{\Xi}^{\,\nu,\zeta,\Lambda^{(2)}}_{(R)}\,.
\end{equation}
In order to show \eqref{Lambda^{(1)}_Lambda^{(2)}_A}, we start from 
\begin{equation}
\label{Lambda^{(1)}_Lambda^{(2)}_B}
\overline{\Xi}^{\,\nu,\zeta,\Lambda^{(1)} \cup \Lambda^{(2)}}_{(R)} =1 + \sum_{n=1}^{\infty} \frac{1}{n!} \int \overline{\mathbb{L}}^{\,\nu,\zeta,\Lambda^{(1)} \cup \Lambda^{(2)}}(\dd \omega_1) \cdots\, \overline{\mathbb{L}}^{\,\nu,\zeta,\Lambda^{(1)} \cup \Lambda^{(2)}}(\dd \omega_n) \,\exp\bigl(-\mathcal{V}^{\nu,\infty}_{(R)}(\vec{\omega})\bigr)\,.
\end{equation}
By continuity and by \eqref{Lambda^{(1)}_Lambda^{(2)}}, each of the loops $\omega_j$ appearing in the integral \eqref{Lambda^{(1)}_Lambda^{(2)}_B} lies either in $\Lambda^{(1)}$ or $\Lambda^{(2)}$. Using this observation as well as \eqref{single_loop_measure_bar_*}, it follows that 
\begin{multline}
\label{Lambda^{(1)}_Lambda^{(2)}_C}
\eqref{Lambda^{(1)}_Lambda^{(2)}_B} \geq \sum_{n=0}^{\infty} \sum_{k=0}^{n} \binom{n}{k}\, \frac{1}{n!}\, 
\int \overline{\mathbb{L}}^{\,\nu,\zeta,\Lambda^{(1)}}(\dd \omega'_1)\, \cdots \overline{\mathbb{L}}^{\,\nu,\zeta,\Lambda^{(1)}}(\dd \omega'_k)  
\\
\times \overline{\mathbb{L}}^{\,\nu,\zeta,\Lambda^{(2)}}(\dd \omega''_1)\, \cdots \overline{\mathbb{L}}^{\,\nu,\zeta,\Lambda^{(2)}}(\dd \omega''_{n-k})  \,\exp\bigl(-\mathcal{V}^{\nu,\infty}_{(R)}(\vec{\omega' \omega''})\bigr)\,.
\end{multline}
where $\vec{\omega' \omega''} \equiv (\omega'_1, \ldots, \omega'_k, \omega''_1, \ldots, \omega''_{n-k})$. 
Note that, when passing from \eqref{Lambda^{(1)}_Lambda^{(2)}_B} to \eqref{Lambda^{(1)}_Lambda^{(2)}_C}, we reduce spatial integrations over $\Lambda^{(1)} \cup \Lambda^{(2)}$ to either $\Lambda^{(1)}$ or $\Lambda^{(2)}$, hence we only have an inequality.

Using \eqref{v_{(R)}}, \eqref{Lambda^{(1)}_Lambda^{(2)}}, and recalling the construction \eqref{loop_interaction}--\eqref{self_interaction} and \eqref{loop_interaction_2} of $\mathcal{V}^{\nu,\infty}_{(R)}(\cdot)$, it follows that 
\begin{equation*}
\mathcal{V}^{\nu,\infty}_{(R)}(\vec{\omega' \omega''})=\mathcal{V}^{\nu,\infty}_{(R)}(\vec{\omega'})\,\mathcal{V}^{\nu,\infty}_{(R)}(\vec{\omega''})
\end{equation*}
for loops as in the integral given by \eqref{Lambda^{(1)}_Lambda^{(2)}_C}. Hence, we deduce that
\begin{multline}
\label{Lambda^{(1)}_Lambda^{(2)}_D}
\eqref{Lambda^{(1)}_Lambda^{(2)}_C} =\sum_{n=0}^{\infty} \sum_{k=0}^{n} \frac{1}{k!}\, \frac{1}{(n-k)!}\, 
\int \overline{\mathbb{L}}^{\,\nu,\zeta,\Lambda^{(1)}}(\dd \omega'_1)\, \cdots \overline{\mathbb{L}}^{\,\nu,\zeta,\Lambda^{(1)}}(\dd \omega'_k)  \,\exp\bigl(-\mathcal{V}^{\nu,\infty}_{(R)}(\vec{\omega'})\bigr)
\\
\times \overline{\mathbb{L}}^{\,\nu,\zeta,\Lambda^{(2)}}(\dd \omega''_1)\, \cdots \overline{\mathbb{L}}^{\,\nu,\zeta,\Lambda^{(2)}}(\dd \omega''_{n-k})  \,\exp\bigl(-\mathcal{V}^{\nu,\infty}_{(R)}(\vec{\omega''})\bigr)
=\overline{\Xi}^{\,\nu,\zeta,\Lambda^{(1)}}_{(R)}\,\overline{\Xi}^{\,\nu,\zeta,\Lambda^{(2)}}_{(R)}\,.
\end{multline}
We hence deduce \eqref{Lambda^{(1)}_Lambda^{(2)}_A}
 from \eqref{Lambda^{(1)}_Lambda^{(2)}_D}.
 
Consider now $L > l \geq1$. We write 
\begin{equation}
\label{L_ell_1}
L=k(l+R) + a\,,\quad k \in \N^*\,,\, 0 \leq a <l+R\,.
\end{equation}
Using \eqref{L_ell_1}, in $\Lambda_L$, we find boxes $\Lambda^{(i)}, i=1, \ldots, k^d$ of size $l$ separated by strips of width $R$ as in Figure \ref{Rectangle_example} below.

\begin{figure}[!ht]
\begin{center}
\includegraphics[scale=0.35]{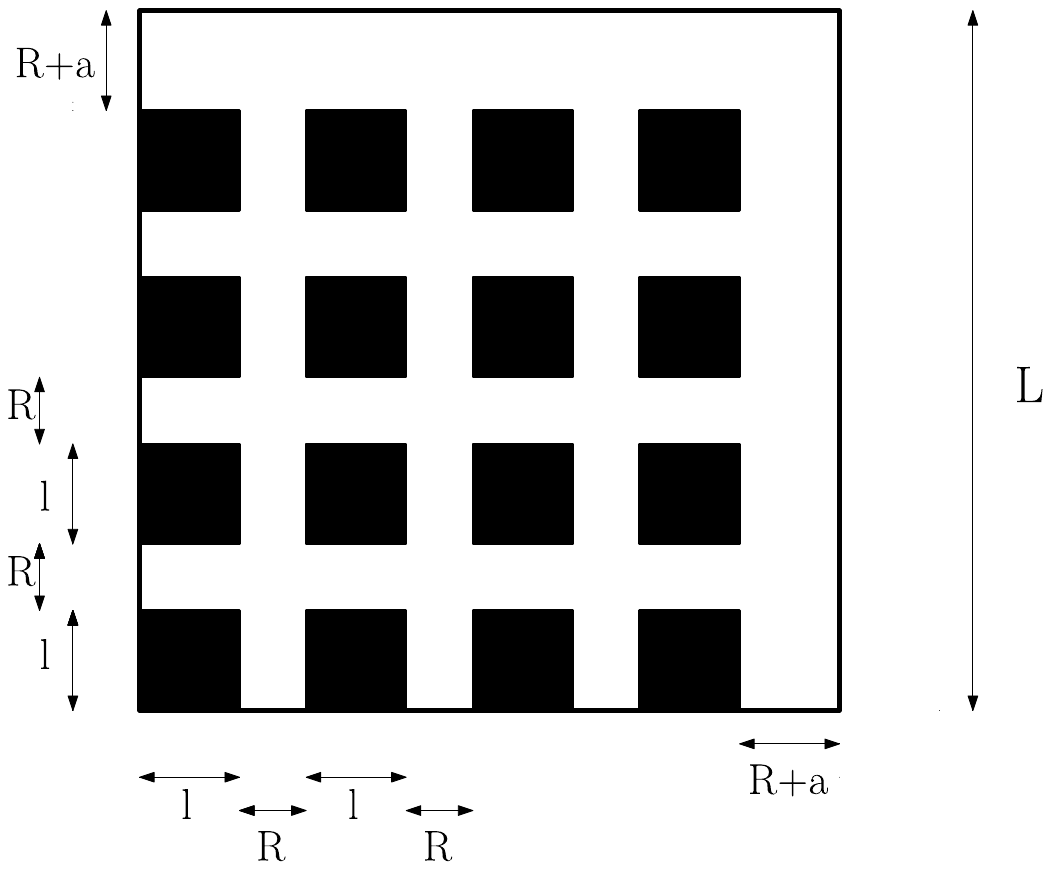}
\end{center}
\caption{Here, we take $d=2$ and $k=4$. The general argument works in the same way. $\Lambda^{(1)},\ldots,\Lambda^{(k^d)}$ are the shaded smaller boxes of sidelength $l$ contained in the big box of sidelength $L$. The unshaded region can be covered by a union of $O_d(k^d)$ rectangles of volume $(l+R)^{d-1} R$ and by $d$ rectangles of volume $a L^{d-1} \leq (l+R)L^{d-1}$. Here, we recall \eqref{L_ell_1}.
\label{Rectangle_example}
}
\end{figure}

Using \eqref{single_loop_measure_bar_*}--\eqref{Ginibre_representation_1_bar_*}, and iteratively applying \eqref{Lambda^{(1)}_Lambda^{(2)}_A}, followed by \eqref{Ginibre_representation_1_bar_*_2},  we have
\begin{multline}
\label{Lambda^{(1)}_Lambda^{(2)}_E}
\frac{\log \overline{\Xi}_{(R)}^{\,\nu,\zeta,L}}{|\Lambda_L|}=\frac{1}{L^d}\,\log \overline{\Xi}_{(R)}^{\,\nu,\zeta,\Lambda_L} 
\geq \frac{1}{L^d}\,\log \overline{\Xi}_{(R)}^{\,\nu,\zeta, \Lambda^{(1)} \cup \cdots \cup \Lambda^{(k^d)}} 
\geq \frac{1}{L^d}\,\log \left[\prod_{i=1}^{k^d} \overline{\Xi}_{(R)}^{\,\nu,\zeta, \Lambda^{(i)}} \right]
\\
=\frac{1}{L^d}\,\log \left[\left(\overline{\Xi}_{(R)}^{\,\nu,\zeta, l}\right)^{k^d} \right]=\frac{k^d}{L^d}\,\log \overline{\Xi}_{(R)}^{\,\nu,\zeta, l}=\frac{k^d l^d}{L^d}\,\frac{\log \overline{\Xi}_{(R)}^{\,\nu,\zeta, l}}{|\Lambda_{l}|}\,.
\end{multline}
Recalling \eqref{L_ell_1}, we have that
\begin{equation}
\label{Lambda^{(1)}_Lambda^{(2)}_F}
1 \geq \frac{k^d l^d}{L^d} \geq 1-\frac{C_d k^d (l+R)^{d-1}R + d (l+R)L^{d-1}}{L^d} \stackrel{L \rightarrow \infty}{\longrightarrow} 1-\frac{dR}{l+R}\,.
\end{equation}
In \eqref{Lambda^{(1)}_Lambda^{(2)}_F}, we used the observation that for a suitable constant $C_d>0$ depending on $d$, we have that
\begin{equation}
\label{Lambda^{(1)}_Lambda^{(2)}_F_2}
k^d l^d +C_d k^d(l+R)^{d-1}R+d(l+R) L^{d-1} \geq L^d\,.
\end{equation}
For a justification of \eqref{Lambda^{(1)}_Lambda^{(2)}_F_2}, see Figure \ref{Rectangle_example} below.

By \eqref{Lambda^{(1)}_Lambda^{(2)}_E}--\eqref{Lambda^{(1)}_Lambda^{(2)}_F}, it follows that 
\begin{equation}
\label{Lambda^{(1)}_Lambda^{(2)}_G}
\liminf_{L \rightarrow \infty} \frac{\log \overline{\Xi}_{(R)}^{\,\nu,\zeta,L}}{|\Lambda_L|} \geq \left(1-\frac{dR}{l+R}\right)\,\frac{\log \overline{\Xi}_{(R)}^{\,\nu,\zeta, l}}{|\Lambda_{l}|}\,.
\end{equation}
From \eqref{Lambda^{(1)}_Lambda^{(2)}_G}, we deduce that 
\begin{equation}
\label{Lambda^{(1)}_Lambda^{(2)}_H}
\liminf_{L \rightarrow \infty} \frac{\log \overline{\Xi}_{(R)}^{\,\nu,\zeta,L}}{|\Lambda_L|} \geq \limsup_{l \rightarrow \infty}\frac{\log \overline{\Xi}_{(R)}^{\,\nu,\zeta, l}}{|\Lambda_{l}|}\,.
\end{equation}
We then deduce the existence of the limit \eqref{Infinite_volume_Thm_1_Star_3_B} from \eqref{Lambda^{(1)}_Lambda^{(2)}_H}. Claim (i) now follows.

Let us now prove claim (ii). 
Let us define
\begin{equation}
\label{Ginibre_representation_1_check}
\widecheck{\Xi}^{\nu,\zeta,L} \coloneqq 1 + \sum_{n=1}^{\infty} \frac{1}{n!} \int \mathbb{L}^{\nu,\zeta,L}(\dd \omega_1) \cdots\, \mathbb{L}^{\nu,\zeta,L}(\dd \omega_n) \,\exp\bigl(-\mathcal{V}^{\nu,\infty}(\vec{\omega})\bigr)\,.
\end{equation}
By repeating Steps 1--4 of the proof of Theorem \ref{free_energy_infinite_volume}, we see that 
\begin{equation}
\label{Ginibre_representation_1_check_convergence}
\lim_{L\rightarrow \infty} \left(\frac{\log \widecheck{\Xi}^{\,\nu,\zeta,L}}{|\Lambda_L|}-\frac{\log \Xi^{\,\nu,\zeta,L}}{|\Lambda_L|}\right)=0
\end{equation}
uniformly in $\nu \in (0,\nu_0]$. Therefore, by \eqref{Ginibre_representation_1_check_convergence}, we have that \eqref{Infinite_volume_Theorem_1_Star_1_B} follows if we prove that
\begin{equation}
\label{Infinite_volume_Theorem_1_Star_1_B}
\lim_{L \rightarrow \infty} \left(\frac{\log \widecheck{\Xi}^{\nu,\zeta,L}}{|\Lambda_L|}-\frac{\log \overline{\Xi}^{\,\nu,\zeta,L}}{|\Lambda_L|}\right)=0\,,
\end{equation}
uniformly in $\nu \in (0,\nu_0]$.

We now show \eqref{Infinite_volume_Theorem_1_Star_1_B}. To this end, let us define the following loop measures.
\begin{align}
\label{check_measure_mu_1_A}
\widecheck{\mu}^{L}(\dd \omega) &\coloneqq  \nu \sum_{T\in \nu {\N}^*} \frac{z^T}{T} \, \int_{\Lambda_L}\dd x\, \int \mathbb{W}^{L,T}_{x,x}(\dd \omega)\,\ee^{-\mathcal{V}^{\nu,\infty}(\omega)}= \nu \sum_{T\in \nu {\N}^*} \frac{z^T}{T} \,  \mathbb{W}^{L,T}(\dd \omega)
\,\ee^{-\mathcal{V}^{\nu,\infty}(\omega)}\,.
\\
\label{check_measure_mu_1_C}
\mu^{L,0}(\dd \omega) &\coloneqq  \nu \sum_{T\in \nu {\N}^*} \frac{z^T}{T} \, \int_{\Lambda_L}\dd x\, \int \mathbb{W}^{L,T}_{x,x}(\dd \omega)=\nu \sum_{T\in \nu {\N}^*} \frac{z^T}{T} \,  \mathbb{W}^{L,T}(\dd \omega)\,.
\end{align}
Note that \eqref{check_measure_mu_1_A} differs from \eqref{measure_mu_1} in the self-interaction term $\ee^{-\mathcal{V}^{\nu,\infty}(\omega)}$ (it is also equal to \eqref{measure_mu_1_star} above). Furthermore, \eqref{check_measure_mu_1_C} corresponds to \eqref{check_measure_mu_1_A} with $v=0$.
By \eqref{check_measure_mu_1_A}--\eqref{check_measure_mu_1_C} and the nonnegativity of $v$, we obtain 
\begin{equation}
\label{check_measure_mu_1_D}
\widecheck{\mu}^{L}(\dd \omega) \leq \mu^{L,0}(\dd \omega)\,.
\end{equation}
Similarly to \eqref{R_cluster_expansion_1}, the arguments from the proof of Proposition \ref{gamma_p_cluster_expansion} (i) allow us to deduce that 
\begin{equation}
\label{check_cluster_expansion_1}
\frac{\log \widecheck{\Xi}^{\,\nu,\zeta,L}}{|\Lambda_L|}=\sum_{n \geq 1} \frac{1}{|\Lambda_L|} \,\int \widecheck{\mu}^L(\dd \omega_1)\,\cdots\,\widecheck{\mu}^L(\dd \omega_n)\,\varphi^{\infty}(\omega_1,\ldots,\omega_n)\,.
\end{equation}

Analogously to the reduction from \eqref{R_cluster_expansion_3} to \eqref{R_cluster_expansion_5}, we use the dominated convergence theorem, the cluster expansions  \eqref{R_cluster_expansion_1} and \eqref{check_cluster_expansion_1}, as well as a telescoping argument to see that \eqref{Infinite_volume_Theorem_1_Star_1} follows if we show that for all $n \in \N^*$ and for all trees $\mathcal{T}$ on $[n]=\{1,\ldots,n\}$, we have
\begin{equation}
\label{check_cluster_expansion_5}
\left|\frac{1}{|\Lambda_L|}\,\int \left(\widecheck{\mu}^L(\dd \omega_1)-\overline{\mu}^{L}(\dd \omega_1)\right)\,\widecheck{\mu}^{L}_{\sharp}(\dd \omega_2)\,\cdots\,\widecheck{\mu}^{L}_{\sharp}(\dd \omega_n)\,
\prod_{\{i,j\} \in \mathcal{T}} \xi^{\infty}(\omega_i,\omega_j)\right| \lesssim_{d,\zeta,v,\mathcal{T}}o_L(1)
\,,
\end{equation}
uniformly in $L$. In \eqref{check_cluster_expansion_5}, $\widecheck{\mu}^{L}_{\sharp}(\dd \omega_j)$ denotes either $\widecheck{\mu}^{L}(\dd \omega_j)$ or $\overline{\mu}^{L}(\dd \omega_j)$ given by \eqref{check_measure_mu_1_A} or \eqref{measure_mu_1_A} respectively.
As in the reduction from \eqref{R_cluster_expansion_5} to \eqref{R_cluster_expansion_6}, we use \eqref{measure_mu_1_D}, \eqref{check_measure_mu_1_D} and the integration algorithm from Lemma \ref{integration_algorithm}, where the measure $\mu^L(\dd \omega)$ in \eqref{measure_mu_1} is replaced by $\overline{\mu}^{L,0}(\dd \omega)$ in \eqref{measure_mu_1_C} or $\mu^{L,0}(\dd \omega)$ in \eqref{check_measure_mu_1_C} accordingly. We then see \eqref{check_cluster_expansion_5} follows if we show that for all $q \in \N$, we have 
\begin{equation}
\label{check_cluster_expansion_6}
\frac{1}{|\Lambda_L|}\,\int \left|\widecheck{\mu}^L(\dd \omega_1)-\overline{\mu}^{L}(\dd \omega_1)\right|\,T(\omega_1)^q \lesssim_{\zeta,q,d} o_L(1)\,,
\end{equation}
uniformly in $L$. We now show \eqref{check_cluster_expansion_6}.

Using \eqref{measure_mu_1_A} and \eqref{check_measure_mu_1_A}, the left-hand side of \eqref{check_cluster_expansion_6} can be rewritten as 
\begin{equation*}
\frac{\nu}{|\Lambda_L|} \sum_{T \in \nu \N^*} \frac{z^T}{T}\,\int_{\Lambda_L}\dd x\, 
\left|\int \mathbb{W}^{L,T}_{x,x}(\dd \omega)\,\ee^{-\mathcal{V}^{\nu,\infty}(\omega)}-
\int \mathbb{W}^{\infty,T}_{x,x}(\dd \omega)\,\ee^{-\mathcal{V}^{\nu,\infty}(\omega)}\,\,\mathbf{1}_{\mathscr{I}_L}(\omega)\right|\,T^q\,,
\end{equation*}
which by Lemma \ref{Wiener_measure_pi_L} is 
\begin{multline}
\label{check_cluster_expansion_7}
\leq 
\frac{\nu}{|\Lambda_L|} \sum_{T \in \nu \N^*} \frac{z^T}{T}\,\int_{\Lambda_L}\dd x\, 
\left|\int \mathbb{W}^{\infty,T}_{x,x}(\dd \omega)\,\ee^{-\mathcal{V}^{\nu,\infty}(\pi_L \circ \omega)}-
\int \mathbb{W}^{\infty,T}_{x,x}(\dd \omega)\,\ee^{-\mathcal{V}^{\nu,\infty}(\omega)}\,\,\mathbf{1}_{\mathscr{I}_L}(\omega)\right|\,T^q
\\
+ \frac{\nu}{|\Lambda_L|} \sum_{T \in \nu \N^*} \frac{z^T}{T}\,\int_{\Lambda_L}\dd x\, \sum_{k \in \Z^d \setminus \{0\}}
\int \mathbb{W}^{\infty,T}_{x+Lk,x}(\dd \omega)\,\ee^{-\mathcal{V}^{\nu,\infty}(\pi_L \circ \omega)}\,T^q \eqqcolon I+II\,.
\end{multline}
We now estimate $I \equiv I^{\nu,\zeta,L,q}$ and $II  \equiv II^{\nu,\zeta,L,q}$ from \eqref{check_cluster_expansion_7} separately.

\paragraph{\emph{\textbf{Estimate of $I$ from \eqref{check_cluster_expansion_7}}}}
Let us first estimate $I$ from \eqref{check_cluster_expansion_7}. By \eqref{I^{L}} and \eqref{pi_L}, we have that 
\begin{equation}
\label{pi_L_on_I^L}
\pi_L \circ \omega=\omega \quad \forall \omega \in \mathscr{I}_L\,.
\end{equation}
By \eqref{pi_L_on_I^L} and the nonnegativity of $v$, it follows that 
\begin{equation}
\label{I_check_cluster_expansion_7_1}
I \leq \frac{\nu}{|\Lambda_L|}\,\sum_{T \in \nu \N^*} \frac{z^T}{T}\,\int_{\Lambda_L}\dd x\,\int \mathbb{W}^{\infty,T}_{x,x}(\dd \omega)\,\left(1-\mathbf{1}_{\mathscr{I}_L}(\omega)\right)\,T^q\,.
\end{equation}
Given $\delta>0$ small, we use the splitting \eqref{x_1_average} in \eqref{I_check_cluster_expansion_7_1} and obtain
\begin{multline}
\label{I_check_cluster_expansion_7_2}
I \leq \frac{\nu}{|\Lambda_L|}\,\sum_{T \in \nu \N^*} \frac{z^T}{T}\,\int_{(1-\delta) \Lambda_L}\dd x\,\int \mathbb{W}^{\infty,T}_{x,x}(\dd \omega)\,\left(1-\mathbf{1}_{\mathscr{I}_L}(\omega)\right)\,T^q
\\
+\frac{\nu}{|\Lambda_L|}\,\sum_{T \in \nu \N^*} \frac{z^T}{T}\,\int_{\Lambda_L \setminus (1-\delta) \Lambda_L}\dd x\,\int \mathbb{W}^{\infty,T}_{x,x}(\dd \omega)\,T^q \eqqcolon I_1 + I_2\,.
\end{multline}
Recalling \eqref{I^{L}} and \eqref{D^{L,tilde_L}_c}, we have
\begin{equation}
\label{I_check_cluster_expansion_7_2_1}
I_1 \leq \frac{\nu}{|\Lambda_L|}\,\sum_{T \in \nu \N^*} \frac{z^T}{T}\,\int_{(1-\delta) \Lambda_L}\dd x\,\int \mathbb{W}^{\infty,T}_{x,x}(\dd \omega)\,\mathbf{1}_{\mathscr{D}^{\infty,L}_{\delta/2}}(\omega)\,T^q \,.
\end{equation}
We use translation invariance to write 
\begin{multline*}
\eqref{I_check_cluster_expansion_7_2_1}= \frac{|(1-\delta)\Lambda_L|}{|\Lambda_L|}\,\nu \sum_{T \in \nu \N^*} \frac{z^T}{T}\,\int \mathbb{W}^{\infty,T}_{0,0}(\dd \omega)\,\mathbf{1}_{\mathscr{D}^{\infty,L}_{\delta/2}}(\omega)\,T^q
\\
\leq \nu \sum_{T \in \nu \N^*} z^T\,T^{q-1}\,\int \mathbb{W}^{\infty,T}_{0,0}(\dd \omega)\,\mathbf{1}_{\mathscr{D}^{\infty,L}_{\delta/2}}(\omega)\,,
\end{multline*}
which, by \eqref{heat_kernel_infty}, \eqref{W^T_infty}, \eqref{D^{L,tilde_L}_c}, and the dominated convergence theorem is
\begin{equation}
\label{I_check_cluster_expansion_7_2_B}
 \lesssim_{\zeta,q,d} o_{L}^{(\delta)}(1)\,.
\end{equation}
Here, we also used \eqref{5.27_remark_bound} and \eqref{heat_kernel_infty} to see that
$\nu \sum_{T \in \nu \N^*} z^T\,T^{q-1}\,\psi^{\infty,T}(0)$ is bounded uniformly in $\nu \in (0,\nu_0]$, and we recalled the notation from \eqref{free_energy_infinite_volume_6}.
By translation invariance, we have
\begin{equation}
\label{I_check_cluster_expansion_7_2_C}
I_2 \leq \frac{|\Lambda_L \setminus (1-\delta) \Lambda_L|}{|\Lambda_L|}\,\nu \sum_{T \in \nu \N^*} \frac{z^T}{T}\,\int \mathbb{W}^{\infty,T}_{0,0}(\dd \omega)\,T^q \lesssim_{\zeta,q,d} \delta\,,
\end{equation}
by arguing similarly as for \eqref{I_check_cluster_expansion_7_2_B}.
We now substitute \eqref{I_check_cluster_expansion_7_2_B}--\eqref{I_check_cluster_expansion_7_2_C} into \eqref{I_check_cluster_expansion_7_2} to deduce that 
\begin{equation}
\label{check_cluster_expansion_7_I}
\lim_{L \rightarrow \infty} I^{\nu,\zeta,L,q}=0
\end{equation}
uniformly in $\nu \in (0,\nu_0]$.
\paragraph{\emph{\textbf{Estimate of $II$ from \eqref{check_cluster_expansion_7}}}}
Let us now estimate $II$ from \eqref{check_cluster_expansion_7}. Using the nonnegativity of $v$, \eqref{heat_kernel_integral_infty}, and arguing similarly to \eqref{Lemma_5.20*_5_2}, we deduce that
\begin{equation}
\label{II_check_cluster_expansion_7}
II \leq \nu \sum_{T \in \nu \N^*} z^T\,T^{q-1}\,\sum_{k \in \Z^d \setminus \{0\}} \psi^{\infty,T} (Lk) \lesssim \nu \sum_{T \in \nu \N^*} z^T \,T^{q-1}\, \frac{T}{L^2}=\mathcal{O}_{\zeta,q,d} \left(\frac{1}{L^2}\right).
\end{equation}
In order to deduce the final bound in \eqref{II_check_cluster_expansion_7}, we used \eqref{5.27_remark_bound}.
We conclude that 
\begin{equation}
\label{check_cluster_expansion_7_II}
\lim_{L \rightarrow \infty} II^{\nu,\zeta,L,q}=0
\end{equation}
uniformly in $\nu \in (0,\nu_0]$. Substituting \eqref{check_cluster_expansion_7_I} and \eqref{check_cluster_expansion_7_II} into \eqref{check_cluster_expansion_7}, we deduce \eqref{check_cluster_expansion_6} and claim (ii) follows.
\end{proof}


\begin{thebibliography}{References}
\bibitem{AFP24} Z. Ammari, S. Farhat, S. Petrat, \emph{Expansion of the Many-body Quantum Gibbs State of the Bose-Hubbard Model on a Finite Graph}, Doc. Math. \textbf{30} (2025), no. 2, 475--496.
\bibitem{AR21} Z. Ammari,  A. Ratsimanetrimanana, \emph{High temperature convergence of the KMS boundary conditions: The Bose–Hubbard model on a finite graph}. Commun. Contemp. Math. {\bf 23}
(2021), no. 5, article no. 2050035, 18 pp.
\bibitem{BS02} A.N.~{Borodin},  P~{Salminen}, \emph{Handbook of Brownian Motion - Facts and Formulae}, Birkhäuser Basel (2002).
\bibitem{Bourgain_1994} J. Bourgain, \emph{Periodic nonlinear Schr\"{o}dinger equation and invariant measures}, Comm. Math. Phys. {\bf 166} (1994), no. 1, 1--26.
\bibitem{Bourgain_1996} J. Bourgain, \emph{Invariant measures for the $2D$-defocusing nonlinear Schr\"{o}dinger equation}, Comm. Math. Phys. \textbf{176} (1996), no. 8, 421--445.
\bibitem{Bourgain_1997} J. Bourgain, \emph{Invariant measures for the Gross-Pitaevskii equation}, J. Math. Pures Appl. (9) \textbf{76} (1997), no. 8, 649--702.
\bibitem{CKRT} C. Caraci, A. Knowles, A. Ranallo, P. Torres-Giesteira, \emph{The Euclidean $\Phi^4_2$ theory as a limit of an inhomogeneous Bose gas}, Preprint arXiv:2603.12241 (2026).
\bibitem{Dinh_Rougerie} V. D. Dinh, N. Rougerie, \emph{From bosonic canonical ensembles to non-linear Gibbs measures}, Preprint arXiv: 2412.13597 (2024). 
\bibitem{FKSS17} J. Fr\"{o}hlich, A. Knowles, B. Schlein, V. Sohinger, \emph{Gibbs measures of nonlinear Schr\"odinger equations as limits of many-body quantum states in dimensions} $d \leq 3$, Comm. Math. Phys., {\bf 356} (2017), no. 3, 883--980.
\bibitem{FKSS18} J. Fr\"{o}hlich, A. Knowles, B. Schlein, V. Sohinger, \emph{A microscopic derivation of time-dependent correlation functions of the 1D cubic nonlinear Schr\"{o}dinger equation}, Advances in Mathematics (2019), no. 353, 67--115.
\bibitem{FKSS_2020_1} J. Fr\"{o}hlich, A. Knowles, B. Schlein, and V. Sohinger, \emph{The mean-field limit of quantum Bose gases at positive temperature}, J. Amer. Math. Soc., \textbf{35} (2022), no. 4, 955\textendash1030.
\bibitem{FKSS_2020_2} J. Fr\"{o}hlich, A. Knowles, B. Schlein, and V. Sohinger,  \emph{A path-integral analysis of interacting Bose gases and loop gases}, J. Stat. Phys. \textbf{180} (2020), no. 1\textendash6, 810\textendash831.
\bibitem{FKSS_2022} J. Fr\"{o}hlich, A. Knowles, B. Schlein, and V. Sohinger, \emph{The Euclidean $\varphi^4_2$ theory as a limit of an interacting Bose gas}, J. Eur. Math. Soc. \textbf{27} (2025), no. 11, 4399--4468.
\bibitem{FKSS_2023} J. Fr\"{o}hlich, A. Knowles, B. Schlein, V. Sohinger, \emph{Interacting loop ensembles and Bose gases}, Ann. Henri Poincar\'{e} \textbf{24} (2023), no. 5, 1439\textendash1503.
\bibitem{Ginibre1} J.~{Ginibre}, \emph{Reduced density matrices of quantum gases. I. Limit of infinite volume}, J. Math. Phys. \textbf{6} (1965), 238\textendash251.
\bibitem{Ginibre2} J.~{Ginibre}, \emph{Reduced density matrices of quantum gases. II. Cluster property}, J. Math. Phys. \textbf{6} (1965), 251\textendash261.
\bibitem{Ginibre3} J.~{Ginibre}, \emph{Reduced density matrices of quantum gases. III. Hardcore potentials}, J. Math. Phys. \textbf{6} (1965), 1432-1446. 
\bibitem{Ginibre} J. Ginibre, \emph{Some applications of functional integration in statistical mechanics}, M\'{e}canique statistique et th\'{e}orie quantique des champs, Les Houches (1971), 327\textendash427.
\bibitem{Jougla_Rougerie} L. Jougla, N. Rougerie, \emph{$\Phi^4_2$ limit of a many-body bosonic free energy}, Preprint arXiv: 2512.10704 (2024). 
\bibitem{Knowles} A. Knowles, \emph{Limiting dynamics in large quantum systems}, ETH Z\"urich Doctoral Thesis, 2009, ETHZ e-collection 18517.
\bibitem{LNR15} M. Lewin, P.-T. Nam, N. Rougerie, \emph{Derivation of nonlinear Gibbs measures from many-body quantum mechanics}, Journal de l'\'{E}cole Polytechnique-Math\'{e}matiques \textbf{2} (2015), 65--115.
\bibitem{LNR18} M. Lewin, P.-T. Nam, N. Rougerie, \emph{Gibbs measures based on 1D (an)harmonic oscillators as mean-field limits}, J. Math. Phys. \textbf{59} (2018), no. 4, 041901.
\bibitem{LNR18b} M. Lewin, P.-T. Nam, N. Rougerie, \emph{Classical field theory limit of 2D many-body quantum
Gibbs states}, Preprint arXiv: 1810.08370v1 (2018).
\bibitem{LNR19} M. Lewin, P.-T. Nam, N. Rougerie, \emph{Derivation of renormalized Gibbs measures from equilibrium many-body quantum Bose gases}, J. Math. Phys, \textbf{60} (2019), no.6, 061901, 11 pp.
\bibitem{LNR21}  M. Lewin, P.-T. Nam, N. Rougerie, \emph{Classical field theory limit of many-body quantum {G}ibbs states in 2D and 3D}, Invent. Math. \textbf{224} (2021), no. 2, 315--444.
\bibitem{Nam_Zhu_Zhu} P.-T. Nam, R. Zhu, X. Zhu, \emph{$\Phi^4_3$ theory from many-body quantum Gibbs states}, Preprint arXiv: 2502.04884 (2025). 
\bibitem{Nam_Zhu_Zhu_2} P.-T. Nam, R. Zhu, X. Zhu, \emph{Derivation of Gibbs measure from Gibbs state with the fractional Bessel interaction in Two Dimensions}, Preprint arXiv: 2604.21583 (2026).
\bibitem{RS22} A. Rout, V. Sohinger, \emph{A microscopic derivation of Gibbs measures for the 1D focusing cubic nonlinear Schr\"{o}dinger equation}, Commun. Partial Differ. Equ. \textbf{48} (2023), no. 7--8, 1008--1055.
\bibitem{RS23} A. Rout, V. Sohinger, \emph{A microscopic derivation of Gibbs measures for the 1D focusing quintic nonlinear Schr\"{o}dinger equation}, SIAM J. Math. Anal. \textbf{7} (2025), no. 5, 4680--4755.
\bibitem{Salmhofer} M. Salmhofer, \emph{Functional Integral and Stochastic Representations for Ensembles of Identical Bosons on a Lattice}, Comm. Math. Phys. \textbf{385} (2021), no. 2, 1163--1211.
\bibitem{Sohinger_2019} V. Sohinger, \emph{A microscopic derivation of Gibbs measures for nonlinear Schr\"odinger equations with unbounded interaction potentials}, Int. Math. Res. Not. IMRN (2022), no. 19, 14964--15063.
\bibitem{Stanley_EC2} R.~{Stanley}, \emph{Enumerative Combinatorics, Volume 2}, Cambridge Studies in Advanced Mathematics \textbf{62} (1999).
\bibitem{Ueltschi_2004} D.~{Ueltschi}, \emph{Cluster expansions and correlation functions}, Moscow Math J. \textbf{4} (2004), no.~2, 511--522.
\end{thebibliography}
\end{document}